\RequirePackage[l2tabu, orthodox]{nag}

\documentclass[12pt,phd,a4paper,oneside]{ucl_thesis}

\usepackage{blindtext}

\usepackage{emptypage}

\usepackage{graphicx}

\usepackage{float}

\usepackage{amsmath}

\usepackage{gensymb}
\usepackage{textcomp}

\usepackage{setspace}
\usepackage{multirow}

\usepackage{bibentry}

\usepackage[format=hang,font=small,labelfont=bf]{caption}

\usepackage{etoolbox}

\usepackage{mathrsfs}
\usepackage{amsfonts,amssymb,amsmath,amsthm}
\usepackage{graphicx}
\usepackage[usenames,dvipsnames]{color}
\usepackage{braket}
\usepackage{tabularx,colortbl}
\usepackage{mathtools}
\usepackage{algorithm,algorithmic}
\newcommand{\comment}[1]{}
\newcommand{\idty}[1]{\mathbb{1}}
\newcommand{\ovsqrt}[1]{\frac{1}{\sqrt{2}}}

\newcommand{\tr}[1]{\mathrm{Tr}}
\newcommand{\Ord}[1]{O\left(#1\right)}
\newcommand{\E}{\mathcal{E}}

\newcommand{\abs}[1]{\left\lvert#1\right\rvert}
\newcommand{\norm}[1]{\left\lVert#1\right\rVert}
\newcommand{\id}[1]{\mathbb{I}}

\newcommand{\N}{\mathbb{N}}
\newcommand{\CC}{\mathbb{C}}

\newcommand{\Nystrom}{Nystr\"om}

\newcommand{\ketbra}[2]{\ket{#1}\!\bra{#2}}
\newcommand{\pnorm}[2]{\left\|#2\right\|_#1}
\newcommand{\maxnorm}[1]{\left\|#1\right\|_{\mathrm{max}}}

\newcommand{\gradtheta}[1]{\frac{\partial #1}{\partial \theta}}
\newcommand{\pr}{{\rm Pr}}
\newcommand{\partr}[2]{\mathrm{Tr}_{#1} \left[#2\right]}
\newcommand{\trh}[1]{\mathrm{Tr}_h \left[#1\right]}

\newcommand{\polylog}{\mathrm{polylog}}
\newcommand{\poly}{\mathrm{poly}}

\renewcommand{\tr}[1]{\mathrm{Tr}\left( #1 \right)}
\newcommand{\Ecal}{\mathcal{E}}
\newcommand{\Prob}{\mathrm{Pr}}
\newcommand{\Hyp}{\mathcal{H}}

\usepackage{xparse}
\NewDocumentCommand\Exp{g}{%
    \ensuremath{\mathbb{E}\IfNoValueTF{#1}{}{\left[#1\right]} }%
}

\newtheorem{theorem}{Theorem}
\newtheorem{prop}{Proposition}
\newtheorem{corollary}[theorem]{Corollary}
\newtheorem{definition}{Definition}
\newtheorem{lemma}{Lemma}
\newtheorem{claim}{Claim}

\newtheorem{remark}{Remark}

\newtheorem{researchquestion}{Research Question}
\newtheorem{assumption}{Assumption}

\definecolor{Gray}{gray}{0.85}
\definecolor{LightCyan}{rgb}{0.88,1,1}

\usepackage{bibentry}
\makeatletter\let\saved@bibitem\@bibitem\makeatother
\usepackage[pdftex,hidelinks]{hyperref}
\makeatletter\let\@bibitem\saved@bibitem\makeatother
\makeatletter
\AtBeginDocument{
    \hypersetup{
        pdfsubject={Thesis Subject},
        pdfkeywords={Thesis Keywords},
        pdfauthor={Author},
        pdftitle={Title},
    }
}
\makeatother

\begin{document}

\nobibliography*



\title{Quantum Machine Learning For Classical Data}
\author{Leonard P. Wossnig}
\department{Department of Computer Science}

\maketitle
\makedeclaration

\begin{abstract} 
In this dissertation, we study the intersection of quantum computing and supervised machine learning algorithms,
which means that we investigate quantum algorithms for supervised machine learning that operate on classical data.
This area of research falls under the umbrella of \textit{quantum machine learning},
a research area of computer science which has recently received wide attention.
In particular, we investigate to what extent quantum computers can be used to
accelerate supervised machine learning algorithms.
The aim of this is to develop a clear understanding of the promises and limitations of the current state-of-the-art
of quantum algorithms for supervised machine learning, but also to define directions for future research in this exciting field.
We start by looking at supervised quantum machine learning (QML) algorithms through the lens of statistical learning theory.
In this framework, we derive novel bounds on the computational complexities of
a large set of supervised QML algorithms under the requirement of optimal learning rates.
Next, we give a new bound for Hamiltonian simulation of dense Hamiltonians,
a major subroutine of most known supervised QML algorithms,
and then derive a classical algorithm with nearly the same complexity.
We then draw the parallels to recent `quantum-inspired' results, and will
explain the implications of these results for quantum machine learning applications.
Looking for areas which might bear larger advantages for QML algorithms,
we finally propose a novel algorithm for Quantum Boltzmann machines,
and argue that quantum algorithms for quantum data are one of the most promising applications for QML
with potentially exponential advantage over classical approaches.
\end{abstract}

\begin{acknowledgements}
I want to thank foremost my supervisor and friend Simone Severini, who has
always given me the freedom to pursue any direction I found interesting and promising,
and has served me as a guide through most of my PhD.

Next, I want to thank Aram Harrow, my secondary advisor,
who has always been readily available to answer my questions and discuss
a variety of research topics with me.

I also want to thank Carlo Ciliberto, Nathan Wiebe, and Patrick Rebentrost,
who have worked closely with me and have also taught me most of the mathematical
tricks and methods upon which my thesis is built.

I furthermore want to thank all my collaborators throughout the years.
These are in particular Chunhao Wang, Andrea Rocchetto, Marcello Benedetti,
Alessandro Rudi, Raban Iten, Mark Herbster, Massimiliano Pontil,
Maria Schuld, Zhikuan Zhao, Anupam Prakash, Shuxiang Cao,
Hongxiang Chen, Shashanka Ubaru, Haim Avron, and Ivan Rungger.
Almost in an equal contribution, I also want to thank
Fernando Brand\~ao, Youssef Mroueh, Guang Hao Low, Robin Kothari, Yuan Su, Tongyang Li,
Ewin Tang, Kanav Setia, Matthias Troyer, and Damian Steiger for many helpful discussions,
feedback, and enlightening explanations.

I am particularly grateful to Edward Grant, Miriam Cha, and Ian Horobin, who
made it possible for me to write this thesis.

I want to acknowledge UCL for giving me the opportunity to pursue this PhD thesis,
and acknowledge the kind support of align Royal Society Research grant
and the Google PhD Fellowship, which gave me the freedom to work on these
interesting topics.

Portions of the work that are included in this thesis were completed while
I was visiting the Institut Henri Poincar\'e of the Sorbonne University in Paris.

I particularly want to thank Riam Kim-McLeod for the support and help with the editing of the thesis.

I finally want to thank my family for the continued love and support.
\end{acknowledgements}

\begin{impactstatement}
Quantum machine learning bears promises for many areas, ranging from the
healthcare to the financial industry. In today's world, where data is available
in abundance, only novel algorithms and approaches are enabling us to make
reliable predictions that can enhance our life, productivity, or wealth.
While Moore's law is coming to an end, novel computational paradigms are sought after
to enable a further growth of processing power.
Quantum computing has become one of the prominent candidates, and is
maturing rapidly.
The here presented PhD thesis develops and studies this novel computational
paradigm in light of existing classical solutions and thereby develops
a path towards quantum algorithms that can outperform classical approaches.
\end{impactstatement}

\setcounter{tocdepth}{2}

\tableofcontents
\listoffigures

\chapter{Introduction and Overview}
\label{chap:introduction}

In the last twenty years, due to increased computational power and the availability of vast amounts of data, \textit{machine learning} (ML) has seen an immense success, with applications ranging from computer vision~\cite{krizhevsky2012imagenet} to playing complex games such as the Atari series~\cite{mnih2013playing} or the traditional game of Go~\cite{silver2016mastering}.
However, over the past few years, challenges have surfaced that threaten the end of this revolution.
The main two challenges are the increasingly overwhelming size of the available data sets, and the end of Moore's law~\cite{markov2014limits}.
While novel developments in hardware architectures, such as graphics processing units (GPUs) or tensor processing units (TPUs), enable orders of magnitude improved performance compared to central processing units (CPUs), they cannot significantly improve the performance any longer,
as they also reach their physical limitations.
They are therefore not offering a structural solution to the challenges posed
and new solutions are required.

On the other hand, a new technology is slowly reaching maturity.
Quantum computing, a form of computation that makes use of quantum-mechanical phenomena such as superposition and entanglement,
has been predicted to overcome these limitations on classical hardware.
Quantum algorithms, i.e., algorithms that can be executed on a quantum computer,
have been investigated since the 1980s, and have recently received increasing interest all around the world.

One area that has received particular attention is quantum machine learning
see e.g.~\cite{ciliberto2018quantum}, the combination of quantum mechanics and machine learning.
Quantum machine learning can generally be divided into four distinct areas of research:
\begin{itemize}
    \item Machine learning algorithms that are executed on classical computers (CPUs, GPUs, TPUs) and applied to classical data, the sector CC in Fig.~\ref{fig:4quadrants},
    \item Machine learning algorithms that are executed on quantum computers (QPUs) and applied to classical data, the sector QC in Fig.~\ref{fig:4quadrants},
    \item Machine learning algorithms that are executed on classical computers and applied to quantum data, the sector CQ in Fig.~\ref{fig:4quadrants}, and
    \item Machine learning algorithms that are executed on quantum computers (QPUs) and applied to quantum data, the sector QQ in Fig.~\ref{fig:4quadrants}.
\end{itemize}

\begin{figure}
    \centering
    \includegraphics[scale=0.4]{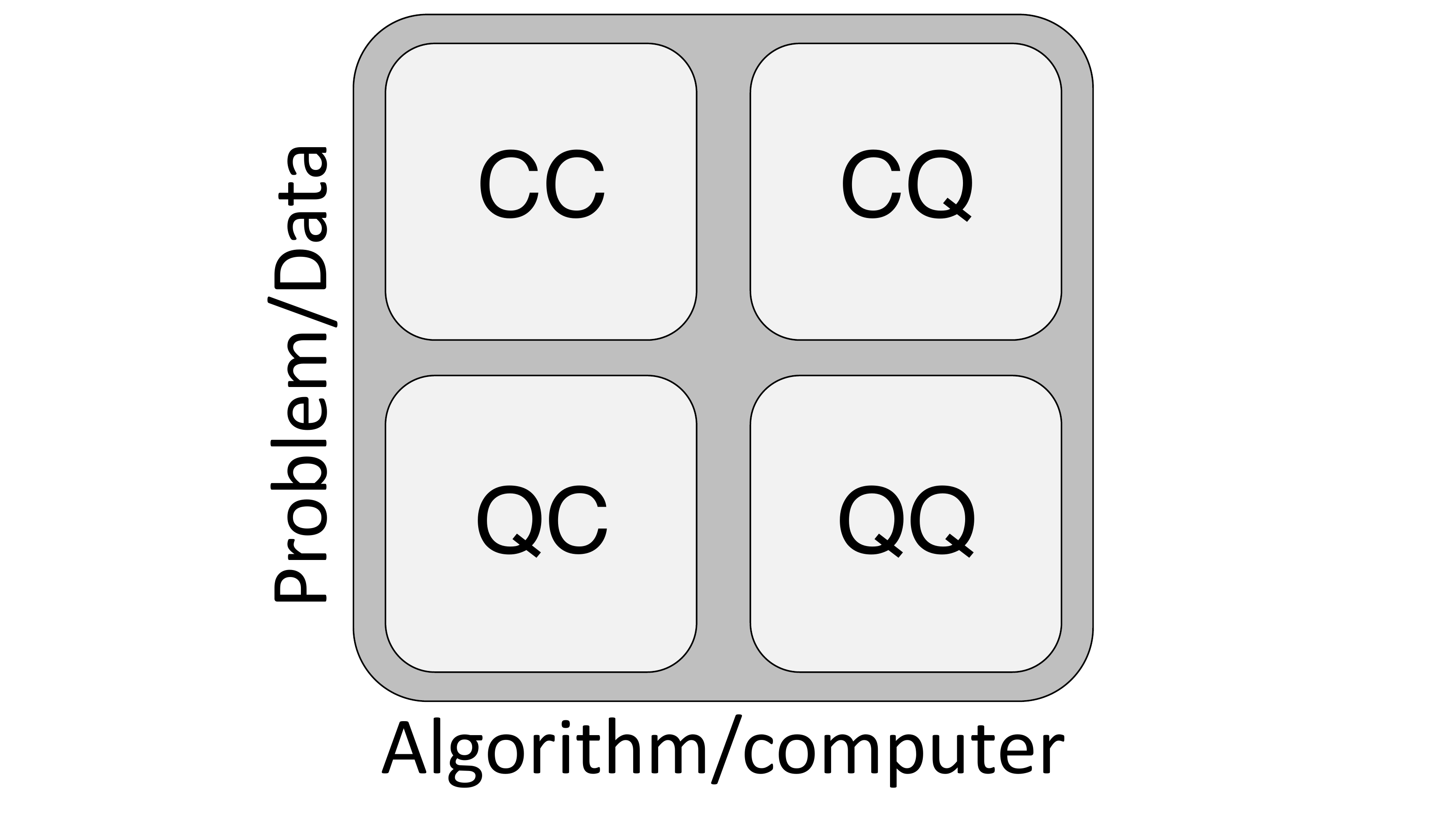}
    \caption{Different fields of study in Quantum Machine Learning. The different areas are
    related to the choice of algorithm, i.e., whether it is executed on a quantum or classical
    computer, and the choice of the target problem, i.e., whether it operates on quantum or classical data.}
    \label{fig:4quadrants}
\end{figure}

The biggest attention has been paid to quantum algorithms that perform
supervised machine learning on classical data.
The main objective in designing such quantum algorithms is to achieve
a computational advantage over any classical algorithm for the same task.
This area is the main focus of this thesis, and we will for brevity throughout this
thesis refer to this area when we speak about QML.

Although becoming an increasingly popular area of research, supervised QML has in
recent years faced two major challenges.

Firstly, initial research was aimed at designing supervised ML algorithms that have a superior performance, where
performance was measured solely in terms of the computational complexity (i.e., theoretical speed) of the algorithm with
respect to the input size (the problem dimension) and the desired error (or accuracy).
When we are speaking about the error in this context,
then we generally refer to the distance (for example in norm)
of the solution that we obtain through the algorithm to the true solution.
For example, for most algorithms we obtain guarantees that if we run it in time
$\epsilon^{-\alpha}$, then the solution will be $\epsilon$-close.
Assuming that the solution is a vector $x$, and our algorithm produces the vector
$\tilde{x}$, this would then imply that $\norm{\tilde{x}-x} \leq \epsilon$.
Note that different algorithms measure this error differently, but in general we
use the spectral or Frobenius norm as measure for the error throughout this thesis.

However, a challenge to this view is that it is today common wisdom in the classical ML community that faster
computation is not necessarily the solution to most practical problems in machine learning.
Indeed, more data and the ability to extrapolate to unseen data is the most essential
component in order to achieve good performance.
Mathematically, the behaviour of algorithms to extrapolate to unseen data and other
statistical properties have been formalised in the field of statistical learning theory.
The technical term of this ability is the so-called generalisation error of an algorithm,
and it is possible to obtain bounds on this error for many common machine learning algorithms.
The first research question of this thesis is therefore the following:

\begin{researchquestion}[QML under the lens of SLT]
\label{rq:slt1}
Under the common assumptions of statistical learning theory, what is
the performance of supervised quantum machine learning algorithms?
\end{researchquestion}

Most quantum algorithms are described in the so-called query model.
Here, the overall computational complexity, i.e., the number of steps an algorithm
requires to completion, is given in terms of a number of calls (uses) of a different
algorithm, which is called the oracle.
Indeed, most known quantum machine learning algorithms heavily rely on such oracles in order to achieve a `quantum advantage'.
A second challenge is therefore posed by the question how much of the
quantum advantage stems from the oracles and how much from the algorithms themselves.
Concretely, most supervised QML methods assume the existence of a fast ($\log(n)$ for input dimension $n$ complexity)
data preparation oracle or procedure, a quantum equivalent of a random access memory, the qRAM~\cite{giovannetti2008qram1,giovannetti2008qram2}.
A closer analysis of these algorithms indeed implies that much of their power stems from
this oracle, which indicates that classical algorithms can potentially achieve similar computational complexities
if they are given access to such a device.
The second research question of this thesis is therefore the following:

\begin{researchquestion}[QML under the lens of RandNLA]
\label{rq:randNLA1}
Under the assumption of efficient sampling processes for the data for both classical and quantum algorithms,
what is the comparative advantage of the latter?
\end{researchquestion}

The goal of this thesis is therefore to investigate the above two challenges to supervised QML, and
the resulting research questions.
We thereby aim to assess the advantages and disadvantages of such algorithms
from the perspective of statistical learning theory and randomised numerical linear algebra.

\section{Synopsis of the thesis}

The thesis is structured as follows.
We begin in Chapter~\ref{chap:prelim} with notation and mathematical preliminaries of the thesis.
In Chapter~\ref{chap:SLT}, we discuss the first challenge and associated Research Question~\ref{rq:slt1}, namely
the ability of quantum algorithms to generalise to data that is not present in the training set.
Limitations arise through the fundamental requirement of the quantum measurement, and the error-dependency of most supervised quantum algorithms.
We additionally discuss the condition number, an additional dependency
in many supervised QML algorithms.

In Chapter~\ref{chap:randNLA} we discuss the second challenge, i.e., Research Question~\ref{rq:randNLA1}
regarding input models. For this we discuss first
how to access data with a quantum computer and how these access models compare to
classical approaches, in particular to randomised classical algorithms.
This also includes a brief discussion about the feasibility of such memory models as well as
potential advantages and disadvantages.
We then propose a new algorithm for Hamiltonian simulation~\cite{lloyd1996universal},
a common subroutine of many supervised QML algorithms.
Next, we propose a randomised classical algorithm for the same purpose, and then
use this to discuss and compare the requirements of both algorithms,
and how they relate to each other.
Based on these insights, we next link our results to subsequent `dequantization' results.
This allows us to finally also make claims regarding the limitations of QML algorithms
that are based on such memory models (or more generally such oracles).

In Chapter~\ref{chap:quantumQML},
we argue that generative quantum machine learning algorithms that are
able to generate quantum data (QQ in Fig.~\ref{fig:4quadrants}),
are one of the most promising approaches for future QML research,
and we propose a novel algorithm for the training of
Quantum Boltzmann Machines (QBMs)~\cite{wiebe2014quantum} with visible and hidden nodes.
We note that quantum data is here defined as data that is generated by an arbitrarily-sized quantum circuit.
We derive error bounds and computational complexities for the training QBMs based
on two different training algorithms, which vary based on the model assumptions.

Finally, in Chapter~\ref{chap:conclusions}, we summarise our insights and discuss the
main results of the thesis, ending with a proposal for further research.

\section{Summary of our contributions}
Our responses to the above-posed research questions are summarised below.

For Research Question~\ref{rq:slt1}, we obtain the following results:
\begin{itemize}
    \item Taking into account optimal learning rates and the generalisation ability,
    we show that supervised quantum machine learning algorithms fail to achieve exponential speedups.
    \item Our analysis is based on the polynomial error-dependency and the repeated
    measurement that is required to obtain a classical output.
    We have the requirement that the (optimisation) error of the algorithm matches the
    scaling of the statistical (estimation) error that is inherent to all
    data-based algorithms and is of $\Ord{1/\sqrt{n}}$,
    for $n$ being the number of samples in the data set.
    We observe that the performance of the algorithms in question
    can in practice be worse than the fastest classical solution.
    \item Such concerns are important for the design of quantum algorithms
    but are not relevant for classical ones.
    The computational complexity of the latter scales typically poly-logarithmic
    in the error $\epsilon$ and therefore only introduces additional $O(\log(n))$
    terms.
    We however acknowledge this is not generally true and it is possible to
    trade-off speed versus accuracy, which is for example used in early stopping.
    Here, one chooses to obtain a lower accuracy in order to obtain an algorithm that converges
    faster to the solution with the best statistical error.
    \item One additional challenge for many quantum algorithms is the problem of preconditioning,
    which most state-of-the-art classical algorithms such as FALCON~\cite{rudi2017falkon}
    do take into account.
    \item While previous results~\cite{clader2013preconditioned} claimed that preconditioning is possible,
    in the quantum case, this turns out not to be true in general as, e.g.
    Harrow and La Placa~\cite{harrow2017limitations} showed.
    Our bounds on the condition number indicate that ill-conditioned QML algorithms
    are even more prone to be outperformed by classical counterparts.
\end{itemize}

For Research Question~\ref{rq:randNLA1}, we obtain the following results:
\begin{itemize}
    \item We first show that data access oracles, such as quantum random access memory~\cite{giovannetti2008qram1} can be
    used to construct faster quantum algorithms for dense Hamiltonian simulation, and obtain an algorithm
    with performance (runtime) depending on the square-root of the dimension $n$ of the input Hamiltonian,
    and linearly on the spectral norm, i.e., $\Ord{\sqrt{n}\norm{H}}$ for Hamiltonian $H \in \mathbb{C}^{n \times n}$.
    For a $s$-sparse Hamiltonian, this reduces to time $\Ord{\sqrt{s}\norm{H}}$,
    where $s$ is the maximum number of non-zero elements in the rows of $H$.
    \item We next show how we can derive a classical algorithm for the same task, which is based on
    a classical Monte-Carlo sampling method called \textit{Nystr\"om method}.
    We show that for a sparse Hamiltonian $H \in \mathbb{C}^{n \times n}$,
    there exists an algorithm that, with probability $1-\delta$ approximates any chosen amplitude of the
    state $e^{i H t} \psi$, for an efficiently describable state $\psi$ in time
    \[\Ord{sq + \frac{t^9\norm{H}^4_F \norm{H}^7}{\epsilon^4}\left(\log(n) + \log\frac{1}{\delta}\right)^2},\]
    where $\epsilon$ determines the quality of the approximation, $\norm{\cdot}_F$ is the Frobenius norm,
   and $q$ is the number of non-zero elements in $\psi$.
    \item While our algorithm is not generally efficient due to the dependency on the Frobenius norm,
    it still removes the explicit dependency on $n$ up to logarithmic factors, i.e., the system dimension.
    We therefore obtain an algorithm that only depends on the rank, sparsity,
    and the spectral norm of the problem, since $\norm{H}_F \leq \sqrt{r} \norm{H}$,
    for $r$ being the rank of $H$.
    For low-rank matrices, we therefore obtain a potentially much faster algorithm,
    and for the case $ q = s = r = \Ord{\mathrm{polylog}(n)}$ our algorithm becomes efficiently executable.
    Note that we here generally assume that $N$ grows exponentially in the number of qubits in
    the system, i.e., we only obtain an efficient algorithm if we do not have an explicit $n$ dependency.
    \item This result is interesting in two different ways:
    Firstly, it is the first known result that applies so-called sub-sampling methods for simulating (general) Hamiltonians on a classical computer.
    Secondly, and more important in context of this thesis, our result indicates that the
    advantage of so-called \textit{quantum machine learning} algorithms may not be as big as promised.
    QML algorithms such as quantum principal component analysis~\cite{lloyd2014quantum},
    or quantum support vector machines~\cite{rebentrost2014quantum} were claimed to be efficient for sparse or low rank
    input data, and our classical algorithm for Hamiltonian simulation is efficient under similar conditions. As a corollary, we indeed show that we can efficiently simulate $\exp(i\rho t)$ for
    density matrix $\rho$, if we can efficiently sample from the rows and columns of $\rho$
    according to a certain probability distribution.
    \item While we did not manage to extend our results to quantum machine learning algorithms,
    shortly after we posted our result to the ArXiv, Ewin Tang used similar methods to `dequantize'~\cite{tang2019quantum}
    the quantum recommendation systems algorithm by Kerenidis and Prakash~\cite{kerenidis2016quantum}.
    While our algorithm requires the input matrix to be efficiently row-computable,
    Tang designed a classical memory model that allows us to sample the rows or columns of the
    input matrix according to their norms.
    Under closer inspection, these requirements are fundamentally equivalent.
    Furthermore, the ability to sample efficiently from this distribution is indeed similar to the
    ability of a quantum random access memory to prepare quantum states of the rows and columns of an input matrix.
    Dequantization then refers to a classical algorithm that can perform the same task as
    a chosen quantum algorithm with an at most polynomial slowdown.
    In the subsequent few months, many other algorithms were published that achieved similar
    results for many other quantum machine learning algorithms including quantum PCA~\cite{tang2018quantum}, Quantum Linear Systems Algorithm~\cite{chia2018quantum,gilyen2018quantum}, and Quantum Semi-Definite Programming~\cite{chia2019quantum}. Most of these algorithms were unified recently in the framework of
    quantum singular value transformations~\cite{chia2020sampling}.
\end{itemize}

A conclusion from the above results and the answers to the research questions is that the hope for
exponential advantages of quantum algorithms for machine learning over their classical counterparts
might be misplaced.
While polynomial advantages may still be feasible, and most results currently indicate a gap,
understanding the limitations of classical and quantum algorithms is still an open research question.
As many of the algorithms we investigate in this thesis are of a theoretical nature, we want to mention
that the ultimate performance test for any algorithm is a benchmark, and only such will give the
final answers to questions of real advantage.
However, this will only be possible once sufficiently large and accurate quantum computers are available,
and will therefore not be possible in the near future.

As our and subsequent results indicate, quantum machine learning algorithms
for classical data to date appear to allow for at most polynomial speedups if any.
We therefore turn to an area where
classical machine learning algorithms might generally be inefficient:
modelling of quantum distributions, and therefore quantum algorithms for quantum data.
We propose a method for fully quantum generative training of quantum Boltzmann machines
that in contrast to prior art have both visible and hidden units.
We base our training on the quantum relative entropy objective function and find
efficient algorithms for training based on gradient estimations under the assumption that
Gibbs state preparation for the model Hamiltonian is efficient.
One interesting feature of these results is that such generative models can in
principle also be used for the state preparation, and might therefore be a useful
tool to overcome qRAM-related issues.

\section{Statement of authorship}

This thesis is based on the following research articles.
Notably, for many articles the order of authors is alphabetical, which is standard in theoretical computer science.
\begin{enumerate}
    \item \bibentry{dervovic2018quantum}
    \item \bibentry{ciliberto2018quantum}
    \item \bibentry{wang2018quantum}
    \item \bibentry{rudi2020approximating}
    \item \bibentry{ciliberto2020fast}
    \item \bibentry{wiebe2019generative}
\end{enumerate}

I have additionally co-authored the following articles that are not included in this thesis:
\begin{enumerate}
    \item \bibentry{chen2018universal}
    \item \bibentry{benedetti2019adversarial}
    \item \bibentry{grant2019initialization}
    \item \bibentry{cao2020cost}
    \item \bibentry{rungger2019dynamical}
    \item \bibentry{patterson2019quantum}
    \item \bibentry{tilly2020computation}
\end{enumerate}

\chapter{Notation And Mathematical Preliminaries}
\label{chap:prelim}

\section{Notation}

We denote vectors with lower-case letters.
For a vector $x\in \mathbb{C}^n$, let $x_i$ denotes the $i$-th element of $x$.
A vector is sparse if most of its entries are $0$.
For an integer $k$, let $[k]$ denotes the set $\{1,\dots, k\}$.

For a matrix $A\in \mathbb{C}^{m \times n}$ let $A^j := A_{:,j}$, $j \in [n]$
denote the $j$-th column vector of $A$,
$A_i := A_{i,:}$, $i \in [m]$ the $i$-th row vector of $A$,
and $A_{ij}:=A(i,j)$ the $(i,j)$-th element.
We denote by $A_{i:j}$ the sub-matrix of $A$ that contains the rows from $i$ to $j$.

The \textit{supremum} is denoted as $\mathrm{sup}$ and the
\textit{infimum} as $\mathrm{inf}$. For a measure space $(X,\Sigma,\mu)$,
and a measurable function $f$ an essential upper bound of $f$ is defined as
$U_f^{\text{ess}} := \{ l \in \mathbb R : \mu(f^{-1}(l, \infty)) = 0 \}$, if the
measurable set $f^{-1} (l, \infty)$ is a set of measure zero, i.e., if $f(x) \leq l$
for almost all $x \in X$.  Then the essential supremum is defined as $\text{ess sup} f := \mathrm{inf}\ U_f^{\text{ess}}$.
We let  the $\mathrm{span}$ of a set $S = \{v_i\}_1^k \subseteq \mathbb C^n$ be defined by
$\mathrm{span}\{S\} := \left\lbrace x \in \mathbb C^n \, | \, \exists \, \{\alpha_i \}_1^k \subseteq \mathbb C\text{ with } x =\sum_{i=1}^k \alpha_i v_i \right\rbrace$. The set is linearly independent if $\sum_i \alpha_i v_i = 0$ if and only if $\alpha_i=0$ for all $i$.
The range of $A \in \mathbb{C}^{m \times n}$ is defined by $\text{range}(A) = \{ y \in \mathbb R^m : y = A x \text{ for some } x \in \mathbb C^n \} = \mathrm{span}(A^1 ,\ldots, A^n)$. Equivalently the range of $A$ is the set of all linear combinations of the columns of $A$.
The nullspace $\mathrm{null}(A)$ (or kernel $\mathrm{ker}(A)$) is the set of vectors such that $A v = 0$. Given a set $S = \{ v_i \}_1^k \subseteq \mathbb C^n$.
The null space of $A$ is $\mathrm{null}(A) = \{ x \in \mathbb R ^b : A x = 0 \}$.

The \textit{rank} of a matrix $A \in \mathbb C^{m \times n}$, $\mathrm{rank}(A)$
is the dimension of $\mathrm{range}(A)$ and is equal to the number of linearly
independent columns of $A$;
Since this is equal to $\mathrm{rank}(A^H)$, $A^H$ being the complex conjugate transpose of $A$,
it also equals the number of linearly independent rows of $A$,
and satisfies $\mathrm{rank}(A) \leq \min \{m,n\}$.
The trace of a matrix is the sum of its diagonal elements $\tr{A} = \sum_i A_{ii}$.
The support of a vector $\mathrm{supp}(v)$ is the set of indices $i$ such that $v_i =0$
and we call it sparsity of the vector. For a matrix we denote the sparsity as the number of non zero entries,
while row or column sparsity refers to the number of non-zero entries per row or column.
A symmetric matrix $A$ is positive semidefinite (PSD) if all its eigenvalues are non-negative.
For a PSD matrix $A$ we write $A \succeq 0$.
Similarly $A \succeq B$ is the partial ordering which is equivalent to $A-B \succeq 0$.

We use the following standard norms. The Frobenius norm
$\norm{A}_F = \sqrt{\sum_{i=1}^m \sum_{j=1}^n A_{ij}A_{ij}^*}$, and the spectral norm $\norm{A} = \sup_{x \in \mathbb C^n, \; x \neq 0} \frac{|A x|}{|x|}$. Note that that $\norm{A}_F^2 = \tr{A^H A} = \tr{A A^H}$. Both norms are submultiplicative and unitarily invariant and they are related to each other as $\norm{A} \leq \norm{A}_F \leq \sqrt{n} \norm{A}$.

The singular value decomposition of $A$ is $A= U \Sigma V^H$ where $U,V$ are unitary matrices and $U^H$ defines the complex conjugate transpose,
also called Hermitian conjugate, of $U$.
We use throughout the thesis $x^H$ or $A^H$ for the Hermitian conjugate as well as for the transpose for real matrices and vectors.
We denote the pseudo-inverse of a matrix $A$ with singular value decomposition $U \Sigma V^H$ as $A^{+}:=V \Sigma^{+} U^H$.

\section{Matrix functional analysis}
While computing the gradient of the average log-likelihood is a straightforward task when training ordinary Boltzmann machines, finding the gradient of the quantum relative entropy is much harder.  The reason for this is that in general $[\partial_\theta H(\theta) , H(\theta)]\ne 0$.  This means that the ordinary rules that are commonly used in calculus for finding the derivative no longer hold.  One important example that we will use repeatedly is Duhamel's formula:
\begin{equation}
\partial_\theta e^{H(\theta)} = \int_{0}^1 \mathrm{d}s e^{H(\theta) s} \partial_\theta H(\theta) e^{H(\theta) (1-s)}.
\end{equation}
This formula can be easily proven by expanding the operator exponential in a Trotter-Suzuki expansion with $r$ time-slices, differentiating the result and then taking the limit as $r\rightarrow \infty$.  However, the relative complexity of this expression compared to what would be expected from the product rule serves as an important reminder that computing the gradient is not a trivial exercise.  A similar formula also exists for the logarithm as shown further below.

Similarly, because we are working with functions of matrices here we need to also work with a notion of monotonicity. We will see that for some of our approximations to hold we will also need to define a notion of concavity (in order to use Jensen's inequality).  These notions are defined below.
    \begin{definition}[Operator monoticity]
      A function $f$ is operator monotone with respect to the semidefinite order if $0 \preceq A \preceq B$, for two symmetric positive definite operators implies, $f(A) \preceq f(B)$.
      A function is operator concave w.r.t. the semidefinite order if $cf(A)+(1-c)f(B) \preceq f(cA+(1-c)B)$, for all positive definite $A,B$ and $c \in [0,1]$.
    \end{definition}

We now derive or review some preliminary equations that we will need in order to obtain a useful bound on the gradients in the main work.

\begin{claim}
  Let $A(\theta)$ be a linear operator which depends on the parameters $\theta$.
  Then
  \begin{equation}
    \label{eq:grad_inverse}
    \frac{\partial}{\partial \theta} A(\theta)^{-1} = - A^{-1} \frac{\partial \sigma}{\partial \theta} A^{-1}.
  \end{equation}
\end{claim}
\begin{proof}
  The proof follows straight forward by using the identity $I$.
  \begin{equation*}
    \gradtheta{I} = 0 = \frac{\partial}{\partial \theta} AA^{-1} = \left(\gradtheta{A}\right) A^{-1} + A \left(\gradtheta{A^{-1}}\right).
  \end{equation*}
  Reordering the terms completes the proof. This can equally be proven using the Gateau derivative.
\end{proof}
In the following we will furthermore rely on the following well-known inequality:
\begin{lemma}[Von Neumann Trace Inequality]
  Let $A \in \mathbb{C}^{n\times n}$ and $B\in \mathbb{C}^{n\times n}$ with singular values $\{\sigma_i(A)\}_{i=1}^n$ and $\{\sigma_i(B)\}_{i=1}^n$ respectively such that $\sigma_i(\cdot) \le \sigma_j(\cdot)$ if $i\le j$.
  It then holds that
  \begin{equation}
    \abs{\tr{AB}} \leq \sum_{i=1}^n \sigma(A)_i \sigma(B)_i.
  \end{equation}
\end{lemma}
Note that from this we immediately obtain
\begin{equation}
  \label{eq:trace_inequality}
  \abs{\tr{AB}} \leq \sum_{i=1}^n \sigma(A)_i \sigma(B)_i \leq \sigma_{max}(B) \sum_i \sigma(A)_i = \norm{B} \sum_i \sigma(A)_i.
\end{equation}
This is particularly useful if $A$ is Hermitian and PSD, since this implies $\abs{\tr{AB}} \leq \norm{B} \tr{A}$ for Hermitian $A$.

Since we are dealing with operators, the common chain rule of differentiation does not hold generally. Indeed the chain rule is a special case if the derivative of the operator commutes with the operator itself.
Since we are encountering a term of the form $\log \sigma(\theta)$, we cannot assume that $[\sigma, \sigma']=0$, where $\sigma':= \sigma^{(1)}$ is the derivative w.r.t., $\theta$.
For this case we need the following identity similarly to Duhamels formula in the derivation of the gradient for the purely-visible-units Boltzmann machine.
\begin{lemma}[Derivative of matrix logarithm \cite{haber2018notes}]
  \label{eq:operator_log_ineq}
  \begin{equation}
    \frac{d}{dt} \log{A(t)} = \int\limits_0^1 [sA + (1-s) I]^{-1} \frac{dA}{dt} [sA + (1-s)I]^{-1}.
  \end{equation}
\end{lemma}
For completeness we here include a proof of the above identity.
\begin{proof}
  We use the integral definition of the logarithm~\cite{higham2008functions} for a complex, invertible, $n\times n$ matrix $A=A(t)$ with no real negative
  \begin{equation}
    \log{A} = (A-I) \int_0^1 ds [s(A-I)+I]^{-1}.
  \end{equation}
  From this we obtain the derivative $$\frac{d}{dt} log{A} = \frac{dA}{dt} \int_0^1 ds [s(A-I)+I]^{-1} + (A-I) \int_0^1 ds \frac{d}{dt}[s(A-I)+I]^{-1}.$$
  Applying \eqref{eq:grad_inverse} to the second term on the right hand side yields
  $$\frac{d}{dt} log{A} = \frac{dA}{dt} \int_0^1 ds [s(A-I)+I]^{-1} + (A-I) \int_0^1 ds [s(A-I)+I]^{-1} s \frac{dA}{dt}[s(A-I)+I]^{-1},$$
  which can be rewritten as
  \begin{eqnarray}
    \frac{d}{dt} log{A} =  \int_0^1 ds [s(A-I)+I][s(A-I)+I]^{-1}\frac{dA}{dt} [s(A-I)+I]^{-1} \\+ (A-I) \int_0^1 ds [s(A-I)+I]^{-1} s \frac{dA}{dt}[s(A-I)+I]^{-1},
  \end{eqnarray}
  by adding the identity $I = [s(A-I)+I][s(A-I)+I]^{-1}$ in the first integral and reordering commuting terms (i.e., $s$).
  Notice that we can hence just substract the first two terms in the integral which yields \eqref{eq:operator_log_ineq} as desired.
\end{proof}

minimises\chapter{Statistical Learning Theory}
\label{chap:SLT}

This chapter applies insights from statistical learning theory to answer the following question:
\setcounter{researchquestion}{0}

\begin{researchquestion}[QML under the lens of SLT]
\label{rq:slt}
Under the common assumptions of statistical learning theory, what is
the performance of supervised quantum machine learning algorithms?
\end{researchquestion}

The main idea in this chapter is to leverage the framework of statistical learning theory
to understand how the minimum number of samples required by a learner to reach a
target generalisation accuracy
influences the overall performance of quantum algorithms.
By taking into account well known bounds on this accuracy,
we can show that quantum machine learning algorithms
for supervised machine learning are unable to achieve polylogarithmic runtimes in the input dimension.
Notably, the results presented here hold only for supervised quantum machine learning algorithms
for which statistical guarantees are available.
Our results show that without further assumptions on the problem,
known quantum machine learning algorithms for supervised learning achieve only moderate
polynomial speedups over efficient classical algorithms - if any.
We note that the quantum machine learning algorithms that we analyse here
are all based on fast quantum linear algebra subroutines~\cite{harrow2009quantum,CGJ19}.
These in particular include quantum quantum support vector machines~\cite{rebentrost2014quantum},
quantum linear regression, and quantum least squares~\cite{wiebe2012quantum,schuld2016prediction,wang2017quantum,kerenidis2017quantum,CGJ19,zhang2018nonlinear}.

Notably, the origin of the `slow down' of quantum algorithms under the above
consideration is twofold.

First, most of the known quantum machine learning
algorithms have at least an inversely linear scaling in the optimisation or
approximation error of the algorithm, i.e., they require a computational time
$O(\epsilon^{-\alpha})$ (for some real positive $\alpha$) to return a solution
$\tilde{x}$ which is $\epsilon$-close in some norm to the true solution $x$ of
the problem. For example, in case of the linear systems algorithm, we obtain a quantum
state $\ket{\tilde{x}}$ such that $\norm{\ket{\tilde{x}}-\ket{x}}\leq \epsilon$,
where $\ket{x}$ is the state that encodes the exact solution to the linear system $\ket{x}=\ket{A^{-1}b}$. This is in stark contrast to classical algorithms which
typically have logarithmic dependency with respect to the error, i.e., for
running the algorithm for time $O(\log(\epsilon))$, we obtain a solution
with error $\epsilon$.

Second, a crucial bottleneck in many quantum algorithms is the requirement to
sample in the end of most quantum algorithms. This implies generally another
error, since we need to repeatedly measure the resulting quantum state in order
to obtain the underlying classical result.

We note that we mainly leverage these two sources of error in the following
analysis, but the extension of this to further include noise in the computation
is straightforward. Indeed, noise in the computation (a critical issue
in the current generation of quantum computers) could immediately be taken into
account by simply adding a linear factor in terms of the error decomposition
that we will encounter in Eq.~\ref{eq:totalerrordec}.
This will be a further limiting factor for near term devices as such
errors need to be surpressed sufficiently in order to obtain good general bounds.

While previous research has already identified a number
of caveats such as the data access~\cite{aaronson2015read} or restrictive
structural properties of the input data, which limit the practicality of these algorithms,
our insights are entirely based on statistical analysis.

As we will discuss in more detail in chapter~\ref{chap:randNLA},
under the assumption that we can efficiently sample
rows and columns of the input matrix (i.e., the input data)
according to a certain distribution,
classical algorithms can be shown to be nearly as efficient as these
quantum machine learning algorithms~\cite{tang2018quantum}.
We note that the scaling of the quantum machine learning algorithms typically
still achieve a high polynomial advantage compared to the classical ones.

This chapter is organised as follows.
First, in Section~\ref{sec:key_res_slt}, we will review existing results from
statistical learning theory in order to allow the reader to follow the subsequent argument.
Section~\ref{ssec:bounds_algorithmic_error} takes into account the
error that is introduced by the algorithm itself.
In Section~\ref{ssec:bounds_sample_error}, we then use this insight to bound the
error that is induced through the sampling process in quantum mechanics.
An additional dependency is typically given by the condition number of the problem,
and we hence derive additional bounds for it in Section~\ref{ssec:bounds_cond_num}.
Finally, in Section~\ref{sec:qml_algo_analysis}, we accumulate these insights and
use them to analyse a range of existing supervised quantum machine learning algorithms.

\section{Review of key results in Learning Theory}
\label{sec:key_res_slt}

Statistical Learning Theory has the aim to statistically quantify the resources
that are required to solve a learning problem~\cite{shalev2014understanding}.
Although multiple types of learning settings exist, depending on the access to data
and the associated error of a prediction, here we primarily focus on supervised learning.
In supervised learning, the goal is to find a function that fits a set of input-output
training examples and, more importantly, guarantees also a good fit on data points
that were not used during the training process.
The ability to extrapolate to data points that are previously not observed is also known
as {\em generalisation ability} of the model.
This is indeed the major difference between machine learning and standard optimisation
processes.
Although this problem can be cast into the framework of optimisation (by optimising
a certain problem instance), setting up this instance to achieve a maximum generalisation
performance is indeed one of the main objectives of machine learning.

After the review of the classical part, i.e., the important points of consideration for any learning algorithm and the assumptions taken,
we also analyse existing quantum algorithms, and their computational speedups within the scope of statistical learning theory.

\subsection{Supervised Learning}

We now set the stage for the analysis by defining the framework of supervised learning.
Let $X$ and $Y$ be probability spaces with distribution $\rho$ on $X \times Y$ from which we sample data points $x \in X$ and the corresponding labels $y \in Y$.
We refer to $X$ and $Y$ as input set and output set respectively.
Let $\ell : Y \times Y \mapsto \mathbb{R}$ be the loss function
measuring the discrepancy between any two points in the input space, which is a point-wise error measure.
There exist a wide range of suitable loss-functions, and choosing the appropriate one in practice is
of great importance.
Typical error functions are the least-squares error, $\ell_{sq}(f(x),y):= (f(x)-y)^2$ over $Y=\mathbb R$ for regression (generally for dense $Y$),
or the $0-1$ loss $\ell_{0-1}(f(x),y):= \delta_{f(x),y}$ over $Y=\{-1,1\}$
for classification (generally for discrete $Y$).
The least squares loss will also use be used frequently throughout this chapter.
For any hypothesis space $\mathcal{H}$ of measurable functions $f: X \mapsto Y$, $f \in \mathcal{H}$,
the goal of supervised learning is then to minimise the \textit{expected risk} or \textit{expected error}
$\mathcal{E}(f):= \Exp_{\rho}\left[\ell(y,f(x))\right]$, i.e.,
\begin{equation}
\label{eq:learning_problem}
    \inf_{f \in \mathcal H} \mathcal{E}(f), \quad \mathcal E (f) = \int_{X \times Y} \ell(f(x), y)) d\rho(x,y).
\end{equation}

We hence want to minimise the expected prediction error for a hypothesis $f:X \rightarrow Y$,
which is the average error with respect to the probability distribution $\rho$.
If the loss function is measurable, then the target space is the space of all measurable functions.
The space of all functions for which the expected risk is well defined is called the \textit{target space},
and typically denoted by $\mathcal{F}$.

For many loss functions we can in practice not achieve the infimum, however, it is still possible to derive a minimizer.
In order to be able to efficiently find a solution to Eq.~\ref{eq:learning_problem}, rather than searching over
the entirety of $\mathcal{F}$, we restrict the search over a restricted hypothesis space $\Hyp$,
which indeed can be infinite.

In summary, a learning problem is defined by the following three components:
\begin{enumerate}
    \item A probability space $X \times Y $ with a Borel probability measure $\rho$.
    \item A measureable loss function $\ell: Y \times Y \mapsto [0,\infty)$.
    \item A hypothesis space $\Hyp$ from which we choose our hypothesis $f\in \Hyp$.
\end{enumerate}
The data or input spaces $X$ can be vector spaces (linear spaces) such as
$X=\mathbb{R}^d$, $d\in \mathbb{N}$, or \textit{structured spaces},
and the output space $Y$ can also take a variety of forms as we mentioned above.

In practice,  the underlying distribution $\rho$ is unknown and we can
only access it through a finite number of observations
This finite set of samples, $S_n = \{(x_i,y_i)\}_{i=1}^n$, $x_i \in X$, $y_i \in Y$,
is called the training set, and we generally assume that these are sampled
identically and independently according to $\rho$.
This is given in many practical cases but it should be noted that this is
not generally true. Indeed, for example for time series, subsequent samples
are typically highly correlated, and the following analysis hence does not
immediately hold.
The assumption of independence can however be relaxed to the assumption that the data does only depend \textit{slightly} on each other via so-called mixing conditions. Under such assumptions, most of the following results for the independent case still hold with only slight adaptations.

The results throughout this chapter rely on the following assumptions.
\begin{assumption}
    The probability distribution on the data space $X\times Y$ can be factorized into a marginal distribution $\rho_X$ on $X$ and a conditional distribution $\rho(\cdot|x)$ on $Y$.
\end{assumption}

We add as a remark the observation that the probability
distribution $\rho$ can take into account a large set of
uncertainty in the data.
The results therefore hold for a range of noise types or partial information.

\begin{assumption}
   The probability distribution $\rho$ is known only through a finite set of samples $S_n=\{x_i,y_i\}_{i=1}^n$, $x_i \in X$, $y_i \in Y$ which are sampled i.i.d. according to the Borel probability measure $\rho$ on the data space $X \times Y$.
\end{assumption}

\subsection{Empirical risk minimization and learning rates}

Under the above assumptions, the goal of a learning algorithm is then
to choose a suitable hypothesis $f_n :X\mapsto Y$, $f_n \in \Hyp$ for the minimizer of the expected risk based on the data set $S_n$.
Empirical Risk Minimization (ERM) approaches this problem by
choosing a hypothesis that minimises the \textit{empirical risk},
\begin{equation}
    \label{eq:emp_risk}
   \inf_{f \in \Hyp} \E_n(f), \quad \E_n(f):= \frac{1}{n} \sum_{(x_i,y_i) \in S_n} \ell(f(x_i),y_i), \quad (\text{Empirical Risk})
\end{equation}
given the i.i.d. drawn data points $\{(x_i,y_i)\}_{i=1}^n\sim \rho^n$.
Note that $$\mathbb{E}_{\{x,y\}_{i=1}^n\sim \rho^n}(\E_n(f)) = \frac{1}{n} \sum_{i=1}^n \mathbb{E}_{(x_i,y_i)\sim \rho}\left[ \ell(f(x_i),y_i) \right] =\E(f),$$
i.e., the expectation of the empirical risk is the expected risk, which implies that we can indeed use the empirical risk as proxy for the expected risk in expectation.

While we would like to use the empirical risk as a proxy for the true risk,
we also need to ensure by minimising the empirical risk, we
actually find a valid solution to the underlying problem, i.e.,
the $f_n$ that we find is approaching the true solution $f_*$ which minimises the empirical risk.
This requirement is termed consistency.
To define this mathematically, let
\begin{equation}
    f_n := \underset{f \in \Hyp}{\mathrm{argmin}} \; \E_n(f)
\end{equation}
be the minimizer of the empirical risk, which exists under weak assumptions on $\Hyp$.
Then, the overall goal of a learning algorithm is to minimise the \textit{excess risk},
\begin{equation}
    \E (f_n) - \inf_{f \in \Hyp} \E (f), \quad (\text{Excess risk})
\end{equation}
while ensuring that $f_n$ is \textit{consistent} for a particular distribution $\rho$, i.e., that
\begin{equation}
    \lim_{n\rightarrow \infty}  \left( \E (f_n) - \inf_{f \in \Hyp} \E (f) \right) = 0. \quad (\text{Consistency})
\end{equation}

Since $S_n$ is a randomly sampled subset, we can analyse this behaviour
in expectation, i.e.,
\begin{equation}
    \lim_{n\rightarrow \infty} \Exp{\E (f_n) - \inf_{f \in \Hyp} \E (f)}=0
\end{equation}
or in probability, i.e.,
\begin{equation}
    \lim_{n\rightarrow \infty} \Prob_{\rho^n}{\left[\E (f_n) - \inf_{f \in \Hyp} \E (f) \right] > \epsilon}=0
\end{equation}
for all $\epsilon >0$.
If the above requirement holds for all distributions $\rho$ on the data space, then we say that it is universally consistent.

However, consistency is not enough in practice, since the convergence
of the risk of the empirical risk minimizer and the minimal
risk could be impracticably slow.

One of the most important questions in a learning setting is therefore
how fast this convergence happens, which is defined through the so-called \textit{learning rate}, i.e., the rate of decay of the excess risk.
We assume that this scales somewhat with respect to $n$, as for example
\begin{equation}
    \Exp{\E (f_n) - \inf_{f \in \mathcal{F}} \E (f)} = O(n^{-\alpha}).
\end{equation}
This speed of course has a practical relevance and hence allows us to compare different algorithms.
The sample complexity $n$ must depend on the error we want to achieve, and
in practice we can therefore define it as follows:
For a distribution $\rho$, $\forall \delta, \epsilon >0$, there exists a $n(\delta, \epsilon)$ such that
\begin{equation}
\Prob_{\rho^n}\left[ \E (f_{n(\epsilon, \delta)}) - \inf_{f \in \mathcal{F}} \E (f) \leq \epsilon \right] \geq 1 - \delta.
\end{equation}
This $n(\delta, \epsilon)$ is called the sample complexity.

One challenge of studying our algorithm's performance through the excess risk is that we
can generally not assess it.
We therefore need to find a way estimate or bound the excess risk solely based on the
hypothesis $f_n$ and the empirical risk $\E_n$.

Let us for this assume the existence of a minimizer
\begin{equation}
    f_*:= \inf_{f \in \Hyp} \E(f)
\end{equation}
over a suitable Hypothesis space $\Hyp$.

We can then decompose the
excess risk as
\begin{equation}
\label{eq:decom_emp_risk}
 \E(f_n) - \E(f_*) =  \E(f_n) - \E_n(f_n) + \E_n(f_n) - \E_n(f_*) + \E_n(f_*) - \E(f_*).
\end{equation}
Now observe that $\E_n(f_n) - \E_n(f_*) \leq 0$, we immediately see that
we can bound this quantity as
\begin{equation}
\label{eq:gen_error_bounds_erm}
   \E(f_n) - \E(f_*) =
   \E(f_n) - \inf_{f \in \Hyp} \E(f)  \leq 2  \sup_{f \in \Hyp} \lvert \E(f) - \E_n(f) \rvert,
\end{equation}
which implies that we can bound the error in terms of the so-called \textit{generalisation error} $\E(f)-\E_n(f)$.

We therefore see that we can study the convergence of the excess risk in terms of the generalisation error.
Controlling the generalisation error is one of the main objectives of statistical learning theory.

A fundamental result in statistical learning theory,
which is often referred as the fundamental theorem of statistical learning, is the following:
\begin{theorem}[Fundamental Theorem of Statistical Learning Theory~\cite{vapnik1998, blumer1989learnability, shalev2014understanding}]
  Let $\Hyp$ be a suitably chosen Hypothesis space of functions $f: X \mapsto Y$, $X\times Y$ be a probability space
  with a Borel probability measure $\rho$ and a measurable loss function $ell: Y \times Y \mapsto [0,\infty)$,
  and let the empirical risk and risk be defined as in Eq.~\ref{eq:learning_problem} and Eq.~\ref{eq:emp_risk} respectively.
  Then, for every $n \in \N$, $\delta \in (0,1)$, and distributions $\rho$, with probability $1-\delta$ it holds that
  \begin{equation}
  \label{eq:errorscaling}
      \sup_{f\in\mathcal{H}}\left\lvert \E_n(f) - \E(f) \right\rvert \leq \Theta \left(  \sqrt{\frac{c\,(\mathcal{H})+\log(1/ \delta)}{n} }
    ~\right),
  \end{equation}
  where $c\,(\mathcal{H})$ is a measure of the complexity of $\mathcal{H}$ (such as the VC dimension, covering numbers, or the Rademacher complexity~\cite{cucker2002mathematical,shalev2014understanding}).
\end{theorem}

We note that lower bounds for the convergence of the empirical to the (expected) risk
are much harder to obtain in general, and generally require to fix the underlying
distribution of the data. We indeed need for the following proof a lower bound on
this complexity which is not generally available. However, for the specific problems we
are treating here this holds true and lower bounds do indeed exist.
For simplicity, we will rely on the above given bound, as the other factors that
occur in these bounds only play a minor role here.

\subsection{Regularisation and modern approaches}

The complexity of the hypothesis space, $c(\mathcal{H})$ in Eq.~\ref{eq:errorscaling} relates to
the phenomenon of \textit{overfitting}, where a large hypothesis space results in a low training
error on the empirical risk, but performs poorly on the true risk.

In the literature, this problem is addressed with so-called \textit{regularisation techniques},
which are able to limit the size of the Hypothesis space and thereby its complexity,
in order to avoid overfitting the training dataset.

A number of different regularisation strategies have been proposed in the literature~\cite{vapnik1998,bishop2006pattern,bauer2007regularisation}, including the well-established
Tikhonov regularisation, which directly imposes constraints on the hypotheses class of candidate predictors.

From a computational perspective, Tikhonov regularisation, and other similar approaches compute a solution for the learning
problem by optimising a constraint objective (i.e., the empirical risk with an additional regularisation term).
The solution is obtained by a sequence of standard linear algebra operations such as matrix multiplication and inversion.
Since the standard matrix inversion time is $\mathcal{O}(n^3)$ for a $n \times n$ square matrix,
most of the solutions such as for GP or SVM, can be found in $\mathcal O(n^3)$ computational time for $n$ data points.
Notably, improvements to this runtime exist based on exploiting sparsity, or trading an approximation error against a lower
computational time.
Additionally, the time to solution typically depends on the conditioning of the matrix, which therefore can be lowered by
using preconditioning methods.

Regularisation is today widely used, and has led to many popular machine learning algorithms such as Regularised Least Squares~\cite{cucker2002mathematical}, Gaussian Process (GP) Regression and Classification \cite{rasmussen2006gaussian}, Logistic Regression~\cite{bishop2006pattern}, and Support Vector Machines (SVM)~\cite{vapnik1998}.
All the above mentioned algorithms fall under the same umbrella of kernel methods~\cite{shawe2004kernel}.

To further reduce the computational cost, modern methods leverage on the fact that
regularisation can indeed be applied implicitly through incomplete optimisation or other forms of approximation.
These ideas have been widely applied in the so-called {\em early stopping} approaches which are today standardly used in practice.
In early stopping, one only performs a limited number of steps of an iterative optimisation algorithm, typically in
gradient based optimisation.
It can indeed be shown for convex functions, that this process avoids overfitting the training set
(i.e., maintains an optimal learning rate), while the computational time is drastically reduced.
All regularisation approaches hence achieve a lower number of required operations,
while maintaining similar or the same generalisation performance of approaches such as Tikhonov regularisation~\cite{bauer2007regularisation}, in some cases provably.

Other approaches include the {\em divide and conquer}~\cite{zhang2013divide} approach or Monte-Carlo sampling (also-called sub-sampling) approaches.
While the former is based on the idea of distributing partitions of the initial training data, training
different predictors on the smaller problem instances, and then combining individual predictors into a joint one,
the latter achieves a form of dimensionality reduction through sampling a subset of the data in a specific manner.
The most well-known sub-sampling methods are random features~\cite{rahimi2007random} and so
called Nystr\"om approaches~\cite{smola2000sparse,williams2001using}.

In both cases, computation benefits from parallelisation and the reduced dimension of the datasets while similarly maintaining statistical guarantees (e.g.,\cite{rudi2015less}).

For all the above mentioned training methods, the computational times can typically be reduced from the $\mathcal O(n^3)$ of standard approaches to $\widetilde{ \mathcal{O}} (n^2)$ or $\widetilde{ \mathcal{O}} (nnz)$, where $nnz$ is the number of non-zero entries in the input data (matrix), while maintaining optimal statistical performance.

Since standard regularisation approaches can trivially be integrated into quantum
algorithms - such as regularised least squares - certain methods appear not to
work in the quantum algorithms toolbox.

For example, preconditioning, as a tool to reduce the condition number and make
the computation more efficient appears not to have an efficient solution yet
in the quantum setting~\cite{harrow2017limitations}. Therefore, more research is
required to give a full picture of the power and limitations of algorithms with
respect to all parameters.
Here we offer a brief discussion of the possible effects of inverting badly
conditioned matrices and how typical cases could affect the computational
complexity, i.e., the asymptotic runtime of the algorithm.

\section{Review of supervised quantum machine learning algorithms}
\label{sec:review_QML}

The majority of proposed supervised quantum machine learning algorithms are based on
fast linear algebra operations.
Indeed, most quantum machine learning algorithms that claim an exponential improvement over classical
counterparts are based on a fast quantum algorithm for solving linear systems of equations~\cite{wiebe2012quantum,schuld2016prediction,rebentrost2014quantum,wang2017quantum,kerenidis2017quantum,CGJ19,zhang2018nonlinear}.
This widely used subroutine is the HHL algorithm, a quantum linear system solver~\cite{harrow2009quantum} (QLSA),
which is named after its inventors Harrow, Hassidim, and Lloyd.
The HHL algorithm takes as input the normalised state
$\ket{b} \in \mathbb{R}^n$ and a $s(A)$-sparse matrix $A \in \mathbb{R}^{n \times n}$,
with spectral norm $\norm{A} \leq 1$ and condition number $\kappa=\kappa(A)$,
and returns as an output a quantum state $\ket{\tilde{w}}$ which encodes an approximation of
the normalised solution $\ket{w} = \ket{A^{-1} b} \in \mathbb{R}^n$ for the linear system $Aw=b$ such that
\begin{equation}
    \norm{\ket{\tilde{w}} - \ket{w}} \leq \gamma,
\end{equation}
for error parameter $\gamma$.
Note that above we assumed that the matrix is invertible, however, the algorithm can in practice
perform the Moore-Penrose inverse (also known as pseudoinverse), which is defined for arbitrary $A \in \mathbb{R}^{n \times m}$ by
$A^{+}:=(A^HA)^{-1}A^H$, and using the singular value decomposition of $A=U\Sigma V^H$, we hence have
\begin{equation}
    A^{+} = (V \Sigma^2 V^H)^{-1} V \Sigma U^H = V \Sigma^{-1} U^H,
\end{equation}
such that $A^{+}A = VV^H = I_{m}$.

The currently best known quantum linear systems algorithm, in terms of computational complexity,
runs in
\begin{equation}
    O(\norm{A}_F \, \kappa\; \mathrm{polylog} (\kappa, n, 1/\gamma)),
\end{equation}
time~\cite{CGJ19}, where $\norm{A}_F \leq \sqrt{n} \norm{A}_2 \leq \sqrt{n}$ is the Frobenius norm of $A$ and $\kappa$ its condition number.
As we will discuss in more detail in Chapter~\ref{chap:randNLA} on randomised numerical linear algebra,
such computations can also be done exponentially faster compared to known classical
algorithms using classical randomised methods in combination with a quantum-inspired memory structure,
by taking advantage of aforementioned Monte-Carlo sampling methods~\cite{RWC+18,tang2018quantum,chia2018quantum,gilyen2018quantum,chia2020sampling}
if the data matrix ($A$) is of low rank.
We note, that for full-rank matrices, an advantage is however still possible.

Before analysing the supervised machine learning algorithms with the above discussed knowledge from
statistical learning theory, we will first recapitulate the quantum least squares algorithm (linear regression)
and the quantum support vector machine (SVM).
Throughout this chapter, we use the least squares problem as a prototypical case to study the behaviour of QML algorithms, but the results extend trivially to the quantum SVM and many other algorithms.
In Section~\ref{sec:qml_algo_analysis} we also summarise the computational complexities taking into account the statistical guarantees
for all other algorithms and hence give explicit bounds.

\subsection{Recap: Quantum Linear Regression and Least Squares}
\label{sec:qml1}

The least squares algorithm minimises the empirical risk with respect to the quadratic loss
\begin{equation}
    \ell^{LS}(f(x),y) = \left( f(x) -y \right)^2,
\end{equation}
for the hypothesis class of linear functions
\begin{equation}
    \Hyp := \{f: X \mapsto Y | \, \exists w \in \mathbb{R}^d: f(x) = w^Hx \},
\end{equation}
with input space $\mathbb{R}^d$ and outputspace $\mathbb{R}$.
Given input-output samples $\{x_i,y_i\}_{i=1}^n$, where $x_i\in \mathbb{R}^d$, and $y_i \in \mathbb{R}$.
The empirical risk is therefore given by
\begin{equation}
    \E_n(f) := \frac{1}{n} \sum_{i=1}^n \left( w^H x_i -y_i \right)^2.
\end{equation}
the least squares problem seeks to find a vector $w$ such that
\begin{equation}
    w = \mathrm{argmin}_{w \in \mathbb{R}^d} \left( \norm{y - Xw}_2^2\right),
\end{equation}
where $y \in \mathbb{R}^n$ and $X \in \mathbb{R}^{n \times d}$.
The closed form solution of this problem is then given by $w = (X^HX)^{-1} X^H y$, and we can hence reformulate
this again into a linear systems problem of the form $Aw=b$, where $A=(X^HX)$ and $b=X^Hy$.
We obtain the solution then by solving the linear system $w=A^{-1}b$.

Since several quantum algorithms for linear regression and in particular least squares have been proposed~\cite{wiebe2012quantum,schuld2016prediction,wang2017quantum,kerenidis2017quantum,CGJ19,zhang2018nonlinear},
which all result in a similar scaling (taking into account the most recent subroutines),
we will in the following analysis use the best known result for the quantum machine learning algorithm.
All approaches have in common that they make use of the quantum linear system algorithm to
convert the state $\ket{\xi} = \ket{X^Hy}$ into the solution $\ket{\tilde{w}} = \ket{(X^HX)^{-1} \xi}$.
The fastest known algorithm~\cite{CGJ19}, indeed solves (regularised) least squares or linear regression problem in time
\[
O\left(\norm{A}_F \, \kappa \; \mathrm{polylog}(n,\kappa,1/\gamma)\right),
\]
where $\kappa^2$ is the condition number of $X^HX$ or $XX^H$ respectively, and $\gamma>0$ is an error parameter for the approximation accuracy.
Notably, this algorithm precludes a physical measurement of the resulting vector $\ket{\tilde{w}}$, since this
would immediately imply a complexity of
\[
O\left(\norm{A}_F \, \kappa/\gamma \; \mathrm{polylog}(n,\kappa,1/\gamma)\right).
\]
In that sense, the algorithm does solve the classical least squares problem only
to a certain extent as the solution is not accessible in that time.
Indeed it prepares a quantum state $\ket{\tilde{w}}$ which is $\gamma$-close to $\ket{w}$, i.e.,
\[
\norm{\ket{\tilde{w}} - \ket{w}} \leq \gamma,
\]
and in order to recover it, we would need to take up to $O(n \log(n))$ samples.
In the current form, we can however immediately observe that the Frobenius norm dependency implies that the algorithm is efficient if $X$ is low-rank (but no necessarily non-sparse).
As for all of the supervised quantum machine learning algorithms for classical input data,
the quantum least squares solver requires a quantum-accessible data structure, such as a qRAM.

Notably, it is assumed that $\frac{1}{\kappa^2} \leq \norm{X^HX} \leq 1$.
The output of the algorithm is then a quantum state $\ket{\tilde{w}}$, such that $\norm{\ket{\tilde{w}}-\ket{w}} \leq \epsilon$, where $\ket{w}$ is the true solution.

We note that other linear regression algorithms based on sample-based Hamiltonian simulation are
possible~\cite{lloyd2014quantum,kimmel2017hamiltonian}, which result in different requirements.
Indeed, for these algorithms we need to repeatedly prepare a density matrix, which is a normalised version
of the input data matrix. While this algorithm has generally worse dependencies on the error $\gamma$,
it is independent of the Frobenius norm~\cite{schuld2016prediction}.
The computational complexity in this case is
\[
O(\kappa^2 \gamma^{-3} \mathrm{polylog}(n)),
\]
which can likely be improved
to $O(\kappa \gamma^{-3} \mathrm{polylog}(n, \kappa, \gamma^{-1}))$.

However, since our analysis will indeed show that a higher polynomial dependency will incur a worse runtime
once we take statistical guarantees into account, we will use the algorithm in the following which has the
lowest dependency.
Notably, the error dependency is in either case polynomial if we require a classical solution to be output.

As a further remark, other linear regression or least squares quantum algorithms exist \cite{wang2017quantum,zhang2018nonlinear},
but we will not include these here as our results can easily be extended to these as well.

Next, we will also recapitulate the quantum support vector machine.

\subsection{Recap: Quantum Support Vector Machine}
\label{sec:qml2}
The second prototypical quantum machine learning algorithm which we want to recapitulate is the
quantum least-squares support vector machine~\cite{rebentrost2014quantum} (qSVM).
As we will see, the procedure for the qSVM is similar to the quantum least squares approach, and
therefore results in very similar runtimes.
The qSVM algorithm is calculating the optimal separating hyperplane by solving again a linear system of equations.
For $n$ points $S_n=\{(x_i,y_i)\}_{i=1}^n$ with $x_i \in \mathbb{R}^d, y_i = \{\pm 1\}$,
and again assuming that we can efficiently prepare states corresponding to the data vectors,
then the least-squares formulation of the solution is given by the linear system of the form
\begin{equation}
    \left(
    \begin{array}{cc}
        0 &  \vec{1}^H \\
        \vec{1} & K + \delta^{-1} I
    \end{array} \right)
    \left(
    \begin{array}{c}
        w_0  \\
        w
    \end{array} \right) = \left(
    \begin{array}{c}
        0 \\
        y
    \end{array} \right),
\end{equation}
where $K_{ij} = x_i^Hx_j$ (or $K_{ij} = \phi(x_i)^H\phi(x_j)$ respectively for a non-linear features) is the kernel matrix, $y = (y_1,\ldots, y_n)^H$, $\vec{1}$ is the all-ones vector, and $\delta$ is a user specified parameter.
We note that certain authors argue that a least square support vector machine is
not truly a support vector machine, and their practical use highly restricted.
The additional row and column in the matrix on the left hand side arise because of a non-zero offset.
Notably, $w^Hx+w_0 > 1$ or $w^H x+w_0 < -1$, with $w^Hx=\sum_{j=1}^n w_j x_j$ determines the hyperplanes.
The solution is hence obtained by solving the linear systems using the HHL algorithm
based on the density matrix exponentiation~\cite{lloyd2014quantum} method previously mentioned.
The only adaptation which is necessary is to use the normalised Kernel $\hat{K} = K/\mathrm{tr}(K)$.
However, since the smallest eigenvalues of $\hat{K}$ will be of $O(1/n)$ due to the normalisation,
the quantum SVM algorithm truncates the eigenvalues which are below a certain threshold $\delta_K$, s.t.,
$\delta_K \leq \lvert \lambda_i \rvert \leq 1$, which results in an effective condition number
$\kappa_{eff} = 1/\delta_K$, thereby effectively implementing a form of spectral filtering.

The runtime of the quantum support vector machine is given by
\[
O(\kappa_{eff}^3 \gamma^{-3} \mathrm{polylog}(nd,\kappa,1/\gamma)),
\]
and outputs a state $\ket{\tilde{w}_n}$ that approximates the solution $w_n:=[w_0, \, w]^H$, such that
$\norm{\ket{\tilde{w}_n} - \ket{w_n}} \leq \gamma$.
Similar as for the least squares algorithm, we cannot retrieve the parameters without an overhead,
and the quantum SVM therefore needs to perform immediate classification.

\section{Analysis of quantum machine learning algorithms}
\label{sec:qml_analysis}

The quantum algorithms we analyse throughout this chapter rely on a range of parameters, which include
the input dimension $n$, which corresponds to the number of data points in a sample (the dimension of the
individual data point is typically small so we focus on this part), the error of the algorithm with respect to
the final prediction $\gamma$, and the condition number $\kappa$ of the input data matrix.
Our main objective is to understand the performance of these algorithms if we want to achieve an overall
generalisation error of $\Theta(n^{-1/2})$.

To start, we therefore first need to return to our previous assessment of the risk, and use in the following
a standard error decomposition.
Let $f$ be a hypothesis, and let $\mathcal{F}$ is the space of all measurable functions $f:X \mapsto Y$.
We define by $\E^* := \inf_{f  \in \mathcal{F}} \E(f)$ the Bayes risk, and want to limit the distance $\E(f) - \E^*$.
Let now $\E_{\Hyp} := \inf_{f \in \Hyp} \E(f)$, i.e., the best risk attainable by any function in the hypothesis space $\Hyp$, where we assume in the following for simplicity that $\E_{\Hyp}$ always admits a minimizer $f_{\Hyp} \in \Hyp$.
Note, that it is possible to remove this assumption by leveraging regularisation.
We then decompose the error as:
\begin{align}
\label{eq:totalerrordec}
    \E( f) - \E^* &= \underbrace{\Ecal(f) - \E( \hat{f})}_{\textrm{Optimisation error}} ~+~ \underbrace{\Ecal( \hat{f}) - \E_{\mathcal{H}}}_{\textrm{Estimation error}} ~+~ \underbrace{ \E_{\mathcal{H}} - \E^*}_{\textrm{Irreducible error}} \\
     &= \xi + \Theta(1/\sqrt{n}) + \mu.
\end{align}

The first term in Eq.~\ref{eq:totalerrordec} is the so-called \textit{optimisation error} which
indicates how good the optimisation procedure which generates $f$ is, with respect to the
actual minimum (infimum) of the empirical risk.
This error stems from the approximations an algorithm typically makes,
and relates to the $\gamma$ in previous sections.
The optimisation error can result from a variety of approximations, such as a finite number of steps in an iterative optimisation process or a sample error introduced through a non-deterministic process.
This error is discussed in detail in Section~\ref{ssec:bounds_algorithmic_error}.
The second term is the \textit{estimation error} which is due to taking the empirical risk as a proxy
for the true risk by using samples from the distribution $\rho$.
This can be bound by the generalisation bound we discussed in Eq.~\ref{eq:errorscaling}.
The last term is the \textit{irreducible error} which measures how well the hypothesis space describes the problem.
If the true solution is not in our hypothesis space, there will always be an irreducible error
that we indicate with the letter $\mu$.
If $\mu=0$, i.e., irreducible error is zero, then we call $\mathcal{H}$ is universal.
For simplicity, we assume here that $\mu = 0$, as it also will not impact the results of this paper much.

From the error decomposition in Eq.~\ref{eq:totalerrordec} we see that in order to achieve the best possible generalisation error overall, we need to make sure that the different error contributions are of the same order.
We therefore in particular need to ensure that the optimisation error matches the scaling of the estimation error.
Since for most known classical algorithms, with the exceptions of e.g., Monte-Carlo algorithms,
the optimisation error typically scales with $O(\log(1/\epsilon))$ and matching the bounds is usually trivial.
However, many quantum algorithms,including some of the quantum linear regression and least squares algorithms
we discussed in the previous section (e.g.~\cite{rebentrost2014quantum,schuld2016prediction}), have a polynomial
dependency on the optimisation error.
In the next section we discuss the implications of matching the bounds, and how they affect the algorithms
computational complexity.

Notably, other quantum algorithms have only a polylogarithmic error dependency,
such as~\cite{CGJ19}, and therefore the error matching does not impose any critical slowdown.
In these cases, however, we will see that quantum algorithms argument still cannot achieve a polylogarithmic runtime
in the dimension of the training set due to the error resulting from the finite sampling process that is
required to extract a classical output from a given quantum state.

Finally, to take into account all dependencies of the quantum algorithms, we also analyse the condition number.
Here, we show that with high probability the condition number has a polynomial dependency on the number of samples in the training set as well, which therefore indicates a certain scaling of the computational complexity.

We do the analysis in the following exemplary for the least squares case, and summarise the resulting
computational complexities of a range of supervised quantum machine learning algorithms next to the classical ones
then in Fig.~\ref{fig:summary_algos}.

\subsection{Bound on the optimisation error}
\label{ssec:bounds_algorithmic_error}

As previously mentioned, we will use the quantum least squares algorithms~\cite{wiebe2012quantum,schuld2016prediction,wang2017quantum} as an example case to demonstrate
how the matching of the error affects the algorithm.
The results we obtain can easily be generalised to other algorithms and instances, and in particularly hold for
all kernel methods. As we try to remain general, we will do the analysis with a general
algorithm with the computational complexity
\begin{equation}
    \Omega \left(n^{\alpha} \gamma^{-\beta} \kappa^{c} \log(n) \right)
\end{equation}
We show that in order to have a total error that scales as $n^{-1/2}$, the quantum algorithm
will pick up a polynomial $n$-dependency.

The known quantum least squares algorithms have a $\gamma$ error guarantee for the final output state $\ket{\tilde{w}}$,
i.e., $\norm{\ket{w} - \ket{\tilde{w}}}_2 \leq \gamma$, where $\ket{w}$ is the true solution.
The computational complexity (ignoring all but the error-dependency), is for all algorithms of the form
$O(\gamma^{-\beta})$ for some $\beta$, for example \cite{schuld2016prediction} with $\beta=3$,
or \cite{li2019sublinear} with $\beta=4$.

Since the quantum algorithms require the input data matrix to be either Hermitian or encoded in a
larger Hermitian matrix, the dimensionality of the overall matrix is $n+d$ for $n$ data points in $\mathbb{R}^d$.
For simplicity, we here assume that the input matrix is a $n \times n$ Hermitian matrix, and neglect this step.
In order to achieve the best possible generalisation error, as discussed previously, we want to match the errors of the
incomplete optimisation to the statistical ones.
In the least squares setting, $\tilde{w} = w_{\gamma,n}$ is the output of the algorithm corresponding to
the optimal parameters fitted to the $S_n$ data points, which exhibits at most a $\gamma$-error.
Therefore, $\tilde{w}$ corresponds to the estimator $f_{\gamma,n}$ in the previous notation
that we saw in Eq.~\ref{eq:totalerrordec} (and Eq.~\ref{eq:gen_error_bounds_erm}).
Concretely, we can see that the total error of an estimator $f_{\gamma,n}$ on $n$ data points with precision $\epsilon$ is given by
\begin{align}
\label{eq:error_contrib}
    \E(f_{\gamma,n}) - \E(f_{n}) &= \nonumber \\
    &= \underbrace{\E(f_{\gamma,n}) - \E_n(f_{\gamma,n})}_{\textrm{generalisation error}} ~+~ \underbrace{\E_n(f_{\gamma,n}) - \E_n(f_\gamma)}_{\textrm{Optimisation error}} ~+~ \underbrace{ \E_n(f_\gamma) - \E(f_\gamma)}_{\textrm{generalisation error}} \nonumber \\
     &= \Theta(n^{-1/2}) + \underbrace{\E_n(f_{\gamma,n}) - \E_n(f_\gamma)}_{\textrm{Optimisation error}},
\end{align}
where the first contribution is a result of Eq.~\ref{eq:errorscaling}, i.e., the generalisation performance,
and the second comes from the error of the quantum algorithm, which we will show next.
In order to achieve the best statistical performance, which means to achieve the lowest
generalisation error, the algorithmic error must scale at worst as the worst statistical error.
We will next show that the optimisation error of a quantum algorithm in terms of the prediction results in a
$\gamma$ error, which is inherited from weights $\ket{\tilde{w}}$, which the quantum algorithm produces.
Recalling, that in least squares classification is performed via the inner product, i.e.,
\begin{equation}
    y_{pred} := \tilde{w}^Hx,
\end{equation}
for model $\tilde{w}$ and data point $x$ which corresponds to $f_{\gamma,n}^H(x)$ in the general notation.
This then will result in the expected risk of the estimator $f_{\gamma,n}$ to be
\begin{equation}
  \E_n(f_{\gamma,n}) = \frac{1}{n} \sum_{i=1}^n \left(\tilde{w}^H x_i -y_i \right)^2.
\end{equation}
Therefore, assuming the output of the quantum algorithm is a state $\tilde{w}$, while the exact minimizer of the empirical risk is $w$, s.t., $\norm{\ket{\tilde{w}} - \ket{w}}_2 \leq \gamma$, and assuming that $|X|$ and $|Y|$ are bounded,
then we find that
\begin{align}
    \lvert \E_n(f_{n,\gamma})-\E_n(f_{n})\rvert &\leq \frac{1}{n} \sum_{i=1}^n  \left\lvert \left(\tilde{w}^H x_i -y_i \right)^2 - \left( w^H x_i -y_i \right)^2 \right\rvert \nonumber \\
    &\leq \frac{1}{n} \sum_{i=1}^n L \left\lvert \left( \tilde{w} - w\right)^H x_i \right\rvert \nonumber \\
    &\leq \frac{1}{n} \sum_{i=1}^n L \norm{\tilde{w} - w}_2 \norm{x_i}_2 \leq  k \cdot \gamma = O(\gamma),
\end{align}
where $k>0$ is a constant, and we used Cauchy-Schwartz, and the fact that that for the least-square it holds that
\begin{equation}
    \lvert \ell^{LS}(f(x_i),y_i) - \ell^{LS}(f(x_j),y_j) \rvert \leq L \lvert (f(x_i)-y_i) - (f(x_j)-y_j)\rvert,
\end{equation}
since $|X|$, and $|Y|$ bounded.

A few remarks.
In the learning setting the number of samples is fixed, and hence cannot be altered, i.e., the statistical error
(generalisation error) is fixed to $\Theta(1/\sqrt{n})$, and the larger $n$ is taken, the better the guarantees we are
able to obtain for future tasks.
Therefore, it is important to understand how we can reduce the other error contributions in Eq.~\ref{eq:error_contrib} in order to guarantee that we have the lowest possible overall error, or accuracy.

To do so, we match the error bounds of the two contributions, so that the overall performance of the algorithm is maximised, which means that the optimisation error should not surpass the statistical error.
We hence set $\gamma = n^{-1/2}$, and see that the overall scaling of the algorithm will need be of the order
$O \left(n^{\beta/2} \right)$, ignoring again all other contributions.
To take a concrete case, for the algorithm in \cite{schuld2016prediction}
the overall runtime is then of at least $O(n^{3/2})$.
The overall complexity of the algorithm then has the form
\begin{equation}
    \Omega \left(n^{\alpha} n^{\beta/2} \log(n) \kappa^{c} \right)
\end{equation}
for some constant $c, \beta, \alpha$

This straightforward argument from above can easily be generalised to arbitrary kernels by replacing
the input data $x$ with feature vectors $\phi(x)$, where $\phi(\cdot)$ is a chosen feature map.

We have so far only spoken about algorithms which naturally have a polynomial $1/\gamma$-dependency.
However, as we previously mentioned, not all quantum algorithms have such an error.
For algorithms which only depend polylogarithmically on $1/\gamma$, however, the quantum mechanical nature
will incur another polynomial $n$ dependency as we will see next.

\subsection{Bounds on the sampling error}
\label{ssec:bounds_sample_error}

So far we have ignored any error introduced by the measurement process.
However, we will always need to compute a classical estimate of the output of the quantum algorithm, which
is based on a repeated sampling of the output state.
As this is an inherent process which we will need to perform for any quantum algorithm,
the following analysis applies to any QML algorithm with classical output.
Since we estimate the result by repeatedly measuring the final state of our quantum computation in
a chosen basis, our resulting estimate for the desired output is a random variable.
It is well known from the central limit theorem, that the sampling error for such a random variable
scales as $O(1/\sqrt{m})$, where $m$ is the number of independent measurements.
This is known as the standard quantum limit or the so-called shot-noise limit.
Using so-called quantum metrology it is sometimes possible to overcome this limit
and obtain an error that scales with $1/m$.
This however poses the ultimate limit to measurement precision which is a direct consequence of the
Heisenberg uncertainty principle~\cite{giovannetti2008qram1, giovannetti2008qram2}.

Therefore, any output of the quantum algorithm will have a measurement error $\tau$.
Let us turn back to our least squares quantum algorithm.
It produces a state $\ket{\tilde{w}}$ which is an approximation to the true solution $\ket{w}$.
Using techniques such as quantum state tomography we can produce a classical estimate $\hat{w}$ of the vector
$\tilde{w}$ with accuracy
\begin{equation}
    \norm{\tilde{w} - \hat{w}}_2 \leq \tau = \Omega(1/m),
\end{equation}
where $m$ is the number of measurements performed.
If $y = w^Hx$ is the error-free (ideal) prediction, then we can hence only produce an approximation
$\hat{y}= \hat{w}^Hx$, such that
\begin{align}
|y-\hat{y}| &= |w^H x - \hat{w}^H x| \\
& \leq \norm{w-\tilde{w} + \tau} \, \norm{x} \\
& \leq (\gamma + \tau)\, \norm{x},
\end{align}
using again Cauchy-Schwartz.
Similar to the previous approach, we need to make sure that the contribution coming from the measurement
error scales at most as the worst possible generalisation error, and hence set $\tau=n^{-1/2}$.
From this, we immediately see that any quantum machine learning algorithm which is
to reach optimal generalisation performance, will require a number of $m=\Omega(n^{1/2})$ repetitions,
which is hence a lower bound for all supervised quantum machine learning algorithms.
For algorithms which do not take advantage of forms of advanced quantum metrology, this might even be $\Omega(n)$.

Putting things together, we therefore have a scaling of any QML algorithm of
\begin{equation}
    \Omega \left(n^{\alpha+(1+\beta)/2} \log(n) \kappa^{c} \right) ,
\end{equation}
which for the state-of-the-art quantum algorithm for quantum least squares~\cite{CGJ19} result in
\begin{equation}
    \Omega \left(n^{2} \kappa \log(n) \right).
\end{equation}
In order to determine the overall complexity, we hence only have one parameter left: $\kappa$.
However, already now we observe that the computational complexity is similar or even worse compared to the best classical machine learning algorithms.

\subsection{Bounds on the condition number}
\label{ssec:bounds_cond_num}

In the following we will do the analysis of the last remaining depedency of the quantum algorithms.
The condition number.
Let the condition number dependency of the QML algorithm again be given by $\kappa^c$ for some constant $c \in \mathbb{R}^+$.
Note that the best known result has a $c=1$ dependency, ignoring logarithmic dependencies.
We can think of the following three scenarios for the condition number.
\begin{enumerate}
    \item Best case scenario:
    In the best case setting, the condition number is one or sufficiently close to one.
    This is the lower bound and can only ever happen if the data is full rank and all the eigenvalues
    are of very similar size, i.e., $\lambda_i \approx \lambda_j$ for all $i,j$.
    However, for such cases, it would be questionable whether a machine learning algorithm would be useful, since this would imply that the data lacks any strong signal.
    In these cases the quantum machine learning algorithms could be very fast and might give a quantum advantage if the $n$-scaling due to the error-dependency is not too high.
    \item Worst case scenario:
    On the other extreme, the condition number could be infinite, as could be the case for very badly conditioned matrices with smallest eigenvalues approaching $0$.
    This can be the case if we have one or a few strong signals (i.e., eigenvalues which are closer to $1$), and a small additional noise which results in the smallest eigenvalues being close to $0$. Such ill conditioned systems do indeed occur in practice, but can generally be dealt with by using spectral-filtering or
    preconditioning methods, as for example discussed in~\cite{gerfo2008spectral}.
    Indeed, the quantum SVM~\cite{rebentrost2014quantum} or the HHL algorithm~\cite{harrow2009quantum}
    do or can readily make use of such methods.
    Concretely they do only invert eigenvalues which are above a certain threshold. This hence gives a new,
    effective condition number $\kappa_{eff}=\sigma_{max}/\sigma_{threshold}\leq 1/\sigma_{threshold}$ which is
    typically way smaller compared to the actual $\kappa$, and makes algorithms practically useful.
    However, it should be noted that quantum algorithms which perform such steps need to be compared against
    corresponding classical methods.
    Note, that such truncations (filters) typically introduce an error, which then needs to
    be taken into account separately.
    Having covered these two extreme scenarios, we can now focus on a typical case.
    \item A plausible case:
    While the second case will appear in practice, these bounds give little insight into the actual performance
    of the quantum machine learning algorithms, since we cannot infer any scaling of $\kappa$ from them.
    However, for kernel based methods, we can derive a plausible case which can give us some intuition of
    how bad the $\kappa$-scaling typically can be.
    We will in the following show that with high probability,
    the condition number for a kernel method can have a certain $n$-dependency.
    Even though this result gives only a bound in probability, it is a plausible case with concrete $n$-dependency (the dimension of the input matrix) rather than the absolute worst case of $\kappa=\infty$, which gives an impractical upper bound.
    As a consequence, a quantum kernel method which scales as $O(\kappa^3)$ could pick up a factor of $O(n\sqrt{n})$ in the worst case which has the same complexity as the classical state of the art.
\end{enumerate}

In the following, we now prove a lower bound for the condition number of a covariance or kernel matrix assuming that we have at least one strong signal in the data.
The high level idea of this proof is that the sample covariance should be close to the true covariance with increasing number of samples, which we can show using concentration of measure.
Next we use that the true covariance is known to have converging eigenvalues,
as it constitutes a converging series.
This means we know that the $k$-th eigenvalue of the true covariance must have
an upper bound in terms of its size which is related to $k$.
Since we also know that the eigenvalues of the two matrices will be close,
and assuming that we have a few strong signals (i.e., $O(1)$ large eigenvalue),
we can then bound the condition number as the ratio of the largest over the smallest eigenvalue.

For the following analysis, we will first need to recapitulate some well known results about Mercer's kernels, which can be found e.g.,
in \cite{cucker2002mathematical}.
If $f \in L^2_{\nu}(X)$ is a function in the Hilbert space of square integrable functions on $X$ with Borel measure $\nu$, and $\{\phi_1, \phi_2,\ldots\}$
is a Hilbert basis of $L_{\nu}^2(X)$, $f$ can be uniquely written as $f=\sum_{k=1}^{\infty} a_k \phi_k$, and the partial sums $\sum_{k=1}^N a_k \phi_k$ converge
to $f$ in $L_{\nu}^2(X)$.
If this convergence holds in $C(X)$, the space of continuous functions on $X$, we say that the series converges uniformly to $f$.
If furthermore $\sum_k |a_k|$ converges, then we say that the series $\sum_k a_k$ converges absolutely.
Let now $K:X\times X\rightarrow \mathbb{R}$ be a continuous function. Then the linear map $$L_K: L_{\nu}^2(X) \rightarrow C(X)$$ given by the following integral transform
$$(L_Kf)(x)=\int K(x,x') f(x') d\nu(x')$$
is well defined.
It is well known that the integral operator and the kernel have the following relationship:
\begin{theorem}[\cite{cucker2002mathematical}[Theorem 1, p. 34; first proven in \cite{hochstadt2011integral}]]
 \label{thm:integral_kernel}
 Let $\mathcal{X}$ be a compact domain or a manifold, $\nu$ be a Borel measure on
 $\mathcal{X}$, and $K : \mathcal{X} \times \mathcal{X} \rightarrow \mathbb{R}$ a Mercer kernel.
 Let $\lambda_k$ be the $k$th eigenvalue of $L_K$ and $\{\phi_k\}_{k \geq 1}$
 be the corresponding eigenvectors. Then, we have for all $x, x' \in \mathcal{X}$
 \[ K(x,x') = \sum_{k=1}^{\infty} \lambda_k \phi_k(x) \phi_k(x'), \]
 where the convergence is absolute (for each $x,x' \in \mathcal{X} \times \mathcal{X}$) and
 uniform (on $\mathcal{X} \times \mathcal{X}$).
\end{theorem}
Note that the kernel here takes the form $K=\Phi_{\infty}\Phi_{\infty}^H$, which
for the linear Kernel has the form $K=XX^H$.
Furthermore, the kernel matrix must have a similar spectrum to the empirical or sample kernel matrix $K_n=X_nX_n^H$, and indeed for $\lim_{n\rightarrow \infty} K_n \rightarrow K$,
where we use here the definition $X_n \in \mathbb{R}^{d
times n}$, i.e., we have $n$ vectors of dimension $d$ and therefore $\mathcal{X}=\mathbb{R}^d$.
The function $K$ is said to be the kernel of $L_K$ and several properties of $L_K$ follow from the properties of $K$.
Since we want to understand the condition number of $K_n$, the
sample covariance matrix, we need to study the behaviour of its eigenvalues.
For this, we start by studying the eigenvalues of $K$.
First, from Theorem~\ref{thm:integral_kernel} the next corollary follows.
\begin{corollary}[\cite{cucker2002mathematical}, Corollary 3]
\label{cor:convergence_evs}
The sum $\sum_k \lambda_k$ is convergent, and
\begin{equation}
  \sum_{k=1}^{\infty} \lambda_k = \int_X K(x,x) \leq \nu(X) C_K,
\end{equation}
where $C_K=\sup_{x,x'\in X} \lvert K(x,x')\rvert$ is an upper bound on the kernel.
Therefore, for all $k\geq 1$, $\lambda_k \leq \left(\frac{\nu(X) C_K}{k}\right)$.
\end{corollary}

As we see from Corollary~\ref{cor:convergence_evs}, the eigenvalue $\lambda_k$ (or singular value,
since $K$ is SPSD) of $L_K$ cannot decrease slower than $O(1/k)$,
since for convergent series of real non-negative numbers $\sum_k \alpha_k$,
it must holds that $\alpha_k$ must go to zero faster than $1/k$.

Recalling that $K$ is the \textit{infinite} version of the kernel matrix $K_n=X_nX_n^H$
(or generally $K=\Phi \Phi^H$ for arbitrary kernels) we now
need to relate the finite sized Kernel $K_n$ to the kernel $K$.
Leveraging on concentration inequalities for random matrices, we will now show how $K_n \in \mathbb{R}^{d \times d}$ for $X_n \in \mathbb{R}^{d \times n}$ converges
to $K$ as $n$ grows, and therefore the spectra (i.e., eigenvalues must match),
which implies that the decay of the eigenvalues of $K_n$.
Indeed, we will see that the smallest eigenvalue $\lambda_n = O(1/n)$ with high probability.
From this we obtain immediately upper bounds on the condition number in high probability.
This is summarised below.

\begin{theorem}
  \label{lem:cond_num_dep}
  The condition number of a Mercer kernel $K$ for a finite number of samples $n$ is with high probability lower bounded by $\Omega(\sqrt{n})$.
\end{theorem}

\begin{proof}
We will in the following need some auxiliary results.

\begin{theorem}[Matrix Bernstein~\cite{tropp2015introduction}]
  \label{thm:matrix_bernstein}
Consider a finite sequence $\{X_k\}$ of independent, centered, random, Hermitian $d$-dimensional matrices, and assume that $\Exp X_k =0$, and $\norm{X_k} \leq R$, for all $k\in [n]$.
Let $X:=\sum_k X_k$ and $\Exp X=\sum_k \Exp X_k$, and let
\begin{align}
\sigma(X) &= \max \left\lbrace \norm{\Exp [XX^H]}, \norm{\Exp [X^HX]} \right\rbrace \nonumber \\
  &= \max \left\lbrace \norm{\sum_{k=1}^n\Exp [X_kX_k^H]}, \norm{\sum_{k=1}^n\Exp [X_k^HX_k]} \right\rbrace.
\end{align}
Then,
\begin{equation}
    \Pr{\left[ \norm{X} \geq \epsilon \right]} \leq 2d \exp{\left(\frac{-\epsilon^2/2}{\sigma(X)+R\epsilon/3}\right)}, \forall \epsilon \geq 0,\,
\end{equation}
\end{theorem}
We can make use of this result to straightforwardly bound the largest
eigenvalue of the sample covariance matrix, which is a well known result in the random matrix literature.
As outlined above, the sample covariance matrix is given by
\begin{equation}
  K_n = \frac{1}{n} \sum_{k=1}^n x_kx_k^H,
\end{equation}
for $n$ centred (zero mean) samples in $\mathbb{R}^d$,
and $K:=\Exp{xx^H} \in \mathbb{R}^{d \times d}$.
We look in the following at the matrix, i.e., $A := K_n - K$, and assume
\begin{equation}
  \label{eq:assumption_bound_samples}
  \norm{x_i}_2^2\leq r, \quad \forall i \in [n]
\end{equation}
i.e., the sample norm is bounded by some constant $r$.
Typically data is sparse, and hence independent of the dimension,
although both scenarios are possible. Here we assume a dependency on $d$.
Under this assumption, we let
\begin{equation}
  A_k := \frac{1}{n} \left( x_k x_k^H - K \right),
\end{equation}
for each $k$ and hence $A = \sum_{k=1}^n A_k$.
With the assumption in Eq.~\ref{eq:assumption_bound_samples} we obtain then
\begin{equation}
  \norm{K} = \norm{\Exp{xx^H}} \leq \Exp{\norm{xx^H}} =\Exp{\norm{x}^2} \leq r,
\end{equation}
using Jensen's inequality.
The matrix variance statistic $\sigma(A)$ is therefore given by
\begin{equation}
    0 \leq \sigma(A) \leq \frac{R}{n} \norm{C_{\infty}} \leq \frac{r^2}{n},
\end{equation}
which follows from straight calculation.
Taking into account that $\norm{A_k}\leq \frac{r}{n}$, and by invoking
Thm.~\ref{thm:matrix_bernstein}, and using the above bounds.
Assuming $r=C\cdot d$ for some constant $C$, i.e., the norm of the vector will be dependent on
the data dimension $d$ (which might not be the case for sparse data!), we hence obtain that
\begin{equation}
  \Exp{\norm{A}} = \Exp{\norm{K_n-K}} \leq O\left(\frac{d}{\sqrt{n}} + \frac{d}{n} \right) =  O\left(\frac{d}{\sqrt{n}}\right).
\end{equation}
Note that this is essentially $O(\sqrt{1/n})$ for $n$ the number of samples, assuming $n \gg d$,
and similarly for sparse data with sparsity $s = O(1)$ the norm will not be proportional to $d$.

We next need to relate this to the eigenvalues $\lambda_{min}$ and $\lambda_{max}$, for which we will need the following lemma.
\begin{lemma}
  \label{lem:inf_bound}
  For any two bounded functions $f,g$, it holds that
  \begin{equation}
      \lvert \inf_{x \in X} f(x) - \inf_{x\in X} g(x) \rvert \leq \sup_{x \in X} \lvert f(x)-g(x) \rvert
  \end{equation}
\end{lemma}
\begin{proof}
  First we show that $\abs{\sup_X f - \sup_X g} \leq \sup_X \abs{f-g}$.
  For this, take $$\sup_X (f \pm g) \leq \sup_X (f-g) + \sup_X g \leq \sup_X \abs{f-g} + \sup_X g,$$
  and $$\sup_X (g \pm f) \leq \sup_X (g-f) + \sup_X f \leq \sup_X \abs{g-g} + \sup_X f,$$
  and the result follows. Next we can proof the Lemma by replacing $f=-f$ and $g=-g$ and using that $\inf_X f = \sup_X(-f)$, the claim follows.
\end{proof}

Using that with high probability $\norm{K_n-K} \leq O\left(d/\sqrt{n} + d/n\right)$, and therefore
\begin{align}
  &\abs{\lambda_{max}(K_n)-\lambda_{max}(K)} \leq O \left(d/\sqrt{n} \right), \quad \text{and}\\
  &\abs{\lambda_{min}(K_n)-\lambda_{min}(K)} \leq O \left(d/\sqrt{n} \right),
\end{align}
which follows from Lemma~\ref{lem:inf_bound} by taking $f(x)= (x^HK_nx)/x^Hx$ and $g(x)= (x^HK x)/x^Hx$.
Therefore, ignoring the data dimension dependency, and recalling that the Kernel matrices
are positive semi-definite, we obtain that with high probability
\begin{align}
  \kappa(K_n) &= \frac{\lambda_{max}(K_n)}{\lambda_{min}(K_n)} \geq \frac{\norm{K_n}}{ \lambda_{min}(K) + O\left(1/\sqrt{n}\right)} \\
  &= \frac{\norm{K_n}}{ O(1/n) + O\left(1/\sqrt{n}\right)} = O(\norm{K_n}\sqrt{n}),
\end{align}
where we used the bounds on the smallest eigenvalue in Corollary~\ref{cor:convergence_evs}.
Therefore, in high probability the condition number of the problem is at least $O(\norm{K_n}\sqrt{n})$, and if we assume $\norm{K_n}=1$ this is $O(\sqrt{n})$ for any Kernel method.
\end{proof}

We hence can assume that an additional $\sqrt{n}$ dependency from $\kappa$ can appear as a
plausible case.
If we additionally take into account the other sources of errors, which we discussed before, typical QML algorithm result in runtimes which are significantly worse than the classical counterpart.
We hence learn from this analysis, that if we desire practically relevant QML algorithms with provable guarantees, we either need to reduce the condition number dependency (e.g., through preconditioning or filtering),
or apply the quantum algorithm only to data sets for which we can guarantee that $\kappa$ is small enough.
Since some quantum algorithms allow for spectral filtering methods, and therefore limit
the condition number by $\kappa_{eff}$, we will not include the condition number scaling
in the following analysis.
We leave it open to the reader to apply this scaling to algorithms which exhibit a
high $\kappa$-dependency.
Notably, it would furthermore require a more involved analysis of the error since such a
truncation immediately imposes an error, e.g., $\norm{A - F(A)}_2^2$ for Hermitian
$A=U\Sigma U^H$, and filter $F$, could result in an error $\norm{\Sigma - F(\Sigma)}_2^2$.
If the filter in the simplest form cuts of the eigenvalues below $\sigma_{eff}$ then the final
error would be $\sigma_{eff}$, and it would then be a matter of the propagation of this error
through the algorithm.

\section{Analysis of supervised QML algorithms}
\label{sec:qml_algo_analysis}

Our analysis is based entirely on the dependency of the
statistical guarantees of the estimator on the size of the data set.
We leverage on the above discussed impacts on the algorithmic error and the
measurement based error as well as the previously derived results in statistical learning theory.
In particular, by using that the accuracy parameter of a supervised learning problem
scales inverse polynomially with the number of samples which are used for
the training of the algorithm, we showed that the errors of quantum algorithms will
results in $\mathrm{poly}(n)$ scalings.
The runtimes for a range of quantum algorithms which we now derive based on these requirements
indicate that these algorithms can therefore not achieve exponential
speedups over their classical counterparts.

We note that this does not rule out exponential advantages for learning problems
where no efficient classical algorithms are known, as there exist learning problems
for which quantum algorithms have a superopolynomial advantage~\cite{grilo2017learning, kanade2018learning}.
One nice feature of our results is that they are independent of the model of access to the
training data, which means that the results hold even if
debated access such as quantum random access memory is used.
Finally, we note that our results do not assume any prior
knowledge on the function to be learned, which allows us to
make statements on virtually every possible learning algorithm, including neural networks.
Under stronger assumptions (e.g., more knowledge of the target function) the dependency
of the accuracy in terms of samples can be derived.

We summarise the results of our analysis in Table~\ref{fig:summary_algos}.
We omit the $\kappa$ dependency which would generally decrease the performance
of the quantum algorithms further.
Notably, while this is classically not an issue due to preconditioning, no
efficient general quantum preconditioning algorithm exists.
We note that \cite{clader2013preconditioned} introduced a SPAI preconditioner, however without providing
an efficient quantum implementation for its construction and without any performance analysis.
We additionally note that recently~\cite{tong2020fast} proposed a different mechanism for constructing efficient preconditioners, called fast inversion.
The main idea is based on the fact that fast inversion of 1-sparse matrices can
be done efficiently on a quantum computer. The algorithm works for matrices of the
form $A+B$, where $\norm{B} = \norm{A^{-1}} = \norm{(A+B)^{-1}}= O(1)$, and $A$
can be inverted fast. It results
in a condition number of the QLSA of $\kappa(M(A+B))=\kappa(I+A^{-1}B)$ once
the preconditioner $M=A^{-1}$ is applied.

\begin{figure}
\begin{center}
    \begin{tabular}{| l | l | l | l |  }
    \hline
    \rowcolor{Gray}
     & Algorithm & Train time & Test time \\
    \hline
    Classical & LS-SVM / KRR & $n^3$ & $n$ \\
    & KRR~\cite{yang2017randomised,ma2017diving,gonen2016solving,avron2017faster,fasshauer2012stable} & $n^2$ & $n$ \\
     & Divide and conquer~\cite{zhang2013divide} & $n^2$ & $n$\\
    & Nystr\"om~\cite{williams2001using,rudi2015less} & $n^2$ & $\sqrt{n}$\\
    & FALKON \cite{rudi2017falkon} & $n \sqrt{n}$ & $\sqrt{n}$ \\
     \hline
    Quantum & QKLS / QKLR \cite{CGJ19}& $ \sqrt{n}$ & $ n \sqrt{n}$ \\
    & QSVM \cite{rebentrost2014quantum} & $n \sqrt{n}$ & $n^2 \sqrt{n}$ \\
    \hline
    \end{tabular}
\end{center}
\caption{Summary of time complexities for training and testing of different classical and quantum algorithms when statistical guarantees are taken into account. We omit $\polylog(n,d)$ dependencies for the quantum algorithms.
We assume $\epsilon = \Theta(1/\sqrt{n})$ and count the effects of measurement errors.
The acronyms in the table refer to: least square support vector machines (LS-SVM), kernel ridge regression (KRR), quantum kernel least squares (QKLS), quantum kernel linear regression (QKLR), and quantum support vector machines (QSVM).
Note that for quantum algorithms the state obtained after training cannot be maintained or copied and the algorithm must be retrained after each test round.
This brings a factor proportional to the train time in the test time of quantum algorithms.
Because the condition number may also depend on $n$ and for quantum algorithms
this dependency may be worse, the overall scaling of the quantum algorithms may be slower than the classical.}
\label{fig:summary_algos}
\end{figure}

\section{Conclusion}

Quantum machine learning algorithms promise to be exponentially faster than their classical counter parts.
In this chapter, we showed by relying on standard results from statistical learning theory that such
claims are not well founded, and thereby rule out QML algorithms with
polylogarithmic time complexity in the input dimensions.
As any practically-used machine learning algorithms have polynomial runtimes,
our results effectively rule out the possibility of exponential advantages for
supervised quantum machine learning.
Although this holds for polynomial runtime classical algorithms, we note that our analysis
does not rule out an exponential advantage over classical algorithms with superpolynomial runtime.
Furthermore, since we do not make any assumptions on the hypothesis space $\Hyp$ of the learning
problem, we note that generally faster error rates are possible if more prior knowledge exists.
It is hence possible to obtain faster convergence rates than $1/\sqrt{n}$, which would imply
a potential quantum advantage for such problems.

As future directions, it is worth mentioning that it may be possible strengthen our results by analysing the $n$ dependency of the condition number. Previous results in this direction are discussed in~\cite{cucker2002mathematical,hochstadt2011integral}.

\chapter{randomised Numerical Linear Algebra}
\label{chap:randNLA}

The research question we answer in this chapter is the following:

\begin{researchquestion}[QML under the lens of RandNLA]
\label{rq:randNLA1}
Under the assumption of efficient sampling processes for the data for both classical and quantum algorithms,
what is the comparative advantage of the latter?
\end{researchquestion}

We will show how the requirement of a fast memory model can hide much of
the computational power of an algorithm.
In particular, by allowing a classical algorithm to sample according
to a certain probability distribution,
we can derive classical algorithms with computational complexities which
are independent of the input dimensions, and therefore only polynomially slower
compared to the best known quantum algorithm for the same task.
We use such Monte Carlo algorithms to construct
fast algorithms for Hamiltonian simulation, and also connect our research to
the recent so-called \textit{quantum-inspired} or \textit{dequatisation} results~\cite{tang2018quantum}.

We start by defining a range of memory models, which are used in quantum algorithms, randomised
numerical linear algebra, and quantum-inspired algorithms.
We continue by introducing the main ideas of Monte Carlo methods for numerical linear algebra,
and as an exemplary case study the randomised matrix multiplication algorithm to
grasp the main ideas of such approximate methods.
For a detailed introduction and overview, we refer the reader to the review of David Woodruff~\cite{woodruff2014sketching}.

We then show that such a memory structure immediately leads to faster simulation algorithms
for dense Hamiltonians in Sec.~\ref{ssec:dense_ham_sim}.

In the next step, we then use these ideas in Sec.~\ref{ssec:ham_sim_nystrom}
to construct symmetric matrices which
are approximations for the Hamiltonian and use these to perform fast
classical Hamiltonian simulation.
We do so by first finding a randomised low-rank approximation $H_k$
of the Hamiltonian $H$, and then applying a form of approximation to the
series expansion of the time evolution operator $\exp(-i H_k t)$.

We thereby show that the ability to sample efficiently from such matrices
immediately allows for faster classical algorithms as well.

Finally we briefly discuss how our results relate to the recent stream
of \textit{dequantisation} results, and show that indeed we can achieve
exponentially faster algorithms if we make similar assumptions.

\section{Introduction}
\label{sec:intro_randNLA}

In general, there exist two different approaches to performing approximate
numerical linear algebra operations, and a closer inspection shows
that there are many parallels between them.
The first stream is based on random sampling, also called sub-sampling, while the second stream is based on so-called random projections.
Random projections are themselves based on the John-Lindenstrauss transformation,
which allow us to embed vectors in a lower dimensional subspace while
preserving certain distance metrics.
Roughly speaking, random projections correspond to uniform sampling in randomly
rotated spaces.
The main differences between different approaches are the resulting error bounds,
where the best known ones are multiplicative relative-error bounds.

Taking as example the problem of a randomised low-rank approximation,
based on the randomised projector $\hat{P}_{k}$, which projects into some rank-$k$
subspace of a matrix, and denoting with $P_{k}$ the projector into the subspace
with containing the best rank-$k$ approximation of $A$ (e.g., the left eigenspace
of the top $k$ singular values of $A$), then we define additive-error bounds to be
of the form
\[
\norm{A - \hat{P}_{k}A}_F \leq \norm{A- P_{k}A}_F + \epsilon \norm{A}_F,
\]
since we have an additional error term to the best rank-$k$ approximation of $A$.
These bounds are typically weaker than the so-called
\textit{multiplicative-error bounds}, which take the form
\[
\norm{A - \hat{P}_{k}A} \leq  f(m,n,k,\ldots) \norm{A-P_{k}A},
\]
where $f(\cdot)$ denotes a function depending on the dimensions $m$ and $n$ of
the matrix $A \in \mathbb{R}^{m \times n}$, the rank, or other parameters.
The best known bounds on $f$ are independent of $m$ or $n$, and are referred to
as \textit{constant-factor bounds}.
The strongest of these bounds are given for $f=1+\epsilon$, for some error parameter
$\epsilon$, i.e.,
\[
\norm{A - \hat{P}_{k}A} \leq  (1+\epsilon) \norm{A-P_{k}A}.
\]

\section{Memory models and memory access}
\label{sec:memory_models}

In computer science, many different models for data access and storage are used.
Here we briefly recapitulate a number of memory models, which are common in
randomised numerical linear algebra, quantum algorithms, and quantum inspired (or dequantised) algorithms.
We will introduce the different memory structures, and define their properties.
Next, we will introduce the basic ideas behind numerical linear algebra on the example of matrix multiplication.

\subsection{The pass efficient model}
\label{ssec:pass_efficient_model}

Traditionally, in randomised numerical linear algebra, the so-called pass-efficient model was used to describe memory access.
In this model, the only access an algorithm has to the data is via a so-called pass,
which is a sequential read of the entire input data set.
An algorithm is then called pass-efficient if it uses only a small or a
constant number of passes over the data,
while it can additionally use RAM space and additional computation sublinear
with respect to the data stream to compute the solution.

The data storage can take several forms, as for example one could store only the index-data pairs as a sparse-unordered representations of the data, i.e., $((i,j),A_{ij})$
for all non-zero entries of the matrix $A$.
In practice, these types of storage are for example implemented in the Intel MKL
compressed sparse row (CSR) format, which is specified by four arrays (values,
columns, pointer 1, and pointer 2).

Stronger results, however, can sometimes be obtained in different input models,
yet these memory models are generally inaccessible and hence should be used
with care.

In order to efficiently sample from the data set using this memory model,
we need to be able to select random samples in a pass-efficient manner.
In order to do so, we can rely on the so-called SELECT algorithm, which is presented below.

\begin{algorithm}
\caption{The SELECT Algorithm \cite{drineas2006fast}}
\label{select_algo}
\begin{flushleft}
        \textbf{INPUT:}  $\{a_1,\ldots,a_n\}$, $a_i \geq 0$, i.e., one sequential read over the data.\\
        \textbf{OUTPUT:}   $i^*,a_{i^*}$.
\end{flushleft}
\begin{algorithmic}[1]
\STATE $D=0$
\FOR{\texttt{i=1 to n}}
    \STATE $D = D+a_i$
    \STATE With probability $a_i/D$, let $i^* = i$ and let $a_{i^*} = a_i$.
\ENDFOR
\RETURN $i^*,a_{i^*}$.
\end{algorithmic}
\end{algorithm}

The SELECT algorithm has then the following properties.
\begin{lemma}[\cite{mahoney2016lecture}]
  Suppose that $\{a_1, \ldots, a_n\}, a_i \geq 0$, are read in one pass, then the SELECT algorithm returns the index $i^*$ with probability $\mathbb{P}[i^*=i] = a_i/\sum_{j=1}^n a_{j}$,
  and requires $O(1)$ additional storage space.
\end{lemma}
\begin{proof}
  The proof is by induction. For the base case we have that the first element $a_1$ with $i^*=1$ is selected with probability $a_1/a_1=1$.
  The induction step is then performed by letting $D_k = \sum_{j=1}^k a_j$, i.e., the first $k$ elements have been read, and the algorithm reads the element $k+1$.
  Hence the probability to have selected any prior $i^*=i$ is $\mathbb{P}[i^*=i]= a_i/D_k$. Then the algorithm selects the index $i^*= k+1$ with probability $a_{k+1}/D_{k+1}$,
  and retains the previous selection otherwise, which is done for $i \in [k]$ with probability $\mathbb{P}[i^*=i] = \mathbb{P}(\text{Entry was selected prior})*\mathbb{P}(\text{Current entry was not selected})$.
  Since by the induction hypothesis it holds that the any entry prior was selected with $\mathbb{P}[i^*=i] = \frac{a_i}{D_k}$ and the probability that the new entry was not selected is $\left(1-\frac{a_{k+1}}{D_{k+1}}\right)$,
  hence we have $\mathbb{P}[i^*=i] =  \frac{a_i}{D_k}\left(1-\frac{a_{k+1}}{D_{k+1}} \right)= \frac{a_i}{D_{k+1}}$. By induction this result holds for all $i$ and hence $l+1=n$.
  The storage space is limited to the space for keeping track of the sum, which is $O(1)$, and hence concludes the proof.
\end{proof}

It is important to note that we can therefore use the SELECT algorithm to perform
importance sampling according to distributions over the rows $A_{i}$ or columns
$A^{i}$ of a matrix $A$, as for example if we want to sample according to
the probabilities
\[
\mathbb{P}[i^*=i] = \norm{A_i}_2^2/\norm{A}_F^2.
\]
We will use this sampling scheme to demonstrate how we can already obtain very
fast algorithms in randomised linear algebra.
In particular, we will use it to derive one of the most fundamental
randomised algorithms in Section~\ref{sec:basic_mm}: The basic matrix multiplication algorithm.

\subsection{Quantum random access memory}
\label{ssec:qram}

A classical random-access memory (RAM) is a device that stores the content of a memory location in a
memory array. A random-access memory importantly allows data items to be read or written in almost the
same amount of time irrespective of the physical location of data inside the memory.
Practically this is due to a binary tree which allows the bus to traverse the memory consisting
of $N$ elements in $\log(N)$ computational steps by traversing the tree.
Quantum random access memory (qRAM) in a similar fashion allows us to access and load data
in superposition from all the memory sites.
qRAMs with $n$-bit addresses can therefore access $2^n$ memory sites and hence
require $\Ord{2^n}$ two-bit logic gates.
Note that we, with abuse of previous notation, define therefore the dimensions
of a Hermitian matrix in this chapter with $N$, unlike in previous chapters.
We do this to comply with the standard notation in the quantum computing literature.
Using, e.g.\ the so-called \textit{bucket brigade} architecture the number of
two-qubit physical interactions during each qRAM call can then in principle reduced to
$\Ord{n}$~\cite{giovannetti2008qram1,giovannetti2008qram2},
which is hence polylogarithmic in the input dimension,
assuming a data array of size $N=2^n$.

A variety of different qRAM architectures have been proposed, and we will here mainly focus on the specific architecture introduced by Kerenidis and Prakash~\cite{kerenidis2016quantum}.
We will use this algorithm later on in context of Hamiltonian simulation for
dense Hamiltonians on a quantum computer.
In the next subsection, we will also show that such a memory model has a powerful classical
correspondence which can be used to construct much faster randomised algorithms for numerical
linear algebra.
In the following, we introduce an adaptation of the architecture proposed
in~\cite{kerenidis2016quantum}, which is suitable for the application that we
will mainly study in this chapter, namely Hamiltonian simulation.
Notably, while most quantum memory structures store the squares of the
entries (e.g., of the matrix or data in general), the here presented one
stores the absolute values only.

\begin{definition}[Quantum Data Structure]
  \label{def:datastructure}
  Let $H\in\mathbb{C}^{N\times N}$ be a Hermitian matrix (where $N = 2^n$), $\pnorm{1}{H}$ being the maximum absolute row-sum norm, and $\sigma_j:= \sum_k \lvert H_{jk} \rvert$. Each entry $H_{jk}$ is represented with $b$ bits of precision. Define $D$ as an array of $N$ binary trees $D_j$ for $j \in \{0, \ldots, N-1\}$. Each $D_j$ corresponds to the row $H_j$, and its organization is specified by the following rules.
  \begin{enumerate}
    \item The leaf node $k$ of the tree $D_j$ stores the value\footnote{Note that the conjugation here is necessary. See Eq.~\eqref{eq:Hij}.} ${H_{jk}^{*}}$ corresponding to the index-entry pair $(j,k,H_{jk})$.
	\item For the level immediately above the bottom level, i.e., the leaves, and any node level above the leaves, the data stored is determined as follows: suppose the node has two children storing data $a$, and $b$ respectively (note that $a$ and $b$ are complex numbers). Then the entry that is stored in this node is given by $(|a|+|b|)$.
  \end{enumerate}
\end{definition}
We show an example of the above data structure in Fig.~\ref{fig:bt-4}.
As we immediately see from the Definition~\ref{def:datastructure}, while
each binary tree $D_j$ in the data structure contains a real number in each internal
(non-leaf) node, the leaf nodes store a complex number.
The root node of each tree $D_j$ then store the absolute column norm, i.e., the value
$\sum_{k=0}^{N-1}|H_{jk}^{*}|$, and we can therefore calculate the value
$\pnorm{1}{H}-\sigma_j$ in constant time.
Furthermore, the $\pnorm{1}{H}$ can be obtained as the maximum among all the roots of the
binary trees, which can be done during the construction of the data structure, or through
another binary search through the tree structure which ends in the $D_j$'s.

\begin{figure}[ht]
    \centering
    \includegraphics[width=0.8\textwidth]{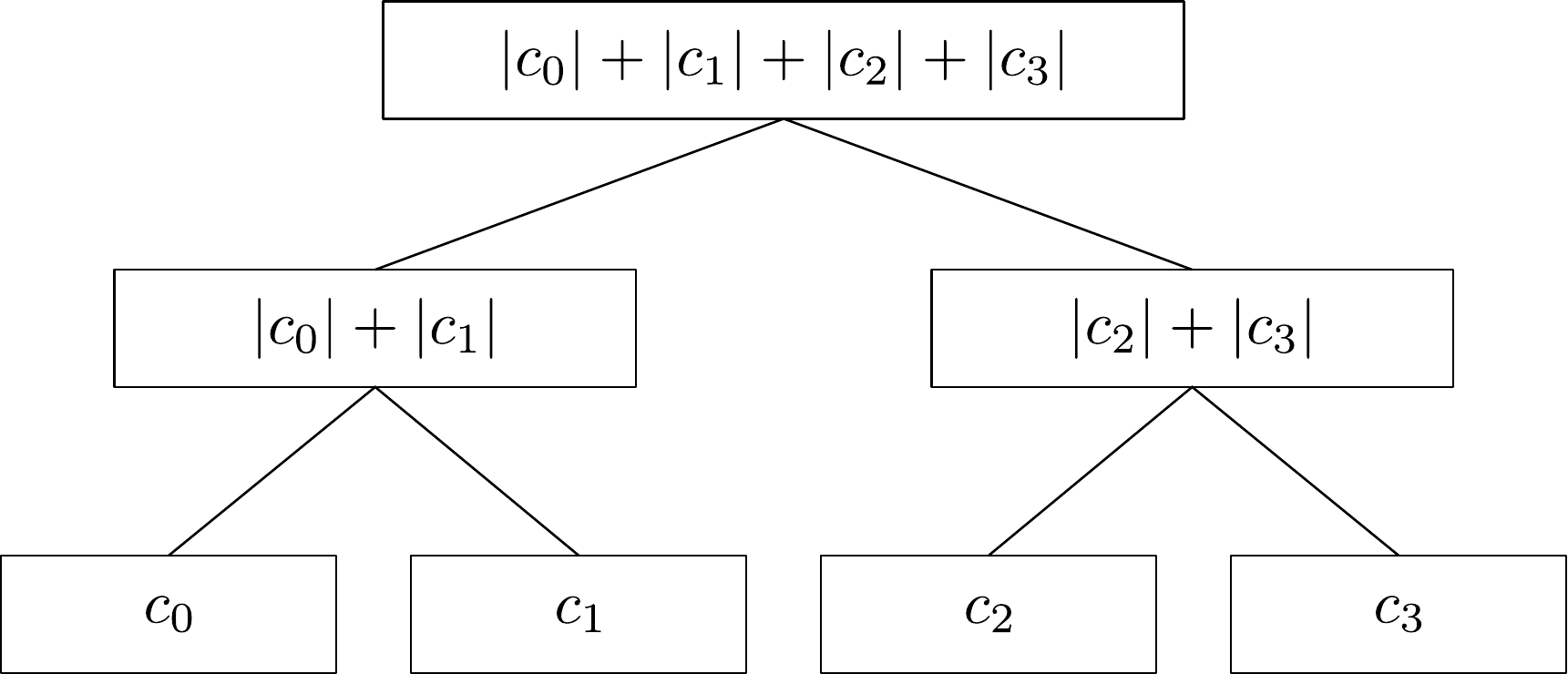}
	\caption{An example of the data structure that allows for efficient state preparation using a logarithmic number of conditional rotations.}
    \label{fig:bt-4}
\end{figure}

With this data structure, we can efficiently perform a state preparation which we will
require in the Hamiltonian simulation algorithm in Sec.~\ref{ssec:dense_ham_sim}.
Notably, as we will see in the next subsection,
such a fast quantum memory has a classical equivalent which allows us to sample
very efficiently from the data.

There are several challenges with such a memory architecture, including a large overhead
from the controlled operations, and possible challenges from decoherence.
For a discussion of some caveats of such qRAM architectures, we point the reader
to the reviews~\cite{adcock2015advances} or \cite{ciliberto2018quantum}.
An important notion that we want to mention here is the separation between quantum algorithms
which use quantum data (i.e.\ data which is accessible in superposition) and ones which
are operating on classical data which needs to be accessed efficiently to allow
polylogarithmic runtimes.
In general most quantum algorithms are discussed in the so-called query model, i.e.,
the computational complexity is given in terms of numbers of queries of an oracle.
The actual oracle could then be another quantum algorithm, or a memory model such as qRAM.
The computational time is then given by the query complexity times the time for each query
call, which could for example be $T_{query} \times \log(N)$ if we call the qRAM $T_{query}$
times.
If a quantum algorithm requires qRAM or a similar solution, then we need to
understand the associated limitations and assumptions and assess these with care.
Particularly, the superiority w.r.t.\ classical algorithms is here in question when
comparing to classical PRAM machines \cite{steiger2016perimeter}.
Results that consider error probabilities in certain quantum RAM architectures further
indicate that a fully-error-corrected quantum RAM is necessary to maintain speedups, which
might result in huge overheads and questionable advantages~\cite{regev2008impossibility,arunachalam2015robustness}.

\subsection{Quantum inspired memory structures}
\label{ssec:qi_memory}
Recent quantum inspired algorithms rely on a data structure
which is very similar to the qRAM used in many QML algorithms.
While QML algorithms rely on qRAM or 'quantum access' to data
in order to allow for state preparation with linear gate count but
polylogarithmic depth, quantum inspired algorithms achieve their
polylogarithmic complexities through \textit{sampling and query access} of the data via a
\textit{dynamic data structure}~\cite{chia2020sampling}.

The \textit{Sampling and query access} model can be
thought of as a classical analogue of the above introduced
qRAM model for state preparation.
The ability to prepare a state $\ket{v}$ which is proportional
to some input vector $v \in \mathbb{C}^{N}$ (such as a column of $H \in \mathbb{C}^{N \times M}$)
from memory is equivalent to the ability make the following queries:
\begin{enumerate}
    \item Given an index $i \in [N]$, output the corresponding entry $v_i$ of the vector $v$.
    \item Sample the index $j \in [N]$ with probability $\lvert v_j \rvert^2/\norm{v}_2^2$.
    \item Output the spectral norm of the vector $\norm{v}_2$.
\end{enumerate}
For a matrix $H \in \mathbb{C}^{M \times N}$ this extends to the ability
to perform the following queries:
\begin{enumerate}
    \item Sample and query access for each vector $H_i$, i.e., for all rows $i \in [N]$ we
    can output each entry, sample indices with probability proportional to the magnitude of the entry,
    and output its spectral norm.
    \item Sample and query access to the vector $h$ with $h_i = \norm{H_i}$, i.e., the vector of row norms.
\end{enumerate}
Notably, if the input data is given in a classical form,
classical algorithms can be run efficiently in the sampling and query model
whenever the corresponding quantum algorithms require qRAM access to the data,
and both all state preparations can be performed efficiently
(i.e., be performed in logarithmic time in the input dimensions).

The classical data structure which enables such sampling and query access
to the data is similar to the one described in Fig.~\ref{fig:bt-4}.
It stores a matrix $H \in \mathbb{C}^{M \times N}$, again in form of
a set of binary trees, where each tree contains the absolute values
(or to be precise, the square of the absolute values) of the entries
of one row or column of the input matrix $H$.
The (time) cost of a query to any entry in $H$ is then $O(1)$ and
sampling can be performed in time $O(\log(MN))$ for any entry $H \in \mathbb{C}^{M \times N}$.
We summarise this data structure below in Definition~\ref{def:classical_datastructure}.

\begin{definition}[Classical Data Structure]
  \label{def:classical_datastructure}
  Let $H\in\mathbb{C}^{N\times M}$ be a Hermitian matrix (where $N = 2^n$), $\pnorm{F}{H}$ being the Frobenius norm,
  and $h_i=\pnorm{2}{H_i}$ being the spectral norm of row $i$ of $H$.
  Each entry $H_{jk}$ is represented with $b$ bits of precision.
  Define $D$ as an array of $N$ binary trees $D_j$ for $j \in \{0, \ldots, N-1\}$.
  Each $D_j$ corresponds to the row $H_j$, and its organization is specified by the following rules.
  \begin{enumerate}
    \item The leaf node $k$ of the tree $D_j$ stores the value $\frac{H_{jk}}{\abs{H_{jk}}}$ corresponding to the index-entry pair $(j,k,H_{jk})$.
	\item For the level immediately above the bottom level, i.e., the leaves, the entry $k$ of the tree $D_j$ is given by $\abs{H_{jk}}^2$
	\item For any node level above the leaves, the data stored is determined as follows: suppose the node has two children storing data $a$, and $b$ respectively (note that $a$ and $b$ are squares values of complex numbers). Then the entry that is stored in this node is given by $(a+b)$.
  \end{enumerate}
  The root nodes of the binary tree $D_j$ is then given by $h_j$, and the memory structure is completed by
  applying the same tree structure where now the leaves $j$ of the tree are now given by the root nodes of $D_j$.
\end{definition}
Figure~\ref{fig:cl-mem-st} demonstrates an example of this structure.

\begin{figure}[ht]
    \centering
    \includegraphics[width=0.8\textwidth]{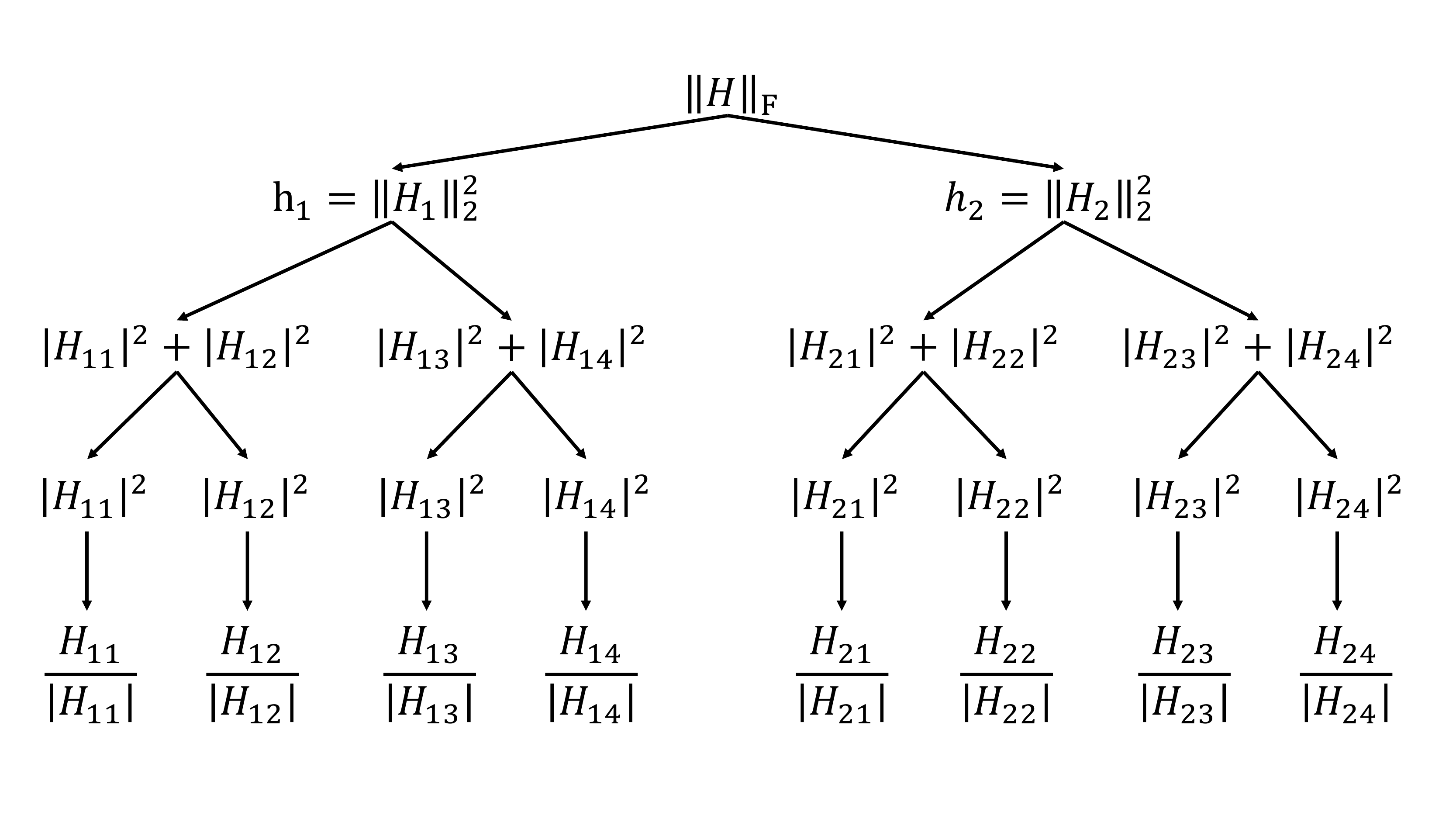}
	\caption{An example of the classical (dynamic) data structure that enables efficient
	sample and query access for the example of $H \in \mathbb{C}^{2 \times 4}$.}
    \label{fig:cl-mem-st}
\end{figure}

Although we do not explicitly rely on such a data structure for the results of section~\ref{ssec:ham_sim_nystrom},
since we use a less restrictive requirement which we call row-computability and row-searchability.
However, the data structure in Def.~\ref{def:classical_datastructure} immediately allows us to perform these
operations.
We will in this thesis just refer to this requirement as \textit{query and sample access},
since it has been established as the commonly used term in the dequantisation literature and
suffices in practice. There are generally exceptions, for example, since row-computability and row-searchability can also be achieved through structural properties of the matrix which would not require the memory structure.

\section{Basic matrix multiplication}
\label{sec:basic_mm}
In this section we will introduce some of the main concepts from randomised numerical linear algebra
on the example of the randomised matrix multiplication algorithm.

We will denote the $i$-th column of some matrix $A$ by $A^i$ and the $j$-th row of $A$ by $A_j$.
Let in the following $A \in \mathbb{R}^{m \times n}$ and $B \in \mathbb{R}^{n \times p}$.
With this notation, recall that $(AB)_{ij} = \sum_{j=1}^n A_{ij} B_{jk} = A_i B^j$, and that we
can write the product $AB= \sum_{i=1}^n A^i B_i$ via the sum of outer products, i.e., a sum
of rank-$1$ matrices, where $A^i B_i \in \mathbb R^{m \times p}$.
The representation of the product $AB$ into the sum of outer products implies that
$AB$ could be decomposed into a sum of random variables, which, if appropriately chosen,
would in expectation result in the product $AB$.
This suggests that we could sample such terms to approximate the product,
specifically we could use an approximation of the form
\begin{equation}
  \label{eq:matrix_multiplication}
  AB = \sum_{i=1}^n A^i B_i \approx \frac{1}{c} \sum_{t=1}^c \frac{A^{i_t} B_{i_t}}{p_{i_t}} = CR,
\end{equation}
where $\{p_i\}_{i=1}^n$ are the sampling probabilities.
We could do this via random uniform sampling, but this would lead to a very high variance.
Hence we will further optimally like to optimise the sampling probabilities to obtain a good result.\\

In general, instead of working with probabilities, it is easier to work with matrices, and we can
define a standardized matrix notation called \textit{sampling matrix formalism} according to~\cite{mahoney2016lecture}.
Notably, we can treat many other approaches in the framework of matrix-multiplication, and will therefore also
rely on it in the following.
For this we let $S \in \mathbb{R}^{n \times c}$ be a matrix such
that
\begin{equation}
  S_{ij} := \left\lbrace \begin{array}{c c}
  1 & \text{ if the } i \text{-th column of }A \text{ is chosen in the }j\text{-th independent trial}\\
  0 & \text{ otherwise,}\\
\end{array} \right.
\end{equation}
and let $D \in \mathbb{R}^{c \times c}$ be a diagonal matrix such that
\begin{equation}
  D_{tt} = 1/\sqrt{c p_{i_t}}.
\end{equation}
Using this, we can write the output of the sampling process described
in Eq.~\ref{eq:matrix_multiplication} via the simplified notation
\begin{equation}
  CR = ASD(SD)^TB \approx AB.
\end{equation}
The matrix multiplication algorithm is given below in Algorithm~\ref{algo:basic_matrix_multiplication},
and below we will next give a proof of its correctness.

\begin{algorithm}
\caption{The BASIC-MATRIX-MULTIPLICATION Algorithm \cite{drineas2006fast}}
\begin{flushleft}
        \textbf{INPUT:} $A \in \mathbb{R}^{m \times n}$, $B \in \mathbb{R}^{n \times p}$, integer $c>0$ and sampling probabilities $\{p_i\}_{i=1}^n$.\\
        \textbf{OUTPUT:} $C$ and $R$ s.t. $CR \approx AB$.
\end{flushleft}
\begin{algorithmic}[1]
  \label{algo:basic_matrix_multiplication}
\STATE $C,R$ as all-zeros matrices.
\FOR{\texttt{i=1 to c}}
        \STATE Sample an index $i_t \in \{ 1, \ldots, n\}$ w. prob. $\mathbb{P}[i_t =k]=p_k$, i.i.d with replacement.
        \STATE Set $C^t = A^{i_t}/\sqrt{c p_{i_t}}$ and $R_t = B_{i_t}/\sqrt{c p_{i_t}}$.
      \ENDFOR
\RETURN $C$ and $R$.
\end{algorithmic}
\end{algorithm}

In the following analysis we will use the fact that for rank-$1$ matrices it holds
that $\norm{A^i B_i}_2 = \norm{A^i}_2 \norm{B_i}_2$.
In order to see why this is the case observe that
\[
\norm{A^i B_i}_2  = \sqrt{(A^iB_i)^T(A^iB_i)} = \sqrt{B_i^T(A^i)^T A^iB_i},
\]
but this is simply the inner product and hence
\[
\sqrt{\sigma_{max} (\norm{B_i}_2^2 \norm{A^i}_2^2)} = \sqrt{\norm{B_i}_2^2 \norm{A^i}_2^2}= \norm{B_i}_2 \norm{A^i}_2.
\]

We now prove a Lemma that states that $CR$ is an unbiased estimator for $AB$, element-wise,
and calculate the variance of that estimator. This variance is strongly dependent on the sampling probabilities.
Based on this, we can then derive optimal sampling probabilities which minimise the variance.
In practice these will not be accessible and one typically needs to find approximations for these.
In practice, we therefore use other distributions, such as the row or column norms, or so-called leverage scores~\cite{mahoney2016lecture,woodruff2014sketching}.

\begin{lemma}[\cite{drineas2006fast}]
  \label{lem:CR_exp}
  Given two matrices $A \in \mathbb{R}^{m \times n}$ and $B \in \mathbb{R}^{n \times p}$, construct matrices $C$ and $R$ with the matrix multiplication algorithm from above.
  Then,
  \begin{align}
  & \mathrm{E}{[(CR)_{ij}]} = (AB)_{ij}, \\
  & \mathrm{Var}[(CR)_{ij}] = \frac{1}{c} \sum_{k=1}^n \frac{A_{ik}^2 B_{kj}^2}{p_k} - \frac{1}{c} (AB)_{ij}^2
  \end{align}
\end{lemma}
\begin{proof}
  Fix an index-pair $(i,j)$ and define the random variable $X_t = \left( \frac{A^{i_t} B_{i_t}}{c p_{i_t}} \right)_{ij} = \left( \frac{A_{ii_t} B_{i_tj}}{c p_{i_t}} \right)$,
  and observe that $(CR)_{ij}=\sum_{t=1}^c X_t$. Then we have that $\mathrm{E}[X_t] = \sum_{k=1}^n p_k \frac{A_{ik} B_{kj}}{c p_k} = \frac{1}{c} (AB)_{ij}$,
  and $\mathrm{E}[X_t^2] = \sum_{k=1}^n \frac{A_{ik}^2 B_{kj}^2}{c^2 p_k}$. Furthermore we just need to sum up in order to obtain $\mathrm{E}[(CR)_{ij}]$,
  and hence obtain \[\mathrm{E}[(CR)_{ij}] = \sum_{t=1}^c \mathrm{E}[X_t] = (AB)_{ij},\] and $\mathrm{Var}[(CR)_{ij}] = \sum_{t=1}^c \mathrm{Var}{[X_t]}$, where
  we can determine the variance of $X_t$ easily from $\mathrm{Var}{[X_t]} = \mathrm{E}[X_t^2] - \mathrm{E}[X_t]^2$, and the lemma follows.
\end{proof}

We can use this directly to establish the expected error in terms of the Frobenius norm, i.e., $\mathrm{E}[\norm{AB - CR}_F^2]$,
by observing that we can treat this as a sum over each of the individual entries, i.e.,
\[
\mathrm{E}[\norm{AB - CR}_F^2] = \sum_{i=1}^m \sum_{j=1}^p \mathrm{E} \left[ (AB - CR)_{ij}^2 \right].
\]
Doing so, we obtain the following result.
\begin{lemma}[Basic Matrix Multiplication \cite{drineas2006fast}]
  \label{lem:basic_matrix_multiplication}
  Given matrices $A$ and $B$, construct matrices $C$ and $R$ with the matrix multiplication algorithm from above. Then it holds that,
  \begin{equation}
  \mathrm{E}{[\norm{AB-CR}_F^2]} = \sum_{k=1}^n \frac{\norm{A^{k}}_2^2 \norm{B_{k}}_2^2}{c p_k} - \frac{1}{c} \norm{AB}_{F}^2
  \end{equation}
  Furthermore, if
  \[
  p_k =  p_k^{optimal} = \frac{\norm{A^{k}}_2 \norm{B_{k}}_2}{\sum_{k'} \norm{A^{k'}}_2 \norm{B_{k'}}_2},
  \]
  are the sampling probabilities used, then
  \begin{equation}
    \label{eq:result_matrix_mult}
  \mathrm{E}{[\norm{AB-CR}_F^2]} =  \frac{ \left(\sum_{k=1}^n  \norm{A^{k}}_2 \norm{B_{k}}_2 \right)^2}{c} - \frac{1}{c} \norm{AB}_{F}^2.
  \end{equation}
\end{lemma}
\begin{proof}
  As mentioned above, first note that \[\mathrm{E}[\norm{AB - CR}_F^2] = \sum_{i=1}^m \sum_{j=1}^p \mathrm{E} \left[ (AB - CR)_{ij}^2 \right].\]
  Oberve that the squared terms give $(AB -CR)^2_{ij} = (AB)_{ij}^2 - (AB)_{ij}(CR)_{ij} - (CR)_{ij}(AB)_{ij} + (CR)_{ij}^2$, and using the results for $\mathrm{E}[(CR)_{ij}]$
  from Lemma~\ref{lem:CR_exp} and since $\mathrm{E}[(ABCR)_{ij}] = (AB)_{ij}^2$, we have
  \[\mathrm{E}[(AB -CR)^2_{ij}] = \mathrm{E}[(CR)_{ij}^2] - (AB)_{ij}^2= \mathrm{E}[(CR)_{ij}^2] - \mathrm{E}[(CR)_{ij}]^2 = \mathrm{Var}[(CR)_{ij}].\]
  Therefore we have $\mathrm{E}[\norm{AB - CR}_F^2] = \sum_{i=1}^m \sum_{j=1}^p \mathrm{Var}[(CR)_{ij}]$. Using the result from Lemma~\ref{lem:CR_exp} again for the
  variance, we then directly obtain the result by using that $\left( \sum_i A_{ik}^2\right) = \norm{A^k}_2^2$ and $\left( \sum_j B_{kj}^2\right) = \norm{B_k}_2^2$.
\end{proof}
The sampling probabilities that we have used here are optimal in the sense that they minimize the expected error. To show this we can define the function
\[f(\{p_i\}_{i=1}^n) = \sum_{k=1}^n \frac{\norm{A^{k}}_2^2 \norm{B_{k}}_2^2}{p_k},\] which captures the $p_k$-dependent part of the error.
In order to find the optimal probabilities we minimise this with the constraint that $\sum_k p_k = 1$, i.e., we define the function $g = f(\{p_i\}_{i=1}^n) + \lambda \left(\sum_k p_k -1 \right)$.
Setting the derivative of this function w.r.t. the $p_k$ to zero and correctly normalising them then gives the probabilities in Lemma~\ref{lem:basic_matrix_multiplication}.
Note a few important points about this result.
\begin{itemize}
  \item Note the results depends strongly on $\{p_k\}_{k=1}^n$. In particular we need a strategy to obtain these probabilities.
  \item With Markov's inequality we can remove the expectation (using the assumption that $p_k \geq \beta * p_k^{optimal}$).
  To do this we reformulate Eq.~\ref{eq:result_matrix_mult} and incorporate the second term such that we obtain for some $1 \geq \beta >0$ and hence nearly optimal probabilities for $\beta$ sufficiently close to $1$,
  \begin{equation}
    \label{eq:new_result_matrix_mult}
    \mathrm{E}{[\norm{AB-CR}_F^2]} = \frac{\left(\sum_{k=1}^n \norm{A^{k}}_2 \norm{B_{k}}_2 \right)^2}{\beta c} - \frac{1}{c} \norm{AB}_{F}^2 \leq \frac{1}{\beta c} \norm{A}_F^2 \norm{B}_F^2 ,
  \end{equation}
  where we neglect the last term and used the Cauchy-Schwarz inequality, i.e.,
  \[
  \left(\sum_{k=1}^n \norm{A^{k}}_2 \norm{B_{k}}_2 \right)^2 \leq \sum_{k=1}^n \norm{A^{k}}^2_2 \sum_{k=1}^n \norm{B_{k}}^2_2.
  \]
  Now we can apply Jensen's inequality such that $\mathrm{E}{[\norm{AB-CR}_F]} \leq \frac{1}{\sqrt{\beta c}} \norm{A}_F \norm{B}_F$.
  Hence if we take the number of samples $c \geq 1/(\beta \epsilon^2)$, then we obtain a bound of the error $\mathrm{E}{[\norm{AB-CR}_F]}  \leq \epsilon \norm{A}_F \norm{B}_F$.
  We can now use Markov's inequality to \textit{remove the expectation} from this bound, and in some cases this will be good enough.
  Let
  \[
  \mathbb{P} \left[ \norm{AB-CR}_F  > \frac{\alpha}{\sqrt{\beta c}} \norm{A}_F \norm{B}_F \right],
  \]
  we obtain using Markov's inequality that
  \[
  \delta \leq  \frac{\mathrm{E}{[\norm{AB-CR}_F]}}{\frac{\alpha}{\sqrt{\beta c}} \norm{A}_F \norm{B}_F} \leq \frac{1}{\alpha}.
  \]
  And hence with probability $\geq 1-\delta$ if $c \geq \beta/ \delta^2 \epsilon^2$ we have that $\norm{AB-CR}_F \leq  \epsilon \norm{A}_F \norm{B}_F$.
\end{itemize}

While the above bounds are already good, if we desire a small failure probability $\delta$,
the number of samples grow too rapidly to be useful in practice.
However, the above results can be exponentially improved using much stronger Chernoff-type bounds.
These indeed allow us to reduce the number of samples to $O(\log(1/\delta))$, rather than $\mathrm{poly}(1/\delta)$, and therefore lead to much more practically useful bounds.
We do not go into much detail here, as we believe the main ideas are covered
above and the details are beyond the scope of this thesis.
Next, we first use the quantum random access memory from Def.~\ref{def:datastructure} to
derive a fast quantum algorithm for Hamiltonian simulation, and then show how randomised
numerical linear algebra can be used to design fast classical algorithms for the same task.

\section{Hamiltonian Simulation}
In this section, we will now derive two results.
First, we will derive a quantum algorithm for Hamiltonian simulation, which is based on the qRAM
model that we described above in Def.~\ref{def:datastructure}.
Next, we will design a classical version for Hamiltonian simulation which is independent of
the dimensionality which depends on the Frobenius norm, and therefore on the spectral norm and the rank of the
input Hamiltonian. The classical algorithms relies on the classical data structure~\ref{def:classical_datastructure}.

\subsection{Introduction}
Hamiltonian simulation is the problem of simulating the dynamics of quantum systems, i.e., how a quantum system evolves over time.
Using quantum computers to describe these dynamics was the original motivation by Feynman for quantum computers~\cite{feynman1982simulating,feynman1986quantum}.
It has been shown that quantum simulation is BQP-hard, and therefore it was conjectured that no classical algorithm can solve
it in polynomial time, since such an algorithm would be able to efficiently solve any problem which can be solved efficiently with a quantum algorithm, including integer factorization~\cite{shor1999polynomial}.

One important distinction between the quantum algorithm which we present in this thesis, and quantum algorithms
which have been traditionally developed is the input model.
A variety of input models have been considered in previous quantum algorithms for simulating Hamiltonian evolution. The \emph{local Hamiltonian} model is defined by a number of local terms of a given Hamiltonian.
On the other hand, the \emph{sparse-access} model for a Hamiltonian $H$ with sparsity $s$, i.e., the number of non-zero
entries per row or column, is specified by the following two oracles:
    \begin{align}
      O_{S} \ket{i,j} \ket{z} &\mapsto \ket{i,j} \ket{z \oplus S_{i,j}}, \text{ and} \\
      \label{eq:oh}
      O_{H} \ket{i,j} \ket{z} &\mapsto \ket{i,j} \ket{z \oplus H_{i,j}},
    \end{align}
    for $i,j \in [N]$, where $S_{i,j}$ is the $j$-th nonzero entry of the $i$-th row of $H \in \mathbb{R}^{N \times N}$ and $\oplus$ denotes the bit-wise XOR. Note that $H$ is a Hermitian matrix, and therefore square.
    The \emph{linear combination of unitaries} (\emph{LCU}) model decomposes the Hamiltonian into a linear combination of unitary matrices and we are given the coefficients and access to an implementation of each unitary.

The first proposal for an implementation of Hamiltonian simulation on a quantum computer came from Lloyd~\cite{lloyd1996universal}, and was based on local Hamiltonians.
Later, Aharonov and Ta-Shma described an efficient algorithm for an arbitrary
sparse Hamiltonian~\cite{aharonov2003adiabatic}, which was only dependent on the sparsity instead of the dimension $N$.
Subsequently, a wide range of algorithms have been proposed, each improving the runtime~\cite{berry2007efficient,berry2009black,berry2017exponential,childs2010relationship,childs2012hamiltonian,poulin2011quantum,wiebe2011simulating,low2016hamiltonian,berry2016corrected,low2017hamiltonian}.
Most of these algorithms have been defined in the sparse-access model, and have lead to optimal dependence on all (or nearly all) parameters for sparse Hamiltonians over the recent years~\cite{berry2015hamiltonian,low2017optimal,low2018hamiltonian}.

The above-mentioned input models are highly relevant when we want to simulate a physical system, for example to obtain
the ground state energy of a small molecule. However, these models are not generally used in modern machine learning
or numerical linear algebra algorithms, as we have seen above. Here it can be more convenient to work with access to
a \emph{quantum random access memory} (\emph{qRAM}) model, which we have described in Def.~\ref{def:datastructure}.
In this model, we assume that the entries of a Hamiltonian are stored in a binary tree data structure
\cite{kerenidis2016quantum}, and that we have quantum access to the memory.
Here, quantum access implies that the qRAM is able to efficiently prepare quantum states corresponding to the input data.
The use of the qRAM model has been successfully demonstrated in many applications such as quantum principal component analysis~\cite{lloyd2014quantum}, quantum support vector machines~\cite{rebentrost2014quantum}, and quantum recommendation systems~\cite{kerenidis2016quantum}, among many other quantum machine learning algorithms.

In contrast to prior work, we consider the qRAM model in order to simulate not necessarily sparse Hamiltonians.
The qRAM allows us to efficiently prepare states that encode the rows of the Hamiltonian.
Using the combination of this ability to prepare states in combination with a quantum walk~\cite{berry2015hamiltonian},
we derive the \emph{first} Hamiltonian simulation algorithm in the qRAM model whose time complexity has
$\widetilde{O}(\sqrt{N})$ dependence, where $\widetilde{O}(\cdot)$ hides all poly-logarithmic factors,
for non-sparse Hamiltonians of dimensionality $N$.
Our results immediately imply~\cite{childs2017quantum} a quantum linear system algorithm in the qRAM model
with square-root dependence on dimension and poly-logarithmic dependence on precision,
which exponentially improves the precision dependence of the quantum
linear systems algorithm by~\cite{wossnig2018quantum}.

The main hurdle in quantum-walk based Hamiltonian simulation is to efficiently prepare the states
which allow for a quantum walk corresponding to $e^{-iH/\norm{H}_1}$.
These states are substantially different to those which have been used in previous quantum algorithms~\cite{kerenidis2016quantum,berry2009black,berry2015hamiltonian}.
Concretely, the states required in previous algorithms, such as~\cite{berry2015hamiltonian} allow for a quantum walk
corresponding to $e^{-iH/(s\maxnorm{H})}$, where $s$ is the row-sparsity of $H$.
These states can be prepared with $O(1)$ queries to the sparse-access oracle.
However, the states we are required to prepare for the non-sparse Hamiltonian simulation algorithm cannot use
structural features such as sparsity of the Hamiltonian, and it is not known how to prepare such state
efficiently in the sparse-access model.
In the qRAM model on the other hand, we are able to prepare such states with time complexity (circuit depth)
of $O(\polylog(N))$, as we will demonstrate below.

Using the efficient state preparation procedure in combination with a linear combination of quantum walks,
we are able to simulate the time evolution for non-sparse Hamiltonians with only poly-logarithmic dependency on the precision.
The main result of this chapter is summarised in the following theorem, which we prove in Sec.~\ref{par:lcu}.
\begin{theorem}[Non-sparse Hamiltonian Simulation]
  \label{thm:densehamiltoniansim}
  Let $H\in\mathbb{C}^{N\times N}$ (with $N=2^n$ for a $n$ qubit system) be a Hermitian matrix stored in the data structure as specified in Definition~\ref{def:datastructure}. There exists a quantum algorithm for simulating the evolution of $H$ for time $t$ and error $\epsilon$ with time complexity (circuit depth)
  \begin{align}
    O\left(t\pnorm{1}{H} n^2\log^{5/2}(t\pnorm{1}{H}/\epsilon)\frac{\log(t\norm{H}/\epsilon)}{\log\log(t\norm{H}/\epsilon)}\right).
  \end{align}
\end{theorem}

Here $\pnorm{1}{\cdot}$ denotes the induced 1-norm (i.e.,~maximum absolute row-sum norm), defined as
$\pnorm{1}{H} = \max_j\sum_{k=0}^{N-1}|H_{jk}|$, $\norm{\cdot}$ denotes the spectral norm.
In the following we will also need the max norm $\maxnorm{\cdot}$, defined by $\maxnorm{H} = \max_{i,j}|H_{jk}|$.

Since it holds that $\pnorm{1}{H}\leq\sqrt{N}\norm{H}$ (see~\cite{CK10}), we immediately obtain the following corollary.
\begin{corollary} \label{cor:densehamiltoniansim}
  Let $H\in\mathbb{C}^{N\times N}$ (where $N=2^n$) be a Hermitian matrix stored in the data structure as specified in Definition~\ref{def:datastructure}. There exists a quantum algorithm for simulating the evolution of $H$ for time $t$ and error $\epsilon$ with time complexity (circuit depth)
  \begin{align}
	O\left(t\sqrt{N}\norm{H}\, n^2\log^{5/2}(t\sqrt{N}\norm{H}/\epsilon)\frac{\log(t\norm{H}/\epsilon)}{\log\log(t\norm{H}/\epsilon)}\right).
  \end{align}
\end{corollary}

\noindent
\textbf{Remarks:}
\begin{enumerate}
  \item As we can see from Corollary~\ref{cor:densehamiltoniansim}, the circuit depth scales as $\widetilde{O}(\sqrt{N})$.
  However, the gate complexity can in principle scale as $O(N^{2.5}\log^2(N))$ due to the addressing scheme
  of the qRAM, as defined in Definition~\ref{def:datastructure}.
  For structured Hamiltonians $H$, the addressing scheme can of course be implemented more efficiently.

  \item If the Hamiltonian $H$ is $s$-sparse, i.e., $H$ has at most $s$ non-zero entries in each row or column,
  then the time complexity (circuit depth) of the algorithm is given by
    \begin{align}
    	O\left(t\sqrt{s}\norm{H}\, n^2\log^{5/2}(t\sqrt{s}\norm{H}/\epsilon)\frac{\log(t\norm{H}/\epsilon)}{\log\log(t\norm{H}/\epsilon)}\right).
    \end{align}
    To show this, we need the following proposition.
    \begin{prop}
    \label{prop:norm_rel}
      If $H \in\mathbb{C}^{N\times N}$ has at most $s$ non-zero entries in any row, it holds that $\pnorm{1}{A} \leq \sqrt{s}\norm{A}$.
    \end{prop}
    \begin{proof}
      First observe that $\norm{A}^2 \leq \sum_i \lambda_i(A^{H}A) = \tr{A^{H}A} = \norm{A}^2_F$, and furthermore we have that
      $\sum_{ij} |a_{ij}|^2 \leq s  \max_{j \in [N]} \sum_i  |a_{ij}|^2$.
      From this we have that $\norm{A} \leq \sqrt{s} \norm{A}_{1}$ for a $s$-sparse $A$. By~\cite[Theorem 5.6.18]{horn1990matrix}, we have that $\norm{A}_{1} \leq C_M(1,*)\norm{A}$ for $C_M(1,*) = \max_{A \neq 0} \frac{\norm{A}}{\norm{A}_{1}}$ and using the above we have $C_M(1,*) \leq \sqrt{s}$. Therefore we find that $\norm{A}_{1} \leq \sqrt{s} \norm{A}$ as desired.
    \end{proof}
    The result then immediately follows since $\pnorm{1}{A} \leq \sqrt{N}\norm{A}$ for dense $A$, i.e., using
    $\pnorm{1}{H} \leq \sqrt{s}\norm{H}$ from Proposition~\ref{prop:norm_rel} with $s=N$, and the result from Theorem~\ref{thm:densehamiltoniansim}.

  \item We note that we could instead of qRAM also rely on the sparse-access model in order to prepare the states
  given in Eq.~\eqref{eq:mapping-ds}. The time complexity of the state preparation for the states in
  Eq.~\eqref{eq:mapping-ds} in the sparse-access model results in an additional $O(s)$ factor (for computing $\sigma_j$)
  compared to the qRAM model, and therefore in a polynomial slowdown.
  Therefore, in order to simulate $H$ for time $t$ in the sparse-access model, the time complexity of the algorithm
  in terms of $t$, $d$, and $\norm{H}$ results in $O(ts^{1.5}\norm{H})$ as $\norm{H}_1 \leq \sqrt{s}\norm{H}$,
  which implies that the here presented methods do not give any advantage over previous results in the sparse-access model.
\end{enumerate}

Similarly to the quantum algorithm, we have also designed a fast classical algorithm for Hamiltonian simulation
which indeed also only depends on the sparsity of the matrix and the norm.
Our classical algorithm requires $H$ to be row-searchable and row-computabile, which we will define later.
In short, these requirements are both fulfilled if we are given access to the memory structure described in
Definition~\ref{def:classical_datastructure}.
Informally our results can be summarised as follows:

\begin{theorem}[Hamiltonian Simulation With The Nystr\"om Method (Informal version of Theorem~\ref{thm:main})]
\label{thm:nystrominformal}
Let $H \in \mathbb{C}^{N \times N}$ be a Hermitian matrix, which is stored in the memory
structure given in Definition~\ref{def:classical_datastructure} and at most
$s$ non-zero entries per row, and if $\psi \in \mathbb{C}^N$ is an
$n$-qubit quantum state with at most $q$ entries.
Then, there exists an algorithm that, with probability of at least $1-\delta$,
approximates any chosen amplitude of the state $e^{i H t} \psi$ in time
\[
O\left(sq + \frac{t^9\norm{H}^4_F \norm{H}^7}{\epsilon^4}\left(n + \log\frac{1}{\delta}\right)^2\right),
\]
up to error $\epsilon$ in spectral norm, where $\norm{\cdot}_F$ is the Frobenius norm, $\norm{\cdot}$ is the spectral norm, $s$ is the maximum number of non-zero elements in the rows or columns of $H$, and $q$ is the number of non-zero elements in $\psi$.
\end{theorem}
Our algorithm is efficient in the low-rank and sparse regime, so if we compare our algorithm to the best
quantum algorithm in the sparse-input model or the black-box Hamiltonian simulation model, only
a polynomial slowdown occurs.
A result which was recently published~\cite{chia2020sampling} is efficient when $H$ is only low-rank and in comparison
their time complexity scales as $\poly(t, \norm{H}_F, 1/\epsilon)$, where $\norm{H}_F$ is the Frobenius norm of $H$.
Notably, our algorithm has therefore stricter requirements, such as the sparsity of the Hamiltonian and the
sparsity of the input state $\ket{\psi}$, however we achieve much lower polynomial dependencies.

By analysing the dependency of the runtime on the Frobenius norm we can determine under which
conditions we can obtain efficient (polylogarithmic runtime in the dimensionality).
Hamiltonian simulations. Informally, we obtain:
\begin{corollary}[Informal]
\label{co:maininformal}
If $H$ is a Hamiltonian on $n$ qubits with at most $s=O(\mathrm{polylog}(N))$ entries per row such that $\norm{H}^2_F - \frac{1}{N}\tr{H}^2 \leq O(\polylog(N))$, and if $\psi$ is an $n$-qubit quantum state with at most $q=O(\mathrm{polylog}(N))$ entries, then
there exists an efficient algorithm that approximates any chosen amplitude of the state $e^{i H t} \psi$.
\end{corollary}

While these results have rule out the possibility of exponential speedups of our quantum algorithm
in the low-rank regime, we note that the complexity of the quantum algorithm has a much lower degree in the polynomials
compared to the classical algorithms.
The quantum algorithm has hence still a large polynomial speedup over the classical algorithm
(for low-rank Hamiltonians and dense Hamiltonians).

\subsection{Related work}
\label{ssec:related}

As already discussed earlier, there are a range of previous results.
Many of these are given in different data access models so we will in the following
briefly discuss the main results in the respective access model.

\paragraph{Hamiltonian simulation with $\norm{H}$ dependence.}
The black-box model for Hamiltonian simulation is a special case of the sparse-access model.
It is suitable for non-sparse Hamiltonians and allows the algorithm to query the oracle with an index-pair $\ket{i,j}$,
which then returns the corresponding entry of the Hamiltonian $H$, i.e., $O_H$ defined in Eq.~\eqref{eq:oh}.
Given access to a Hamiltonian in the black-box model allows to simulate $H$ with error $\epsilon$ with query complexity
\[
O((\norm{H}t)^{3/2}N^{3/4}/\sqrt{\epsilon})
\]
for dense Hamiltonians~\cite{berry2009black}.
Empirically, the authors of~\cite{berry2009black} additionally observed, that for several classes of Hamiltonians,
the actual number of queries required for simulation even reduces to
\[
O(\sqrt{N}\log (N)).
\]
However, this $\widetilde{O}(\sqrt{N})$ dependence does not hold provably for all Hamiltonians, and was
therefore left as an open problem.
After the first version of this work was made public, this open problem was almost resolved by
Low~\cite{low2018hamiltonian}, who proposed a quantum algorithm for simulating black-box Hamiltonians with time complexity
\[
O((t\sqrt{N}\norm{H})^{1+o(1)}/\epsilon^{o(1)}).
\]
We note that the qRAM model is stronger than the black-box model, and therefore
the improvements could indeed stem from the input model.
However, our model to date gives the best performance for quantum machine learning applications,
and therefore might overall allow us to achieve superior performance in terms of computational complexity.

\paragraph{Hamiltonian simulation with $\maxnorm{H}$ dependence.}
As previously noted, the qRAM model is stronger than the black-box model and the sparse-access model.
One immediate implication of this is that prior quantum algorithms such as~\cite{berry2009black,berry2015hamiltonian}
can immediately be used to simulate Hamiltonians in the qRAM model.
Using black-box Hamiltonian simulation for a $s$-sparse Hamiltonian, the circuit depth is then given by
$\widetilde{O}(ts\maxnorm{H})$~\cite{berry2009black, berry2015hamiltonian}.
For non-sparse $H$, this implies a scaling of $\widetilde{O}(tN\maxnorm{H})$.
Since $\norm{H}\leq\sqrt{N}\maxnorm{H}$ implies $\widetilde{O}(t\sqrt{H}\norm{H}) = \widetilde{O}(tN\maxnorm{H})$,
our results do not give an advantage if the computational complexity is expressed in form of the $\maxnorm{H}$.
However, we can express the computational complexity in terms of the $\norm{H}$, which plays a
crucial role in solving linear systems~\cite{harrow2009quantum,childs2017quantum}, and black-box unitary
implementation~\cite{berry2009black}.
In this setting, we achieve a quadratic improvement w.r.t. the dimensionality dependence,
as the inequality $\maxnorm{H}\leq\norm{H}$ implies $\widetilde{O}(tN\maxnorm{H}) = \widetilde{O}(tN\norm{H})$.

\paragraph{Hamiltonian simulation in the qRAM model.}
Shortly after the first version of our results were made available, Chakraborty, Gily{\'e}n, and Jeffery~\cite{CGJ19}
independently proposed a quantum algorithm for simulating non-sparse Hamiltonians.
Their algorithm makes use of a qRAM input model similar to the one proposed by
Kerenidis~\cite{kerenidis2016quantum}, and achieved the same computational complexity as our result.
However, their work is also generalising our results since it uses a more general input model, namely,
the block-encoding model in which it first frames the results.
This was first proposed in~\cite{low2016hamiltonian}, and it assumes that we are given a unitary
$\bigl(\begin{smallmatrix}H/\alpha & \cdot\\\cdot&\cdot\end{smallmatrix}\bigr)$ that contains the Hamiltonian $H/\alpha$
(where we want to simulate $H$) in its upper-left block.
The time evolution $e^{-iHt}$ is then performed in $\widetilde{O}(\alpha\norm{H}t)$ time.
Using the sparse-access model for $s$-sparse Hamiltonian $H$, a block-encoding of $H$ with $\alpha=d$ can be efficiently implemented, implying an algorithm for Hamiltonian simulation with time complexity~\cite{low2016hamiltonian}
\[
\widetilde{O}(s\norm{H}t).
\]
The main result for Hamiltonian simulation described in~\cite{CGJ19} then takes the qRAM model for a $s$-sparse
Hamiltonian $H$, and shows that a block-encoding with $\alpha = \sqrt{s}$ can be efficiently implemented.
This yields an algorithm for Hamiltonian simulation with time complexity
\[
\widetilde{O}(\sqrt{s}\norm{H}t).
\]
We note that the techniques introduced in~\cite{CGJ19} have been generalised in~\cite{GSLW19} to a quantum framework for implementing singular value transformation of matrices.
Furthermore,~\cite{CGJ19} gives a detailed analysis which applies their results to the quantum linear systems
algorithm, which we have omitted in our analysis.

\paragraph{Quantum-inspired classical algorithms for Hamiltonian simulation.}

The problem of applying functions of matrices has been studied intensively in numerical linear algebra.
The exponential function is one application which has received particular interest in the
literature~\cite{higham2005scaling,higham2009scaling,al2009new,al2011computing}.
For arbitrary Hermitian matrices, at this point in time, no known algorithm exists that
exhibits a runtime logarithmic in the dimension of this input matrix.
However, such runtimes are required if we truly want to simulate the time evolution of quantum systems,
as the dimensions of the matrix that governs the evolution scale exponentially with the number
of quantum objects (such as atoms or orbitals) in the system.

More recently, the field of randomised numerical linear algebra techniques, from which we have previously
seen the example of randomised matrix multiplication (c.f. Section~\ref{sec:basic_mm}) has enabled
new approaches to such problems.
These methods, along with results from spectral graph theory, have culminated in a range of new classical algorithms for matrix functions.
In particular, they have also recently given new algorithms for approximate matrix
exponentials~\cite{drineas2006fast,drineas2011faster,mahoney2011randomised,woodruff2014sketching,rudi2015less}.
Previous results typically hold for matrices, which offer some form of structure.
For example, Orrecchia \textit{et al.}~\cite{orecchia2012approximating} combined function approximations
with the spectral sparsifiers of Spielmann and Teng~\cite{spielman2004nearly,spielman2011spectral}
into a new algorithm that can approximate exponentials of strictly diagonally dominant matrices
in time almost linear in the number of non-zero entries of $H$.
A standard approach is to calculate low-rank approximations of matrices, and then
use these within function-approximations to implement fast algorithms
for matrix functions~\cite{drineas2006fast2}.
Although these methods have many practical applications, they are not suitable to the application at hand,
i.e., Hamiltonian simulation, since they produce sketches that do not generally preserve
the given symmetries of the input matrix.

For problems where the symmetry of the sketched matrix is preserved, alternative methods
such as the \Nystrom{} method have been proposed.
The \Nystrom{} method has originally been developed for the approximation of kernel matrices in
statistical learning theory.
Informally, \Nystrom methods construct a lower-dimensional, symmetric, positive semidefinite approximation
of the input matrix by sampling from the input columns.
More specifically, let $A \in \mathbb{R}^{N \times N}$ be a symmetric, rank $r$,
positive semidefinite matrix, $A^j$ the $j$-th column vector of $A$, and $A_i$ the $i$-th row vector of $A$,
with singular value decomposition $A = U \Sigma U^{H}$ ($\Sigma = \operatorname{diag} (\sigma_1, \dots, \sigma_r)$), and Moore-Penrose pseudoinverse
$A^+ = \sum_{t=1}^r \sigma_t ^{-1} U^i (U^i)^{H}.$
The \Nystrom{} method then finds a low-rank approximation for the input matrix
$A$ which is close to $A$ in spectral norm (or Frobenius norm), which also
preserves the symmetry and positive semi-definiteness property of the matrix.
Let $C$ denote the $n \times l$ matrix formed by (uniformly) sampling $l \ll n$ columns of $A$, $W$ denote the $l \times l$ matrix consisting of the intersection of these $l$ columns with the
corresponding $l$ rows of $A$, and $W_k$ denote the best rank-$k$ approximation of $W$, i.e.,
 \[
 W_k = \operatorname{argmin}_{V\in \mathbb{R}^{l \times l}, \operatorname{rank}(V) = k} \norm{V-W}_F.
 \]
The {\Nystrom} method therefore returns a rank-$k$ approximation $\tilde{A}_k$ of $A$ for $k < n$ defined by:
\[
\tilde{A}_k = C W_k ^+ C^{\top} \approx A
\]
The running time of the algorithm is $O(nkl)$~\cite{kumar2012sampling}.
There exist many ways of sampling from the initial matrix $A$ in order to
create the approximation $\tilde{A}_k$,
and in particular non-uniform sampling schemes enable us to improve the performance,
see for example Theorem $3$ in~\cite{drineas2005nystrom}.
The \Nystrom{} method results in particularly good approximations for matrices
which are approximately low rank.
The first applications it was developed for were regression and classification problems
based on Gaussian processes, for which Williams and Seeger developed a sampling-based
algorithm~\cite{williams2001using,williams2001using,williams2002observations}.
Since the technique exhibit similarities to a method for solving linear integral equations
which was developed by \Nystrom{}~\cite{nystrom1930praktische}, they denoted their
result as \textit{\Nystrom{} method}.
Other methods, which have been developed more recently include the \textit{\Nystrom{} extension},
which has found application in large-scale machine learning problems, statistics,
and signal processing~\cite{williams2001using,williams2002observations,zhang2010clustered,talwalkar2008large,fowlkes2004spectral,kumar2012sampling,belabbas2007fast,belabbas2008sparse,kumar2009sampling,li2010making,mackey2011divide,zhang2010clustered,zhang2008improved}.

After the here presented quantum algorithm and our \Nystrom{} algorithm  for
Hamiltonian simulation which can use the memory model from Def.~\ref{def:classical_datastructure} have
been made public, a general framework was proposed to perform
such quantum-inspired classical algorithms which are based on the classical memory structure from Def.~\ref{def:classical_datastructure}~\cite{chia2020sampling},
namely \emph{sampling and query access}. Notably, this memory structure was first proposed by Tang~\cite{tang2018quantum}.
Although their framework generalises all quantum-inspired algorithm using the input model from
Def.~\ref{def:classical_datastructure}, our classical algorithm does not require the specified input model.
Furthermore our algorithm achieves significantly lower polynomial dependencies.
While the time complexity of their algorithm has a $36$-th power dependency on
the Frobenius norm of the Hamiltonian, our algorithm scales with a $4$-th power.
This however comes at the cost of a sparsity requirement on the Hamiltonian and the input state, which
can be restrictive in practical cases, i.e., in non-sparse but low-rank cases.
Therefore, the results by~\cite{chia2020sampling} might generally better suitable
for dense Hamiltonians.
A big difference is that the algorithm of Tang and others are solving the
same problem as the quantum algorithm that we developed, i.e.,
it allows us to sample from the output distribution of the time evolution.
Taking this into account, we believe our classical \Nystrom{} based algorithm
can be translated into this framework as well by using rejection sampling approaches
similar to the ones introduced in~\cite{tang2018quantum}.
However, we leave this as an open question for future work.

\paragraph{Summary of related results.}

To summarise this subsection, we provide Table~\ref{tab:relatedwork} for state-of-the-art algorithms for Hamiltonian simulation which include quantum and classical ones.

\begin{figure}
  \resizebox{\textwidth}{!}{
  \begin{tabular}{|l|r|r|}
\hline
\rowcolor{Gray}
\parbox[c][4em]{4cm}{Model} & \parbox[c][4em]{4cm}{State-of-the-art} & \parbox[c][3em]{5.6cm}{Advantage of our quantum algorithm} \\ \hline
\parbox[c][4em]{4cm}{Sparse-access with \\ $\maxnorm{H}$ dependence} & $\widetilde{O}(ts\maxnorm{H})$~\cite{berry2015hamiltonian} & No advantage \\ \hline
\parbox[c][4em]{4cm}{Sparse-access with \\ $\norm{H}$ dependence} & $O((t\sqrt{s}\norm{H})^{1+o(1)}/\epsilon^{o(1)})$~\cite{low2018hamiltonian} & \parbox[c][3em]{5.6cm}{Subpolynomial improvement in $t, s$; exponential improvement in $\epsilon$} \\ \hline
\parbox[c][4em]{4cm}{qRAM} & $\widetilde{O}(t\sqrt{s}\norm{H})$~\cite{CGJ19} & Same result \\ \hline
\parbox[c][4em]{4cm}{Classical sampling \\ and query access} & $\poly(t, \norm{H}_F, 1/\epsilon)$~\cite{chia2020sampling} & Polynomial speedup \\ \hline
\parbox[c][4em]{4cm}{Classical sampling \\ and query access, sparse input} & $\poly(s, t, \norm{H}_F, 1/\epsilon)$, [Thm.\ref{thm:nystrominformal}] & Polynomial speedup \\ \hline
\end{tabular}
}
\caption{Comparing our result $O(t\sqrt{d}\norm{H}\,\polylog(t,d,\norm{H}, 1/\epsilon))$ with other quantum and classical algorithms for different models. Since the qRAM model is stronger than the sparse-access model and the classical sampling and query access model, we consider the advantage of our algorithm against others when they are directly applied to the qRAM model.}
\label{tab:relatedwork}
\end{figure}

\subsection{Applications}
\label{ssec:applications}
In the following, we also lay out a few applications of Hamiltonian simulation.
Since the main focus of this thesis is quantum machine learning we in particular
refer to applications in this area.
For this, we will discuss the quantum linear systems algorithm, which is a key
subroutine in most of the existing quantum machine learning algorithms with an acclaimed exponential speedup.
Of course, other applications such as the estimation of properties such as ground state energies of chemical
systems exist~\cite{kassal2011simulating}.

\paragraph{Unitary implementation.} One application which follows directly from the ability to simulate non-sparse
Hamiltonians is the ability to approximately implement an arbitrary unitary matrix, the so-called
unitary implementation problem: given access to the entries of a unitary $U$ construct a quantum circuit $\tilde{U}$
such that $\norm{U - \tilde{U}} \leq \epsilon$ for some fixed $\epsilon$.
Unitary implementation can be reduced to Hamiltonian simulation as shown by Berry and others~\cite{berry2009black,jordan2009efficient}.
For this, consider the Hamiltonian of the form
\begin{align}
  H = \begin{pmatrix}
    0 & U\\
    U^{H} & 0
  \end{pmatrix}.
\end{align}
By performing the time evolution according to $e^{-iH \pi/2}$ on the state $\ket{1}\ket{\psi}$,
we are able to implement the operation $e^{-iH\pi/2}\ket{1}\ket{\psi} = -i\ket{0}U\ket{\psi}$,
which means that we can apply $U$ to $\ket{\psi}$.
If the entries of $U$ are stored in a data structure similar to Definition~\ref{def:datastructure},
using our Hamiltonian simulating algorithm we can implement the unitary $U$ with time complexity (circuit depth) $O(\sqrt{N}\,\polylog(N,1/\epsilon))$.

\paragraph{Quantum linear systems solver.} As previously mentioned, one of the major subroutines of the quantum linear
systems algorithm is Hamiltonian simulation.
Hamiltonian simulation, in particular in combination with the phase estimation~\cite{kitaev1997quantum} algorithm is
used to retrieve the eigenvalues of the input matrix, which can then be inverted to complete the matrix inverse.
Assuming here for simplicity that $\ket{b}$ is entirely in the column-space of $A$,
then the linear systems algorithm solves a linear system of the form $Ax=b$ and outputs an approximation
to the normalised solution $\ket{x} = \ket{A^{-1}b}$.
For $s$-sparse matrices~\cite{childs2017quantum} showed that $A^{-1}$ can be approximated as a linear combination
of unitaries of the form $e^{-iAt}$. Notably for non-Hermitian $A$, we can just apply the same ideas to the
Hamiltonian
\begin{align}
  H = \begin{pmatrix}
    0 & A\\
    A^{H} & 0
  \end{pmatrix},
\end{align}
which is Hermitian by definition, and due to the logarithmic runtimes results in only a factor $2$ overhead.
In order to implement these unitaries we can then use any Hamiltonian simulation algorithm, such as
\cite{berry2015hamiltonian}, which results in an overall (gate) complexity of
$O(s\kappa^2\polylog(N, \kappa/\epsilon))$.
For a non-sparse input matrix $A$, the algorithm however scales as $\widetilde{O}(N)$ (again ignoring logarithmic factors).
In some of our earlier results, we used a similar data structure to Definition~\ref{def:datastructure} to
design a quantum algorithm for solving linear systems for non-sparse matrices~\cite{wossnig2018quantum} which
resulted in a time complexity (circuit depth) of $O(\kappa^2\sqrt{N}\polylog(N)/\epsilon)$.
Using again this data structure, similarly to~\cite{kerenidis2016quantum,wossnig2018quantum},
our new Hamiltonian simulation algorithm, and the linear combinations of unitaries (LCU) decompositions
from~\cite{childs2017quantum}, we can describe a quantum algorithm for solving linear systems for non-sparse matrices
with time complexity (circuit depth)
\[
O(\kappa^2\sqrt{N}\polylog(\kappa/\epsilon)),
\]
which is an exponential improvement in the error dependence compared to our previous result~\cite{wossnig2018quantum}.
Notably, one drawback of our implementation, which has been resolved by \cite{CGJ19} and subsequent works is the
high condition number dependency of our algorithm which seems restrictive for practical applications.

\subsection{Hamiltonian Simulation for dense matrices}
\label{ssec:dense_ham_sim}

We now show how we can use the data structure from Definition~\ref{def:datastructure} to derive a fast quantum algorithm
for non-sparse Hamiltonian simulation.
We will in particular proof the results from Theorem~\ref{thm:densehamiltoniansim}.

We prove the result in multiple steps.
First, we show how to use the data structure to prepare a certain state.
Next, we use the state preparation to perform a quantum walk.
In the final step, we show how the quantum walk can be used to implement Hamiltonian simulation.
This then leads to the main result.

\paragraph{State Preparation.}
\label{par:state_prep}

Using the data structure from Definition~\ref{def:datastructure}, we can efficiently perform the mapping described
in the following technical lemma for efficient state preparation.

\begin{lemma}[State Preparation]
  \label{lemma:statepreparation}
  Let $H\in\mathbb{C}^{N\times N}$ be a Hermitian matrix (where $N = 2^n$ for a $n$ qubit Hamiltonian) stored in
  the data structure as specified in Definition~\ref{def:datastructure}.
  Each entry $H_{jk}$ is represented with $b$ bits of precision.
  Then the following holds
  \begin{enumerate}
	\item Let $\pnorm{1}{H}=\max_j \sum_{k=0}^{N-1}\abs{H_{jk}}$ as before.
	A quantum computer that has access to the data structure can perform the following mapping for $j \in \{ 0, \ldots , N-1 \}$,
	  \begin{align}
        \label{eq:mapping-ds}
        \ket{j}\ket{0^{\log N}}\ket{0} \mapsto \frac{1}{\sqrt{\pnorm{1}{H}}}\ket{j}\sum_{k=0}^{N-1}\ket{k}\left(\sqrt{H_{jk}^{*}}\ket{0} + \sqrt{\frac{\pnorm{1}{H} - \sigma_j}{N}}\ket{1}\right), 	  \end{align}
	  with time complexity (circuit depth) $O(n^2b^{5/2}\log b)$, where $\sigma_j = \sum_k|H_{jk}|$, and the square-root satisfies $\sqrt{H_{jk}}\bigl(\sqrt{H_{jk}^*}\bigr)^* = H_{jk}$.
	\item The size of the data structure containing all $N^2$ complex entries is $O(N^2\log^2(N))$.
  \end{enumerate}
\end{lemma}
In order to perform this mapping, we will need the following Lemma.
We will use it in order to efficiently implement the conditional rotations of the qubits with complex numbers.
\begin{lemma}
  \label{lemma:conditional-rotation}
  Let $\theta, \phi_0, \phi_1 \in \mathbb{R}$ and let $\widetilde{\theta}, \widetilde{\phi_0}, \widetilde{\phi_1}$ be the $b$-bit finite precision representation of $\theta, \phi_0$, and $\phi_1$, respectively. Then there exists a unitary $U$ that performs the following mapping:
  \begin{align}
	U:\ket{\widetilde{\phi_0}}\ket{\widetilde{\phi_1}}\ket{\widetilde{\theta}}\ket{0} \mapsto \ket{\widetilde{\phi_0}}\ket{\widetilde{\phi_1}}\ket{\widetilde{\theta}}\left(e^{i\widetilde{\phi_0}}\cos(\widetilde{\theta})\ket{0} + e^{i\widetilde{\phi_1}}\sin(\widetilde{\theta})\ket{1}\right).
  \end{align}
  Moreover, $U$ can be implemented with $O(b)$ 1- and 2-qubit gates.
\end{lemma}
\begin{proof}
  Define $U$ as
  \begin{align}
	U = \left(\sum_{\widetilde{\phi_0}\in\{0,1\}^b}\ketbra{\widetilde{\phi_0}}{\widetilde{\phi_0}}\otimes e^{i\ketbra{0}{0}\widetilde{\phi_0}}\right)
	\left(\sum_{\widetilde{\phi_1}\in\{0,1\}^b}\ketbra{\widetilde{\phi_1}}{\widetilde{\phi_1}}\otimes e^{i\ketbra{1}{1}\widetilde{\phi_1}}\right) \nonumber \\
	\left(\sum_{\widetilde{\theta}\in\{0,1\}^b}\ketbra{\widetilde{\theta}}{\widetilde{\theta}}\otimes e^{-iY\widetilde{\theta}}\right),
  \end{align}
  where $Y = \bigl(\begin{smallmatrix}0&-i\\i&0\end{smallmatrix}\bigr)$ is the Pauli $Y$ matrix.

  To implement the operator $\sum_{\widetilde{\theta}\in\{0,1\}^b}\ketbra{\widetilde{\theta}}{\widetilde{\theta}}\otimes e^{-iY\widetilde{\theta}}$, we use one rotation controlled on each qubit of the first register, with the rotation angles halved for each successive bit. The other two factors of $U$ can be implemented in a similar way. Therefore, $U$ can be implemented with $O(b)$ 1- and 2-qubit gates.
\end{proof}

Before proving Lemma~\ref{lemma:statepreparation}, we first describe the construction and the size of the data structure. Readers may refer to~\cite{kerenidis2016quantum} for more details.
  \begin{itemize}
  \item The data structure is built from $N$ binary trees $D_i, i \in \{0, \dots, N-1\}$ and we start with an empty tree.
\item When a new entry $(i,j,H_{ij})$ arrives, we create or update the leaf node $j$ in the tree $D_i$, where the adding of the entry takes $ O(\log(N))$ time, since the depth of the tree for $H \in  \mathbb C^{N \times N}$ is at most $\log(N)$. Since the path from the root to the leaf is of length at most $\log(N)$ (under the assumption that $N=2^n$), we have furthermore to update at most $\log(N)$ nodes, which can be done in $O(\log(N))$ time if we store an ordered list of the levels in the tree.
\item The total time for updating the tree with a new entry is given by $\log(N)\times\log(N)= \log^2(N)$.
\item The memory requirements for $k$ entries are given by $O(k \log^2(N))$ as for every entry $(j,k,H_{jk})$ at least $\log(N)$ nodes are added and each node requires at most $O(\log(N))$ bits.
\end{itemize}

Now we are ready to prove Lemma~\ref{lemma:statepreparation}.
\begin{proof}[Proof of Lemma~\ref{lemma:statepreparation}]
With this data structure, we can perform the mapping specified in Eq.~\eqref{eq:mapping-ds}, with the following steps. For each $j$, we start from the root of $D_j$. Starting with the initial state $\ket{j}\ket{0^{\log N}}\ket{0}$, first apply the rotation (according to the value stored in the root node and calculating the normalisation in one query) on the last register to obtain the state
  \begin{align}
    \frac{1}{\sqrt{\pnorm{1}{H}}}\ket{0}\ket{0^{\log N}}\left(\sqrt{\sum_{k=0}^{N-1}|H_{jk}^{*}|}\ket{0} + \sqrt{\pnorm{1}{H} - \sigma_j}\ket{1}\right),
  \end{align}
  which is normalised since $\sigma_j = \sum_k \abs{H_{jk}}$, and we have by definition that
  $\sqrt{\sum_{k=0}^{N-1}|H_{jk}^{*}|} \cdot \left(\sqrt{\sum_{k=0}^{N-1}|H_{jk}^{*}|}\right)^{*} = \sum_{k}|H_{jk}| = \sigma_j$.
  Then a sequence of conditional rotations is applied on each qubit of the second register to obtain the state as in Eq.~\eqref{eq:mapping-ds}. At level $\ell$ of the binary tree $D_j$, a query to the data structure is made to load the data $c$ (stored in the node) into a register in superposition, the rotation to perform is proportional to $\bigl(\sqrt{c}, \sqrt{(\pnorm{1}{H}-\sigma_j)/2^\ell}\bigr)$ (assuming at the root, $\ell = 0$, and for the leaves, $\ell = \log N$). Then the rotation angles will be determined by calculating the square root and trigonometric functions on the output of the query: this can be implemented with $O(b^{5/2})$ 1- and 2-qubit gates using simple techniques based on Taylor series and long multiplication as in~\cite{berry2015hamiltonian}, where the error is smaller than that caused by truncating to $b$ bits. Then the conditional rotation is applied by the circuit described in Lemma~\ref{lemma:conditional-rotation}, and the cost for the conditional rotation is $O(b)$. There are $n=\log(N)$ levels, so the cost excluding the implementation of the oracle is $O(nb^{5/2})$. To obtain quantum access to the classical data structure, a quantum addressing scheme is required. One addressing scheme described in~\cite{giovannetti2008qram1} can be used. Although the circuit size of this addressing scheme is $\widetilde{O}(N)$ for each $D_j$, its circuit depth is $O(n)$. Therefore, the time complexity (circuit depth) for preparing the state in Eq.~\eqref{eq:mapping-ds} is $O(n^2b^{5/2}\log n)$.

  We use the following rules to determine the sign of the square-root of a complex number: if $H_{jk}$ is not a negative real number, we write $H_{jk} = re^{i\varphi}$ (for $r\geq 0$ and $-\pi\leq\varphi\leq\pi$) and take $\sqrt{H_{jk}^*} = \sqrt{r}e^{-i\varphi/2}$; when $H_{jk}$ is a negative real number, we take $\sqrt{H_{jk}^*} = \mathrm{sign}(j-k)i\sqrt{|H_{jk}|}$ to avoid the sign ambiguity. With this recipe, we have $\sqrt{H_{jk}}\bigl(\sqrt{H_{jk}^*}\bigr)^* = H_{jk}$.
\end{proof}

In order to convey the working of the data structure better, we give in the following a
example of the state preparation procedure based on the data structure in Fig.~\ref{fig:bt-4}.
For the sake of comprehensibility and simplicity, we only take an example with $4$ leaves $\{c_1,c_2,c_3,c_4\}$,
and hence $H=[c_1\; c_2\; c_3\; c_4]^t$.
The initial state (omitting the first register) is $\ket{00}\ket{0}$. Let $\sigma = |c_0|+|c_1|+|c_2|+|c_3|$. Apply the first rotation, we obtain the state
\begin{align}
  \frac{1}{\sqrt{\pnorm{1}{H}}}\ket{00}\left(\sqrt{|c_0|+|c_1|+|c_2|+|c_3|}\ket{0} + \sqrt{\pnorm{1}{H}-\sigma_j}\ket{1}\right) = \nonumber \\
    \frac{1}{\sqrt{\pnorm{1}{H}}}\ket{00}\left(\sqrt{\sigma}\ket{0} + \sqrt{\pnorm{1}{H}-\sigma}\ket{1}\right) = \nonumber \\
  \frac{1}{\sqrt{\pnorm{1}{H}}}\left(\sqrt{\sigma}\ket{00}\ket{0} + \sqrt{\pnorm{1}{H}-\sigma}\ket{00}\ket{1}\right).
\end{align}
Then, apply a rotation on the first qubit of the first register conditioned on the last register, we obtain the state
\begin{align}
  \frac{1}{\sqrt{\pnorm{1}{H}}} & \Bigg(\left(\sqrt{|c_0|+|c_1|}\ket{00}+\sqrt{|c_2|+|c_3|}\ket{10}\right)\ket{0} \nonumber \\
   &+ \left(\sqrt{\frac{\pnorm{1}{H}-\sigma_j}{2}}(\ket{00} + \ket{10}) \right) \ket{1}\Bigg).
\end{align}
Next, apply a rotation on the second qubit of the first register conditioned on the first qubit of the first register and last register, we obtain the desired state:
\begin{align}
  \frac{1}{\sqrt{\pnorm{1}{H}}}
  \Bigg(\sqrt{c_0}\ket{00}\ket{0}+\sqrt{c_1}\ket{01}\ket{0} + \sqrt{c_2}\ket{10}\ket{0} + \sqrt{c_3}\ket{11}\ket{0}  \nonumber \\
  + \sqrt{\frac{\pnorm{1}{H}-\sigma_j}{4}}\ket{00}\ket{1} +\sqrt{\frac{\pnorm{1}{H}-\sigma_j}{4}}\ket{01}\ket{1} \nonumber \\
  + \sqrt{\frac{\pnorm{1}{H}-\sigma_j}{4}}\ket{10}\ket{1} + \sqrt{\frac{\pnorm{1}{H}-\sigma_j}{4}}\ket{11}\ket{1}\Bigg) \nonumber \\
  = \frac{1}{\sqrt{\pnorm{1}{H}}}\Bigg(\ket{00}\left(\sqrt{c_0}\ket{0} + \sqrt{\frac{\pnorm{1}{H}-\sigma_j}{4}}\ket{1}\right) \nonumber \\
  + \ket{01}\left(\sqrt{c_1}\ket{0} + \sqrt{\frac{\pnorm{1}{H}-\sigma_j}{4}}\ket{1}\right)  \nonumber \\
  + \ket{10}\left(\sqrt{c_2}\ket{0} + \sqrt{\frac{\pnorm{1}{H}-\sigma_j}{4}}\ket{1}\right)  \nonumber \\
  + \ket{11}\left(\sqrt{c_3}\ket{0} + \sqrt{\frac{\pnorm{1}{H}-\sigma_j}{4}}\ket{1}\right)\Bigg).
\end{align}

\paragraph{The Quantum Walk Operator.}
\label{par:q_walk}

Based on the data structure specified in Definition~\ref{def:datastructure} and the efficient state preparation in Lemma~\ref{lemma:statepreparation}, we construct a quantum walk operator for $H$ as follows. First define the isometry $T$ as
  \begin{align}
	T = \sum_{j=0}^{N-1}\sum_{b\in\{0,1\}}(\ketbra{j}{j}\otimes\ketbra{b}{b})\otimes\ket{\varphi_{jb}},
  \end{align}
  with $\ket{\varphi_{j1}} = \ket{0}\ket{1}$ and
  \begin{align}
    \label{eq:varphi_j0}
	\ket{\varphi_{j0}} = \frac{1}{\sqrt{\pnorm{1}{H}}}\sum_{k=0}^{N-1}\ket{k}\left(\sqrt{H_{jk}^{*}}\ket{0} + \sqrt{\frac{\pnorm{1}{H} - \sigma_j}{N}}\ket{1}\right),
  \end{align}
where $\sigma_j = \sum_{k=0}^{N-1}|H_{jk}|$.
Let $S$ be the swap operator that maps $\ket{j_0}\ket{b_0}\ket{j_1}\ket{b_1}$ to $\ket{j_1}\ket{b_1}\ket{j_0}\ket{b_0}$, for all $j_0, j_1\in\{0, \ldots, N-1\}$ and $b_0, b_1\in\{0,1\}$. We observe that
\begin{align}
  \label{eq:Hij}
  \bra{j}\bra{0}T^{H}ST\ket{k}\ket{0} = \frac{\sqrt{H_{jk}}\left(\sqrt{H_{jk}^*}\right)^*}{\pnorm{1}{H}} = \frac{H_{jk}}{\pnorm{1}{H}},
\end{align}
where the second equality is ensured by the choice of the square-root as in the proof of Lemma~\ref{lemma:statepreparation}.
This implies that
\begin{align}
  (I\otimes\bra{0})T^{H}ST(I\otimes\ket{0}) = \frac{H}{\pnorm{1}{H}}.
\end{align}

The quantum walk operator $U$ is defined as
\begin{align}
  \label{eq:quantumwalk}
  U = iS(2TT^{H} - I).
\end{align}
A more general characterization of the eigenvalues of quantum walks is presented in~\cite{szegedy2004quantum}.
Here we give a specific proof on the relationship between the eigenvalues of $U$ and $H$ as follows.

\begin{lemma}
  Let the unitary operator $U$ be defined as in Eq.~\eqref{eq:quantumwalk}, and let $\lambda$ be an eigenvalue of $H$ with eigenstate $\ket{\lambda}$. It holds that
  \begin{align}
	U\ket{\mu_{\pm}} = \mu_{\pm}\ket{\mu_{\pm}},
  \end{align}
  where
  \begin{align}
	\ket{\mu_{\pm}} =& (T+i\mu_{\pm}ST)\ket{\lambda}\ket{0}, \mbox{ and}\\
	\label{eq:mupm}
	\mu_{\pm} =& \pm e^{\pm i \arcsin(\lambda/\pnorm{1}{H})}.
  \end{align}
\end{lemma}
\begin{proof}
  By the fact that $T^{H}T = I$ and
  $(I\otimes\bra{0})T^{H}ST(I\otimes\ket{0}) = H/\pnorm{1}{H}$, and $(I \otimes \bra{1})T^{H} ST (I \otimes \ket{0}) = 0$, we have
  \begin{align}
	U\ket{\mu_{\pm}} = \mu_{\pm}T\ket{\lambda}\ket{0}+i\left(1+\frac{2\lambda i}{\pnorm{1}{H}}\mu_{\pm}\right)ST\ket{\lambda} \ket{0}.
  \end{align}
  In order for this state being an eigenstate, it must hold that
  \begin{align}
	1 + \frac{2\lambda i}{\pnorm{1}{H}}\mu_{\pm} = \mu_{\pm}^2,
  \end{align}
  and the solution is
  \begin{align}
	\mu_{\pm} = \frac{\lambda i}{\pnorm{1}{H}} \pm \sqrt{1 - \frac{\lambda^2}{\pnorm{1}{H}^2}} =\pm e^{\pm i\arcsin(\lambda/\pnorm{1}{H})}.
  \end{align}
\end{proof}

\paragraph{Linear combination of unitaries and Hamiltonian simulation.}
\label{par:lcu}

Next, we need to convert the quantum walk operator $U= i S(2TT^H -I)$ into an operator for
Hamiltonian simulation.
For this, we consider the generating functions for the Bessel functions, denoted by $J_m(\cdot)$.
From~\cite[(9.1.41)]{abramowitz1964handbook}, we know that it holds that
\begin{align}
\label{eq:relation_sum_exp}
  \sum_{m=-\infty}^{\infty}J_m(z)\mu_{\pm}^m = \exp\left(\frac{z}{2}\left(\mu_{\pm}-\frac{1}{\mu_{\pm}}\right)\right) = e^{iz\lambda/\pnorm{1}{H}},
\end{align}
where the second equality follows from Eq.~\eqref{eq:mupm} and the fact that $\sin(x)=(e^{ix}-e^{-ix})/2i$. This leads to the following linear combination of unitaries:
\begin{align}
  \label{eq:precise-V}
  V_{\infty} = \sum_{m=-\infty}^{\infty}\frac{J_m(z)}{\sum_{j={-\infty}}^{\infty}J_j(z)}U^m = \sum_{m=-\infty}^{\infty}J_m(z)U^m = e^{izH/\pnorm{1}{H}},
\end{align}
where the second equality follows from the fact that $\sum_{j={-\infty}}^{\infty}J_j(z) = 1$.

Since we cannot in practice implement the infinite sum,
we will in the following find an approximation to
$e^{-izH/\pnorm{1}{H}}$ by truncation the sum in Eq.~\eqref{eq:precise-V}:
\begin{align}
  \label{eq:lcu-vk}
  V_k = \sum_{m=-k}^k\frac{J_m(z)}{\sum_{j=-k}^{k}J_j(z)}U^m.
\end{align}
Here the coefficients are normalised by $\sum_{j=-k}^k J_j(z)$ so that they sum to 1. This will minimize the approximation error (see the proof of Lemma~\ref{lemma:truncate}), and the normalisation trick was originated in~\cite{berry2015hamiltonian}). The eigenvalues of $V_k$ are
\begin{align}
  \sum_{m=-k}^k\frac{J_m(z)}{\sum_{j=-k}^{k}J_j(z)}\mu_{\pm}^m.
\end{align}
Note that each eigenvalue of $V_k$ does not depend on $\pm$ as $J_{-m}(z) = (-1)^mJ_m(z)$.

To bound the error in this approximation, we require the following technical lemma.
\begin{lemma}
  \label{lemma:truncate}
  Let $V_k$ and $V_{\infty}$ be defined as above. There exists a positive integer $k$ satisfying $k \geq |z|$ and
  \begin{align}
    k = O\left(\frac{\log(\norm{H}/(\pnorm{1}{H}\epsilon))}{\log\log(\norm{H}/(\pnorm{1}{H}\epsilon))}\right),
  \end{align}
  such that
  \begin{align}
	\norm{V_k-V_{\infty}} \leq \epsilon.
  \end{align}
\end{lemma}

\begin{proof}
The proof outlined here follows closely the proof of Lemma~8 in~\cite{berry2015hamiltonian}.
Recalling the definition of $V_k$ and $V_{\infty}$, we define the weights in $V_k$ by
\begin{align}
  \alpha_m := \frac{J_m(z)}{C_k},
\end{align}
where $C_k = \sum_{l=-k}^k J_l(z)$. The normalisation here is chosen so that $\sum_m a_m =1$ which will give the best result~\cite{berry2015hamiltonian}.\\
Since
\begin{equation}
\sum_{m=-\infty}^{\infty} J_m(z) = \sum_{m=[-\infty:-k-1;k+1:\infty]} J_m(z)  + \sum_{m=-k}^{k} J_m(z) = 1 ,
\end{equation}
observe that we have two error sources. The first one comes from the truncation of the series, and the second one comes from the different renormalisation of the terms which introduces an error in the first $\lvert m \rvert \leq k$ terms in the sum.
We therefore start by bounding the normalisation factor $C_k$. For the Bessel-functions for all $m$ it holds that $\lvert J_m(z) \rvert \leq \frac{1}{\lvert m\lvert !}\left\lvert \frac{z}{2} \right\rvert^{|m|}$, since $J_{-m}(z) = (-1)^m J_m(z)$~\cite[(9.1.5)]{abramowitz1964handbook}.
For $\lvert m \rvert \leq k$ we can hence find the following bound on the truncated part
\begin{align}
\sum_{m=[-\infty:-k-1;k+1:\infty]} \lvert J_m(z) \rvert &= 2 \sum_{m=k+1}^{\infty} \lvert J_m(z) \rvert \leq  2 \sum_{m=k+1}^{\infty} \frac{\lvert z/2 \rvert^m}{m!} \nonumber\\
 &= 2 \frac{\lvert z/2 \rvert^{k+1}}{(k+1)!} \left(1+\frac{|z/2|}{k+2} + \frac{|z/2|^2}{(k+2)(k+3)}  + \cdots\right) \nonumber\\
 &< 2 \frac{\lvert z/2 \rvert^{k+1}}{(k+1)!} \sum_{m=k+1}^{\infty} \left(\frac{1}{2}\right)^{m-(k+1)} = \frac{ 4 \lvert z/2 \rvert^{k+1}}{(k+1)!}.
\end{align}
Since $\sum_m J_m(z) =1$, based on the normalisation, we hence find that
\begin{equation}
\sum_{m=-k}^k J_m(z) \geq \left( 1- \frac{ 4 \lvert z/2 \rvert^{k+1}}{(k+1)!} \right),
\end{equation}
which is a lower bound on the normalisation factor $C_k$. Since $a_m = \frac{J_m(z)}{C_k}$, the correction is small, which implies that
\begin{equation}
\label{eq:dep_am}
a_m = J_m(z) \left( 1 + O \left( \frac{\lvert z/2 \rvert^{k+1}}{(k+1)!} \right) \right),
\end{equation}
and we have a multiplicative error based on the renormalisation.

Next we want to bound the error in the truncation before we join the two error sources.\\
From Eq.~(\ref{eq:relation_sum_exp}) we know that
\begin{equation}
e^{iz \lambda/\Lambda} - 1 = \sum_{m=-\infty}^{\infty} J_m(z) ( \mu_{\pm}^m -1),
\end{equation}
by the normalisation of $\sum_m J_m(z)$. From this we can see that we can hence obtain a bound on the truncated $J_m(z)$ as follows.
\begin{equation}
\label{eq:reordering_exp_sum}
\sum_{m=-k}^{k} J_m(z) ( \mu_{\pm}^m -1) = e^{iz \lambda/\Lambda} - 1 - \sum_{\substack{m = [-\infty: -(k+1);\\ (k+1):\infty]}} J_m(z)  ( \mu_{\pm}^m -1).
\end{equation}
Therefore we can upper bound the left-hand side in terms of the exact value of $V_{\infty}$, i.e.\ $e^{iz\lambda/\Lambda}$ if we can bound the right-most term in Eq.~\eqref{eq:reordering_exp_sum}. Using furthermore the bound in Eq.~\eqref{eq:dep_am} we obtain
\begin{align}
\label{eq:difference1}
\sum_{m=-k}^{k} a_m(z) ( \mu_{\pm}^m -1) = \left( e^{iz \lambda/\Lambda} - 1 - \sum_{\substack{m = [-\infty: -(k+1);\\ (k+1):\infty]}} J_m(z)  ( \mu_{\pm}^m -1)\right) \nonumber \\
\left( 1 + O \left( \frac{\lvert z/2 \rvert^{k+1}}{(k+1)!} \right) \right),
\end{align}
which reduced with $\left\lvert 2^{iz \lambda/\Lambda} -1 \right\rvert \leq \lvert z \lambda/\Lambda \rvert$ and $|z| \leq k$ to
\begin{equation}
\label{eq:difference2}
\sum_{m=-k}^{k} a_m(z) ( \mu_{\pm}^m -1) = e^{iz \lambda/\Lambda} - 1 -  O \left(\sum_{\substack{m = [-\infty: -(k+1);\\ (k+1):\infty]}} J_m(z)  ( \mu_{\pm}^m -1)\right).
\end{equation}
We can then obtain the desired bound $\norm{ V_{\infty} - V_k}$ by reordering the above equation, and using that $\sum_{m=-k}^k a_m(z)=1$ such that we have
\begin{equation}
\label{eq:final_error_bound}
\norm{ V_{\infty} - V_k} = \left\lvert \sum_{m=-k}^k a_m \mu_{\pm}^m - e^{iz \lambda/\Lambda} \right\rvert = O \left( \sum_{\substack{m = [-\infty: -(k+1);\\ (k+1):\infty]}} J_m(z)  ( \mu_{\pm}^m -1)\right).
\end{equation}
We hence only need to bound the right-hand side.\\
For $\mu_+$ we can use that $\lvert \mu_{+}^m - 1 \rvert \leq 2 \lvert m \lambda/ \Lambda \rvert =: 2 \lvert m \nu \rvert$ and obtain the bound $ 2 \frac{\lvert \nu \rvert}{k!}\left\lvert\frac{z}{2} \right\rvert^{k+1} $~\cite{berry2015hamiltonian}. For the $\mu_-$ case we need to refine the analysis and will show that the bound remains the same.
Let $\nu := \lambda/ \Lambda$ as above. First observe that
$J_m(z) \mu_-^m + J_{-m} (z) \mu_i^{-m} = J_{-m} (z) \mu_+^{-m} + J_m (z) \mu_+^m$, and it follows that
\begin{align}
 \sum_{m = -\infty}^{-(k+1)} J_m(z)  ( \mu_{-}^m -1)  &+ \sum_{m = k+1}^{\infty} J_m(z)  ( \mu_{-}^m -1) \nonumber \\
 &=\sum_{m = -\infty}^{-(k+1)} J_m(z)  ( \mu_{+}^m -1)  + \sum_{m = k+1}^{\infty} J_m(z)  ( \mu_{+}^m -1).
\end{align}
Therefore we only need to treat the $\mu_+$ case.
\begin{align}
\left\lvert \sum_{\substack{m = [-\infty: -(k+1);\\ (k+1):\infty]}} J_m(z)  ( \mu_{+}^m -1) \right\rvert &\leq 2 \left\lvert \sum_{m = k+1} ^{\infty} J_m(z)  ( \mu_{+}^m -1)  \right\rvert \nonumber\\
&\leq 2 \sum_{m= k+1}^{\infty} \lvert J_m(z) \rvert  \lvert \mu_{+}^m- 1 \rvert  \nonumber\\
& = 2 \sum_{m= k+1}^{\infty} \frac{1}{|m|!}\left\lvert\frac{z}{2} \right\rvert^{|m|}  \lvert \mu_{+}^m- 1 \rvert \nonumber\\
& \leq 4 \sum_{m= k+1}^{\infty} \frac{1}{|m|!}\left\lvert\frac{z}{2} \right\rvert^{|m|} m \lvert \nu \rvert \nonumber\\
&<   \frac{8 \lvert \nu \rvert }{(k+1)!}\left\lvert\frac{z}{2} \right\rvert^{k+1} (k+2) .
\end{align}
Using this bound, we hence obtain from Eq.~\eqref{eq:final_error_bound},
\begin{equation}
\label{eq:final_error_bound_2}
\norm{ V_{\infty} - V_k} = \left\lvert \sum_{m=-k}^k a_m \mu_{\pm}^m - e^{iz \lambda/\Lambda} \right\rvert \leq O \left(
\frac{\lambda}{k!\ \Lambda } \left\lvert\frac{z}{2} \right\rvert^{k+1} \right) = O\left(\frac{\norm{H}(z/2)^{k+1}}{\Lambda k!}\right).
\end{equation}
In order for the above equation being upper-bounded by $\epsilon$, it suffices to choose some $k$ that is upper bounded as claimed.
\end{proof}

As we can see from the above discussion, we can hence use the operator $V_k$ to implement
the time evolution according to $e^{-izH/\pnorm{1}{H}}$.
We can also immediately see from this, that $V_k$ is a linear combination of unitaries (LCU).
We proceed now to implement this LCU by using some well-known results.

In the following, we provide technical lemmas for implementing linear combination of unitaries.
Suppose we are given the implementations of unitaries $U_0$, $U_1$, \ldots, $U_{m-1}$, and coefficients $\alpha_0, \alpha_1, \ldots, \alpha_{m-1}$. Then the unitary
\begin{align}
  V = \sum_{j=0}^{m-1}\alpha_j U_j
\end{align}
can be implemented probabilistically by the technique called linear combination of unitaries (LCU)~\cite{kothari2014efficient}. Provided $\sum_{j=0}^{m-1}|\alpha_j| \leq 2$, $V$ can be implemented with success probability $1/4$. To achieve this, we define the multiplexed-$U$ operation, which is denoted by $\text{multi-}U$, as
\begin{align}
  \text{multi-}U\ket{j}\ket{\psi} = \ket{j}U_j\ket{\psi}.
\end{align}

The probabilistic implementation of $V$ is summarised in the following lemma.
\begin{lemma}
  \label{lemma:lcu-const}
  Let $\mathrm{multi}$-$U$ be defined as above. If $\sum_{j=0}^{m-1}|\alpha_j| \leq 2$, then there exists a quantum circuit that maps $\ket{0}\ket{0}\ket{\psi}$ to the state
  \begin{align}
	\frac{1}{2}\ket{0}\ket{0}\left(\sum_{j=0}^{m-1}\alpha_jU_j\ket{\psi}\right) + \frac{\sqrt{3}}{2}\ket{\Phi^{\bot}},
  \end{align}
  where $(\ketbra{0}{0}\otimes\ketbra{0}{0}\otimes I)\ket{\Phi^{\bot}} = 0$.
  Moreover, this quantum circuit uses $O(1)$ applications of $\mathrm{multi}$-$U$ and $O(m)$ 1- and 2-qubit gates.
\end{lemma}

\begin{proof}
  Let $s = \sum_{j=0}^{m-1}|\alpha_j|$.
  We first define the unitary operator $B$ to prepare the coefficients:
  \begin{align}
	B\ket{0}\ket{0} = \left(\sqrt{\frac{s}{2}}\ket{0} + \sqrt{1-\frac{s}{2}}\ket{1}\right)\otimes \frac{1}{\sqrt{s}}\sum_{j=0}^{m-1}\sqrt{\alpha_j}\ket{j}.
  \end{align}
  Define the unitary operator $W$ as $W = (B^{H}\otimes I)(I \otimes \mbox{multi-}U)(B\otimes I)$. We claim that $W$ performs the desired mapping, as
  \begin{align}
	W\ket{0}\ket{0}\ket{\psi} =& (B^{H}\otimes I)(I\otimes \mbox{multi-}U)(B\otimes I) \ket{0}\ket{0}\ket{\psi} \nonumber\\
	=& \frac{1}{\sqrt{2}}(B^{H}\otimes I)\ket{0}\sum_{j=0}^{m-1}\sqrt{\alpha_j}\ket{j}U_j\ket{\psi} \nonumber \\
  +& \sqrt{\frac{2-s}{2s}}(B^{H}\otimes I)\ket{1}\sum_{j=0}^{m-1}\sqrt{\alpha_j}\ket{j}U_j\ket{\psi} \nonumber\\
	=& \frac{1}{2}\ket{0}\ket{0}\sum_{j=0}^{m-1}\alpha_jU_j\ket{\psi} + \sqrt{\gamma}\ket{\Phi^{\bot}},
  \end{align}
  where $\ket{\Phi^{\bot}}$ is a state satisfying $(\ketbra{0}{0}\otimes\ketbra{0}{0}\otimes I)\ket{\Phi^{\bot}} = 0$, and $\gamma$ is some normalisation factor.

  The number of applications of $\mbox{multi-}U$ is constant, as in the definition of $W$. To implement the unitary operator $B$, $O(m)$ 1- and 2-qubit gates suffice.
\end{proof}

Let $W$ be the quantum circuit in Lemma~\ref{lemma:lcu-const}, and let $P$ be the projector defined as $P = \ketbra{0}{0}\otimes\ketbra{0}{0}\otimes I$. We have
\begin{align}
  PW\ket{0}\ket{0}\ket{\psi} = \frac{1}{2}\ket{0}\ket{0}\sum_{j=0}^{m-1}\alpha_jU_j\ket{\psi}.
\end{align}
If $\sum_{j=0}^{m-1}\alpha_jU_j$ is a unitary operator, one application of the oblivious amplitude amplification operator $-W(1-2P)W^{H}(1-2P)W$ implements $\sum_{j=0}^{m-1}U_j$ with certainty~\cite{berry2017exponential}. However, in our application, the unitary operator $\widetilde{W}$ implements an approximation of $V_{\infty}$ in the sense that
\begin{align}
  \label{eq:approx-lcu}
  P\widetilde{W}\ket{0}\ket{0}\ket{\psi} = \frac{1}{2}\ket{0}\ket{0}V_k\ket{\psi},
\end{align}
with $\norm{V_k - V_{\infty}} \leq \epsilon$. The following lemma shows that the error caused by the oblivious amplitude amplification is bounded by $O(\epsilon)$.
\begin{lemma}
  \label{lemma:oaa}
  Let the projector $P$ be defined as above. If a unitary operator $\widetilde{W}$ satisfies $P\widetilde{W}\ket{0}\ket{0}\ket{\psi} = \frac{1}{2}\ket{0}\ket{0}\widetilde{V}\ket{\psi}$ where $\norm{\widetilde{V} - V} \leq \epsilon$. Then $\norm{-\widetilde{W}(I-2P)\widetilde{W}^{H}(I-2P)\widetilde{W}\ket{0}\ket{0}\ket{\psi} - \ket{0}\ket{0}V\ket{\psi}} = O(\epsilon)$.
\end{lemma}

\begin{proof}
We have
\begin{align}
  -\widetilde{W}(I-2P) &\widetilde{W}^{H}(I-2P)\widetilde{W}\ket{0}\ket{0}\ket{\psi} \nonumber \\
  =& (\widetilde{W} + 2P\widetilde{W} - 4\widetilde{W}P\widetilde{W}^{H}P\widetilde{W})\ket{0}\ket{0}\ket{\psi} \nonumber\\
  =& (\widetilde{W} + 2P\widetilde{W} - 4\widetilde{W}P\widetilde{W}^{H}PP\widetilde{W}P)\ket{0}\ket{0}\ket{\psi} \nonumber\\
  =& \widetilde{W}\ket{0}\ket{0}\ket{\psi} + \ket{0}\ket{0}\widetilde{V}\ket{\psi} - \widetilde{W}\left(\ket{0}\ket{0}\widetilde{V}^{H}\widetilde{V}\ket{\psi}\right)
\end{align}
Because $\norm{\widetilde{V}-V} \leq \epsilon$ and $V$ is a unitary operator, we have $\norm{\widetilde{V}^{H}\widetilde{V} - I} = O(\epsilon)$. Therefore, we have
\begin{align}
  \norm{-\widetilde{W}(I-2P)\widetilde{W}^{H}(I-2P)\widetilde{W}\ket{0}\ket{0}\ket{\psi} - \ket{0}\ket{0}\widetilde{V}\ket{\psi}} = O(\epsilon).
\end{align}
Thus
\begin{align}
  \norm{-\widetilde{W}(I-2P)\widetilde{W}^{H}(I-2P)\widetilde{W}\ket{0}\ket{0}\ket{\psi} - \ket{0}\ket{0}V\ket{\psi}} = O(\epsilon).
\end{align}
\end{proof}

Now we are ready to prove Theorem~\ref{thm:densehamiltoniansim}.
\begin{proof}[Proof of Theorem~\ref{thm:densehamiltoniansim}]
The proof we outline here follows closely the proof given in~\cite{berry2015hamiltonian}. The intuition of this algorithm is to divide the simulation into $O(t\pnorm{1}{H})$ segments, with each segment simulating $e^{-iH/2}$. To implement each segment, we use the LCU technique to implement $V_k$ defined in Eq.~\eqref{eq:lcu-vk}, with coefficients $\alpha_m = J_m(z)/\sum_{j=-k}^kJ_j(z)$.
When $z=-1/2$, we have $\sum_{j=-k}^k|\alpha_j| <2$. Actually, this holds for all $|z| \leq 1/2$ because
\begin{align}
    \sum_{j=-k}^k |\alpha_{j}| &\leq \sum_{j=-k}^k \frac{\lvert J_j(z)\rvert }{1-4 \frac{|z/2|^{k+1}}{(k+1)!}}
    \leq \sum_{j=-k}^k \frac{\left\lvert z/2\right\rvert^{|j|}}{|j|!} \left(1 - \frac{4|z/2|^{k+1}}{(k+1)!}\right)^{-1}  \nonumber \\
    &< \frac{8}{7} + \frac{16}{7} \sum_{j=1}^{\infty} \frac{1}{4^j j!} = \frac{8}{7} + \frac{16}{7} \left(\sqrt[4]{e}-1 \right)
    < 2,
\end{align}
where the first inequality follows from the fact that $\sum_{j=-k}^kJ_j(z) \geq 1-4|z/2|^{k+1}/(k+1)!$ (see~\cite{berry2015hamiltonian}), the second inequality follows from the fact that $|J_m(z)| \leq |z/2|^{|m|}/|m|!$ (see~\cite[(9.1.5)]{abramowitz1964handbook}), and the third inequality uses the assumption that $|z|\leq 1/2$.
Now, Lemmas~\ref{lemma:lcu-const} and~\ref{lemma:oaa} can be applied. By Eq.~\eqref{eq:relation_sum_exp} and using Lemma~\ref{lemma:truncate}, set
\begin{align}
  \label{eq:k}
  k = O\left(\frac{\log(\norm{H}/(\pnorm{1}{H}\epsilon'))}{\log\log(\norm{H}/(\pnorm{1}{H}\epsilon'))}\right),
\end{align}
and we obtain a segment that simulates $e^{-iH/(2\pnorm{1}{H})}$ with error bounded by $O(\epsilon')$. Repeat the segment $O(t\pnorm{1}{H})$ times with error $\epsilon' = \epsilon/(t\pnorm{1}{H})$, and we obtain a simulation of $e^{-iHt}$ with error bounded by $\epsilon$. It suffices to take
\begin{align}
  k = O\left(\frac{\log(t\norm{H}/\epsilon)}{\log\log(t\norm{H}/\epsilon)}\right).
\end{align}

By Lemma~\ref{lemma:lcu-const}, each segment can be implemented by $O(1)$ application of $\text{multi-}U$ and $O(k)$ 1- and 2-qubit gates, as well as the cost for computing the coefficients $\alpha_m$ for $m\in\{-k,\ldots,k\}$. The cost for each $\text{multi-}U$ is $k$ times the cost for implementing the quantum walk $U$. By Lemma~\ref{lemma:statepreparation}, the state in Eq.~\eqref{eq:varphi_j0} can be prepared with time complexity (circuit depth) $O(n^2b^{5/2})$, where $b$ is the number of bit of precision. To achieve the overall error bound $\epsilon$, we choose $b=O(\log(t\pnorm{1}{H}/\epsilon))$. Hence the time complexity for the state preparation is $O(n^2\log^{5/2}(t\pnorm{1}{H}/\epsilon))$, which is also the time complexity for applying the quantum walk $U$. Therefore, the time complexity for one segment is
\begin{align}
  O\left(n^2\log^{5/2}(t\pnorm{1}{H}/\epsilon)\frac{\log(t\norm{H}/\epsilon)}{\log\log(t\norm{H}/\epsilon)}\right).
\end{align}
Considering $O(t\pnorm{1}{H})$ segments, the time complexity is as claimed.
\end{proof}

Note that the coefficients $\alpha_{-k}, \ldots, \alpha_k$ (for $k$ defined in Eq.~\eqref{eq:k}) in Lemma~\ref{lemma:lcu-const} can be classically computed using the methods in~\cite{british1960bessel,olver1964error}, and the cost is $O(k)$ times the number of bits of precision, which is $O(\log(t\norm{H}/\epsilon)$. This is no larger than the quantum time complexity.

\paragraph{Discussion.} We have hence seen that we can design a quantum Hamiltonian simulation algorithm
which hast time complexity $\widetilde{O}(\sqrt{N})$ even for non-sparse Hamiltonians.
The algorithm relies heavily on the access to a seemingly powerful input model.
While it is questionable that it is even possible to implement such a data structure physically due to the
exponential amount of quantum resources~\cite{aaronson2015read,adcock2015advances,ciliberto2018quantum},
and this requirement might even be further increased through a potentially strong requirements in terms of hte
error rate per gate of $O(1/\mathrm{poly}(N))$ to retain a feasible error rate for
applications~\cite{arunachalam2015robustness}, for us there is an even more important question:
How fast can classical algorithms be if they are given a similarly
powerful data structure?
We investigate this question in the next section, where we use the so-called Nystr\"om approach to
simulate a Hamiltonian on a classical computer.

\subsection{Hamiltonian Simulation with the Nystr\"om method}
\label{ssec:ham_sim_nystrom}

In this section, we derive a classical, randomised algorithm for the strong
simulation of quantum Hamiltonian dynamics which is based on the \Nystrom{} method
that we introduced in Section~\ref{ssec:related}[Quantum-inspired
classical algorithms for Hamiltonian simulation].
We particularly prove the results in Theorem~\ref{thm:nystrominformal},
and the Corollary~\ref{co:maininformal}.
For this, we develop an algorithm for so-called
strong quantum simulation, where our objective is to obtain an algorithm which
can compute the amplitude of a particular outcome and hence the entries of the
final state after the evolution for time $t$.
This is in contrast to the weak simulation, where we require the algorithm
to only be able to sample from the output distribution of a quantum circuit.
Informally, in the one case we are hence able to query the algorithm with an
index $i$, and the algorithm will be able to return $\psi[i]$ (or more precisely
the projection of the state vector in the $i$-th element of the computational basis)
while in the other case we can only obtain the basis state $i$
with probability $\abs{\psi[i]}^2$.
Strong simulation is known to have unconditional and exponential lower
bounds~\cite{huang2018explicit}.
This implies that it is therefore in general hard for both classical and quantum
computers.
On the other hand, weak simulation can be performed efficiently by a quantum
computer (for circuits of polynomial size).
From the perspective of computational complexity, we are therefore
attempting to solve a stronger problem compared to the Hamiltonian simulation
that the quantum algorithm performs, since the latter can only sample
from the probability distribution induced by measurements on the output state.
Surprisingly, we are still able to find cases in which the time evolution can be simulated efficiently.
Our algorithm is able to efficiently perform the Hamiltonian simulation and output the
the requested amplitude (c.f., Theorem~\ref{thm:nystrominformal}) if we grant it access
to a memory structure which fulfills the following requirements.
Note that these are in principle more general than the assumption of the
classical memory structure from
Definition~\ref{def:classical_datastructure} but the memory structure immediately fulfills
these requirements.
On the other hand, we also require input and row sparsity in order to perform
the simulation efficiently (i.e., in $O(\log(N))$ time for a
$n$-qubit system with dimension $N=2^n$.
We next discuss these requirements in more detail.

\paragraph{Input Requirements}
\label{par:input_requirements}
We assume throughout the chapter that the input matrix and state are sparse, i.e.,
we assume that every row of $H$ has at most $s$ non-zero entries, and that
the input state $\psi$ has at most $q$ non-zero entries.
In order for algorithm to work, we then require the following assumptions:

\begin{enumerate}
  \item We require $H$ to be {\em row-computable}, i.e.,
  there exists a classical efficient algorithm that, given a row-index $i$,
  outputs a list of the non-zero entries (and their indices) of the row.
  If we are not having access to a memory structure such as described in
  Def.~\ref{def:classical_datastructure} or a structured Hamiltonian,
  then this condition is in general only fulfilled if every row of $H$ has at most a
  number $s = O(\polylog(N))$ of non-zero entries.
  \item We require that the entries of the initial state $\psi$ are row-computable, i.e.,
  there exists a classical efficient algorithm that outputs a list of the non-zero entries similar as above.
  In order for this to be generally valid, we require the state to have at most $q=O(\polylog(N))$ non-zero entries. However, we can also use the data structure given in
  Def.~\ref{def:classical_datastructure} or a structured input.
  \item We require $H$ to be efficiently {\em row-searchable}.
  This condition informally states that we can efficiently sample randomly selected indices of the rows of $H$
  in a way proportional to the norm of the row (general case) or the diagonal element of the row (positive semidefinite case).
  This is fundamentally equivalent to the sample-and-query access
  discussed earlier, and indeed, the data structure allows us to immediately perform this
  operation as well.
  We will show this relationship later on.
\end{enumerate}

While the notion of row-searchability is commonly assumed to hold in the randomised numerical
linear algebra literature (e.g.,~\cite{frieze2004fast}[Section $4$]), in the context of quantum
systems this is not given in general, since we are dealing with exponentially sized matrices.
It is of course reasonable to assume that for a polynomially sized matrix we are indeed able
to evaluate all the row-norms efficiently in a time (number of steps)
proportional to the number of non-zero entries.

In order to obtain an efficient algorithm (i.e., one which only depends logarithmically on the dimension $N = 2^n$ for a $n$-qubit system),
we require sparsity of the Hamiltonian $H$ and the input state $\psi$.
In general we therefore require $s=q=O(\polylog(N))$ for the non-zero entries
in order for the algorithm to be efficient.

\paragraph{High-level description of the algorithm}
The algorithm proceeds by performing a two steps approximation.
First, we approximate the Hamiltonian $H$ in terms of a low rank operator $\widehat{H}$
which is small, and therefore more amendable for the computations which we perform next.
We obtain $\widehat{H}$ by sampling the rows with a probability proportional to the row norm,
and hence do rely on a \Nystrom{} scheme.
Second, we approximate the time evolution of the input state $e^{i\widehat{H} t}\psi$ via
a truncated Taylor expansion of the matrix exponential.
This can be efficiently performed since we can use the small operator $\widehat{H}$ and
the spectral properties of the truncated exponential..

As an addition, we separately consider the case of a generic Hamiltonian $H$ and the restricted case of positive semidefinite Hamiltonians.
As is the case in general, for the more restricted case of the PSD Hamiltonian, we are able to
derive tighter bounds.

We will briefly discuss the sampling scheme to obtain $\tilde{H}$ to give the reader an
idea of the overall procedure before going into the detailed proofs.

In both cases, our algorithm leverages on a low-rank approximation of the Hamiltonian $H$ to
efficiently approximate the matrix exponential $e^{i Ht}$ which are obtained by randomly
sampling $M = O(\polylog(N))$ rows according to the magnitude of the row norms, and then
collating them in a matrix $A\in \mathbb C^{M \times N}$.
Let now $A\in \mathbb C^{M \times N}$ be the matrix obtained by sub-sampling $M$ rows of $H$,
and $B\in\mathbb C^{M \times M}$ the matrix obtained by selecting the columns of $A$
whose indices correspond to those $M$ indices for the rows of $H$.

For SPD $H$, we then use an approximation of the form $\widehat{H}= AB^+ A^H$,
where $B^+$ is the pseudoinverse of the SPD matrix $B$.

Next, in order to perform the time evolution $e^{i\widehat{H}t}\psi$ by truncating the Taylor
series expansion of the matrix exponential function after the $K$-th order.
In order to apply the individual terms in the truncated Taylor series, we then
make use of the structure of $\widehat{H}$, and formulate the operator only in terms of linear
operations involving the matrices $A^HA$, $B^+$ and $B$ and the vector $A^H\psi$.
Under the above assumptions (Paragraph~\ref{par:input_requirements}[Input Requirements]),
and for $s=q=O(\mathrm{polylog}(N))$ all these operations can be performed efficiently.

In the general Hermitian setting, we form $A$ by first sampling the rows of $H$
and then rescale the sampled rows according to their sampling probability.
Unlike in the case of SPD $H$, we approximate the Hamiltonian by the
approximation $\widehat{H}^2 := AA^H$ to approximate $H^2$.
This approximation is useful here, since we decompose the matrix exponential
into two auxiliary functions, and then for each of these evaluate its truncated
Taylor series expansion.
By doing so, we can formulate the final approximation solely in terms of linear
operations involving $A^HA$ and $A^H\psi$.
These operations can then again be performed efficiently under the initial assumptions
oh $H$ and $\psi$.

\paragraph{Row-searchability implies efficient row-sampling}
\label{par:sampling}

As we see in the assumptions, all our algorithms require the Hamiltonian
to be row-searchable.
In this section we describe an efficient algorithm for sampling
rows of a row-searchable Hamiltonians according to some probability distribution.
Let $n \in \N$.
We first introduce a binary tree of subsets spanning $\{0,1\}^n$.
In the following, with abuse of notation, we identify binary tuples with the associated binary number.
Let $L$ be a binary string with $|L| \leq n$, where $|L|$ denotes the length of the string.
We denote with $S(L)$ the set
\begin{align}
S(L) = \{L\} \times \{0,1\}^{n - |L|} = \{(L_1,\dots, L_{|L|}, v_1,\dots,v_{n - |L|}) ~|~ v_1,\dots,v_{n - |L|} \in \{0,1\} \}.
\end{align}
We are now ready to state the row-searchability property for a matrix $H$.

\begin{definition}[Row-searchability]
  \label{cond:row-search}
  Let $H$ be a Hermitian matrix of dimension $2^n$, for $n \in \N$.
  $H$ is {\em row-searchable} if, for any binary string $L$ with $|L| \leq n$,
  it is possible to compute the following quantity in $O(\poly(n))$
  \begin{equation}
    \label{norm_diag}
    w(S(L)) = \sum_{i \in S(L)} h(i,H_{:,i}),
  \end{equation}
  where $h$ is the function computing the weight associated to the $i$-th column $H_{:,i}$.
  For positive semidefinite $H$ we use $h(i,H_{:,i}) = H_{i,i}$,
  i.e.\ the diagonal element $i$ while
  for general Hermitian $H$ we use $h(i, H_{:,i}) = \|H_{:,i}\|^2$.
\end{definition}
Row-searchability intuitively works as follows:
We are working with the following binary tree corresponding to a Hamiltonian $H$.
We start at the leaves of the tree, which contain the individual probabilities
according to which we want to sample from $H$.
The parents at each level are then given (and computed)
by the marginals over their children nodes, i.e., the sum over the probabilities
of the children (assuming discrete probabilities).
Using this tree, we can for a randomly sampled number in $[0,1]$ traverse through
the levels of the tree in $\log(N)$ time to find the leave node that is sampled,
i.e.\ the indices of the column of $H$.
More specifically, row-searchability then requires the evaluation of $w(S(L)))$
as defined in Eq.~(\ref{norm_diag}) which computes marginals of the diagonal of $H$
or the norm $\|H_{:,i}\|^2$ in the general case,
where the co-elements, \textit{i.e.} the elements where we are not summing over,
are defined by the tuple $L$.
Hence, for empty $L$, $w(S(L)) = \tr{H}$.
Note that this assumption is obviously closely related to the data structure given
in Definition~\ref{def:classical_datastructure}, and indeed we can use this
data structure to immediately perform these operations.

%

\begin{algorithm}
\caption{MATLAB code for the sampling algorithm\label{alg:sampling}}
\begin{flushleft}
{\bf Input:} \texttt{wS(L)} corresponds to the function $w(S(L))$ defined in Eq.~\ref{norm_diag}.\\
{{\bf Output:} $L$ is the sampled row index}
\end{flushleft}
\begin{center}
\begin{verbatim}
L = [];
q = rand()*wS(L);

for i=1:n
  if q >= wS([L 0])
      L = [L 0];
  else
      L = [L 1];
      q = q - wS([L 0]);
  end
end
\end{verbatim}
\end{center}
\end{algorithm}

Note that the function $h$ that we described above is indeed related to leverage score sampling,
which is a widely used sampling process in randomised numerical linear algebra~\cite{mahoney2011randomised,woodruff2014sketching}.
Leverage scores allow us to efficiently obtain sample probabilities which have a
sufficiently low variance to obtain fast algorithms with low error.

Alg.~\ref{alg:sampling} describes an algorithm, that, given a row-searchable $H$,
is able to sample an index with probability $p(j) = h(j, H_{:,j})/w(\{0,1\}^n)$.
Let $q$ be a random number uniformly sampled in $[0,T]$, where $T = w(\{0,1\}^n)$
is the sum of the weights associated to all the rows.
The algorithm uses logarithmic search, starting with $L$ empty and adding
iteratively $1$ or $0$, to find the index $L$ such that
$w(\{0, \dots, L-1\}) \leq q \leq w(\{0, \dots, L\})$.
The total time required to compute one index, is $O(n Q(n))$ where $Q(n)$ is
the maximum time required to compute a $w(S(L))$ for $L \in \{0,1\}^n$.
Note that if $w(S(L))$ can be computed efficiently for any $L \in \{0,1\}^n$,
then $Q(n)$ is polynomial and the cost of the sampling procedure will be polynomial.

\begin{remark}[Row-searchability more general than sparsity]
Note that if $H$ has a polynomial number of non-zero elements,
then $w(S(L))$ can be always computed in polynomial time.
Indeed given $L$, we go through the list of elements describing $H$
and select only the ones whose row-index starts with $L$ and then compute $w(S(L))$,
both step requiring polynomial time.
However $w(S(L))$ can be computed efficiently even for Hamiltonians that are not
polynomially sparse.
For example, take the diagonal Hamiltonian defined by $H_{ii} = 1/i$ for $i \in [2^n]$.
This $H$ is not polynomially sparse and in particular it has an exponential
number of non-zero elements, but still $w(S(L))$ can be computed in polynomial time,
here in particular in $O(1)$.
As it turns out, the data structure from Definition~\ref{def:classical_datastructure},
which was first proposed by Tang~\cite{tang2018quantum} to do a similar sampling
process efficiently, also allows us to perform this operation efficiently with the
difference, that this holds then true for arbitrary Hamiltonians.
\end{remark}

\paragraph{Algorithm for PSD row-searchable Hermitian matrices}
\label{par:PSD}

Given a $2^n \times 2^n$ (i.e., $N \times N$) matrix $H \succeq 0$,
the algorithm should output an approximation for the state given by the
time evolution
\begin{equation}\label{eq:true-state}
\psi(t) = \exp{(i H t)} \psi.
\end{equation}
The algorithm will do this through an expression of the form $\widehat{\exp}(i \widehat{H} t) \psi$,
where $\widehat{\exp}$ and $\widehat{H}$ are an approximation of the exponential function and
the low rank approximation of $H$ respectively.

The first algorithm we describe here applies for for $H\succeq 0$.
We will then generalise this result in the following section to arbitrary Hermitian $H$.
All our results hold under the row-searchability condition, i.e.
if condition~\ref{cond:row-search} is fulfilled.

Let $h$ be the diagonal of the positive semidefinite $H$ and let $t_1,\dots, t_M$, with $M \in \N$
be indices i.i.d. sampled with repetition from $\{1,\dots, 2^n\}$
according to the probabilities
\begin{equation}
  \label{eq:probabilities}
  p(q) = h_q / \sum_i h_i,
\end{equation}
e.g. via Alg.~\ref{alg:sampling}.
Then, let $B \in \CC^{M \times M}$ where $B_{i,j} = H_{t_i, t_j}$, for $1 \leq i,j \leq M$.
Let furthermore $A \in \CC^{2^n \times M}$ be the matrix with $A_{i,j} = H_{i,t_j}$
for $1 \leq i \leq 2^n$ and $1 \leq j \leq M$.
We then define the approximation for $H$ by $\widehat{H} = A B^{+} A^H$,
where $(\cdot)^+$ denotes again the pseudoinverse.
We now define a function $g(x) = (e^{i t x} - 1)/x$.
We can reformulate this, and immediately have $e^{i t x} = 1 + g(x)x$,
and note that $g$ is an analytic function, i.e.,
it has a series expansion
\[
g(x) = \sum_{k\geq 1} \frac{(i t)^k}{k!} x^{k-1}.
\]
Then
\begin{align}
e^{i\widehat{H}t} = I + g(\widehat{H})\widehat{H} = I + g(A B^{+} A^H)A B^{+} A^H = I + A g(B^{+} A^H A) B^{+} A^H,
\end{align}
where the last step is due to the fact that for any analytic function
$q(x) = \sum_{k \geq 0} \alpha_k x^k$, it holds that
\begin{align*}
q(A B^{+} A^H)A B^{+} A^H &= \sum_{k\geq 1} \alpha_k (A B^{+} A^H)^{k} A B^{+} A^H \\
& = A \sum_{k\geq 1} \alpha_k (B^{+} A^H A)^k B^{+} A^H = A q(B^{+} A^H A) B^{+} A^H.
\end{align*}
By writing $D = B^{+} A^H A$, the algorithm then performs the
operation
\[
\widehat{\psi}_M(t) = \psi +  A g(D) B^{+} A^H \psi.
\]

Next, in order to make the computation feasible,
we truncate the series expansion after a finite number of terms.
To do this, we hence approximate $g$ with $g_K(x)$,
which limits the series defining $g$ to the first $K$ terms, for $K \in \N$.
Moreover, we can evaluate the function $g_K(D)B^{+}(A^H\psi)$
in an iterative fashion, and for this chose
\[
b_j = \frac{(i t)^{K-j}}{(K-j)!}v + D b_{j-1}, \quad v =  B^{+} (A^H \psi),
\]
where $b_0 = \frac{(i t)^{K}}{K!}v$ and so
$b_{K-1} = g(D) B^{+} A^H \psi $.
Then, the new approximate state is given by
\begin{equation}\label{eq:algo}
\widehat{\psi}_{K, M}(t) = \psi +  A b_{K-1}.
\end{equation}
\begin{algorithm}[t]
\caption{MATLAB code for approximating Hamiltonian dynamics when $H$ is PSD\label{alg:Nystrom}}
\begin{flushleft}
{\bf Input:} $M$, $T = t_1,\dots, t_M$ list of indices computed via Alg.~\ref{alg:sampling}. The function \texttt{compute\_H\_subMatrix}, given two lists of indices, computes the associated submatrix of $H$. \texttt{compute\_psi\_subVector}, given a list of indices, computes the associated subvector of $\psi$.\\
{\bf Output:} vector $b$, s.t. $b=e^{iHt}\psi$.
\end{flushleft}
\begin{center}
\begin{verbatim}
B = compute_H_subMatrix(T, T);

D = zeros(M,M);
v = zeros(M,1);
for i=1:(2^n/M)
    E = compute_H_subMatrix((i-1)*M+1:M*i, T);
    D = D + E'*E;
    v = v + E'*compute_Psi_subVector((i-1)*M+1:M*i);
end
u
D = B\D;

b = zeros(M,1);
for j=1:K
    b = (1i*t)^(K-j)/factorial(K-j) * v + D*b;
end
\end{verbatim}
\end{center}
\end{algorithm}
We summarise the algorithm in form of a MATLAB implementation in Alg.~\ref{alg:Nystrom}. \\
Next, we analyse the cost of this algorithm.
Let the row sparsity $s$ of $H$ be of order $\poly(n)$.
Then the total cost of applying this operator is given by $O(M^2 \poly(n) + KM^2 + M^3) )$
time complexity, where the terms $M^3$ and $M^2 \poly(n)$ are resulting from the calculation
of $D$ and the inverse.
To compute the total cost in terms of space/memory,
note that we do not have to save $H$ or $A$ in memory,
but only $B, D$ and the vectors $v, b_j$, which requires a total cost of $O(M^2)$.
Indeed $D$ can be computed in the following way:
Assuming, without loss of generality, to have $2^n/M \in \N$, then
\[
D = B^{-1} \sum_{i=1}^{2^n/M} A_{M(i-1)+1:Mi}^HA_{M(i-1)+1:Mi},
\]
where $A_{a:b}$ is the submatrix of $A$ containing the rows from $a$ to $b$.
A similar reasoning holds for the computation of the vector $v$.
In the above computation we have assumed that we can efficiently sample
from the matrix $H$ according to the probabilities in Eq.~\ref{eq:probabilities}.
In order to do this, we are relying on the sampling algorithm which is
summarised in Algorithm~\ref{alg:sampling}.

We next analyse the errors and complexity of the algorithm in more detail,
i.e., we derive bounds on $K$ and $M$ for a concrete approximation error
$\epsilon$.
The results are summarised in the following theorem,
which hold for the SPD case of $H$.

\begin{theorem}[Algorithm for simulating PSD row-searchable Hermitian matrices]
\label{thm:main_psd}
Let $\epsilon, \delta \in (0,1]$, let $K, M \in \N$ and $t > 0$. Let $H$ be positive semidefinite, where $K$ is the number of terms in the truncated series expansions
of $g(\hat H)$ and $M$ the number of samples we take for the approximation.
Let $\psi(t)$ be the true evolution (Eq.~\ref{eq:true-state}) and let $\widehat{\psi}_{K, M}(t)$ be the output of our Alg.~\ref{alg:Nystrom} (Eq.~\ref{eq:algo}).
When
\begin{equation}
K \geq e\, t\norm{H} + \log\frac{2}{\epsilon},\qquad
\quad M \geq \max\left(405 \tr{H},~ \frac{72\tr{H} t}{\epsilon} \log\frac{36\tr{H} t}{\epsilon\delta}\right),
\end{equation}
then the following holds with probability $1-\delta$,
\[\norm{\psi(t) - \widehat{\psi}_{K,M}(t)} \leq \epsilon.\]
\end{theorem}
Note that with the result above, we have that $\widehat{\psi}_{K,M}(t)$ in Eq.~\eqref{eq:algo} (Alg.~\ref{alg:Nystrom}) approximates $\psi(t)$, with error at most $\epsilon$ and with probability at least $1-\delta$, requiring a computational cost that is $\Ord{\frac{s t^2\tr{H}^2}{\epsilon^2} \log^2\frac{1}{\delta}}$ in time and $\Ord{\frac{t^2\tr{H}^2}{\epsilon^2} \log^2\frac{1}{\delta}}$ in memory.

In the following we now prove the first main result of this work.
To prove Theorem~\ref{thm:main_psd} we decompose the error into multiple
contributions. Lemma~\ref{lemma:base-decomposition} performs a basic decomposition
of the error in terms of the distance between $H$ and
the approximation $\widehat{H}$ as well as in terms of the approximation $g_K$ with respect to $g$.
Lemma~\ref{lemma:nystrom-analytic-bound} then provides an analytic bound on the
distance between $H$ and $\widehat{H}$, expressed in terms of the expectation of
eigenvalues or related matrices which are then concentrated in Lemma~\ref{lemma:nystrom-probabilistic-bound}.

\begin{lemma}\label{lemma:base-decomposition}
Let $K, M \in \N$ and $t > 0$, then
\[\norm{\psi(t) - \widehat{\psi}_{K, M}(t)} \leq t \norm{H - \widehat{H}} + \frac{(t\norm{\widehat{H}})^{K+1}}{(K+1)!}.\]
\end{lemma}
\begin{proof}
By definition we have that $e^{ixt} = 1 + g(x)x$ with
$g(x) = \sum_{k \geq 1} x^{k-1}(i t)^k/k!$ and $g_K$ is the truncated version of $g$. By adding and subtracting $e^{i \widehat{H} t}$, we have
\begin{equation}
\norm{e^{iHt}\psi - (I+g_K(\widehat{H})\widehat{H})\psi}
\leq \norm{\psi}\ (\norm{e^{iHt} - e^{i\widehat{H}t}} + \norm{e^{i\widehat{H}t} - (I + g_K(\widehat{H})\widehat{H})}).
\end{equation}
By \cite{nakamoto2003norm},
\begin{align}
\norm{e^{iHt} - e^{i\widehat{H}t}} \leq t \norm{H - \widehat{H}},
\end{align}
moreover, by \cite{mathias1993approximation}, and since $\widehat{H}$ is Hermitian and hence all the eigenvalues are real, we have
\begin{align}
\norm{e^{i\widehat{H}t} - (I + g_K(\widehat{H})\widehat{H})} \leq \frac{(t\norm{\widehat{H}})^{K+1}}{(K+1)!}\sup_{l \in [0,1]}\norm{i^{K+1}e^{i l \widehat{ H}t}} \leq \frac{(t\norm{\widehat{H}})^{K+1}}{(K+1)!}.
\end{align}
Finally note that $\norm{\psi} = 1$.
\end{proof}
To study the norm $\norm{H - \widehat{H}}$ note that, since $H$ is positive semidefinite, there exists an operator $S$ such that $H = S S^H$, so $H_{i,j} = s_i^H s_j$ with $s_i,s_j$ the $i$-th and $j$-th row of $S$.
Denote with $C$ and $\widetilde{C}$ the operators
\[C = S^HS, \quad \widetilde{C} = \frac{1}{M}\sum_{j=1}^M \frac{\tr{H}}{h_{t_j}} s_{t_j}s_{t_j}^H.\]
We then obtain the following result.
\begin{lemma}\label{lemma:nystrom-analytic-bound}
The following holds with probability $1$. For any $\tau > 0$,
\begin{equation}
\norm{H - \widehat{H}} \leq \frac{\tau}{1-\beta(\tau)}, \quad
\quad \beta(\tau) = \lambda_{\max}((C+\tau I)^{-1/2}(C- \widetilde{C})(C + \tau I)^{-1/2}),
\end{equation}
moreover $\norm{\widehat{H}} \leq \norm{H}$.
\end{lemma}
\begin{proof}
Define the selection matrix $V \in \CC^{M\times 2^n}$, that is always zero except for one element in each row which is $V_{j, t_j} = 1$ for $1 \leq j \leq M$. Then we have that
\[A = HV^H, \quad B = V H V^H,\]
i.e., $A$ is again given by the rows according to the sampled indiced $t_1,\ldots, t_M$ and $B$ is the submatrix obtained from taking the rows and columns according to the same indices.
In particular by denoting with $\widehat{P}$ the operator $\widehat{P} = S^HV^H(VSS^HV^H)^+ VS$, and recalling that $H=SS^H$ and $C=S^HS$, we have
\[\widehat{H} = A B^+ A^H = SS^HV^H(VSS^HV^H)^+ VSS^H = S \widehat{P} S^H.\]
By definition $\widehat{P}$ is an orthogonal projection operator, indeed it is symmetric and, by definition $Q^+ Q Q^+ = Q^+$, for any matrix $Q$, then
\begin{align}
\widehat{P}^2  = S^HV^H[(VSS^HV^H)^+ (VSS^HV^H)(VSS^HV^H)^+] VS = S^HV^H (VSS^HV^H)^+  VS = \widehat{P}.
\end{align}
Indeed this is a projection in the row space of the matrix $R:=S^HV^H$, since with the singular value decomposition $R := U_R \Sigma_R V_R^H$ we have
\begin{align}
\hat{P} = R^H(R^HR)^+ R = V_R \Sigma_R U_R^H U_R \Sigma^{-2}_R U_R^H U_R \Sigma_R V_R^H = V_RV_R^H,
\end{align}
which spans the same space as $R$.
Finally, since $(I - \widehat{P}) = (I - \widehat{P})^2$, and $\norm{Z^HZ} = \norm{Z}^2$, we have
\begin{align}
\norm{H - \widehat{H}} = \norm{S(I - \widehat{P})S^H} = \norm{S(I - \widehat{P})^2S^H} = \norm{(I - \widehat{P})S^H}^2.
\end{align}
Note that $\widetilde{C}$ can be rewritten as $\widetilde{C} = S^HV^HLVS$, with $L$ a diagonal matrix, with $L_{jj} = \frac{\tr{H}}{M h_{t_j}}$. Moreover $t_j$ is sampled from the probability $p(q) = h_q/\tr{H}$, so $h_{t_j} > 0$ with probability 1, then $L$ has a finite and strictly positive diagonal, so $\widetilde{C}$ has the same range of $\widehat{P}$.
Now, with  $C = S^HS$, we are able to apply Proposition 3 and Proposition 7 of \cite{rudi2015less}, and obtain
\begin{equation}
\norm{(I - \widehat{P})S^H}^2 \leq \frac{\tau}{1-\beta(\tau)}, \quad
\quad \beta(\tau) = \lambda_{\max}((C+\tau I)^{-1/2}(C- \widetilde{C})(C + \tau I)^{-1/2}).
\end{equation}
Finally, note that, since $\widehat{P}$ is a projection operator we have that $\norm{\widehat{P}} = 1$, so
\[\norm{\widehat{H}} = \norm{S \widehat{P} S^H} \leq \norm{\widehat{P}}\norm{S}^2 \leq \norm{S}^2 = \norm{H},\]
where the last step is due to the fact that $H = SS^H$.
\end{proof}

\begin{lemma}\label{lemma:nystrom-probabilistic-bound}
Let $\delta \in (0,1]$ and $\tau > 0$.
When
\begin{equation}
M \geq \max\left(405 \tr{H},~ 67 \tr{H} \log \frac{\tr{H}}{2\delta}\right), \quad
\quad \tau = \frac{9\tr{H}}{M} \log \frac{M}{2\delta},
\end{equation}
then with probability $1-\delta$ it holds that
\[\lambda_{\max}((C+\tau I)^{-1/2}(C- \widetilde{C})(C + \tau I)^{-1/2}) \leq \frac{1}{2}.\]
\end{lemma}
\begin{proof}
Define the random variable $\zeta_j = \sqrt{\frac{\tr{H}}{h_{t_j}}} s_{t_j},$
for $1 \leq j \leq M$. Note that
\[\norm{\zeta_j} \leq \sqrt{\frac{\tr{H}}{h_{t_j}}}\norm{s_{t_j}} \leq \sqrt{\tr{H}},\]
almost surely. Moreover,
\begin{equation}
\mathbb{E} \zeta_j \zeta_j^H = \sum_{q=1}^{2^n} p(q) \frac{\tr{H}}{h_q} s_{q} s_q^H
= \sum_{q=1}^{2^n} s_{q} s_q^H
= S^HS
= C.
\end{equation}
By definition of $\zeta_j$, we have
\[ \widetilde{C} = \frac{1}{M} \sum_{j=1}^M \zeta_j \zeta_j^H.\]
Since $\zeta_j$ are independent for $1 \leq j \leq M$, uniformly bounded, with expectation equal to $C$, and with $\zeta_j^H (C + \tau I)^{-1} \zeta_j \leq \norm{\zeta_j}^2 \tau^{-1} \leq \tr{H}\tau^{-1}$, we can apply Proposition 8 of \cite{rudi2015less}, that uses non-commutative Bernstein inequality for linear operators \cite{tropp2012user}, and obtain
\begin{align}
\lambda_{\max}((C+\tau I)^{-1/2}(C- \widetilde{C})(C + \tau I)^{-1/2}) \leq \frac{2 \alpha}{3 M} + \sqrt{\frac{2\alpha}{M t}},
\end{align}
with probability at least $1 - \delta$, with $\alpha = \log \frac{4 \tr C}{\tau \delta}$.
Since \[\tr{C} = \tr{S^HS} = \tr{SS^H} = \tr{H},\] by Remark~1 of \cite{rudi2015less}, we have that
\[\lambda_{\max}((C+\tau I)^{-1/2}(C- \widetilde{C})(C + \tau I)^{-1/2}) \leq \frac{1}{2},\]
with probability $1-\delta$, when
$M \geq \max(405 \kappa^2, 67 \kappa^2 \log \frac{\kappa^2}{2\delta})$ and $\tau$ satisfies $\frac{9\kappa^2}{M} \log \frac{M}{2\delta} \leq \tau \leq \norm{C}$ (note that $\norm{C} = \norm{H}$), where $\kappa^2$ is a bound for the following quantity
\begin{equation}
\begin{split}
\inf_{\tau > 0}[(\norm{C} + \tau) (\mathrm{ess}\sup \zeta_j^H (C + \tau I)^{-1} \zeta_j)] \\
\leq \tr{H} \inf_{\tau > 0} \frac{\norm{H} + \tau}{\tau} \leq \tr{H} := \kappa^2,
\end{split}
\end{equation}
where $\mathrm{ess}\sup$ here denotes the essential supremum.
\end{proof}
Now we are ready to prove Theorem~\ref{thm:main_psd}.
\begin{proof}[Proof of Theorem~\ref{thm:main_psd}]
By Lemma~\ref{lemma:base-decomposition}, we have
\[\norm{e^{iHt}\psi - (I+g_K(\widehat{H})\widehat{H})\psi} \leq t \norm{H - \widehat{H}} + \frac{(t\norm{\widehat{H}})^{K+1}}{(K+1)!}.\]
Let $\tau > 0$. By Lemma~\ref{lemma:nystrom-analytic-bound}, we know that $\norm{\widehat{H}} \leq \norm{H}$ and that
\[\norm{H - \widehat{H}} \leq \frac{\tau}{1-\beta(\tau)},\]
\[\quad \beta(\tau) = \lambda_{\max}((C+\tau I)^{-1/2}(C- \widetilde{C})(C + \tau I)^{-1/2}),\]
with probability $1$.
Finally by Lemma~\ref{lemma:nystrom-probabilistic-bound}, we have that the following holds with probability $1-\delta$,
\[\lambda_{\max}((C+\tau I)^{-1/2}(C- \widetilde{C})(C + \tau I)^{-1/2}) \leq \frac{1}{2},\]
when
\[ M \geq \max\left(405 \tr{H},~ 67 \tr{H} \log \frac{\tr{H}}{2\delta}\right)\]
and
\[\quad \tau = \frac{9\tr{H}}{M} \log \frac{\tr{H}}{2\delta}.\]
So we have
\[ \norm{e^{iHt}\psi - (I+g_K(\widehat{H})\widehat{H})\psi} \leq \frac{18 \tr{H}t}{M} \log\frac{M}{2\delta} + \frac{(t\norm{H})^{K+1}}{(K+1)!},\]
with probability $1-\delta$.\\
Now we select $K$ such that $\frac{(t\norm{H})^{K+1}}{(K+1)!} \leq \frac{\epsilon}{2}$. Since, by the Stirling approximation, we have
\[(K+1)! \geq \sqrt{2\pi} (K+1)^{K+3/2} e^{-K-1} \geq  (K+1)^{K+1} e^{-K-1}.\]
Since \[(1+x)\log(1/(1+x)) \leq -x,\] for $x > 0$ we can select $K = e t \norm{H} + \log\frac{2}{\epsilon} - 1$, such that we have

\begin{align*}
\log\left(\frac{(t\norm{H})^{K+1}}{(K+1)!}\right) &\leq (K+1) \log \frac{e t \norm{H}}{K+1} \\
&\leq e t \norm{H}\left(1 + \frac{\log\frac{2}{\epsilon}}{e t \norm{H}}\right) \log \frac{1}{1 + \frac{\log\frac{2}{\epsilon}}{e t \norm{H}}} \\
&\leq \log\frac{\epsilon}{2}.
\end{align*}
Finally we require $M$, such that \[\frac{18 \tr{H}t}{M} \log\frac{M}{2\delta} \leq \frac{\epsilon}{2},\] and select \[M = \frac{72\tr{H} t}{\epsilon} \log\frac{36\tr{H} t}{\epsilon\delta}.\]
Then we have that
\begin{equation*}
\frac{18 \tr{H}t}{M} \log\frac{M}{2\delta}
\leq \frac{\epsilon}{2} \frac{\log\frac{36\tr{H} t}{\epsilon\delta} +  \log\log\frac{36\tr{H} t}{\epsilon\delta}}{2 \log\frac{36\tr{H} t}{\epsilon\delta}}
\leq \frac{\epsilon}{2}.
\end{equation*}

\end{proof}

Next, we generalise this results to arbitrary Hermitian matrices
under the assumption that these are row-searchable,
i.e.\ assuming the ability to sample according to some leverage of the rows.
This will lead to our second result for classical Hamiltonian simulation with the
\Nystrom{} method.

\paragraph{Algorithm for row-searchable Hermitian matrices}
\label{par:hermitian_case}

As mentioned, in this section we now generalise our previous result and derive
an algorithm for simulating arbitrary Hermitian matrices.
We again provide guarantees on the runtime and errors of the algorithm,
generally under the assumption that $H$ is row-searchable.
Our algorithm for the general case has slightly worse guarantees compared to the
SPD case, which is to be expected.
Let in the following again $s$ be the maximum number of non-zero elements
in any of the rows of $H$,
$\epsilon$ be the error in the approximation of the output states of the
algorithm w.r.t.\ the ideal $\psi(t)$, and
$t$ the evolution time of the simulation.
Let further $K$ be the order of the truncated series expansions and $M$ the number
of samples we take for the approximation.

In the following we again start by describing the algorithm, and then
derive bounds on the runtime and error.

For arbitrary matrices $H$ we will use the following algorithm.
Sample $M \in \N$ independent indices $t_1, \dots t_M$, with probability
$p(i) = \frac{\|h_i\|^2}{\|H\|_F^2}$, $1 \leq i \leq 2^n$, where $h_i$ is the
$i$-th row of $H$ (sample via Alg.~\ref{alg:sampling}).
Let $A \in \mathbb{C}^{2^n \times M}$ be the matrix defined by
\[
A = \left[\frac{1}{\sqrt{Mp(t_1)}}h_{t_1}, \dots, \frac{1}{\sqrt{Mp(t_1)}} h_{t_M}\right].
\]
We then approximate $H$ via the matrix $\hat{H}^2=AA^H$.\\
Next, we again define two functions that we will use to approximate the
exponential $e^{ix}$,
\[
f(x) = \frac{\cos(\sqrt{x}) - 1}{x}, \quad g(x) = \frac{\sin(\sqrt{x}) - \sqrt{x}}{x\sqrt{x}},
\]
and denote with $f_K$ and $g_K$ the $K$-truncated Taylor expansions of $f$ and $g$,
for $K \in \N$, i.e.,
\[
f_K(x) = \sum_{j=0}^K \frac{(-1)^{j+1}x^j}{(2j+2)!}, \quad g_K(x) = \sum_{j=0}^K \frac{(-1)^{j+1}x^j}{(2j+3)!}.
\]
In particular note that
\[
e^{ix} = 1 + ix + f(x^2)x^2 + i g(x^2)x^3.
\]

Similar in  spirit to the previous approach for SPD $H$,
we hence approximate $e^{ix}$ via the functions $f_K$ and $g_K$.
The final approximation is then given by
\begin{equation}\label{eq:algo-frob}
\widehat{\psi}_{K,M}(t) = \psi + i t u + t^2A f_K(t^2A^HA)v + i t^3 A g_K(t^2A^HA) z,
\end{equation}
where $u = \hat H\psi$, $v = A^H\psi$, $z = A^Hu$.
The products $f_k(A^HA)v$ and $A g_K(A^HA) z$ are done by again exploiting
the Taylor series form of the two functions and performing only matrix vector products
similar to Alg.~\ref{alg:Nystrom}.
Recall that $s$ is the maximum number of non-zero elements in the rows of $H$,
and $q$ the number of non-zero elements in $\psi$.
The algorithm then requires $O(sq)$ in space and time to compute $u$,
$O(M\min(s,q))$ in time and $O(M)$ in space to compute $v$ and $O(Ms)$ in time and space to compute $z$.
We therefore obtain a total computational complexity of
\begin{align}
&{\rm time:}~~O\left(sq + M\min(s,q) + sMK\right), \\
&\quad {\rm space:}~~O\left(s(q + M)\right).
\end{align}

Note that if $s > M$ is it possible to further reduce the memory requirements at
the cost of more computational time, by computing $B = A^HA$, which can be done
in blocks and require $O(sM^2)$ in time and $O(M^2)$ in memory, and then compute
\[
\widehat{\psi}_{K,M}(t) = \psi + i t u + t^2A f_K(t^2B)v + i t^3 A g_K(t^2B) z.
\]
In this case the computational cost would be
\begin{align}
\label{eq:Timquation}
&{\rm time:}~~O\left(sq + M\min(s,q) + M^2(s +K)\right), \\
\label{eq:Spquation}
&\quad {\rm space:}~~O\left(sq + M^2\right).
\end{align}
The properties of the this algorithm are summarised in the following theorem,
which is a formal statement of Theorem~\ref{thm:nystrominformal}:

\begin{theorem}[Algorithm for simulating row-samplable Hermitian matrices]
\label{thm:main}
Let $\delta, \epsilon \in (0,1]$. Let $t > 0$ and $K, M \in \mathbb{N}$, where $K$ is the number of terms in the truncated series expansions
of $g(\widehat{H})$ and $M$ the number of samples we take for the approximation, and let $t > 0$.
Let $\psi(t)$ be the true evolution (Eq.~\ref{eq:true-state}) and let $\widehat{\psi}_{K,M}(t)$ be computed as in Eq.~\ref{eq:algo-frob}. When
\begin{align}
\label{eq:Mquation}
&M \geq \frac{256t^4(1+t^2\norm{H}^2)\norm{H}^2_F \norm{H}^2}{\epsilon^2} \log \frac{4\norm{H}^2_F}{\delta\norm{H}^2},\\
\label{eq:Kquation}
&\quad K \geq 4t \sqrt{\norm{H}^2 + \epsilon} + \log \frac{4(1+t\norm{H})}{\epsilon},
\end{align}
then
\[\norm{\widehat{\psi}_{K,M}(t) - \psi(t)} \leq \epsilon,\]
with probability at least $1-\delta$.
\end{theorem}

Note that with the result above, we have that $\widehat{\psi}_{K,M}(t)$ in Eq.~\eqref{eq:algo-frob} approximates $\psi(t)$, with error at most $\epsilon$ and with probability at least $1-\delta$, requiring a computational cost that is $O\left(sq + M\min(s,q) + M^2(s +K)\right)$ in time and $O\left(sq + M^2\right)$ is memory.

Combining Eq.~\ref{eq:Timquation}  and~\ref{eq:Spquation} with Eq.~\ref{eq:Mquation} and~\ref{eq:Kquation}, the whole computational complexity of the algorithm described in this section, is

\begin{align}
&{\rm time:}~~O\left(sq + \frac{t^9\norm{H}^4_F \norm{H}^7}{\epsilon^4}\left(n + \log\frac{1}{\delta}\right)^2\right),\\
&{\rm space:}~~O\left(sq + \frac{t^8\norm{H}^4_F \norm{H}^6}{\epsilon^4}\left(n + \log\frac{1}{\delta}\right)^2\right),
\end{align}
where the quantity $\log \frac{4\norm{H}^2_F}{\delta\norm{H}^2}$ in Eq.~\ref{eq:Mquation} was bounded using the following inequality
\[
\log \frac{\norm{H}^2_F}{\norm{H}^2} \leq \log \frac{2^n \lambda_{MAX} ^2}{\lambda_{MAX} ^2} = n,
\]
where $\lambda_{MAX}$ is the biggest eigenvalue of $H$.

Observe now that simulation of the time evolution of $\alpha I$ does only change the phase of the time evolution, where $I \in \mathbb C^{N \times N}$ is the identity matrix and $\alpha$ some real parameter. We can hence perform the time evolution of $\tilde{H} := H - \alpha I$, since for any efficient classical description of the input state we can apply the time evolution of the diagonal matrix $e^{-i\alpha I t}$. We can then optimise the parameter $\alpha$ such that the Frobenius norm of the operator $\tilde{H}$ is minimized, i.e.
\begin{align}
 \alpha =\underset{\alpha}{\text{argmin}} \norm{\tilde H}^2_F = \underset{\alpha}{\text{argmin}} \norm{H - \alpha I}^2_F,
\end{align}
from which we obtain the condition $\alpha = \frac{\tr{H}}{2^n}$.
Since, in order for the algorithm to be efficient, we require that $\norm{\tilde H}_F$
is bounded by $\polylog N$.
Using the spectral theorem, and the fact that the Frobenius norm is unitarily invariant,
this in turn gives us after a bit of algebra the condition
\begin{align}
\label{eq:new_bound}
\norm{H}_F^2 - \frac{1}{N}\tr{H}^2 \leq \Ord{\polylog(N)},
\end{align}
for which we can simulate the Hamiltonian $H$ efficiently.\\
We now prove the second main result of this work and establish the correctness of the above results.

\begin{proof}[Proof of Theorem~\ref{thm:main}]
Denote with
\begin{align*}
\widehat{Z}_K(Ht,At) &= I + i t H + t^2A f_K(t^2A^HA)A^H + i t^3 A g_K(t^2A^HA) A^HH, \\
\widehat{Z}(Ht, At) &= I + i t H + t^2A f(t^2A^HA)A^H + i t^3 A g(t^2A^HA) A^HH.
\end{align*}
By definition of $\widehat{\psi}_{K,M}(t)$ and the fact that $\norm{\psi} = 1$, we have
\begin{align*}
\norm{\widehat{\psi}_{K,M}(t) - \psi(t)}  \leq\norm{\widehat{Z}_K(At,Ht) - e^{iHt}}\norm{\psi} \\
\leq \norm{\widehat{Z}_K(At,Ht) - \widehat{Z}(Ht, At)} +  \norm{\widehat{Z}(Ht, At) - e^{iHt}}.
\end{align*}
We first study $\norm{\widehat{Z}(Ht, At) - e^{iHt}}$.
Define $l(x) = f(x) x $ and $m(x) = g(x) x$.
Note that, by the spectral theorem, we have
\begin{align*}
\widehat{Z}(Ht, At) &= I + i t H + t^2A f(t^2A^HA)A^H + i t^3 A g(t^2A^HA) A^HH \\
 &= I + i t H + t^2 f(t^2AA^H)AA^H + i t^3 g(t^2AA^H) AA^H \\
 &= I + i t H + l(t^2 AA^H) + i t m(t^2AA^H) H.
\end{align*}
Since
\[
e^{ixt} = 1 + ixt + l(t^2x^2) + it m(t^2x^2)x,
\]
we have
\begin{align*}
\norm{\widehat{Z}(Ht, At) - e^{iHt}} &= \norm{l(t^2 AA^H) - l(t^2H^2) + it m(t^2 AA^H) H - itm(t^2 H^2) H} \\
 &\leq \norm{l(t^2 AA^H) - l(t^2H^2)} + t \norm{m(t^2 AA^H) - m(t^2 H^2)}\norm{H}.
\end{align*}
To bound the norms in $l, m$ we will apply Thm.~1.4.1 of \cite{aleksandrov2016operator}.
The theorem state that if a function $f \in L^\infty(\mathbb{R})$, i.e.\ $f$ is in the function space which elements are the essentially bounded measurable functions,
it is entirely on $\mathbb{C}$ and satisfies $|f(z)| \leq e^{\sigma |z|}$ for any $z \in \mathbb{C}$. Then
$\norm{f(A) - f(B)} \leq \sigma \norm{f}_{L^\infty(\mathbb{R})} \norm{A - B}$.
Note that
\begin{align*}
|l(z)| &= \left|\sum_{j=1}^{\infty}  \frac{(-1)^j z^j}{(2j)!}\right| \leq \sum_{j=1}^{\infty}  \frac{|z|^j}{(2j)!} \leq \sum_{j=1}^{\infty}  \frac{|z|^j}{j!} \leq e^{|z|}, \\
|l(z)| &= \left|\sum_{j=1}^{\infty} \frac{(-1)^j z^{j}}{(2j+1)!}\right| \leq \sum_{j=1}^{\infty}  \frac{|z|^{j}}{(2j+1)!} \leq \sum_{j=1}^{\infty}  \frac{|z|^j}{j!} \leq e^{|z|}.
\end{align*}
Moreover it is easy to see that $\norm{l}_{L^{\infty}(\mathbb{R})}, \norm{m}_{L^{\infty}(\mathbb{R})} \leq 2.$ So
\begin{align*}
\norm{\widehat{Z}(Ht, At) - e^{iHt}} &\leq 2(1+t\norm{H}) \norm{t^2AA^H - t^2H^2} \nonumber \\
&= 2t^2(1+t\norm{H}) \norm{AA^H - H^2}.
\end{align*}
Now note that, by defining the random variable $\zeta_i = \frac{1}{p(t_i)} h_{t_i} h_{t_i}^H,$ we have that
\begin{align*}
&AA^H = \frac{1}{M} \sum_{i=1}^M \zeta_i, \\
& \mathbb{E} [\zeta_i] = \sum_{q=1}^{2^n} p(q) \frac{1}{p(q)} h_{t_i} h_{t_i}^H = H^2, ~~ \forall i.
\end{align*}
Let $\tau > 0$. By applying Thm.~1 of \cite{hsu2014weighted} (or Prop.~9 in \cite{rudi2015less}), for which
\[\norm{AA^H - H^2} \leq \sqrt{\frac{\norm{H}^2_F \norm{H}^2 \tau}{M}} + \frac{2\norm{H}^2_F \norm{H}^2 \tau}{M},\]
with probability at least $1 - 4\frac{\norm{H}^2_F}{\norm{H}^2}\tau/(e^\tau - \tau - 1)$.
Now since \[1 - 4\frac{\norm{H}^2_F}{\norm{H}^2}\tau/(e^\tau - \tau - 1) \geq 1 - e^{\tau},\]
when $\tau \geq e$, by selecting \[\tau = 2\log\frac{4\norm{H}^2_F}{\norm{H}\delta},\]
we have that the equation above holds with probability at least $1-\delta$.\\
Let $\eta > 0$, by selecting \[M = 4\eta^{-2}\norm{H}^2_F\norm{H}^2\tau,\] we then obtain
\[\norm{AA^H - H^2} \leq \eta,\]
with probability at least $1-\delta$.\\

Now we study $\norm{\widehat{Z}_K(At,Ht) - \widehat{Z}(Ht, At)}$.
Denote with $a$, $b$ the functions $a(x) = l(x^2)$, $b(x) = m(x^2)$ and with $a_K,
b_K$ the associated $K$-truncated Taylor expansions.
Note that $a(x) = \cos(x) - 1$, while $b(x) = (\sin(x) - x)/x$.
Now by definition of $\widehat{Z}_K$ and $\widehat{Z}$, we have
\begin{align*}
\norm{\widehat{Z}_K(At,Ht) - \widehat{Z}(Ht, At)} &\leq \norm{a_K(t \sqrt{AA^H}) - a(t \sqrt{AA^H})} \nonumber \\
&+ t \norm{H}\norm{b_K(t \sqrt{AA^H}) - b(t \sqrt{AA^H})}.
\end{align*}
Note that, since \[\sum_j x^{2j}/(2j)! = \cosh(x) \leq 2 e^{|x|},\] and
\begin{align*}
|(a_K-a)(x)| &= \left|\sum_{j=K+2}(-1)^j \frac{x^{2j}}{(2j)!}\right| \\
&\leq \frac{|x|^{2K + 4}}{(2K + 4)!} \sum_{j=0} \frac{|x|^{2j}}{(2j)!} \frac{(2K + 4)!(2j)!}{(2j + 2K + 4)!} \\
&\leq \frac{2|x|^{2K + 4}e^{|x|}}{(2K + 4)!},\\
|(b_K-b)(x)| &= \left|\sum_{j=K+2}(-1)^j \frac{x^{2j}}{(2j+1)!} \right| \\
&\leq \frac{|x|^{2K + 4}}{(2K + 4)!} \sum_{j=0} \frac{|x|^{2j}}{(2j)!} \frac{(2K + 4)!(2j)!}{(2j + 2K + 5)!} \\
&\leq \frac{2|x|^{2K + 4}e^{|x|}}{(2K + 4)!}.
\end{align*}
Let $R > 0, \beta \in (0,1]$. Now note that, by Stirling approximation, $c! \geq e^{c \log \frac{c}{e}}$, so by selecting $K = \frac{e^2}{2} R + \log(\frac{1}{\beta})$, we have for any $|x| \leq R$,
\begin{align*}
\log\left|\frac{2|x|^{2K + 4}e^{|x|}}{(2K + 4)!}\right| &\leq |x| + (2K+4)\log\frac{e|x|}{2K + 4} \\
&\leq R + (2K+4)\log\frac{eR}{2K + 4} \\
& \leq R - \left(e^2 R + 4 + \log\frac{1}{\beta}\right)\log\left(e + \frac{4+\log{1}{\beta}}{eR}\right) \\
& \leq R - \left(e^2 R + 4 + \log\frac{1}{\beta}\right) \\
& \leq -(e^2 - 1)R - 4 - \log \frac{1}{\beta} \\
& \leq -\log \frac{1}{\beta}.
\end{align*}
So, by choosing $K \geq \log\frac{1}{\beta} +e^2R/2$, we have $|a_K(x) - a(x)|,|b_K(x) - b(x)| \leq \beta$.
With this we finally obtain
\[\norm{\widehat{Z}_K(At,Ht) - \widehat{Z}(Ht, At)} \leq 2(1+t\norm{H})\beta,\]
when $K \geq e^2t\norm{A}/2 + \log \frac{1}{\beta}$, and therefore we have
\[\norm{\widehat{\psi}_{K,M}(t) - \psi(t)} \leq 2t^2(1+t\norm{H})\eta + 2(1 + t\norm{H}) \beta,\]
with probability at least $1-\delta$, when
\[M \geq 8\eta^{-2}\norm{H}^2_F\norm{H}^2\log\frac{4\norm{H}^2_F}{\norm{H}^2\delta},\quad K \geq 4t\norm{A} + \log \frac{1}{\beta}.\]
In particular, by choosing $\eta = \epsilon/(4t^2(1+t\norm{H})$ and $\beta = \epsilon/(4(1+t\norm{H}))$,
we have
\[\norm{\widehat{\psi}_{K,M}(t) - \psi(t)} \leq \epsilon,\]
with probability at least $1-\delta$.

With this result in mind note that, in the event where $\norm{AA^H - H^2} \leq \epsilon$, we have that
\[|\norm{AA^H} - \norm{H}^2| \leq \norm{AA^H - H^2} \leq \epsilon,\]
and therefore, $\norm{A} \leq \sqrt{\norm{H}^2 + \epsilon}$.
\end{proof}

\subsection{Beyond the Nystr\"om method}
While the \Nystrom{} method has been applied very successful over the past decade or
so, in recent years, it also has been improved upon using better approximations
for symmetric matrices.
The results of \cite{yang2012nystrom} indicated that the \Nystrom{} method has advantages
over the random feature method~\cite{rahimi2009weighted}, the closest competitor, both theoretically and empirically.
However, as it has recently been established, even the \Nystrom{} method
does not attain high accuracy in general.
A model which improves the accuracy of the \Nystrom{} method, is the
so-called prototype model~\cite{halko2011finding,wang2013improving}.
The prototype model performs first a random sketch on the input matrix $H$, i.e.,
$C = HS$, where $S$ is a sketching or sampling matrix which samples $M$ rows of $H$, and then computes the
intersection matrix $U^*$ as
\[
 U^* := \mathrm{argmin}_U \norm{H - C U C^H}_F^2 = C^+ H (C^+)^H \in \mathbb{R}^{M \times M}.
\]
The model then approximates $H$ by $CU^*C^H$.
While the intial versions of the prototype model were not efficient due to the
cost of calculating $U^*$, \cite{wang2016towards} improved these results by approximately
calculating the optimal $U^*$, which led to a higher accuracy compared to the
\Nystrom{} method, but also to an improved runtime.
As future  work, we leave it open, whether these methods can be used to obtain
improved classical algorithms for Hamiltonian simulation and quantum machine learning
in general.

\subsection{Conclusion}
\label{ssec:dequantisation}

As we mentioned in the introduction of the chapter, our results are closely related
to the so-called 'quantum-inspired' or 'dequantisation' results, which followed the work by Tang~\cite{tang2018quantum}.
It is obvious that our \textit{row-searchable} condition is fundamentally equivalent to the
\textit{sample and query} access requirement, and both our and other results indicate that
the state preparation condition requires a careful assessment in order to determine whether a
quantum algorithm provides an actual advantage over its classical equivalent.
In particular, our sampling scheme for Hermitian matrices is equivalent to the
sampling scheme based on the classical memory structure which we introduced earlier,
although for PSD matrices we use a more computationally efficient variant.
While we also provide a method based on binary-trees to compute the sampling probabilities,
the tree structure of the memory immediately allows us to execute such a sampling
process in practice.
The main difference is therefore that we use a traditional memory structure and provide a
fast way to calculate the marginals using a binary tree (see Sec.~\ref{par:sampling}),
while \cite{tang2018quantum} assumes a memory structure, which allows one to sample efficiently according to this distribution.

Notably, we believe our results can be further improved by using the rejection sampling
methods as is done in all quantum-inspired approaches.
The rejection sampling is required in order to achieve faster algorithms as they
allow us to sample from the output distribution rather than putting out the
entirety of it -- as is the case with our algorithm.

Based on our results for Hamiltonian simulation, and
supported by the large amount of recent quantum-inspired results for QML applications,
we believe it less likely that QML algorithms will admit any exponential speedups
compared to classical algorithms.
The main reason is because, unlike their classical counterparts, most QML
algorithms require strong input assumptions.
As we have seen, these caveats arise due to two reasons.
First, the fast loading of large input data into a quantum computer is a very powerful
assumption, which indeed leads to similarly powerful classical algorithms.
Second, extracting the results from an output quantum state is hard in general.
Tang's~\cite{tang2018quantum} breakthrough result in 2018 for the quantum recommendation systems algorithm,
which was previously believed to be one of the strongest candidates for a practical exponential speedup in QML,
indeed implied that the quantum algorithm does not give an exponential speedup.
Tang's algorithm solves the same problem as the quantum algorithm and just
incurs a polynomial slow-down.
To do so, it relies on similar methods to the ones we described for Hamiltonian simulation
with the \Nystrom{} method, and combines these with classical rejection sampling.

While exponential speedups therefore appear not to be possible in general,
in specific cases, such as problems with sparsity assumptions,
these dequantisation techniques cannot be applied and QML algorithms
still might offer the long-sought exponential advantage.

The most well known quantum algorithm of this resistant type is the HHL algorithm
for sparse matrix inversion~\cite{harrow2009quantum}, which indeed is BQP-complete.
Although HHL has a range of caveats, see e.g.\ \cite{aaronson2015read}, perhaps
it can be useful in certain instances.
Works that build on top of the HHL, such as Zhao et al.\ on Gaussian process regression
\cite{zhao2015quantum} and Lloyd et al. on topological data analysis \cite{lloyd2014topological}
might indeed overcome these drawbacks and achieve a super-polynomial quantum speedup.

The question whether such quantum or quantum-inspired algorithms will be useful
in practice is, however, even harder to answer.
Indeed, for algorithms with a theoretical polynomial speedup an actual advantage can only be established
through proper benchmarks and performance analysis such as a scaling analysis.
The recent work by Arrazola et al.~\cite{arrazola2019quantum}, for example, implemented and tested the quantum-inspired algorithms
for regression and recommendation systems and compared these against existing state-of-the-art implementations.
While giving various insights, however, it is unclear whether the analysis in question is valid.
In order to assure that the algorithms are tested properly, more modern sketching algorithms would need to
be used, since the ones that we and also other authors rely on do not necessarily
provide performance close to the current state-of-the-art.
The reason for this is that they were chosen to obtain complexity theoretic results,
i.e., to demonstrate that from a theoretical perspective the quantum algorithms do not
have an exponential advantage.
As was also pointed out in~\cite{chia2020sampling} Dahiya, Konomis, and Woodruff \cite{dahiya2018empirical}
already conducted an empirical study of sketching algorithms for low-rank approximation
on both synthetic datasets and the movielens dataset, and reported that
their implementation ``finds a solution with cost at most 10 times the optimal one
[...] but does so 10 times faster'', which is in contrast to the results of
Arrazola et al.~\cite{arrazola2019quantum}.
Additionally, using self-implemented algorithms and comparing these in a
non-optimised way to highly optimised numerical packages such as LAPACK, BLAS, or others,
is usually not helpful when evaluating the performance of algorithms.

We therefore conclude that quantum machine learning algorithms still might have a
polynomial advantage over classical ones, but in practice many challenges exist
and need to be overcome. It is hence unclear as of today whether they will be
able to provide any meaningful advantage over classical machine learning algorithms.
In the future we hope to see results that implement the quantum algorithms
and analyse the resulting overhead, and ultimately benchmark these against their
quantum-inspired counterparts.
However, we believe that such benchmarks will not be possible in the near future,
as the requirements on the quantum hardware are far beyond what is possible today.

\chapter{Promising avenues for QML}
\label{chap:quantumQML}
In the previous chapters, we have tried to answer the main questions of this thesis,
i.e., whether (a) quantum machine learning algorithms can offer an advantage over
their classical counterparts from a statistical perspective, and (b) whether
randomised approaches can be used to achieve similarly powerful classical algorithms.
We answered both of these questions in the context of supervised quantum machine
learning algorithms for classical data, i.e., data that is stored in some
form of a classical memory.
However, data could also be obtained as the immediate output of a quantum
process, and the QML algorithm could therefore also directly use the resulting
quantum states as input that could in principle resolve qRAM-related challenges.
Indeed, in this regime (QQ in Fig.~\ref{fig:4quadrants}), it is very likely
that classical machine learning algorithms will have difficulties, since they
are fundamentally not able to use the input except through prior sampling
(measurement) of the state, which implies a potential loss of information.
The ability to prepare, i.e., generate arbitrary input states with a polynomially
sized circuit could therefore enable useful quantum machine learning algorithms.

In this chapter, we now briefly discuss a possible future avenue for QML,
namely the generative learning of quantum distributions and quantum states.
In particular, we provide a method for fully quantum generative training of
quantum Boltzmann machines with both visible and hidden units.
To do this, we rely on the quantum relative entropy as an objective function,
which is significant since prior methods cannot do so due to mathematical challenges.
The mathematical challenges are a result of the gradient evaluation which is
required in the training process.
Our method is highly relevant, as it allows us to efficiently estimate gradients
even for nearly parallel states, which has been impossible for all previous methods.
We can therefore use our algorithm even to approximate cloning and state preparation
for arbitrary input states.
We present in the following two novel methods for solving this problem.
The first method given an efficient algorithm for a class of restricted quantum
Boltzmann machines with mutually commuting Hamiltonians on the hidden units.
In order to train it with gradient descent and the quantum relative entropy
as an objective function, we use a variational upper bound.
The second one generalises the first result to generic quantum Boltzmann machines
by using high-order divided difference methods and linear-combinations
of unitaries to approximate the exact gradient of the relative entropy.
Both methods are efficient under the assumption that Gibbs state preparation is
efficient and assuming that the Hamiltonian is a sparse, row-computable matrix.

\section{Generative quantum machine learning}

One objective of QML is to design models that can learn in quantum mechanical
settings~\cite{biamonte2017quantum,ciliberto2018quantum,servedio2004equivalences,arunachalam2018optimal,lloyd2013quantum,wiebe2019language},
i.e., models which are able to quickly identify patterns in data,
i.e., quantum state vectors, which inhabit an exponentially
large vector space~\cite{mcclean2018barren,schuld2018circuit}.
The sheer size of these vectors therefore restricts the computations that can be performed efficiently.
As we have seen above, one of the main bottlenecks for many QML applications
(e.g.\ supervised learning) is the input data preparation.
While qRAM appears not to be a solution to the common data read-in problem,
other methods for the fast preparation of quantum states might enable useful
QML applications.

The clearest cases where quantum machine learning can therefore provide
an advantage is when the input is already in the form of quantum data, i.e., quantum states.
An alternative way to qRAM to efficiently generate quantum states
according to a desired distribution are generative models.
Generative models can derive concise descriptions of large quantum states~\cite{kieferova2016tomography,schuld2019quantum,romero2017quantum,benedetti2019adversarial},
and be able to prepare states with in principle polynomially sized circuits.

The most general representation for the input, which is the natural analog of a quantum training set,
would be a density operator $\rho$,
which is a positive semi-definite trace-$1$ Hermitian matrix that,
roughly speaking, describes a probability distribution over the input quantum state vectors.
The goal in quantum generative training is to learn a unitary operator $V: \ket{0} \mapsto \sigma$,
which takes as input a quantum state $\ket{0}$, and prepares the density matrix $\sigma$,
by taking a small (polynomial) number of samples from $\rho$,
with the condition that under some chosen distance measure $D$, $D(\rho, \sigma)$, is small.
For example, for the $L_1$ norm this could be $D(\rho,\sigma)= \norm{\rho -\sigma}_1$.
In the quantum information community this task is known as partial tomography~\cite{kieferova2016tomography}
or approximate cloning~\cite{chefles1999strategies}.
A similar task is to replicate a conditional probability distribution over a label subspace.
Such approaches play a crucial role in QAOA-based quantum neural networks.

\subsection{Related work}
\label{ssec:related_work_QBMs}

Various approaches have been put forward to solve the challenge of
generative learning in the quantum domain~\cite{kieferova2016tomography,romero2017quantum,kappen2018learning,amin2018quantum,crawford2016reinforcement},
but to date, all proposed solutions suffer from drawbacks due to data input requirements, vanishing gradients,
or an inability to learn with hidden units.
In this chapter, we present two novel approaches for training quantum Boltzmann machines~\cite{kieferova2016tomography,benedetti2017quantum,amin2018quantum}.
Our approaches resolve all restrictions of prior art, and therefore address
a major open problem in generative quantum machine learning.
Figure~\ref{fig:tabular_results} summarises the key results and prior art.

\begin{figure}
  \resizebox{\textwidth}{!}{
  \begin{tabular}{c|c|c|c|c}
    \hline
    \rowcolor{Gray}
    Work & Data & Algorithm & Objective & Hidden Units \\
     \hline
    Wiebe et al. \cite{wiebe2015quantum} & Classical & Quantum (Classical model) & Maximum likelihood & Yes \\
    Amin et al. \cite{amin2018quantum} & Classical & Quantum (Quantum model) & Maximum likelihood & Yes (But cannot be trained) \\
    Kieferova et al. \cite{kieferova2016tomography} & Classical (Tomography) & Quantum & Maximum likelihood (KL-divergence) & Yes \\
    Kieferova et al. \cite{kieferova2016tomography} & Quantum & Quantum  & Relative entropy & No \\
    \rowcolor{LightCyan}
    Our work & Quantum & Quantum (Restricted H) & Relative entropy (Variational bound) & Yes \\
    \rowcolor{LightCyan}
    Our work & Quantum & Quantum & Relative entropy & Yes \\
  \end{tabular}
  }
  \caption{Comparison of previous training algorithms for quantum Boltzmann machines.
  The models have a varying cost function (objective),
  contain (are able to be trained with) hidden units, and have different input data (classical or quantum).}
  \label{fig:tabular_results}
\end{figure}

\section{Boltzmann and quantum Boltzmann machines}
\label{sec:bm_qbm_intro}

Boltzmann machines are a physics inspired class of neural network~\cite{aarts1988simulated,salakhutdinov2007restricted,tieleman2008training,le2008representational,salakhutdinov2009deep,salakhutdinov2010efficient,lee2009convolutional,hinton2012practical},
which have gained increasing popularity and found numerous applications over the last decade~\cite{lee2009unsupervised,lee2009convolutional,mohamed2011acoustic,srivastava2012multimodal}.
More recently, they have been used in the generative description of complex quantum systems~\cite{carleo2017solving,torlai2016learning,nomura2017restricted}.
Boltzmann machines are an immediate approach when we want to perform generative learning,
since they are a physically inspired model which resembles many common quantum systems.
Through this similarity, they are also highly suitable to be implemented on quantum computers.
To be more precise, a Boltzmann machine is defined by the energy which arises from
the interactions in the physical system.
As such, it prescribes an energy to every configuration of this system,
and then generates samples from a distribution which assigns probabilities
to the states according to the exponential of the states energy.
This distribution is in classical statistical physics known as the canonical ensemble.
The explicit model is given by
\begin{equation}
  \label{eq:hidden_unit_QBM}
\sigma_v(H) = {\rm Tr}_h \left(\frac{e^{-H}}{Z} \right)= \frac{{\rm Tr}_h~ e^{-H}}{{\rm Tr}~ e^{-H}},
\end{equation}
where ${\rm Tr}_h(\cdot)$ is the partial trace over the so-called
\textit{Hidden} subsystem, an auxillary sub-system which allows the model to
build correlations between different nodes in the \textit{Visible} subsystem.
In the case of classical Boltzmann machines, the Hamiltonian $H$ is an energy function
that assigns to each state an energy, resulting in a diagonal matrix.
In the quantum mechanical case, however, we can obtain superpositions (i.e., linear combinations)
of different states, and the energy function therefore turns into the Hamiltonian
of the quantum Boltzmann, i.e., a general Hermitian matrix which has off-diagonal entries.

As mentioned earlier, for a quantum system the dimension of the Hamiltonian
grows exponentially with the number of units such as qubits or orbitals, denoted $n$,
of the system, i.e., $H \in \mathbb{C}^{2^n \times 2^n}$.

Generative quantum Boltzmann training is then defined as the task of finding the
parameters or so-called weights $\theta$, which parameterise the Hamiltonian $H=H(\theta)$,
such that
\[
H = {\rm argmin}_{H(\theta)} \left({\rm D}(\rho,\sigma_v(H))\right),
\]
where $D$ is again an appropriately chosen distance, or divergence, function.
For example, the quantum analogue of an all-visible Boltzmann machine with $n_v$ units would then take the form
\begin{equation}
    H(\theta) = \sum_{n=1}^{n_v} \theta_{2n-1} {\sigma_x}^{(n)} + \theta_{2n} \sigma_z^{(n)} + \sum_{n>n'} \theta_{(n,n')} \sigma_z^{(n)} \sigma_z^{(n')}.
\end{equation}
Here $\sigma_z^{(n)}$ and $\sigma_x^{(n)}$ are Pauli matrices acting on qubit (unit) $n$.

We now provide a formal definition of the quantum Boltzmann machine.
\begin{definition}
A quantum Boltzmann machine to be a quantum mechanical system that acts on a tensor product of Hilbert spaces $\mathcal{H}_v\otimes \mathcal{H}_h \in \mathbb{C}^{2^n}$
that correspond to the visible and hidden subsystems of the Boltzmann machine.  It further has a Hamiltonian of the form
$H \in \mathbb{C}^{2^n \times 2^n}$ such that $\norm{H - \text{diag}(H)} > 0$.
The quantum Boltzmann machine takes these parameters and then outputs a state of the form ${\rm Tr}_h \left(\frac{e^{-H}}{{\rm Tr}(e^{-H})} \right)$.
\end{definition}

For classical data we hence want to minimise the distance between two different
distributions, namely the input and output distribution.
The natural way to measure such a distance is the Kullback-Leibler (KL) divergence.
In the case of the quantum Boltzmann machine, in contrast, we want to
minimise the distance between two quantum states (density matrices),
and the natural notion of distance changes to the quantum relative entropy:
\begin{equation}
  \label{eq:quant_rel_ent}
S(\rho | \sigma_v) = {\rm Tr}\left( \rho \log \rho \right) - {\rm Tr}\left( \rho \log \sigma_v \right).
\end{equation}
For diagonal $\rho$ and $\sigma_v$ this indeed reduces to the (classical)
KL divergence, and it is zero if and only if $\rho = \sigma_v$.

For Boltzmann machines with visible units, the gradient of the relative entropy
can be computed in a straightforward manner because here
$\sigma_v = e^{-H}/Z$ and since $\log(e^{-H}/Z) = -H -\log(Z)$,
we can easily compute the matrix derivatives.
This is however not the case in general, and no methods are known for the generative
training of Boltzmann machines for the quantum relative entropy loss function, if hidden units are present.
The main challenge which restricts an easy solution in this case is
the evaluation of the partial trace in $\log({\rm Tr}_h e^{-H} /Z)$,
which prevents us from simplifying the logarithm term which is required for the gradient.

In the following, we will provide two practical methods for training quantum Boltzmann machines
with both hidden and visible units.
We provide two different approaches for achieving this.
The first method only applies to a restricted case,
however it allows us to propose a more efficient scheme that
relies on a variational upper bound on the quantum relative entropy
in order to calculate the derivatives.
The second approach is more general method, but requires more resources.
For the second approach we rely on recent techniques from quantum simulation
to approximate the exact expression for the gradient.
For this we rely on a Fourier series approximation and high-order divided
difference formulas in place of the analytic derivative.
Under the assumption that Gibbs state preparation is efficient,
which we expect to hold in most practical cases, both methods are efficient.

We note however, that this assumption is indeed not valid in general,
in particular, since the possibility of efficient Gibbs state preparation
would imply $\mathrm{QMA}\subseteq \mathrm{BQP}$ which is unlikely to hold.

While the here presented results might be interesting for several readers,
in order to preserve the flow and keep the presentation simple, we will
include the proofs of technical lemmas in the appendix of the thesis.

\section{Training quantum Boltzmann machines}
We now present methods for training quantum Boltzmann machines.\\
We begin by defining the quantum relative entropy,
which we use as cost function for our quantum Boltzmann machine (QBM) with hidden units.
Recall that $\rho$ is our input distribution,
and $\sigma = \partr{h}{e^{-H}/ \tr{e^{-H}}}$ is the output distribution,
where we perform the partial trace over the hidden units. With this,
the cost function is hence given by
\begin{equation}
  \label{eq:obj_QBM_hidden}
 \mathcal{O}_{\rho}(H) = S \left(\rho \Big| \partr{h}{ e^{-H}/\tr{e^{-H}}} \right),
\end{equation}
where $S(\rho | \sigma_v)$ is the quantum relative entropy as defined in eq.~\ref{eq:quant_rel_ent}.
A few remarks.
First, note that it is possible to add a regularisation term
in order to penalise undesired quantum correlations in the model~\cite{kieferova2016tomography}.
Second, since we will only derive gradient based algorithms for the training,
we need to evaluate the gradient of the cost function,
which requires the gradient of the quantum relative entropy.

As previously mentioned, this is possible to do in a closed-form expression
for the case of an all-visible Boltzmann machine, which
corresponds to $\text{dim}(\mathcal{H}_h)=1$.
The gradient in this case takes the form
\begin{equation}
 \frac{\partial \mathcal{O}_{\rho}(H)}{\partial \theta} = -\tr{ \frac{\partial}{\partial \theta} \rho \log \sigma},
\end{equation}
which can be simplified using $\log(\exp(-H)) = -H$ and Duhamels formula
to obtain the following equation for the gradient, denoting with $\partial_{\theta} := \partial/\partial \theta$,
\begin{equation}
  \tr{\rho \partial_{\theta} H} - \tr{e^{-H}\partial_{\theta} H}/\tr{e^{-H}}.
\end{equation}
However, this gradient formula is does not hold any longer if we include hidden units.
If we hence want to include hidden units, then we need to additionally
trace out the subsystem.
Doing so results in the majorised distribution from eq.~\ref{eq:hidden_unit_QBM},
which also changes the cost function into the form described in eq.~\ref{eq:obj_QBM_hidden}.
Note that $H=H(\theta)$ is depending on the training weights.
We will adjust these in the training process with the goal to match $\sigma$ to
$\rho$ -- the target density matrix.
Omitting the $\tr{\rho \log \rho}$ since it is a constant, we therefore obtain
\begin{equation}
  \label{eq:grad_QBM_hidden_units}
 \frac{\partial \mathcal{O}_{\rho}(H)}{\partial \theta} = -\tr{ \frac{\partial}{\partial \theta} \rho \log \sigma_v},
\end{equation}
for the gradient of the objective function.

In the following sections, we now present our two approaches for evaluating
the gradient in eq.~\ref{eq:grad_QBM_hidden_units}.
The first one is less general, but gives us an easy implementable algorithm
with strong bounds.
The second one, on the hand, can be applied to arbitrary problem instances,
and therefore presents a general purpose gradient based algorithm
for training quantum Boltzmann machines with the relative entropy objective.
The results are to be expected, since the no-free-lunch theorem suggests
that no good bounds can be obtained without assumptions on the problem instance.
Indeed, our general algorithm can in principle result in an exponentially worse complexity
compared to the specialised one.
However, we are expecting that for most practical applications
this will not be the case, and our algorithm will in hence be applicable for most
cases of practical interest.
To our knowledge, we hence present the first algorithm of this kind,
which is able to train arbitrary arbitrary quantum Boltzmann machines efficiently.

\subsection{Variational training for restricted Hamiltonians}

Our first approach optimises a variational upper bound of the objective function,
for a restricted but still practical setting.
This approach results in a fast and easy way to implement quantum algorithms which,
however, is less general due to the input assumptions we are making.
These assumptions are required in order to overcome challenges related the
evaluation of matrix functions, for which classical calculus fails.

Concretely, we express the Hamiltonian in this case as
\begin{equation}
H = H_v + H_h + H_{\rm int},
\end{equation}
i.e., a decomposition of the Hamiltonian into a part acting on the visible layers,
the hidden layers and a third interaction Hamiltonian
that creates correlations between the two.
In particular, we further assume for simplicity that there are two sets of
operators $\{v_k\}$ and $\{h_k\}$ composed of $D=W_v+ W_h+ W_{\rm int}$ terms such that
\begin{align}
H_v &= \sum_{k=1}^{W_v} \theta_k v_k \otimes I,\qquad &H_h = \sum_{k=W_v+1}^{W_v +W_h} \theta_k I \otimes h_k \nonumber\\
H_{\rm int} &= \sum_{k=W_v+W_h +1}^{W_v +W_h + W_{\rm int}} \theta_k v_k \otimes h_k,&[h_k,h_j]  = 0~\forall~j,k,\label{eq:assumptions}
\end{align}
which implies that the Hamiltonian can in general be expressed as
\begin{equation}
\label{eq:ham_form}
H=\sum_{k=1}^D \theta_k v_k \otimes h_k.
\end{equation}
We break up the Hamiltonian into this form to emphasise the qualitative
difference between the types of terms that can appear in this model.
Note that we generally assume throughout this article that $v_k,h_k$
are unitary operators, which is indeed given in most common instances.

By choosing this particular form of the Hamiltonian in~\eqref{eq:assumptions},
we want to force the non-commuting terms, i.e., terms for which it holds that
the commutator $[v_k,h_k]\neq 0$, to act only on the visible units of the model.
On the other hand, only commuting Hamiltonian terms act on the hidden register,
and we can therefore express the eigenvalues and eigenvectors for the Hamiltonian as
\begin{equation}
H \ket{v_h} \otimes \ket{h} = \lambda_{v_h, h} \ket{v_h} \otimes \ket{h}.
\end{equation}
Note that here both the conditional eigenvectors and eigenvalues for the visible subsystem
are functions of the eigenvector $\ket{h}$ in the hidden register,
and we hence denote these as $v_h,\lambda_{v_h,h}$ respectively.
This allows the hidden units to select between eigenbases to interpret the input
data while also penalising portions of the accessible Hilbert space that
are not supported by the training data.
However, since the hidden units commute they cannot be used to construct a non-diagonal eigenbasis.
While this division between the visible and hidden layers on the one hand helps
us to build an intuition about the model, on the other hand -- more importantly --
it allows us to derive a more efficient training algorithms that exploit this fact.

As previously mentioned, for the first result we rely on a variational bound
of the objective function in order to train the quantum Boltzmann machine weights
for a Hamiltonian $H$ of the form given in \eqref{eq:ham_form}.
We can express this variational bound compactly in terms of a thermal
expectation against a fictitious thermal probability distribution.
We define this expectation below.
\begin{definition}
\label{def:prob_distr}
Let $\tilde{H}_h = \sum_k \theta_k \tr{\rho v_k} h_k$ be the Hamiltonian acting conditioned on the visible subspace only on the hidden subsystem of the Hamiltonian $H:=\sum_k \theta_k v_k \otimes h_k$.
Then we define the expectation value over the marginal distribution over the hidden variables $h$ as
	\begin{equation}
	\label{eq:def_exp}
	\mathbb{E}_h(\cdot) =  \sum_h   \frac{(\cdot) e^{- \tr{\rho \tilde{H}_h }}}{\sum_h e^{- \tr{\rho \tilde{H}_h}}}.
	\end{equation}
\end{definition}

Using this we derive an upper bound on $S$ in section~\ref{sec:variational_bound_derivative} of the Appendix, which leads to the following lemma.

\begin{lemma}
\label{def:Stilde}
    Assume that the Hamiltonian $H$ of the quantum Boltzmann machine takes the
    form described in eq.~\ref{eq:ham_form}, where $\theta_k$ are the parameters
    which determine the interaction strength and $v_k,h_k$ are unitary operators.
    Furthermore, let $h_k \ket{h}=E_{h,k}\ket{h}$ be the eigenvalues of the hidden
    subsystem, and  $\mathbb{E}_h(\cdot)$ as given by Definition~\ref{def:prob_distr},
    i.e., the expectation value over the effective Boltzmann distribution of
    the visible layer with $\tilde{H}_h :=\sum_k E_{h,k} \theta_k v_k$.
    Then, a variational upper bound $\widetilde{S}$ of the objective function,
    meaning that $\widetilde{S}(\rho|H)\ge  S(\rho | e^{-H}/Z)$, is given by
    \begin{align}
        \label{eq:varbd}
  		 \widetilde{S}(\rho|H)  := \tr{\rho \log \rho} +  \tr{\rho { \sum_{k} \mathbb{E}_{h} \left[ E_{h,k} \theta_k v_k \right]
       +  \mathbb{E}_{h} \left[ \log \alpha_h\right] } } + \log Z,
  	\end{align}
    where
    \[
    \alpha_h =  \frac{e^{- \tr{\rho \widetilde{H}_h }}}{\sum_h e^{- \tr{\rho \widetilde{H}_h}}}
    \]
    is the corresponding Gibbs distribution for the visible units.
\end{lemma}

The proof that~\eqref{eq:varbd} is a variational bound proceeds in two steps.
First, we note that for any probability distribution $\alpha_h$
\begin{equation}
\tr{\rho \log \left(\sum_{h=1}^N e^{-\sum_{k} E_{h,k} \theta_k v_k} \right)}
= \tr{\rho \log \left(\sum_{h=1}^N \alpha_h \frac{e^{-\sum_{k} E_{h,k} \theta_k v_k}/\alpha_h }{\sum_{h'} \alpha_{h'}} \right)}
\end{equation}
We then apply Jensen's inequality and minimise the result over all $\alpha_h$.
This not only verifies that $\widetilde{S}(\rho|H) \ge S(\rho|H)$ but also yields a variational bound.
The details of the proof can be found in Equation~\ref{eq:def_variational_bound} in section~\ref{sec:variational_bound_derivative} of the appendix.

Using the above assumptions we derive the gradient of the variational upper
bound of the relative entropy in the Section~\ref{sec:var_grad_estimation}
of the Appendix.
We summarise the result in Lemma~\ref{lem:gradient_visible_layer}.

\begin{lemma}
\label{lem:gradient_visible_layer}
  Assume that the Hamiltonian $H$ of the quantum Boltzmann machine takes the form
  described in eq.~\ref{eq:ham_form}, where $\theta_k$ are the parameters which
  determine the interaction strength and $v_k,h_k$ are unitary operators.
  Furthermore, let $h_k \ket{h}=E_{h,k}\ket{h}$ be the eigenvalues of the hidden subsystem,
  and $\mathbb{E}_h(\cdot)$ as given by Definition~\ref{def:prob_distr}, i.e.,
  the expectation value over the effective Boltzmann distribution of the visible layer with $\tilde{H}_h :=\sum_k E_{h,k} \theta_k v_k$.
  Then, the derivatives of $\widetilde{S}$ with respect to the parameters of the Boltzmann machine are given by
    \begin{align}
       \label{eq:gradient_var_bound}
      \frac{\partial \widetilde{S}(\rho|H)}{\partial_{\theta_p}}
      = \mathbb{E}_{h} \left[ \tr{\rho E_{h,p} v_p}\right] - \tr{\frac{\partial H}{\partial \theta_p} \frac{e^{-H}}{Z}}.
    \end{align}
\end{lemma}

As a sanity check of our results, we can consider the case of no interactions
between the visible and the hidden layer.
Doing so, we observe that the gradient above reduces to the case of the visible Boltzmann machine,
which was treated in \cite{kieferova2016tomography}, resulting in the gradient
\begin{equation}
  \tr{\rho \partial_{\theta_p}H} - \tr{\frac{e^{-H}}{Z} \partial_{\theta_p} H},
\end{equation}
under our assumption on the form of $H$, $\partial_{\theta_p}H = v_p$.\\

From Lemma~\ref{lem:gradient_visible_layer}, we know the form of the derivatives
of the relative entropy w.r.t.\ any parameter $\theta_p$ via Eq.~\ref{eq:gradient_var_bound}.
To understand the complexity of evaluating this gradient, we approach each term separately.
The second term can easily be evaluated by preparing the Gibbs state $\sigma_{Gibbs} := e^{-H}/Z$,
and then evaluating the expectation value of the operator $\partial_{\theta_j}H$ w.r.t. this Gibbs state.
In practice we can do so using amplitude estimation for the Hadamard test~\cite{aharonov2009polynomial},
which is a standard procedure and we describe it in algorithm~\ref{alg:algo0}
in section~\ref{sec:eval_dm_operator} of the appendix.
The computational complexity of this procedure is easy to evaluate.
If $T_{Gibbs}$ is the query complexity for the Gibbs state preparation,
the query complexity of the whole algorithm including the phase estimation step
is then given by $O(T_{Gibbs}/\epsilon)$ for an $\epsilon$-accurate estimate of phase estimation.

Next, we derive an algorithm to evaluate the first term, which requires a more involved process.
For this, note first that we can evaluate each term $\tr{\rho v_k}$ independently
from $\mathbb{E}_{h} \left[ E_{h,p}\right]$, and individually for all $k \in [D]$,
i.e., all $D$ dimensions of the gradient.
This can be done via the Hadamard test for $v_k$ which we recapitulate in
section~\ref{sec:eval_dm_operator} of the appendix, assuming $v_k$ is unitary.
More generally, for non-unitary $v_k$ we could evaluate this term using a linear combination of unitary operations.

Therefore, the remaining task is to evaluate the terms $\mathbb{E}_{h} \left[ E_{h,p}\right]$ in \eqref{eq:gradient_var_bound},
which reduces to sampling elements according to the distribution $\{\alpha_h\}$,
recalling that $h_p$ applied to the subsystem has eigenvalues $E_{h,p}$.
For this we need to be able to create a Gibbs distribution for the effective
Hamiltonian $\tilde{H}_h = \sum_k \theta_k \tr{\rho v_k} h_k$ which contains
only $D$ terms and can hence be evaluated efficiently as long as $D$ is small,
which we can generally assume to be true.
In order to sample according to the distribution $\{\alpha_h\}$,
we first evaluate the factors $\theta_k \tr{\rho v_k}$ in the sum over $k$ via the Hadamard test,
and then use these in order to implement the Gibbs distribution $\exp{(-\tilde{H_h})}/\tilde{Z}$
for the Hamiltonian $$\tilde{H}_h = \sum_k \theta_k \tr{\rho v_k} h_k.$$
The algorithm is summarised in Algorithm~\ref{alg:algo1}.

\begin{algorithm}[tb!]
\caption{Variational gradient estimation - term 1}
\label{alg:algo1}
\begin{algorithmic}
  \STATE {\bfseries Input:} An upper bound $\tilde{S}(\rho | H )$ on the quantum relative entropy, density matrix $\rho \in \mathbb{C}^{2^n\times 2^n}$, and Hamiltonian $H \in \mathbb{C}^{2^n\times 2^n}$.
  \STATE {\bfseries Output:} Estimate $\mathcal{S}$ of the gradient $\nabla \tilde{S}$ which fulfills Thm.~\ref{thm:gradient_results}.
  \STATE {\bfseries 1.} Use Gibbs state preparation to create the Gibbs distribution for the effective Hamiltonian $\tilde{H}_h = \sum_k \theta_k \tr{\rho v_k} h_k$ with sparsity $d$.
  \STATE {\bfseries 2.} Prepare a Hadamard test state, i.e., prepare an ancilla qubit in the $\ket{+}$-state and apply a controlled-$h_k$ conditioned on the ancilla register, followed by a Hadamard gate, i.e.,
  \begin{align}
    \ket{\phi} := \frac{1}{2} \left( \ket{0} \left(\ket{\psi}_{Gibbs} + (h_k \otimes I) \ket{\psi}_{Gibbs} \right) + \ket{1} \left( \ket{\psi}_{Gibbs} -  (h_k \otimes I) \ket{\psi}_{Gibbs} \right) \right)
  \end{align}
  where $\ket{\psi}_{Gibbs} := \sum_{h} \frac{e^{-E_h/2}}{\sqrt{Z}} \ket{h}_A \ket{\phi_h}_B$ is the purified Gibbs state.
  \STATE {\bfseries 3.} Perform amplitude estimation on the $\ket{0}$ state,we need to implement the amplitude estimation with reflector $P:= -2 \ket{0}\bra{0} +I$,
    and operator $G:= \left( 2 \ket{\phi}\bra{\phi} - I \right) (P \otimes I)$.
  \STATE {\bfseries 4.} Measure now the phase estimation register which returns an $\tilde \epsilon$-estimate of the probability $\frac{1}{2} \left( 1 + \mathbb{E}_h[E_{h,k}] \right)$ of the Hadamard test to return $0$
  \STATE {\bfseries 6.} Repeat the procedure for all $D$ terms and output the first term of $\nabla \tilde{S}$.
\end{algorithmic}
\end{algorithm}

The algorithm is build on three main subroutines.
The first one is Gibbs state preparation, which is a known routine which we recapitulate in Theorem~\ref{thm:Gibbs_state_prep} in the appendix.
The two remaining routines are the Hadamard test and amplitude estimation, both are well established quantum algorithms.
The Hadamard test, will allow us to estimate the probability of the outcome.
This is concretely given by
\begin{equation}
          \mathrm{Pr}(0) = \frac{1}{2} \left(1 + \mathrm{Re}\bra{\psi}_{Gibbs} (h_k \otimes I) \ket{\psi}_{Gibbs} \right) =  \frac{1}{2} \left(1 + \sum_h \frac{e^{-E_h} E_{h,k}}{Z} \right),
\end{equation}
i.e., from $\mathrm{Pr}(0)$ we can easily infer the estimate of $\mathbb{E}_h\left[ E_{h,k} \right]$ up to precision $\epsilon$ for all the $k$ terms,
since the last part is equivalent to $\frac{1}{2} \left( 1 + \mathbb{E}_h[E_{h,k}] \right)$.
To speed up the time for the evaluation of the probability $\mathrm{Pr}(0)$, we use amplitude estimation.
We recapitulate this procedure in detail in the suppemental material in section~\ref{sec:amplitude_estimation}.
In this case, we let $P:= -2 \ket{0}\bra{0} +I$ be the reflector, where $I$ is the identity which is just the Pauli $z$ matrix up to a global phase,
and let $G:= \left( 2 \ket{\phi}\bra{\phi} - I \right) (P \otimes I)$, for $\ket{\phi}$ being the state after the Hadamard test prior to the measurement.
The operator $G$ has then the eigenvalue $\mu_{\pm}= \pm e^{\pm i 2 \theta}$ , where $2 \theta = \arcsin{\sqrt{\mathrm{Pr}(0)}}$,
and $\mathrm{Pr}(0)$ is the probability to measure the ancilla qubit in the $\ket{0}$ state.
Let now $T_{Gibbs}$ be the query complexity for preparing the purified Gibbs state (c.f. eq~\eqref{eq:gibbs_state_query_complexity} in the appendix).
We can then perform phase estimation with precision $\epsilon$ for the operator $G$ requiring $O(T_{Gibbs}/\tilde \epsilon)$ queries to the oracle of $H$.

In section~\ref{app:proof_thm_gradient_results} of the appendix we analyse the runtime and error of the above algorithm.
The result is summarised in Theorem~\ref{thm:gradient_results}.

\begin{theorem}
\label{thm:gradient_results}
Assume that the Hamiltonian $H$ of the quantum Boltzmann machine takes the form
described in eq.~\ref{eq:ham_form}, where $\theta_k$ are the parameters which
determine the interaction strength and $v_k,h_k$ are unitary operators.
Furthermore, let $h_k \ket{h}=E_{h,k}\ket{h}$ be the eigenvalues of the hidden subsystem,
and  $\mathbb{E}_h(\cdot)$ as given by Definition~\ref{def:prob_distr},
i.e., the expectation value over the effective Boltzmann distribution of the visible layer with $\tilde{H}_h :=\sum_k E_{h,k} \theta_k v_k$,
and suppose that $I \preceq \tilde{H}_h$ with bounded spectral norm $\norm{\tilde{H}_h(\theta)} \leq \norm{\theta}_1$, and let $\tilde{H}_h$ be $d$-sparse.
Then $\mathcal{S}\in \mathbb{R}^D$ can be computed for any $\epsilon \in (0,\max\{1/3, 4\max_{h,p} |E_{h,p}|\})$ such that
\begin{equation}
	\norm{\mathcal{S} - \nabla \tilde{S}}_{max}  \leq \epsilon,
\end{equation}
with
\begin{equation}
\label{eq:hidden_gradient_query_complexity}
  \widetilde{\mathcal{O}} \left(\sqrt{\xi}\frac{D \norm{\theta}_1 dn^2 }{\epsilon} \right),
\end{equation}
queries to the oracle $O_H$ and $O_{\rho}$ with probability at least $2/3$,
where $\norm{\theta}_1$ is the sum of absolute values of the parameters of the Hamiltonian,
 $\xi := \max[N/z, N_h/z_h]$, $N=2^n$, $N_h=2^{n_h}$, and $z,z_h$ are known lower bounds
 on the partition functions for the Gibbs state of $H$ and $\tilde{H}_h$ respectively.
\end{theorem}

Theorem~\ref{thm:gradient_results} shows that the computational complexity of
estimating the gradient grows the closer we get to a pure state,
since for a pure state the inverse temperature $\beta\rightarrow \infty$,
and therefore the norm $\norm{H(\theta)}\rightarrow \infty$,
as the Hamiltonian is depending on the parameters, and hence the type of state we describe.
In such cases we typically would rely on alternative techniques.
However, this cannot be generically improved because otherwise we would be able
to find minimum energy configurations using a number of queries in $o(\sqrt{N})$,
which would violate lower bounds for Grover's search.
Therefore more precise statements of the complexity will require further
restrictions on the classes of problem Hamiltonians to avoid lower bounds imposed by Grover's search and similar algorithms.

\subsection{Gradient based training for general Hamiltonians}

While our first scheme for training quantum Boltzmann machines is only applicable
in case of a restricted Hamiltonian, the second scheme, which we will present
next, holds for arbitrary Hamiltonians.
In order to calculate the gradient, we use higher order divided difference estimates
for the relative entropy objective based on function approximation schemes.
For this we generate differentiation formulas by differentiating an interpolant.
The main ideas are simple:
First we construct an interpolating polynomial from the data.
Second, an approximation of the derivative at any point is obtained via a direct differentiation of the interpolant.

Concretely we perform the following steps.

We first approximate the logarithm via a Fourier-like approximation, i.e.,
$\log \sigma_v \rightarrow \log_{K,M}\sigma_v,$ where the subscripts $K,M$
indicate the level of truncation similar to~\cite{van2017quantum}.
This will yield a Fourier-like series in terms of $\sigma_v$, i.e., $\sum_m c_m \exp{(im\pi \sigma_v)}$.

Next, we need to evaluate the gradient of the function
\[
\tr{ \frac{\partial}{\partial \theta} \rho \log_{K,M}(\sigma_v)}.
\]
Taking the derivative yields many terms of the form
\begin{equation}
  \label{eq:integral_partial_derivatives}
  \int_0^1 ds e^{(ism\pi \sigma_v)} \frac{\partial \sigma_v}{\partial \theta} e^{(i(1-s)m\pi \sigma_v)},
\end{equation}
as a result of the Duhamel's formula for the derivative of exponentials
of operators (c.f., Sec.~\ref{eq:operator_log_ineq} in the mathematical preliminaries).
Each term in this expansion can furthermore be evaluated separately via a sampling procedure,
since the terms in Eq.~\ref{eq:integral_partial_derivatives} can be approximated by
\[
\mathbb{E}_s \left[ e^{(ism\pi \sigma_v)} \frac{\partial \sigma_v}{\partial \theta} e^{(i(1-s)m\pi \sigma_v)} \right].
\]
Furthermore, since we only have a logarithmic number of terms,
we can combine the results of the individual terms via classical postprocessing once we have evaluated the trace.

Now, we apply a divided difference scheme to approximate the gradient term $\frac{\partial \sigma_v}{\partial \theta}$ which results
in an interpolation polynomial $\mathcal{L}_{\mu,j}$ of order $l$ (for $l$ being
the number of points at which we evaluate the function) in $\sigma_v$ which we can efficiently evaluate.

However, evaluating these terms is still not trivial.
The final step consists hence of implementing a routine which allows us to evaluate
these terms on a quantum device. In order to do so, we once again make use of the Fourier series approach.
This time we take the simple idea of approximating the density operator $\sigma_v$ by the series of itself,
i.e., $\sigma_v \approx F(\sigma_v) := \sum_{m'} c_{m'} \exp{(im \pi m' \sigma_v)}$,
which we can implement conveniently via sample based Hamiltonian simulation~\cite{lloyd2014quantum,kimmel2017hamiltonian}.

Following these steps we obtain the expression in Eq.~\ref{eq:full_approximation}.
The real part of
\begin{equation}
  \label{eq:full_approximation}
  \sum_{m=-M_1}^{M_1} \sum_{m'=-M_2}^{M_2} \frac{i c_m \tilde{c}_{m'} m \pi}{2} \sum_{j=0}^{\mu} \mathcal{L}'_{\mu,j}(\theta) \mathbb{E}_{s \in[0,1]} \left[\tr{ \rho e^{\frac{i s \pi m}{2}\sigma_v} e^{\frac{i \pi m'}{2} \sigma_v(\theta_j)} e^{\frac{i (1-s) \pi m}{2}\sigma_v}} \right].
\end{equation}
 then approximates $\partial_{\theta} \tr{\rho \log \sigma_v}$ with at most $\epsilon$ error,
 where $\mathcal{L}'_{\mu,j}$ is the derivative of the interpolation polynomial which we obtain using divided differences,
and $\{c_i\}_{i},\{\tilde{c}_j\}_{j}$ are coefficients of the approximation polynomials,
which can efficiently be evaluated classically.
We can evaluate each term in the sum separately and combine the results then via classical post-processing,
i.e., by using the quantum computer to evaluate terms containing the trace.

\begin{algorithm}[tb!]
\caption{Gradient estimation via series approximations}
\label{alg:algo2}
\begin{algorithmic}
  \STATE {\bfseries Input:} Density matrices $\rho \in \mathbb{C}^{2^n \times 2^n}$ and $\sigma_v \in \mathbb{C}^{2^{n_v}\times 2^{n_v}}$, precalculated parameters $K,M$ and Fourier-like series for the gradient as described in eq.~\ref{eq:full_approximation}.
  \STATE {\bfseries Output:} Estimate $\mathcal{G}$ of the gradient $\nabla_{\theta}\tr{\rho \log \sigma_v}$ with guarantees in Thm.~\ref{thm:gradient_results}.
  \STATE {\bfseries 1.} Prepare the $\ket{+} \otimes \rho$ state for the Hadamard test.
  \STATE {\bfseries 2.} Conditionally on the first qubit apply sample based Hamiltonian simulation to $\rho$, i.e.,
  for $U:=e^{\frac{i s \pi m}{2}\sigma_v} e^{\frac{i \pi m'}{2} \sigma_v(\theta_j)} e^{\frac{i (1-s) \pi m}{2}\sigma_v}$, apply $\ket{0}\bra{0} \otimes I + \ket{1}{1} \otimes U$.
  \STATE {\bfseries 3.} Apply another Hadamard gate to the first qubit.
  \STATE {\bfseries 4.} Repeat the above procedure and measure the final state each time and return the averaged output.
\end{algorithmic}
\end{algorithm}

The main challenge for the algorithmic evaluation hence to compute the terms
\begin{equation}
  \label{eq:error_formula_sample_based_ham_sim}
	\tr{ \rho e^{\frac{i s \pi m}{2}\sigma_v} e^{\frac{i \pi m'}{2} \sigma_v(\theta_j)} e^{\frac{i (1-s) \pi m}{2}\sigma_v}}.
\end{equation}
Evaluating this expression is done through Algorithm~\ref{alg:algo2},
relies on two established subroutines, namely sample based Hamiltonian simulation~\cite{lloyd2014quantum,kimmel2017hamiltonian},
and the Hadamard test which we discussed earlier.
Note that the sample based Hamiltonian simulation approach introduces an additional
$\epsilon_h$-error in trace norm, which we also need to take into account in the analysis.
In section~\ref{app:proof_general_algo} of the appendix we derive the following guarantees for Algorithm.~\ref{alg:algo2}.

\begin{theorem}
  \label{thm:complexity_general_algo}
  Let $\rho, \sigma_v$ being two density matrices, $\norm{\sigma_v} <1/\pi$,
  and we have access to an oracle $O_{H}$ that computes the locations of non-zero
  matrix elements in each row and their values for the $d$-sparse
  Hamiltonian $H(\theta)$ (as per~\cite{berry2007efficient}) and an oracle $O_{\rho}$
  which returns copies of  purified density matrix of the data $\rho$, and $\epsilon \in (0,1/6)$ an error parameter.
  With probability at least $2/3$ we can obtain an estimate $\mathcal{G}$ of
  the gradient w.r.t. $\theta \in \mathbb{R}^D$ of the relative entropy
  $\nabla_{\theta} \tr{\rho \log \sigma_v}$ such that
  \begin{equation}
    \norm{\nabla_{\theta}\tr{\rho \log \sigma_v} - \mathcal{G}}_{max} \leq \epsilon,
  \end{equation}
  with
\begin{align}
\label{eq:final_query_complexity}
   \tilde{O} \left(   \sqrt{\frac{N}{z}}
    \frac{D \norm{H(\theta)}
      d   \mu^5
    \gamma
    }
    {\epsilon^3}
  \right),
\end{align}
queries to $O_H$ and $O_{\rho}$, where $\mu\in O(n_h + \log(1/\epsilon))$, $\|\partial_{\theta} \sigma_v\| \le e^\gamma$, $\norm{\sigma_v}\geq 2^{-n_v}$ for $n_v$
being the number of visible units and $n_h$ being the number of hidden units, and $$\tilde{O}\left(\text{poly}\left(\gamma, n_v, n_h,\log(1/\epsilon)\right)\right)$$
classical precomputation.
\end{theorem}

In order to obtain the bounds in Theorem~\ref{thm:complexity_general_algo},
we decompose the total error into the errors that we incur at each step of the approximation scheme,
\begin{align}
\label{eq:bounds_approximation_error}
  &\left\lvert \partial_{\theta}\tr{\rho \log \sigma_v} - \partial_{\theta} \tr{\rho \log_{K_1,M_1}^s \tilde{\sigma}_v} \right\rvert \leq \sum_i \sigma_i(\rho) \cdot \left\lVert \partial_{\theta} [\log \sigma_v - \log_{K_1,M_1}^s \tilde{\sigma}_v] \right\rVert \nonumber \\
    &\leq \sum_i \sigma_i(\rho) \cdot \left( \left\lVert \partial_{\theta} [\log \sigma_v - \log_{K_1,M_1} \sigma_v] \right\rVert \right. \nonumber \\
    &+ \left.
    \left\lVert \partial_{\theta} [\log_{K_1,M_1} \sigma_v - \log_{K_1,M_1} \tilde{\sigma}_v] \right\rVert + \left\lVert \partial_{\theta} [\log_{K_1,M_1} \tilde{\sigma_v} - \log^s_{K_1,M_1} \tilde{\sigma}_v] \right\rVert \right).
\end{align}
Then bounding each term separately and adjusting the parameters to obtain an overall error of $\epsilon$ allows us to obtain the above result.
We are hence able to use this procedure to efficiently obtain gradient estimates for a QBM with hidden units, while making minimal assumptions on the input data.

\section{Conclusion}

Generative models may play a crucial role in making quantum machine learning
algorithms practical, as they yield concise models for complex quantum states that have no known a priori structure.
In this chapter, we solved an outstanding problem in the field of generative quantum modelling:
We overcame the previous inability to train quantum generative models
with a quantum relative entropy objective for a QBM with hidden units.
In particular, the inability to train models with hidden units was a substantial drawback.

Our results show that, given an efficient subroutine for preparing Gibbs states and
an efficient algorithm for computing the matrix elements of the Hamiltonian,
one can efficiently train a quantum Boltzmann machine with the quantum relative entropy.

While a number of problems remain with our algorithms,
namely (1) that we have not given a lower bound for the query complexity, and
hence cannot answer the question whether a linear scaling in $\|\theta\|_1$ is optimal,
(2) that our results are only efficient if the complexity of performing the Gibbs state preparation is low,
we still show a first efficient algorithm for training such generative quantum models.
This is important, as such models might in the future open up avenues for fast
state preparation, and could therefore overcome some of the drawbacks related
to the data access, since the here presented model could also be trained in an online manner,
i.e., whenever a data point arrives it could be used to evaluate and update the gradient.
Once the model is trained, we could then use it to generate states from the input
distribution and use these in a larger QML routine.
With further optimisation, it is our hope that quantum Boltzmann machines may
not be just a theoretical tool that can one day be used to model quantum states,
but that they can be used to show an early use of quantum computers for
modelling quantum or classical data sets in near term quantum hardware.

The here discussed results for learning with quantum data lead also to a more philosophical question, which was
raised to us by Iordanis Kerenidis:
Is the learning of quantum distributions from polynomially-sized circuits truly a quantum problem,
or is it rather classical machine learning?

In the following we now briefly attempt to answer this question as well.

Much of the attention of the QML literature for learning of quantum data has been
focused on learning a circuit that can generate a given input quantum distribution.
However, such distributions may be assumed to be derived from a polynomially-sized
quantum circuit, which implies that the data has an efficient (polynomially-sized)
classical description, namely the quantum circuit that generated it (or the circuits).

Under the assumption that there exist polynomially-sized
quantum circuits, which generate the quantum data that we aim to learn,
the problem can, however, be seen as a classical machine learning problem.
This is because the learning can happen entirely with classical means by
learning the gates of the circuit with the slight addition that one can run the quantum circuits as well.
Note that running the quantum circuit can be seen as the query of a cost
function or a sample generation, depending on the context.
In this context, we indeed do not anticipate that a differentiation can be made between
classical or quantum data and hence the learning.

The difference might still occur when we are dealing with quantum states
that are derived from inaccessible (black-box) quantum processes.
In this case, we cannot guarantee that there exists a quantum circuit that has
a polynomially sized circuit, which is able to generate the data.
However, we are still able to learn an approximation to these states, see e.g.~\cite{aaronson2007learnability}.
In this scenario, we can describe the problem again through a two-step process
as a classical machine learning problem.
First, we can for each quantum state from the data set learn a circuit that
approximates the input state.
Then, we take the set of circuits as an input model to the classical machine learning
routine, and learn a classification of these.
In this sense, the problem is entirely classical, with the difference that
we rely on a quantum circuit for the evaluation of the cost function (e.g., the
distance of the quantum states).

Given this, we anticipate that for certain tasks, such as classification tasks,
no exponential advantage should be possible for quantum data.

\chapter{Conclusions}
\label{chap:conclusions}

In this thesis, we have investigated the area of quantum machine learning,
the combination of machine learning and quantum mechanics (in general),
and in particular here the use of quantum computers to perform machine learning.
We have particularly investigated supervised quantum machine learning methods,
and answered two important questions:

The first research question~\ref{rq:slt1} asked what is the performance
of supervised quantum machine learning algorithms
under the common assumptions of statistical learning theory.
Our analysis indicated that most of the established and existing quantum machine
learning algorithms, which are all of a theoretical nature, do not offer any
advantage over their classical counterparts under common assumptions.
In particular, we gave the first results that rigorously study the performance
of quantum machine learning algorithms in light of the target generalisation accuracy
of the learning algorithm, which is of a statistical nature.
We show that any quantum machine learning algorithm will exhibit a polynomial scaling in the number of
training samples -- immediately excluding any exponential speedup.
A more in-depth analysis of the algorithms demonstrated that indeed, taking
this into account, many of the quantum machine learning algorithms do not offer any advantage over
their classical counterparts.

This insight also allows us to immediately understand that noise in current
generations of quantum computers significantly limits our ability to obtain
good solutions, since it would immediately introduce an additional linear term
in error decomposition of Eq.~\ref{eq:totalerrordec}.

A range of questions crucially remain open.
First, while we better understand the limitations of quantum algorithms due to the
generalisation ability, we are unable to investigate these bounds with
respect to the complexity of the Hypothesis space.
A core open question and possible avenue forward to better understand the
discrepancy between classical and quantum methods is the study of the
complexity $c\,(\mathcal{H})$, i.e., the measure of the complexity of $\mathcal{H}$
(such as the VC dimension, covering numbers, or the Rademacher complexity~\cite{cucker2002mathematical,shalev2014understanding}).
This might indeed clarify whether quantum kernels or other approaches could be
used to improve the performance of quantum algorithms.
First attempts in this direction exist in the works of \cite{liu2020rigorous}
and \cite{huang2020power}. In particular, the latter identified with the tools of statistical
learning theory under which circumstances a quantum kernel can give an advantage.
While these results are very encouraging, \cite{huang2020power} also demonstrates that for most problems a quantum advantage
for learning problems is highly unlikely. In particular, the questions remain
for what type of practical problems such a quantum kernel would give an advantage.

Another open question of this thesis is the effect of the condition number.
In practice, conditioning is a crucial aspect of many algorithms, and while we give
here an indication about possible outcomes, a clear quantum advantage would need
to take the condition number also into account.
In particular, since existing algorithms, such as the one by Clader et al.~\cite{clader2013preconditioned},
have been shown not to work~\cite{harrow2017limitations}, new approaches are
needed to overcome the drawbacks of current quantum algorithms compared to their classical counterparts.

The second research question~\ref{rq:randNLA1} on the other hand investigated
the challenge posed by the oracles that are used to establish exponentially
faster (theoretical) computational complexities, i.e., run times.
Most supervised quantum machine learning methods assume the existence of a logarithmic state
preparation process, i.e., a quantum equivalent of a random access memory called qRAM~\cite{giovannetti2008qram1}.
We inspected this integral part of quantum machine learning under the lens of
randomised numerical linear algebra, and asked what is the comparative advantage
of quantum algorithms over classical ones if we can efficiently sample
from the data in both cases.

We established on the example of Hamiltonian simulation, that given the ability
to sample efficiently from the input data -- according to some form of leverage
-- we are able to obtain computational complexities for Hamiltonian simulation
that are independent of the dimension of the input matrix.
To do this, we first derived a quantum algorithm for Hamiltonian simulation,
which is based on a quantum random memory.
Next, we developed a classical algorithm that performed a similar though harder task, but which also did not have any explicit dependency on the dimension.
Our results were some of the first that introduced randomised methods into
the classical simulation of quantum dynamics and one of the first
to study the comparative advantage of classical randomised algorithms over
quantum ones.

While we didn't manage to extend our algorithms immediately to any application
in quantum machine learning, shortly after we published our results, Ewin Tang
derived the now famous quantum-inspired algorithm for recommendation systems~\cite{tang2018quantum},
which uses randomised methods similar to the ones we introduced for Hamiltonian
simulation in the context of quantum machine learning.
We hence also put our results in context by adding a discussion
regarding the literature of quantum-inspired algorithms.

Notably, the larger question remains if such memory structures such as qRAM
will ever be practically feasible.
Older results~\cite{regev2008impossibility} indicate that even small errors
can diminish a quantum advantage for search, and therefore could also crucially
affect quantum machine learning algortihms.
While recent research~\cite{arunachalam2015robustness} claims that such errors could be sustained by quantum machine
learning algorithms, we believe this question is far from clear today.

A further constraint arises from the actual costs -- end-to-end -- to run a quantum
algorithm using such a memory structure.
In reality, the cost to build and run such devices still need to be significantly reduced
in order to make any advantage economic.

In the final chapter of the thesis, we then turned to future avenues for quantum
machine learning. Since we established that the state preparation is one of the
biggest bottlenecks for successful applications and a possible advantage of quantum machine learning
over classical methods, we hence investigate generative quantum models.
We believe that these might hold the key to making quantum machine learning
models useful in the intermediate future by replacing potentially infeasible quantum memory
models such as qRAM.
We resolve a major open question in the area of generative quantum modelling.
Namely, we derive efficient algorithms for the gradient based training of
Quantum Boltzmann Machines with hidden layers using the quantum relative entropy
as an objective function.
We present in particular two algorithms.
One that applies in a more restricted
case and is more efficient.
And a second one that applies to arbitrary instances
but is less efficient.
Our results hold under the assumption that we can efficiently prepare Gibbs distributions,
which is not true in general, but to be expected to hold for most practical applications.

A key question that is not treated in this thesis is how can we establish good
benchmarks for new quantum algorithms such as in quantum machine learning.
This slightly extends the more general questions on how to judge normal optimisation
algorithms for exactly the statistical guarantees that we aim to achieve.

In practice, it might be required to perform the following steps to evaluate a
new algorithm:
\begin{enumerate}
  \item Obtain the theoretical time complexity of an algorithm (asymptotic).
  \item Determine the dependencies (error, input dimensions, condition number, etc.) that
  are in common with the classical algorithm and identify how these compare.
  \item Use known bounds on the target generalisation accuracy and see if the advantage remains.
  \item Estimate overheads (pre-factors) for the algorithm using the best known results to go from the asymptotic scaling to a concrete cost.
  \item Use estimates of condition number and other properties for some common machine learning datasets to see if a practical advantage is feasible
  (since the overhead).
  \item Take other error sources into account and check if the advantage still remains (i.e., take account error rates in the device and how
  they affect the bounds). This must also include overheads from error correction.
\end{enumerate}

Benchmarks as the above one, which are of a theoretical nature, are the only
option to establish a quantum advantage unless fully error-corrected devices are available for testing.
However, defining a common way to perform such benchmarks remains an open question in
the quantum algorithms community.

More recently, the community started to investigate the area of so-called quantum kernel methods.
These might enable us to find new ways to leverage quantum computing to
better identify patterns in certain types of data.

From a practical point of view, however, the ultimate test for any method is
the demonstrated performance on a benchmark data set, such as MNIST.
While today's quantum computers appear to be far away from this, it will be
the only true way to establish a performance advantage for quantum machine learning
over classical methods.

One interesting aspect that we have not considered in this thesis is the
reconciliation of the research questions one and two, namely how classical randomised algorithms
and quantum algortihms fit into the picture of statistical learning theory.
It appears that a much fairer comparison is possible between quantum and classical
randomised methods. Indeed, there exist classical algorithms for machine learning
which leverage randomisation, optimal rates (generalisation accuracy), and
other aspects such as preconditioning to obtain performance improvements.
Such methods typically lead to a better time-scaling while also
resulting in optimal learning rates - take for example FALCON~\cite{rudi2017falkon}.

Whether quantum machine learning will become useful in the future using
generative models for the data access, or whether the polynomial speedup of
quantum machine learning algorithms over their classical (quantum-inspired)
versions will hold remains an open question.

Ultimately, we believe that these questions can only truly be verified through
implementation and benchmarking.
Due to the current state of quantum computing hardware, we however believe that
such benchmarks will unlikely be possible in the next years, and that
researchers will need to continue to make theoretical progress long before
we will be able to resolve the questions of a `quantum advantage' through
an actual benchmark.

\addcontentsline{toc}{chapter}{Appendices}

\appendix
\chapter{Appendix 1: Quantum Subroutines}
\label{sec:appendix}

\section{Amplitude estimation}
\label{sec:amplitude_estimation}
In the following we describe the established \textit{amplitude estimation algorithm}~\cite{brassard2002quantum}:

\begin{algorithm}[tb!]
\caption{Amplitude estimation}
\label{alg:algo3}
\begin{algorithmic}
  \STATE {\bfseries Input:} Density matrix $\rho$, unitary operator $U:\mathbb{C}^{2^n} \rightarrow \mathbb{C}^{2^n}$, qubit registers $\ket{0}\otimes \ket{0}^{\otimes n}$.
  \STATE {\bfseries Output:} An $\tilde \epsilon$ close estimate of $\tr{U \rho}$.
  \STATE {\bfseries 1.} Initialize two registers of appropriate sizes to the state $\ket{0} \mathcal{A} \ket{0}$, where $\mathcal{A}$ is a unitary transformation which prepares the input state, i.e., $\ket{\psi}=\mathcal{A}\ket{0}$.
  \STATE {\bfseries 2.} Apply the quantum Fourier transform $\mathrm{QFT}_N: \ket{x}\rightarrow \frac{1}{\sqrt{N}} \sum_{y=0}^{N-1} e^{2 \pi i x y/N} \ket{y}$ for $0 \leq x < N$, to the first register.
  \STATE {\bfseries 3.} Apply $\Lambda_N(Q)$ to the second register, i.e., let $\Lambda_N(U): \ket{j}\ket{y}\rightarrow \ket{j}(U^j \ket{y})$ for $0\leq j < N$,
                      then we apply $\Lambda_N(Q)$ where $Q:= -\mathcal{A}S_0 \mathcal{A}^{\dagger}S_t$ is the Grover's operator.
  \STATE {\bfseries 4.} Apply $\mathrm{QFT}^{\dagger}_N$ to the first register.
  \STATE {\bfseries 5.} Return $\tilde{a} = \sin^2(\pi \frac{\tilde \theta}{N})$.
\end{algorithmic}
\end{algorithm}

Algorithm~\ref{alg:algo3} describes the amplitude estimation algorithm.
The output is an $\epsilon$-close estimate of the target amplitude.
Note that in step (3), $S_0$ changes the sign of the amplitude if and only if the state is the zero state $\ket{0}$,
and $S_t$ is the sign-flip operator for the target state, i.e., if $\ket{x}$ is the desired outcome, then $S_t := I - 2\ket{x}\bra{x}$.

The algorithm can be summarized as the unitary transformation $$ \left( (\mathrm{QFT}^{\dagger} \otimes I) \Lambda_N(Q) (\mathrm{QFT}_N \otimes I)\right)$$
applied to the state $\ket{0}\mathcal{A}\ket{0}$, followed by a measurement of the first register and classical post-processing returns an estimate $\tilde \theta$ of the amplitude of the desired outcome such that $\lvert \theta - \tilde \theta \rvert \leq \epsilon$ with probability at least $8/\pi^2$.
The result is summarized in the following theorem, which states a slightly more general version.
\begin{theorem}[Amplitude Estimation~\cite{brassard2002quantum}]
\label{thm:amplitude_estimation}
For any positive integer $k$, the Amplitude Estimation Algorithm returns an estimate $\tilde a$ ($0\leq \tilde a \leq 1$) such that $$\lvert \tilde a - a \rvert \leq 2 \pi k \frac{\sqrt{a(1-a)}}{N} + k^2 \frac{\pi^2}{N^2}$$
with probability at least $\frac{8}{\pi^2}\approx 0.81$ for $k=1$ and with probability greater than $1-\frac{1}{2(k-1)}$ for $k\geq 2$. If $a=0$ then $\tilde a=0$ with certainty, and and if $a=1$ and $N$ is even, then $\tilde{a}=1$ with certainty.
\end{theorem}
Notice that the amplitude $\theta$ can hence be recovered via the relation $\theta = \arcsin{\sqrt{\theta_a}}$ as described above which incurs an $\epsilon$-error for $\theta$ (c.f., Lemma 7, \cite{brassard2002quantum}).

\section{The Hadamard test}
\label{sec:eval_dm_operator}

Here we present an easy subroutine to evaluate the trace of products of unitary operators $U$ with a density matrix $\rho$, which is known as the Hadamard test.

\begin{algorithm}[tb!]
\caption{Variational gradient estimation - term 2}
\label{alg:algo0}
\begin{algorithmic}
  \STATE {\bfseries Input:} Density matrix $\rho$, unitary operator $U:\mathbb{C}^{2^n} \rightarrow \mathbb{C}^{2^n}$, qubit registers $\ket{0}\otimes \ket{0}^{\otimes n}$.
  \STATE {\bfseries Output:} An $\tilde \epsilon$ close estimate of $\tr{U \rho}$.
  \STATE {\bfseries 1.} Prepare the first qubits $\ket{+}$ state and initialize the second register to $0$.
  \STATE {\bfseries 2.} Use an appropriate subroutine to prepare the density matrix $\rho$ on the second register to obtain the state $\ket{+}\bra{+}\otimes \rho$.
  \STATE {\bfseries 3.} Apply a controlled operation $\ket{0}\bra{0} \otimes I_{2^n} + \ket{1}\bra{1} \otimes U$, followed by a Hadamard gate.
  \STATE {\bfseries 4.} Perform amplitude estimation on the $\ket{0}$ state, via the reflector $P:= -2 \ket{0}\bra{0} +I$, and operator $G:= \left( 2 \rho - I \right) (P \otimes I)$.
  \STATE {\bfseries 5.} Measure now the phase estimation register which returns an $\tilde \epsilon$-estimate of the probability $\frac{1}{2} \left( 1 + \mathrm{Re}\left[\tr{U\rho}\right] \right) $ of the Hadamard test to return $0$.
  \STATE {\bfseries 6.} Repeat the procedure for an additional controled application of $\exp(i\pi/2)$ in step (3) to recover also the imaginary part of the result.
  \STATE {\bfseries 7.} Return the real and imaginary part of the probability estimates.
\end{algorithmic}
\end{algorithm}

Note that this procedure can easily be adapted to be used for $\rho$ being some Gibbs distribution.
We then would use a Gibbs state preparation routine in step (2).
For example for the evaluation of the gradient of the variational bound, we require this subroutine to evaluate $U=\partial_{\theta}H$ for $\rho$ being the Gibbs distribution corresponding to the Hamiltonian $H$.

\chapter{Appendix 2: Deferred proofs}

First for convenience, we formally define quantum Boltzmann machines below.

\begin{definition}
A quantum Boltzmann machine to be a quantum mechanical system that acts on a tensor product of Hilbert spaces $\mathcal{H}_v\otimes \mathcal{H}_h \in \mathbb{C}^{2^n}$ that correspond to the visible and hidden subsystems of the Boltzmann machine.  It further has a Hamiltonian of the form $H \in \mathbb{C}^{2^n \times 2^n}$ such that $\norm{H - \text{diag}(H)} > 0$.  The quantum Boltzmann machine takes these parameters and then outputs a state of the form ${\rm Tr}_h \left(\frac{e^{-H}}{{\rm Tr}(e^{-H})} \right)$.
\end{definition}
Given this definition, we are then able to discuss the gradient of the relative entropy between the output of a quantum Boltzmann machines and the input data that it is trained with.

\section{Derivation of the variational bound}
\label{sec:variational_bound_derivative}
\begin{proof}[Proof of Lemma~\ref{def:Stilde}]
  Recall that we assume that the Hamiltonian $H$ takes the form $$H:=\sum_k \theta_k v_k \otimes h_k,$$ where $v_k$ and $h_k$ are operators acting on the visible and hidden units respectively and we can assume $h_k=d_k$ to be diagonal in the chosen basis.
  Under the assumption that $[h_i,h_j]=0, \forall i,j$, c.f. the assumptions in \eqref{eq:assumptions}, there exists a basis $\{\ket{h}\}$ for the hidden subspace such that $h_k \ket{h}=E_{h,k}\ket{h}$.
  With these assumptions we can hence reformulate the logarithm as
  \begin{align}
    &\log\trh{e^{-H}} = \log \left( \sum_{v,v',h}  \bra{v,h} e^{-\sum_{k} \theta_k v_k \otimes h_k} \ket{v',h} \ket{v}\bra{v'}\right) \\
    &= \log \left( \sum_{v,v',h}  \bra{v} e^{-\sum_{k} E_{h,k} \theta_k v_k} \ket{v'} \ket{v}\bra{v'} \right) \\
    &= \log \left(\sum_h e^{-\sum_{k} E_{h,k} \theta_k v_k} \right),
  \end{align}
  where it is important to note that $v_k$ are operators and we hence just used the matrix representation of these in the last step.
  In order to further simplify this expression, first note that each term in the sum is a positive semi-definite operator.
  In particularly, note that the matrix logarithm is operator concave and operator monotone, and hence by Jensen's inequality, for any sequence of non-negative number $\{\alpha_i\}: \sum_i \alpha_i =1$ we have that $$\log \left( \frac{\sum_{i=1}^N \alpha_i U_i}{\sum_j \alpha_j} \right) \geq \frac{\sum_{i=1}^N \alpha_i\log \left(  U_i \right)}{\sum_j \alpha_j}.$$
  and since we are optimizing $\tr{\rho \log \rho} -\tr{\rho \log \sigma_v}$ we hence obtain for arbitrary choice of $\{\alpha_i\}_i$ under the above constraints,
  \begin{align}
  \tr{\rho \log \left(\sum_{h=1}^N e^{-\sum_{k} E_{h,k} \theta_k v_k} \right)} &= \tr{\rho \log \left(\sum_{h=1}^N \alpha_h \frac{e^{-\sum_{k} E_{h,k} \theta_k v_k}/\alpha_h }{\sum_{h'} \alpha_{h'}} \right)} \nonumber \\
  & \geq -\tr{\rho \frac{\sum_h \alpha_h \sum_{k} E_{h,k} \theta_k v_k + \sum_h \alpha_h \log \alpha_h}{\sum_{h'} \alpha_{h'}}}.
  \end{align}
  Hence, the variational bound on the objective function for any $\{\alpha_i\}_i$ is
  \begin{align}
  \label{eq:def_variational_bound}
  \mathcal{O}_{\rho}(H) =& \tr{\rho \log \rho} - \tr{\rho \log \sigma_v} \nonumber \\
            &\leq \tr{\rho \log \rho} +  \tr{\rho \frac{\sum_h \alpha_h \sum_{k} E_{h,k} \theta_k v_k + \sum_h \alpha_h \log \alpha_h}{\sum_{h'} \alpha_{h'}}} + \log Z  =: \tilde{S}
  \end{align}
\end{proof}

\section{Gradient estimation}
\label{sec:var_grad_estimation}
For the following result we will rely on a variational bound in order to train the quantum Boltzmann machine weights for a Hamiltonian $H$ of the form given in \eqref{eq:ham_form}.
We begin by proving Lemma~\ref{lem:gradient_visible_layer} in the main work, which will give us an upper bound for the gradient of the relative entropy.

\begin{proof}[Proof of Lemma~\ref{lem:gradient_visible_layer}]
	We first derive the gradient of the normalization term ($Z$) in the relative entropy, which can be trivially evaluated using Duhamels formula to obtain $$\frac{\partial}{\partial \theta_p} \log \tr{e^{-H}} = - \tr{\frac{\partial H}{\partial \theta_p} \frac{e^{-H}}{Z}} =- \tr{\sigma \partial_{\theta_p} H}.$$
	Note that we can easily evaluate this term by first preparing the Gibbs state $\sigma_{Gibbs} := e^{-H}/Z$ and then evaluating the expectation value of the operator $\partial_{\theta_p}H$ w.r.t. the Gibbs state, using amplitude estimation for the Hadamard test.
	If $T_{Gibbs}$ is the query complexity for the Gibbs state preparation, the query complexity of the whole algorithm including the phase estimation step is then given by $O(T_{Gibbs}/\tilde{\epsilon})$ for an $\tilde{\epsilon}$-accurate estimate of phase estimation.
Taking into account the desired accuracy and the error propagation will hence straight forward give the computational complexity to evaluate this part.\\
	We now proceed with the gradient evaluations for the model term.
    Using the variational bound on the objective function for any $\{\alpha_i\}_i$, given in eq.~\ref{eq:def_variational_bound}, we obtain the gradient
    \begin{align}
      \frac{\partial \tilde{S}}{\partial_{\theta_p}} &=- \tr{\frac{\partial H}{\partial \theta_p} \frac{e^{-H}}{Z}} + \tr{ \frac{\partial }{\partial \theta_p} \rho \sum_h \alpha_h \sum_{k} E_{h,k} \theta_k v_k }+   \frac{\partial }{\partial \theta_p} \sum_h \alpha_h \log \alpha_h \\
&= - \tr{\frac{\partial H}{\partial \theta_p} \frac{e^{-H}}{Z}} +  \frac{\partial }{\partial \theta_p} \left( \sum_h \alpha_h \tr{ \rho \sum_{k} E_{h,k} \theta_k v_k }+  \sum_h \alpha_h \log \alpha_h \right)
    \end{align}
    where the first term results from the partition sum.
The latter term can be seen as a new effective Hamiltonian, while the latter term is the entropy.
The latter term hence resembles the free energy $F(h)=E(h)-TS(h)$, where $E(h)$ is the mean energy of the effective system with energies $E(h):=\tr{ \rho \sum_{k} E_{h,k} \theta_k v_k }$, $T$ the temperature and $S(h)$ the Shannon entropy of the $\alpha_h$ distribution. We now want to choose these $\alpha_h$ terms to minimize this variational upper bound.
It is well-established in statistical physics, see for example~\cite{landau1980statistical}, that the distribution which maximizes the free energy is the Boltzmann (or Gibbs) distribution, i.e., $$\alpha_h =  \frac{e^{- \tr{\rho \tilde{H}_h }}}{\sum_h e^{- \tr{\rho \tilde{H}_h}}},$$
where $\tilde{H}_h :=\sum_k E_{h,k} \theta_k v_k$ is a new effective Hamiltonian on the visible units, and the $\{\alpha_i\}$ are given by the corresponding Gibbs distribution for the visible units.

Therefore, our gradients can be taken with respect to this distribution and the bound above, where $\tr{\rho \tilde{H}_h}$ is the mean energy of the the effective visible system w.r.t. the data-distribution.
For the derivative of the energy term we obtain
\begin{align}
  &\frac{\partial }{\partial \theta_p}  \sum_h \alpha_h \tr{ \rho \sum_{k} E_{h,k} \theta_k v_k } = \\
  &= \sum_h \left( \alpha_h \left( \mathbb{E}_{h'}\left[\tr{ \rho E_{h',p} v_p }\right] - \tr{\rho E_{h,p} v_p} \right) \tr{\rho \tilde{H}_h} + \alpha_h \tr{\rho E_{h,p} v_p}\right) \\
  &= \mathbb{E}_{h} \left[\left( \mathbb{E}_{h'}\left[\tr{ \rho E_{h',p} v_p }\right] - \tr{\rho E_{h,p} v_p} \right) \tr{\rho \tilde{H}_h} + \tr{\rho E_{h,p} v_p}\right],
\end{align}
while the entropy term yields
\begin{align}
  \frac{\partial}{\partial \theta_p} \sum_h \alpha_h \log \alpha_h &=  \sum_h \alpha_h \left( \left[ \tr{\rho E_{h,p}v_p} - \mathbb{E}_{h'} \left[\tr{\rho E_{h',p}v_p} \right] \right] \tr{\rho \tilde{H}_h} - \tr{\rho E_{h,p} v_p} \right) \nonumber \\
  &+ \sum_h \alpha_h \left( \tr{\rho E_{h,p}} - \mathbb{E}_{h'} \left[\tr{\rho E_{h',p}v_p} \right] \right) \log  \tr{e^{-\tilde{H}_h}} \nonumber \\
  &+ \mathbb{E}_{h'} \left[\tr{\rho E_{h',p} v_p} \right].
\end{align}
This can be further simplified to
\begin{align}
  &\sum_h \alpha_h \left( \tr{\rho E_{h,p}v_p} - \mathbb{E}_{h'} \left[\tr{\rho E_{h',p}v_p} \right] \right) \tr{\rho \tilde{H}_h} \\
  =&\mathbb{E}_h \left[ \left( \tr{\rho E_{h,p}v_p} - \mathbb{E}_{h'} \left[\tr{\rho E_{h',p}v_p} \right] \right) \tr{\rho \tilde{H}_h} \right].
\end{align}
Te resulting gradient for the variational bound for the visible terms is hence given by
  \begin{align}
        \frac{\partial \tilde{S}}{\partial_{\theta_p}} &= \mathbb{E}_{h} \left[ \tr{\rho E_{h,p} v_p}\right] - \tr{\frac{\partial H}{\partial \theta_p} \frac{e^{-H}}{Z}}
  \end{align}
\end{proof}
Notably, if we consider no interactions between the visible and the hidden layer, then indeed the gradient above reduces recovers the gradient for the visible Boltzmann machine, which was treated in \cite{kieferova2016tomography}, resulting in the gradient
\begin{equation*}
  \tr{\rho \partial_{\theta_p}H} - \tr{\frac{e^{-H}}{Z} \partial_{\theta_p} H},
\end{equation*}
under our assumption on the form of $H$, $\partial_{\theta_p}H = v_p$.

\subsection{Operationalizing the gradient based training}
\label{app:proof_thm_gradient_results}
From Lemma~\ref{lem:gradient_visible_layer}, we know that the derivative of the relative entropy w.r.t.\ any parameter $\theta_p$ can be stated as
    \begin{align}
	\label{eq:gradient_approach_1_recap}
        \frac{\partial \tilde{S}}{\partial_{\theta_p}} =   \mathbb{E}_h \left[ E_{h,p} \right] \tr{\rho v_p} - \tr{\frac{\partial H}{\partial \theta_p} \frac{e^{-H}}{Z}}.
    \end{align}
    Since evaluating the latter part is, as mentioned above, straight forward, we give here an algorithm for evaluating the first part.\\
	Now note that we can evaluate each term $\tr{\rho v_k}$ individually for all $k \in [D]$, i.e., all $D$ dimensions of the gradient via the Hadamard test for $v_k$, assuming $v_k$ is unitary.
	More generally, for non-unitary $v_k$ we could evaluate this term using a linear combination of unitary operations.
	Therefore, the remaining task is to evaluate the terms $\mathbb{E}_{h} \left[ E_{h,p}\right]$ in \eqref{eq:gradient_approach_1_recap}, which reduces to sampling according to the distribution $\{\alpha_h\}$.\\
	For this we need to be able to create a Gibbs distribution for the effective Hamiltonian $\tilde{H}_h = \sum_k \theta_k \tr{\rho v_k} h_k$ which contains only $D$ terms and can hence be evaluated efficiently as long as $D$ is small, which we can generally assume to be true.   	In order to sample according to the distribution $\{\alpha_h\}$, we first evaluate the factors $\theta_k \tr{\rho v_k}$ in the sum over $k$ via the Hadamard test, and then use these in order to implement the Gibbs distribution $\exp{(-\tilde{H_h})}/\tilde{Z}$ for the Hamiltonian $$\tilde{H}_h = \sum_k \theta_k \tr{\rho v_k} h_k.$$
In order to do so, we adapt the results of \cite{van2017quantum} in order to prepare the corresponding Gibbs state (although alternative methods can also be used~\cite{poulin2009sampling,chowdhury2016quantum,yung2012quantum}).

 \begin{theorem}[Gibbs state preparation~\cite{van2017quantum}]
	\label{thm:Gibbs_state_prep}
    Suppose that $I \preceq H$ and we are given $K\in \mathbb{R}_+$ such that $\norm{H}\leq 2K$, and let $H \in \mathbb{C}^{N\times N}$ be a $d$-sparse Hamiltonian, and we know a lower bound $z\leq Z=\tr{e^{-H}}$.
    If $\epsilon \in (0,1/3)$, then we can prepare a purified Gibbs state $\ket{\gamma}_{AB}$ such that
    \begin{equation}
      \norm{\mathrm{Tr}_B \left[\ket{\gamma}\bra{\gamma}_{AB} \right] - \frac{e^{-H}}{Z}} \leq \epsilon
    \end{equation}
    using
    \begin{equation}
	\label{eq:gibbs_state_query_complexity}
      \tilde{\mathcal{O}} \left(\sqrt{\frac{N}{z}} Kd \log \left( \frac{K}{\epsilon}\right) \log \left( \frac{1}{\epsilon} \right) \right)
    \end{equation}
    queries, and
    \begin{equation}
      \tilde{\mathcal{O}} \left(\sqrt{\frac{N}{z}}  Kd \log \left( \frac{K}{\epsilon}\right) \log \left( \frac{1}{\epsilon} \right)  \left[ \log(N) + \log^{5/2} \left( \frac{K}{\epsilon} \right) \right] \right)
    \end{equation}
    gates.
  \end{theorem}
  Note that by using the above algorithm with $\tilde{H}_{sim}/2$, the preparation of the purified Gibbs state will leave us in the state
  \begin{equation}
    \ket{\psi}_{Gibbs} := \sum_{h} \frac{e^{-E_h/2}}{\sqrt{Z}} \ket{h}_A \ket{\phi_h}_B,
  \end{equation}
where $\ket{\phi_j}_B$ are mutually orthogonal trash states, which can typically be chosen to be $\ket{h}$, i.e., a copy of the first register, which is irrelevant for our computation, and $\ket{h}_A$ are the eigenstates of $\tilde{H}$.
Tracing out the second register will hence leave us in the corresponding Gibbs state $$\sigma_h := \sum_h \frac{e^{-E_h}}{Z} \ket{h}\bra{h}_A,$$
and we can hence now use the Hadamard test with input $h_k$ and $\sigma_h$, i.e., the operators on the hidden units and the Gibbs state, and estimate the expectation value $\mathbb{E}_h \left[E_{h,k}\right]$.
We provide such a method below.

\begin{proof}[Proof of Theorem~\ref{thm:gradient_results}]
Conceptually, we perform the following steps, starting with Gibbs state preparation followed by a Hadamard test coupled with amplitude estimation to obtain estimates of the probability of a $0$ measurement.
The proof follows straight from the algorithm described in~\ref{alg:algo1}.

From this we see that the runtime constitutes the query complexity of preparing the Gibbs state $$T_{Gibbs}^V= \tilde{\mathcal{O}} \left(\sqrt{\frac{2^n}{Z}}\frac{\norm{H(\theta)}d}{\epsilon} \log \left( \frac{\norm{H(\theta)}}{\tilde \epsilon}\right) \log \left( \frac{1}{\tilde \epsilon} \right) \right),$$ where $2^n$ is the dimension of the Hamiltonian, as given  in Theorem~\ref{thm:gradient_results} and combining it with the query complexity of the amplitude estimation procedure, i.e., $1/\epsilon$.
However, in order to obtain a final error of $\epsilon$, we will also need to account for the error in the Gibbs state preparation.
For this, note that we estimate terms of the form $\mathrm{Tr}_{AB} \left[\bra{\psi}_{Gibbs} (h_k \otimes I) \ket{\psi}_{Gibbs}^V\right] = \mathrm{Tr}_{AB} \left[ (h_k \otimes I) \ket{\psi}_{Gibbs}^V\bra{\psi}_{Gibbs}^V\right]$.
We can hence estimate the error w.r.t. the true Gibbs state $\sigma_{Gibbs}$ as
\begin{align}
\label{eq:error_prop_Gibbs_state_prep}
&\mathrm{Tr}_{AB} \left[ (h_k \otimes I) \ket{\psi}_{Gibbs}^V\bra{\psi}_{Gibbs}^V\right] -\mathrm{Tr}_{A} \left[ h_k \sigma_{Gibbs}\right]  \nonumber \\
&\qquad= \mathrm{Tr}_{A} \left[ h_k \mathrm{Tr}_{B}\left[\ket{\psi}_{Gibbs}^V\bra{\psi}_{Gibbs}^V\right] - h_k \sigma_{Gibbs}\right] \nonumber \\
&\qquad\leq \sum_i \sigma_i(h_k) \norm{\mathrm{Tr}_{B}\left[\ket{\psi}_{Gibbs}^V\bra{\psi}_{Gibbs}^V\right] - \sigma_{Gibbs}} \nonumber \\
&\qquad \leq \tilde \epsilon \sum_i \sigma_i(h_k).
\end{align}
For the final error being less then $\epsilon$, the precision we use in the phase estimation procedure, we hence need to set $\tilde \epsilon = \epsilon/(2 \sum_i \sigma_i(h_k)) \leq 2^{-n-1} \epsilon$, reminding that $h_k$ is unitary, and similarly precision $\epsilon/2$ for the amplitude estimation, which yields the query complexity of
\begin{align}
 &\mathcal{O} \left(\sqrt{\frac{N_h}{z_h}}\frac{\norm{H(\theta)}d}{\epsilon} \left(n^2+ n \log \left( \frac{\norm{H(\theta)}}{ \epsilon}\right)  + n \log \left( \frac{1}{  \epsilon}\right) + \log \left( \frac{\norm{H(\theta)}}{ \epsilon}\right) \log \left( \frac{1}{  \epsilon}\right) \right) \right),
 \nonumber\\
 &\qquad \in \widetilde{O} \left(\sqrt{\frac{N_h}{z_h}}\left(\frac{n^2 \|\theta\|_1 d}{\epsilon} \right) \right).
\end{align}
where we denote with $A$ the hidden subsystem with dimensionality $2^{n_h} \leq 2^N$, on which we want to prepare the Gibbs state and with $B$ the subsystem for the trash state.

Similarly, for the evaluation of the second part in \eqref{eq:gradient_approach_1_recap} requires the Gibbs state preparation for $H$, the Hadamard test and phase estimation.
Similar as above we meed to take into account the error.
Letting the purified version of the Gibbs state for $H$ be given by $\ket{\psi}_{Gibbs}$, which we obtain using Theorem~\ref{thm:Gibbs_state_prep}, and $\sigma_{Gibbs}$ be the perfect state, then the error is given by
\begin{align}
 &\mathrm{Tr}_{AB} \left[ (v_k \otimes h_k \otimes I) \ket{\psi}_{Gibbs} \bra{\psi}_{Gibbs} \right] -\mathrm{Tr}_{A} \left[ (v_k \otimes h_k)\sigma_{Gibbs}\right]  \nonumber \\
&\qquad= \mathrm{Tr}_{A} \left[ (v_k \otimes h_k) \mathrm{Tr}_{B}\left[\ket{\psi}_{Gibbs}^V\bra{\psi}_{Gibbs}^V\right] - (v_k \otimes h_k) \sigma_{Gibbs}\right] \nonumber \\
&\qquad\leq \sum_i \sigma_i(v_k \otimes h_k) \norm{\mathrm{Tr}_{B}\left[\ket{\psi}_{Gibbs}^V\bra{\psi}_{Gibbs}^V\right] - h_k \sigma_{Gibbs}} \nonumber \\
&\qquad \leq \tilde \epsilon \sum_i \sigma_i(v_k \otimes h_k),
\end{align}
where in this case $A$ is the subsystem of the visible and hidden subspace and $B$ the trash system.
We hence upper bound the error similar as above and introducing $\xi := \max[N/z, N_h/z_h]$ we can find a uniform bound on the query complexity for evaluating a single entry of the $D$-dimensional gradient is in
$$
 \widetilde{O} \left(\sqrt{\zeta}\left(\frac{n^2 \|\theta\|_1 d}{\epsilon} \right) \right),
$$
thus we attain the claimed query complexity by repeating this procedure for each of the $D$ components of the estimated gradient vector $\mathcal{S}$.

Note that we also need to evaluate the terms $\tr{\rho v_k}$ to precision $\hat{\epsilon} \leq \epsilon$, which though only incurs an additive cost of $D/\epsilon$ to the total query complexity, since this step is required to be performed once. Note that $|\mathbb{E}_h(h_p)|\le 1$ because $h_p$ is assumed to be unitary.
To complete the proof we only need to take the success probability of the amplitude estimation process into account.
For completeness we state the algorithm in the appendix and here refer only to Theorem~\ref{thm:amplitude_estimation}, from which we have that the procedure succeeds with probability at least $8/\pi^2$.
In order to have a failure probability of the final algorithm of less than $1/3$, we need to repeat the procedure for all $d$ dimensions of the gradient and take the median.
We can bound the number of repetitions in the following way.

Let $n_f$ be the number of instances of the gradient estimate such that the error is larger than $\epsilon$ and $n_s$ be the number of instances with an error $\leq \epsilon$ for one dimension of the gradient, and the result that we take is the median of the estimates, where we take $n=n_s+n_f$ samples.
The algorithm gives a wrong answer for each dimension if $n_s \leq \left\lfloor \frac{n}{2} \right\rfloor$, since then the median is a sample such that the error is not bound by $\epsilon$.
Let $p=8/\pi^2$ be the success probability to draw a positive sample, as is the case of the amplitude estimation procedure.
Since each instance of the phase estimation algorithm will independently return an estimate, the total failure probability is given by the union bound, i.e.,
\begin{equation}
\label{eq:probability_boost}
\pr_{fail} \leq D \cdot \pr \left[n_s \leq  \left\lfloor \frac{n}{2} \right\rfloor \right] \leq D \cdot e^{- \frac{n}{2p} \left(p - \frac{1}{2} \right)^2} \leq \frac{1}{3},
\end{equation}
which follows from the Chernoff inequality for a binomial variable with $p>1/2$, which is given in our case.
Therefore, by taking $n \geq \frac{2p}{(p-1/2)^2} \log(3D) = \frac{16}{(8-\pi^2/2)^2} \log(3D) = O(\log(3D))$, we achieve a total failure probability of at most $1/3$.

This is sufficient to demonstrate the validity of the algorithm if
\begin{equation}
  \label{eq:trace_est}
  \tr{\rho \tilde H_h}
\end{equation}
is known exactly.  This is difficult to do because the probability distribution $\alpha_h$ is not usually known apriori.  As a result, we assume that
the distribution will be learned empirically and to do so we will need to draw samples from the purified Gibbs states used as input.
This sampling procedure will incur errors.
To take such errors into account assume that we can obtain estimates $T_h$ of \ref{eq:trace_est} with precision $\delta_t$, i.e.,
\begin{equation}
  \left\lvert T_h - \tr{\rho \tilde H_h} \right\rvert \leq \delta_t.
\end{equation}
Under this assumption we can now bound the distance $\lvert \alpha_h - \tilde{\alpha}_h \rvert$ in the following way.
Observe that
\begin{align}
  \label{eq:error_step1}
  \left\lvert  \alpha_h - \tilde{\alpha}_h\right\rvert &= \left\lvert  \frac{e^{-\tr{\rho \tilde{H}_h}}}{\sum_h e^{-\tr{\rho \tilde{H}_h}}} - \frac{T_h}{\sum_h T_h} \right\rvert \nonumber \\
    &\leq \left\lvert \frac{e^{-\tr{\rho \tilde{H}_h}}}{\sum_h e^{-\tr{\rho \tilde{H}_h}}} - \frac{T_h}{\sum_h e^{-\tr{\rho \tilde{H}_h}}} \right\rvert + \left\lvert \frac{T_h}{\sum_h e^{-\tr{\rho \tilde{H}_h}}} - \frac{T_h}{\sum_h T_h} \right\rvert,
\end{align}
and we hence need to bound the following two quantities in order to bound the error.
First, we need a bound on
\begin{align}
  \label{eq:error_first_part_1}
  \left\lvert e^{-\tr{\rho \tilde H_h}} - e^{-T_h} \right\rvert.
\end{align}
For this, let $f(s):= T_h\ (1-s) + \tr{\rho \tilde H_h}\ s$, such that eq.~\ref{eq:error_first_part_1} can be rewritten as
\begin{align}
  \label{eq:error_first_part_2}
  \left\lvert e^{-f(1)} - e^{-f(0)} \right\rvert &= \left\lvert \int_0^1 \frac{d}{ds} e^{-f(s)} ds \right\rvert \nonumber \\
        &= \left\lvert \int_0^1 \dot{f}(s) e^{-f(s)} ds \right\rvert \nonumber \\
        &= \left\lvert \int_0^1 \left(\tr{\rho \tilde H_h} - T_h \right) e^{-f(s)} ds \right\rvert \nonumber \\
        &\leq \delta e^{-\min_s f(s)} \nonumber \\
        &\leq \delta e^{-\tr{\rho
         H_h} + \delta}
\end{align}
and assuming $\delta \leq \log(2)$, this reduces to
\begin{align}
  \left\lvert e^{-f(1)} - e^{-f(0)} \right\rvert \leq 2 \delta e^{-\tr{\rho \tilde H_h}}.
\end{align}
Second, we need the fact that
\begin{align}
  \left\lvert \sum_h e^{-\tr{\rho \tilde H_h}} - \sum_h T_h \right\rvert \leq 2 \delta \sum_h e^{-\tr{\rho \tilde H_h}}.
\end{align}
Using this, eq.~\ref{eq:error_step1} can be upper bound by
\begin{align}
  &\frac{2 \delta e^{-\tr{\rho \tilde H_h}}}{\sum_h e^{-\tr{\rho \tilde H_h}}} + \lvert T_h\rvert \left\lvert \frac{1}{\sum_h e^{-\tr{\rho \tilde H_h}}} - \frac{1}{(1-2\delta)\sum_h e^{-\tr{\rho \tilde H_h}}} \right\rvert \nonumber \\
      &\qquad\leq \frac{2 \delta e^{-\tr{\rho \tilde H_h}}}{\sum_h e^{-\tr{\rho \tilde H_h}}} + \frac{4 \delta \lvert T_h\rvert}{\sum_h e^{-\tr{\rho \tilde H_h}}},
\end{align}
where we used that $\delta \leq 1/4$.
Note that
\begin{align}
  4 \delta \lvert T_h \rvert &\leq 4 \delta \left(e^{-\tr{\rho \tilde H_h}} + 2 \delta e^{-\tr{\rho \tilde H_h}} \right) \nonumber \\
            &= e^{-\tr{\rho \tilde H_h}} \left(4 \delta + 8 \delta^2 \right) \nonumber \\
            &\leq e^{-\tr{\rho \tilde H_h}}(4 \delta + 2 \delta) \nonumber \\
            &\leq 6 \delta e^{-\tr{\rho \tilde H_h}},
\end{align}
which leads to a final error of
\begin{equation}
  \lvert \alpha_h - \tilde{\alpha}_h \rvert \leq 8 \delta \frac{e^{-\tr{\rho \tilde H_h}}}{\sum_h e^{-\tr{\rho \tilde H_h}}}.
\end{equation}
With this we can now bound the error in the expectation w.r.t. the faulty distribution for some function $f(h)$ to be
\begin{align}
  \left\lvert \mathbb{E}_h(f(h)) - \tilde{\mathbb{E}}_h(f(h)) \right\rvert
  &\leq 8 \delta \sum_h \frac{ f(h) e^{-\tr{\rho \tilde H_h}} }{\sum_h e^{-\tr{\rho \tilde H_h}}} \nonumber \\
        &\leq 8 \delta \max_h f(h).
\end{align}
We can hence use this in order to estimate the error introduced in the first term of eq.~\ref{eq:gradient_approach_1_recap} through errors in the distribution $\{ \alpha_h\}$ as
\begin{align}
\left\lvert \mathbb{E}_h[E_{h,p}]\tr{\rho v_p} - \tilde{\mathbb{E}}[E_{h,p}\tr{\rho v_p}] \right\rvert &\leq 8 \delta \max_{h} \lvert E_{h,p} \tr{\rho v_p} \rvert \nonumber \\
  &\leq 8 \delta \max_{h,p} \lvert E_{h,p} \rvert,
\end{align}
where we used in the last step the unitarity of $v_k$ and the Von-Neumann trace inequality.
For an final error of $\epsilon$, we hence choose $\delta_t = \epsilon/[16\max_{h,p}|E_{h,p}|]$ to ensure that this sampling error incurrs at most
half the error budget of $\epsilon$.
Thus we ensure $\delta \le 1/4$ if $\epsilon \le 4 \max_{h,p} |E_{h,p}|$.

We can improve the query complexity of estimating the above expectation by values by using amplitude amplification, sice we obtain the measurement via a Hadamard test.
For this case we require only $O(\max_{h,p}|E_{h,p}|/\epsilon)$ samples in order to achieve the desired accuracy from the sampling.
Noting that we might not be able to even access $\tilde{H}_h$ without any error, we can deduce that the error of the individual terms of $\tilde{H}_h$ for an $\epsilon$-error in the final estimate must be bounded by $\delta_tv\norm{\theta}_1$,
where with abuse of notation, $\delta_t$ now denotes the error in the estimates of $E_{h,k}$.
Even taking this into account, the evaluation of this contribution is however dominated by the second term, and hence can be neglected in the analysis.

\end{proof}

\section{Approach 2: Divided Differences}

In this section we develop a scheme to train a quantum Boltzmann machine using divided difference estimates for the relative entropy error.
The idea for this is straightforward: First we construct an interpolating polynomial from the data.
Second, an approximation of the derivative at tany point can be then obtained by a direct differentiation of the interpolant.
We assume in the following that we can simulate and evaluate $\tr{\rho \log \sigma_v}$.
As this is generally non-trivial, and the error is typically large, we propose in the next section a different more specialised approach which, however, still allows us to train arbitrary models with the relative entropy objective.

In order to proof the error of the gradient estimation via interpolation, we first need to establish error bounds on the interpolating polynomial which can be obtained via the remainder of the Lagrange interpolation polynomial.
The gradient error for our objective can then be obtained by as a combination of this error with a bound on the $n+1$-st order derivative of the objective.
We start by bounding the error in the polynomial approximation.
\begin{lemma}
  \label{lem:remainder}
  Let $f(\theta)$ be the $n+1$ times differentiable function for which we want to approximate the gradient and let $p_n(\theta)$ be the degree $n$ Lagrange interpolation polynomial for points $\{\theta_1, \theta_2, \ldots, \theta_k, \ldots, \theta_n\}$.
  The gradient evaluated at point $\theta_k$ is then given by the interpolation polynomial
  \begin{equation}
    \frac{\partial p(\theta_k)}{\partial \theta} = \sum_{j=0}^n f(\theta_j) \mathcal{L}_{n,j}'(\theta_k),
  \end{equation}
  where $\mathcal{L}_{n,j}'$ is the derivative of the Lagrange interpolation polynomials $\mathcal{L}_{\mu,j}(\theta):= \prod_{\substack{k=0\\ k\neq j}}^{\mu} \frac{\theta - \theta_k}{\theta_j - \theta_k}$, and the error is given by
  \begin{equation}
    \left\lvert \frac{\partial f(\theta_k)}{\partial \theta} - \frac{\partial p_n(\theta_k)}{\partial \theta} \right\rvert \leq  \frac{1}{(n+1)!} \left\lvert f^{(n+1)}(\xi(\theta_k)) \prod\limits_{\substack{j=0 \\ j\neq k}}^n (\theta_j - \theta_k) \right\rvert,
  \end{equation}
  where $\xi(\theta_k)$ is a constant depending on the point $\theta_k$ at which we evaluated the gradient, and $f^{(i)}$ denotes the $i$-th derivative of $f$.
\end{lemma}
Note that $\theta$ is a point within the set of points at which we evaluate.
\begin{proof}
  Recall that the error for the degree $n$ Lagrange interpolation polynomial is given by
  \begin{equation}
    f(\theta) -  p_n(\theta) \leq  \frac{1}{(n+1)!} f^{(n+1)}(\xi_{\theta}) w(\theta),
  \end{equation}
  where $w(\theta) := \prod\limits_{j=1}^n(\theta - \theta_j)$.
  We want to estimate the gradient of this, and hence need to evaluate
  \begin{equation}
    \gradtheta{f(\theta)} -  \gradtheta{p_n(\theta)} \leq \lim\limits_{\Delta \rightarrow 0} \left( \frac{\frac{1}{(n+1)!} f^{(n+1)}(\xi_{\theta + \Delta}) w(\theta + \Delta) - \frac{1}{(n+1)!} f^{(n+1)}(\xi_{\theta}) w(\theta)}{\Delta} \right).
  \end{equation}
  Now, since we do not necessarily want to estimate the gradient at an arbitrary point $\theta$ but indeed have the freedom to choose the point, we can set $\theta$ to be one of the points at which we evaluate the function $f(\theta)$, i.e., $\theta \in \{ \theta_i \}_{i=1}^n$.
  Let this choice be given by $\theta_k$, arbitrarily chosen. Then we see that the latter term vanishes since $w(\theta_k)=0$. Therefore we have
  \begin{equation}
    \gradtheta{f(\theta_k)} -  \gradtheta{p_n(\theta_k)} \leq \lim\limits_{\Delta \rightarrow 0} \left( \frac{\frac{1}{(n+1)!} f^{(n+1)}(\xi_{\theta_k + \Delta}) w(\theta_k + \Delta)}{\Delta} \right),
  \end{equation}
  and noting that $w(\theta_k)$ contains one term $(\theta_k + \Delta - \theta_k) = \Delta$ achieves the claimed result.
\end{proof}

We will perform a number of approximation steps in order to obtain a form which can be simulated on a quantum computer more efficiently, and only then resolve to divided differences at this ``lower level".
In detail we will perform the following steps.
As described in the body of the thesis, we perform the following steps in order to obtain the gradient.
\begin{enumerate}
  \item Approximate the logarithm via a Fourier-like approximation
  \begin{equation}
    \log \sigma_v \rightarrow \log_{K,M}\sigma_v,
  \end{equation}
  which yields a Fourier-like series $\sum_m c_m \exp{(im\pi \sigma_v)}$.
  \item Evaluate the gradient of $\tr{ \frac{\partial}{\partial \theta} \rho \log_{K,M}(\sigma_v)}$,
  yielding terms of the form
  \begin{equation}
    \int_0^1 ds e^{(ism\pi \sigma_v)} \frac{\partial \sigma_v}{\partial \theta} e^{(i(1-s)m\pi \sigma_v)}.
  \end{equation}
  \item Each term in this expansion can be evaluated separately via a sampling procedure, i.e.,
  \begin{equation}
    \int_0^1 ds e^{(ism\pi \sigma_v)} \frac{\partial \sigma_v}{\partial \theta} e^{(i(1-s) m\pi \sigma_v)} \approx \mathbb{E}_s
    \left[ e^{(ism\pi \sigma_v)} \frac{\partial \sigma_v}{\partial \theta} e^{(i(1-s)m\pi \sigma_v)} \right].
  \end{equation}
  \item Apply a divided difference scheme to approximate the gradient $\frac{\partial \sigma_v}{\partial \theta}$.
  \item Use the Fourier series approach to aproximate the density operator $\sigma_v$ by the series of itself,
  i.e., $\sigma_v \approx F(\sigma_v) := \sum_{m'} c_{m'} \exp{(im \pi m' \sigma_v)}$.
  \item Evaluate these terms conveniently via sample based Hamiltonian simulation and the Hadamard test.
\end{enumerate}
In the following we will give concrete bounds on the error introduced by the approximations and details of the implementation.
The final result is then stated in Theorem~\ref{thm:complexity_general_algo}.
We first bound the error in the approximation of the logarithm and then use Lemma~37 of \cite{van2017quantum} to obtain a Fouries series approximation which is close to $\log(z)$.
  The Taylor series of $$\log(x) = \sum_{k=1}^{\infty} (-1)^{k+1}\frac{(x-1)^k}{k} = \sum_{k=1}^{K_1} (-1)^{k+1}\frac{(x-1)^k}{k} + R_{K_1+1}(x-1),$$ for $x \in (0,1)$ and where $R_{K+1}(z)= \frac{f^{K_1+1}(c)}{K!}(z-c)^{K_1}z$ is the Cauchy remainder of the Taylor series, for $-1<z<0$.
  The error can hence be bounded as $$\lvert R_{K_1+1}(z) \rvert = \left\lvert (-1)^{K_1}\frac{z^{K_1+1}(1-\alpha)^{K_1}}{(1+\alpha z)^{K_1+1}}\right\rvert,$$
  where we evaluated the derivatives of the logarithm and $0\leq \alpha \leq 1$ is a parameter. Using that $1+\alpha z \geq 1+z$ (since $z \leq 0$) and hence $0 \leq \frac{1-\alpha}{1+\alpha z} \leq 1$, we obtain the error bound
  \begin{equation}
    \lvert R_{K_1+1}(z) \rvert \leq  \frac{\left\lvert z \right\rvert^{K_1+1}}{1+z}
  \end{equation}
  Reversing to the variable $x$ the error bound for the Taylor series, and assuming that $0<\delta_l <z$ and  $0<|1-z|\leq \delta_u <1$, which is justified if we are dealing with sufficiently mixed states, then the approximation error is given by
  \begin{equation}
    \lvert R_{K_1+1}(z) \rvert \leq  \frac{(\delta_l) ^{K_1+1}}{\delta_u} \overset{!}{\leq} \epsilon_1.
  \end{equation}
  Hence in order to achieve the desired error $\epsilon_1$ we need $$K_1 \geq \frac{\log \left( (\epsilon_1 \delta_u)^{-1} \right)}{\log \left((\delta_l)^{-1}\right)}.$$
  We hence can chose $K_1$ such that the error in the approximation of the Taylor series is $\leq \epsilon_1/4$.
  This implies we can make use of Lemma~37 of \cite{van2017quantum}, and therefore obtain a Fourier series approximation for the logarithm.
  We will restate this Lemma here for completeness:

  \begin{lemma}[Lemma 37, \cite{van2017quantum}]
    \label{lem:vanApeldoorn}
    Let $f:\mathbb{R} \rightarrow \mathbb{C}$ and $\delta,\epsilon \in (0,1)$, and $T(f):= \sum_{k=0}^K a_k x^k$ be a polynomial such that $\left\lvert f(x) - T(f) \right\rvert \leq \epsilon/4$ for all $x \in [-1+\delta, 1-\delta]$.
    Then $\exists c \in \mathbb{C}^{2M+1}:$
    \begin{equation}
      \left\lvert f(x) - \sum_{m=-M}^{M} c_m e^{\frac{i \pi m}{2}x} \right\rvert \leq \epsilon
    \end{equation}
    for all $x \in [-1+\delta, 1- \delta]$, where $M= \max \left(2 \left\lceil \ln \left( \frac{4 \norm{a}_1}{\epsilon} \right) \frac{1}{\delta} \right\rceil, 0\right)$ and $\norm{c}_1 \leq \norm{a}_1$.
    Moreover, $c$ can be efficiently calculated on a classical computer in time $\mathrm{poly}(K,M,\log(1/\epsilon))$.
  \end{lemma}
  In order to apply this lemma to our case, we restrict the approximation rate to the range $(\delta_l, \delta_u)$, where $0 < \delta_l \leq \delta_u <1$.
  Therefore we obtain over this range a approximation of the following form.
  \begin{corollary}
    Let $f:\mathbb{R} \rightarrow \mathbb{C}$ be defined as $f(x)=\log(x)$, $\delta,\epsilon_1 \in (0,1)$, and $\log_{K}(1-x):= \sum_{k=1}^{K_1} \frac{(-1)^{k-1}}{k} x^k$ such that $a_k:=\frac{(-1)^{k-1}}{k}$ and $\norm{a}_1 = \sum_{k=1}^{K_1} \frac{1}{k}$ with $K_1 \geq \frac{\log \left(4(\epsilon_1 \delta_1^u)^{-1} \right)}{\log \left((\delta_l)^{-1}\right)}$ such that $\left\lvert \log(x) -\log_{K}(x)\right\rvert \leq \epsilon_1/4$ for all $x \in [\delta_l, \delta_u]$.
    Then $\exists c \in \mathbb{C}^{2M+1}:$
    \begin{equation}
      \left\lvert f(x) - \sum_{m=-M_1}^{M_1} c_m e^{\frac{i \pi m}{2}x} \right\rvert \leq \epsilon_1
    \end{equation}
    for all $x \in [\delta_l, \delta_u]$, where $M_1= \max \left(2 \left\lceil \ln \left( \frac{4 \norm{a}_1}{\epsilon_1} \right) \frac{1}{1-\delta_u} \right\rceil, 0\right)$ and $\norm{c}_1 \leq \norm{a}_1$.
    Moreover, $c$ can be efficiently calculated on a classical computer in time $\mathrm{poly}(K_1,M_1,\log(1/\epsilon_1))$.
  \end{corollary}
  \begin{proof}
    The proof follows straight forward by combining Lemma~\ref{lem:vanApeldoorn} with the approximation of the logarithm and the range over which we want to approximate the function.
  \end{proof}
  In the following we denote with $\log_{K,M}(x):=\sum_{m=-M_1}^{M_1} c_m e^{\frac{i \pi m}{2}x}$, where we keep the $K$-subscript to denote that classical computation of this approximation is $\mathrm{poly}(K)$-dependent.
  We can now express the gradient of the objective via this approximation as
  \begin{equation}
    \label{eq:exact_gradient_log_approx}
    \tr{\frac{\partial}{\partial \theta} \rho \log_{K,M} \sigma_v} \approx \sum_{m=-M_1}^{M_1} \frac{i c_m m \pi}{2} \int_0^1 ds\ \tr{ \rho e^{\frac{i s \pi m}{2}\sigma_v} \frac{\partial \sigma_v}{\partial \theta} e^{\frac{i (1-s) \pi m}{2}\sigma_v}}.
  \end{equation}
  where we can evaluate each term in the sum individually and then classically post process the results, i.e., sum these up. In particular the latter can be evaluated as the expectation value over $s$, i.e.,
  \begin{equation}
    \int_0^1 ds\ \tr{ \rho e^{\frac{i s \pi m}{2}\sigma_v} \frac{\partial \sigma_v}{\partial \theta} e^{\frac{i (1-s) \pi m}{2}\sigma_v}} = \mathbb{E}_{s \in[0,1]} \left[\tr{ \rho e^{\frac{i s \pi m}{2}\sigma_v} \frac{\partial \sigma_v}{\partial \theta} e^{\frac{i (1-s) \pi m}{2}\sigma_v}} \right],
  \end{equation}
 which we can evaluate separately on a quantum device.
 In the following we hence need to device a method to evaluate this expectation value.\\
 First, we will expand the gradient using a divided difference formula such that $\frac{\partial \sigma_v}{\partial \theta}$ is approximated by the Lagrange interpolation polynomial of degree $\mu-1$, i.e., $$\frac{\partial \sigma_v}{\partial \theta}(\theta) \approx \sum_{j=0}^{\mu}\sigma_v(\theta_j) \mathcal{L'}_{\mu,j}(\theta),$$ where $$\mathcal{L}_{\mu,j}(\theta):= \prod_{\substack{k=0\\ k\neq j}}^{\mu} \frac{\theta - \theta_k}{\theta_j - \theta_k}.$$
 Note that the order $\mu$ is free to chose, and will guarantee a different error in the solution of the gradient estimate as described prior in Lemma~\ref{lem:remainder}. Using this in the gradient estimation, we obtain a polynomial of the form (evaluated at $\theta_j$, i.e., the chosen points)
\begin{equation}
 \sum_{m=-M_1}^{M_1} \frac{i c_m m \pi}{2} \sum_{j=0}^{\mu} \mathcal{L'}_{\mu,j}(\theta_j ) \mathbb{E}_{s \in[0,1]} \left[\tr{ \rho e^{\frac{i s \pi m}{2}\sigma_v} \sigma_v(\theta_j) e^{\frac{i (1-s) \pi m}{2}\sigma_v}} \right],
\end{equation}
where each term again can be evaluated separately, and efficiently combined via classical post processing. Note that the error in the Lagrange interpolation polynomial decreases exponentially fast, and therefore the number of terms we use is sufficiently small to do so.
Next, we need to deploy a method to evaluate the above expressions.
In order to do so, we implement $\sigma_v$ as a Fourier series of itself, i.e., $\sigma_v = \arcsin(\sin(\sigma_v \pi/2)/(\pi/2))$, which we will then approximate similar to the approach taken in Lemma~\ref{lem:vanApeldoorn}. With this we obtain the following result.

\begin{lemma}\label{lem:bounds_z_approx}
  Let $\delta,\epsilon_2 \in (0,1)$, and $\tilde x:= \sum_{m'=-M_2}^{M_2} \tilde{c}_{m'} e^{i \pi m'x/2}$ with  $K_2 \geq \frac{\log(4/\epsilon_2)}{\log(\delta_u^{-1})}$ and $M_2 \geq \left\lceil \log \left(\frac{4}{\epsilon_2} \right) \sqrt{(2\log{\delta_u^{-1}})^{-1}} \right\rceil$ and $x \in [\delta_l, \delta_u]$.
  Then $\exists \tilde c \in \mathbb{C}^{2M+1}:$
  \begin{equation}
    \left\lvert x - \tilde x \right\rvert \leq \epsilon_2
  \end{equation}
  for all $x \in [\delta_l, \delta_u]$, and $\norm{c}_1 \leq 1$.
  Moreover, $\tilde c$ can be efficiently calculated on a classical computer in time $\mathrm{poly}(K_2,M_2,\log(1/\epsilon_2))$.
\end{lemma}

\begin{proof}
Invoking the technique used in~\cite{van2017quantum}, we expand $$\arcsin(z) = \sum_{k'=0}^{K_2} 2^{-2k'} {2k' \choose k'} \frac{z^{2k'+1}}{2k'+1} + R_{K_2+1}(z),$$ wher $R_{K_2+1}$ is the remainder as before.
For $0<z\leq \delta_u \leq 1/2$, remainder can be bound by $ \left\lvert R_{K_2+1} \right\rvert \leq \frac{\left\lvert \delta_u \right\rvert^{K_2+1}}{1/2} \overset{!}{\leq} \epsilon_2/2$, which gives the bound $$K_2 \geq \frac{\log(4/\epsilon_2)}{\log(\delta_u^{-1})}.$$
We then approximate
\begin{equation}
\sin^l(x)=\left(\frac{i}{2}\right)^l \sum_{m'=0}^l (-1)^{m'} {l \choose m'} e^{ix(2m'-l)}
\end{equation} by
\begin{equation}
\sin^l(x)\approx \left(\frac{i}{2}\right)^l \sum_{m'=\lceil l/2\rceil - M_2}^{\lfloor l/2\rfloor + M_2} (-1)^{m'} {l \choose m'} e^{ix(2m'-l)},
\end{equation}
which induces an error of $\epsilon_2/2$ for the choice $$M_2 \geq \left\lceil \log \left(\frac{4}{\epsilon_2} \right) \sqrt{(2\log{\delta_u^{-1}})^{-1}} \right\rceil.$$
This can be seen by using Chernoff's inequality for sums of binomial coefficients, i.e., $$\sum_{m'=\lceil l/2+M_2\rceil}^l 2^{-l} {l \choose m'} \leq e^{-\frac{2M_2^2}{l}},$$ and chosing $M$ appropriately.
Finally, defining $f(z):= \arcsin(\sin(z \pi/2)/(\pi/2))$, as well as $\tilde{f}_1 := \sum_{k'=0}^{K_2} b_{k'} \sin^{2k'+1}(z \pi/2)$ and
\begin{equation}
\tilde{f}_2(z) := \sum_{k'=0}^{K_2} b_{k'}\left(\frac{i}{2}\right)^l \sum_{m'=\lceil l/2\rceil - M_2}^{\lfloor l/2\rfloor + M_2} (-1)^{m'} {l \choose m'} e^{ix(2m'-l)},
\end{equation}
 and observing that
 $$\norm{f -\tilde{f}_2}_{\infty} \leq \norm{f - \tilde{f}_1}_{\infty} + \norm{\tilde{f}_1 - \tilde{f}_2}_{\infty},$$ yields the final error of $\epsilon_2$ for the approximation
$z \approx \tilde z = \sum_{m'} \tilde{c}_{m'} e^{i \pi m' z/2}$.
\end{proof}
Note that this immediately leads to an $\epsilon_2$ error in the spectral norm for the approximation
\begin{equation}
\left\lVert \sigma_v-\sum_{m'=-M_2}^{M_2} \tilde{c}_{m'} e^{i \pi m' \sigma_v/2}\right\rVert_2 \leq \epsilon_2,
\end{equation}
where $\sigma_v$ is the reduced density matrix.

Since our final goal is to estimate $\tr{\partial_{\theta} \rho \log \sigma_v}$, with a variety of $\sigma_v(\theta_j)$ using the divided difference approach, we also need to bound the error in this estimate which we introduce with the above approximations.
Bounding the derivative with respect to the remainder can be done by using the truncated series expansion and bounding the gradient of the remainder. This yields the following result.

\begin{lemma}
    \label{lem:proof_error_bound_grads}
    For the of the parameters $M_1,M_2, K_1,L,\mu,\Delta,s$ given in eq.~(\ref{eq:bound_M1}-\ref{eq:bounds_epss}), and $\rho, \sigma_v$ being two density matrices, we can estimate the gradient of the relative entropy such that
    \begin{equation}
      \left\lvert \partial_{\theta}\tr{\rho \log \sigma_v} - \partial_{\theta} \tr{\rho \log_{K_1,M_1} \tilde{\sigma}_v} \right\rvert \leq \epsilon,
    \end{equation}
    where the function $\partial_{\theta} \tr{\rho \log_{K_1,M_1} \tilde{\sigma}_v}$ evaluated at $\theta$ is defined as
    \begin{equation}
      \label{eq:full_approximation}
    \mathrm{Re}\left[ \sum_{m=-M_1}^{M_1} \sum_{m'=-M_2}^{M_2} \frac{i c_m \tilde{c}_{m'} m \pi}{2} \sum_{j=0}^{\mu} \mathcal{L}'_{\mu,j}(\theta) \mathbb{E}_{s \in[0,1]} \left[\tr{ \rho e^{\frac{i s \pi m}{2}\sigma_v} e^{\frac{i \pi m'}{2} \sigma_v(\theta_j)} e^{\frac{i (1-s) \pi m}{2}\sigma_v}} \right] \right]
    \end{equation}
    The gradient can hence be approximated to error $\epsilon$ with $O(\text{poly}(M_1,M_2,K_1,L,s,\Delta,\mu))$ computation on a classical computer and using only the Hadamard test, Gibbs state preparation and LCU on a quantum device.
\end{lemma}
Notably the expression in \eqref{eq:full_approximation} can now be evaluated with a quantum-classical hybrid device by evaluating each term in the trace separately via a Hadamard test and, since the number of terms is only polynomial, and then evaluating the whole sum efficiently on a classical device.\\

\begin{proof}
  For the proof we perform the following steps.
  Let $\sigma_i(\rho)$ be the singular values of $\rho$, which are equivalently the eigenvalues since $\rho$ is Hermitian.
  Then observe that the gradient can be separated in different terms, i.e., let $\log_{K_1,M_1}^s\sigma_v$ be the approximation as given in \eqref{eq:full_approximation} for a finite sample of the expectation values $\mathbb{E}_{s}$, then we have
  \begin{align}
	\label{eq:bounds_approximation_error}
    &\left\lvert \partial_{\theta}\tr{\rho \log \sigma_v} - \partial_{\theta} \tr{\rho \log_{K_1,M_1}^s \tilde{\sigma}_v} \right\rvert \leq \nonumber \\
      &\leq \sum_i \sigma_i(\rho) \cdot \left\lVert \partial_{\theta} [\log \sigma_v - \log_{K_1,M_1}^s \tilde{\sigma}_v] \right\rVert \nonumber \\
      &\leq \sum_i \sigma_i(\rho) \cdot \left( \left\lVert \partial_{\theta} [\log \sigma_v - \log_{K_1,M_1} \sigma_v] \right\rVert \right. \nonumber \\
      &+ \left. \left\lVert \partial_{\theta} [\log_{K_1,M_1} \sigma_v - \log_{K_1,M_1} \tilde{\sigma}_v] \right\rVert + \left\lVert \partial_{\theta} [\log_{K_1,M_1} \tilde{\sigma_v} - \log^s_{K_1,M_1} \tilde{\sigma}_v] \right\rVert \right)
  \end{align}
  where the second step follows from the Von-Neumann trace inequality and the terms are (1) the error in approximating the logarithm, (2) the error introduced by the divided difference and the approximation of $\sigma_v$ as a Fourier-like series, and (3) is the finite sampling approximation error.
  We can now bound the different term separately, and start with the first part which is in general harder to estimate.
  We partition the bound in three terms, corresponding to the three different approximations taken above.
  \begin{align*}
      &\left\lVert \partial_{\theta} [\log \sigma_v - \log_{K_1,M_1} \sigma_v] \right\rVert \leq \\
      &\leq \left\lVert \partial_{\theta} \sum_{k=K_1+1}^{\infty} \frac{(-1)^k}{k} \sigma_v^k \right\rVert + \left\lVert \partial_{\theta} \sum_{k=1}^{K_1} \frac{(-1)^k}{k} \sum_{l=L}^{\infty} b_l^{(k)} \sin^l(\sigma_v \pi/2) \right\rVert \\
      & +\left\lVert \partial_{\theta} \sum_{k=1}^{K_1} \frac{(-1)^k}{k} \sum_{l=L}^{\infty} b_l^{(k)}  \left(\frac{i}{2} \right)^l \sum_{m \in [0, \lceil l/2\rceil -M_1] \cup [\lfloor l/2 \rfloor +M_1, l]} (-1)^m e^{i (2m-l)\sigma_v \pi/2} \right\rVert
  \end{align*}
  The first term can be bound in the following way:
  \begin{align}
    \leq \sum_{k=K_1+1}^{\infty} \lVert \sigma_v \rVert^{k-1} = \frac{\lVert \sigma_v\rVert^{K_1}}{1-\lVert \sigma_v \rVert},
  \end{align}
  and, assuming $\norm{\sigma_v}<1$, we hence can set
\begin{equation}
	K_1\geq \log((1-\norm{\sigma_v})\epsilon/9)/\log(\norm{\sigma_v})
\end{equation}
 appropriately in order to achieve an $\epsilon/9$ error. The second term can be bound by assuming that $\norm{\sigma_v \pi}<1$, and chosing $$L\geq\log \frac{\left(\frac{\epsilon}{9\pi K \norm{\frac{\partial \sigma_v}{\partial\theta}}}\right)}{\log(\norm{\sigma_v} \pi)},$$ which we derive by observing that
  \begin{align}
    &\leq \sum_{k=1}^K \frac{1}{k} \sum_{l=L}^{\infty} b_l^{(k)} l \left\lVert \sin^{l-1} (\sigma_v \pi/2)\right\rVert \cdot \left\lVert \frac{\pi}{2} \frac{\partial \sigma_v}{\partial \theta} \right\rVert \\
    & < \sum_{k=1}^K \frac{1}{k} \sum_{l=L+1}^{\infty} b_l^{(k)} \pi \left\lVert \sigma_v \pi\right\rVert^{l-1} \cdot \left\lVert \frac{\partial \sigma_v}{\partial \theta} \right\rVert \\
    &\leq \sum_{k=1}^K \frac{1}{k}  \pi \left\lVert \sigma_v \pi\right\rVert^{L} \cdot \left\lVert \frac{\partial \sigma_v}{\partial \theta} \right\rVert,
  \end{align}
where we used in the second step that $l < 2^l$.
  Finally, the last term can be bound similarly, which yields
  \begin{align}
   &\leq \sum_{k=1}^K \frac{1}{k} \sum_{l=1}^{L} b_l^{(k)} e^{-2(M_1)^2/l} \cdot l \cdot \frac{\pi}{2} \norm{\frac{\partial \sigma_v}{\partial \theta}}\\
   &\leq \sum_{k=1}^K \frac{L}{k} \sum_{l=1}^{L} b_l^{(k)} e^{-2(M_1)^2/L} \frac{\pi}{2} \norm{\frac{\partial \sigma_v}{\partial \theta}} \\
   &\leq \sum_{k=1}^K \frac{L}{k} e^{-2(M_1)^2/L} \frac{\pi}{2} \norm{\frac{\partial \sigma_v}{\partial \theta}} \leq \frac{KL\pi}{2}e^{-2(M_1)^2/L} \norm{\frac{\partial \sigma_v}{\partial \theta}},
  \end{align}
  and we can hence chose $$M_1\geq \sqrt{L \log\left(\frac{9 \norm{\frac{\partial \sigma_v}{\partial \theta}} K_1 L \pi}{2\epsilon}\right)} $$ in order to decrease the error to $\epsilon/3$ for the first term in \eqref{eq:bounds_approximation_error}.\\
  For the second term, first note that with the notation we chose, $\left\lVert \partial_{\theta} [\log_{K_1,M_1} \sigma_v - \log_{K_1,M_1} \tilde{\sigma}_v] \right\rVert$ is the difference between the log-approximation where the gradient of $\sigma_v$ is still exact, i.e., \eqref{eq:exact_gradient_log_approx}, and the version where we approximate the gradient via divided differences and the linear combination of unitaries, given in \eqref{eq:full_approximation}.
  Recall that the first level approximation was given by
  \begin{equation*}
    \sum_{m=-M_1}^{M_1} \frac{i c_m m \pi}{2} \int_0^1 ds\ \tr{ \rho e^{\frac{i s \pi m}{2}\sigma_v} \frac{\partial \sigma_v}{\partial \theta} e^{\frac{i (1-s) \pi m}{2}\sigma_v}},
  \end{equation*}
  where we went from the expectation value formulation back to the integral formulation to avoid consideration of potential errors due to sampling.

  Bounding the difference hence yields one term from the divided difference approximation of the gradient and an error from the Fourier series, which we can both bound separately. Denoting $\partial \tilde p(\theta_k)/\partial \theta$ as the divided difference and the LCU approximation of the Fourier series\footnote{which effectively means that we approximate the coefficients of the interpolation polynomial}, and with $\partial p(\theta_k)/\partial \theta$ the divided difference without approximation via the Fourier series,  we hence have
  \begin{align}
    &\left\lVert \partial_{\theta} [\log_{K_1,M_1} \sigma_v - \log_{K_1,M_1} \tilde{\sigma}_v] \right\rVert \leq \\
    &\leq \left\lvert \sum_{m=-M_1}^{M_1} \frac{i c_m m \pi}{2} \int_0^1 ds\ \tr{ \rho e^{\frac{i s \pi m}{2}\sigma_v} \left( \frac{\partial \sigma_v}{\partial \theta} - \frac{\partial \tilde p(\theta_k)}{\partial \theta} \right) e^{\frac{i (1-s) \pi m}{2}\sigma_v}} \right\rvert \\
    &\leq \frac{M_1 \pi \norm{a}_1}{2}  \int_0^1 ds\ \sum_i \sigma_i(\rho) \left\lVert \frac{\partial \sigma_v}{\partial \theta} -  \frac{\partial \tilde p(\theta_k)}{\partial \theta}\right\rVert \\
&\leq \frac{M_1 \pi \norm{a}_1}{2}  \int_0^1 ds\ \sum_i \sigma_i(\rho) \left( \left\lVert \frac{\partial \sigma_v}{\partial \theta} - \frac{\partial p(\theta_k)}{\partial \theta}\right\rVert  +  \left\lVert \frac{\partial p(\theta_k)}{\partial \theta} -  \frac{\partial \tilde p(\theta_k)}{\partial \theta}\right\rVert \right) \\
&\leq \frac{M_1 \norm{a}_1 \pi}{2} \int_0^1 ds\ \sum_i \sigma_i(\rho) \left(\norm{\frac{\partial^{\mu+1} \sigma_v}{\partial \theta^{\mu+1}}} \left( \frac{\Delta}{\mu-1}\right)^{\mu}\frac{ \max_k (\mu-k)!}{(\mu+1)!} + \sum_{j=0}^{\mu} \lvert \mathcal{L}'_{\mu,j}(\theta_j) \rvert \norm{\sigma_v - \tilde{\sigma}_v} \right) \\
    &\leq \frac{M_1 \norm{a}_1 \pi}{2} \left( \norm{\frac{\partial^{\mu+1} \sigma_v}{\partial \theta^{\mu+1}}} \left( \frac{\Delta}{\mu-1}\right)^{\mu} \frac{\mu!}{(\mu+1)!} +\mu \norm{\mathcal{L}'_{\mu,j}(\theta_j)}_{\infty} \epsilon_2 \right) , \label{eq:bounds_ddfs}
  \end{align}
  where $\norm{a}_1 = \sum_{k=1}^{K_1} 1/k$, and we used in the last step the results of Lemma~\ref{lem:bounds_z_approx}. Under appropriate assumptions on the grid-spacing for the divided difference scheme $\Delta$ and the number of evaluated points $\mu$ as well as a bound on the $\mu+1$-st derivative of $\sigma_v$ w.r.t. $\theta$, we can hence also bound this error.
In order to do so, we need to analyze the $\mu+1$-st derivative of $\sigma_v = \trh{e^{-H}}/Z$ with $Z=\tr{e^{-H}}$.
For this we have
  \begin{align}
\label{eq:bound_mu_derivative}
  \norm{\frac{\partial^{\mu+1} \sigma_v}{\partial \theta^{\mu+1}}} &\leq \sum_{p=1}^{\mu+1} {\mu+1 \choose p}  \norm{\frac{\partial^{p} \trh{e^{-H}}}{\partial \theta^{p}}}  \norm{\frac{\partial^{\mu+1-p} Z^{-1}}{\partial \theta^{\mu+1-p}}} \nonumber\\
&\leq  2^{\mu+1} \max_p \norm{\frac{\partial^{p} \trh{e^{-H}}}{\partial \theta^{p}}}  \norm{\frac{\partial^{\mu+1-p} Z^{-1}}{\partial \theta^{\mu+1-p}}}
  \end{align}
We have that
\begin{align}
\norm{\frac{\partial^{p} \trh{e^{-H}}}{\partial \theta^{p}}} &\leq  \mathrm{dim}(H_h) \norm{\frac{\partial^q e^{-H}}{\partial \theta^q}}
\end{align}
where $\mathrm{dim}(H_h) = 2^{n_h}$. In order to bound this, we take advantage of the infinitesimal expansion of the exponent, i.e.,
\begin{align}
	\norm{\frac{\partial^q e^{-H}}{\partial \theta^q}} &= \norm{\frac{\partial^q }{\partial \theta^q} \lim_{r\rightarrow \infty} \prod_{j=1}^{r} e^{-H/r}} \nonumber \\
	&= \norm{\lim_{r \rightarrow \infty} \left(  \frac{\partial^q e^{-H/r}}{\partial \theta^q} \prod_{j=2}^{r} e^{-H/r} +  \frac{\partial^{q-1} e^{-H/r}}{\partial \theta^{q-1}} \frac{\partial e^{-H/r}}{\partial \theta} \prod_{j=3}^{r} e^{-H/r} + \ldots \right)} \nonumber \\
	&\leq \lim_{r \rightarrow \infty} \left(\norm{\frac{\partial H/r}{\partial \theta}}^q \cdot r^q + O\left(\frac{1}{r}\right) \right) \norm{e^{-H}} = \norm{\frac{\partial H}{\partial \theta}}^q\norm{e^{-H}}  ,
\end{align}
where the last step follows from the fact that we have $r^q$ terms and that we used that the error introduced by the commutations above will be of $O(1/r)$.
Observing that $\partial_{\theta_i}H = \partial_{\theta_i}\sum_j \theta_j H_j = H_i$ and assuming that $\lambda_{max}$ is the largest singular eigenvalue of $H$, we can hence bound this by $\lambda_{max}^q\norm{e^{-H}}$.

\begin{align}
\norm{\frac{\partial^{p} \trh{e^{-H}}}{\partial \theta^{p}}} &\leq  \mathrm{dim}(H_h) \norm{\frac{\partial^q e^{-H}}{\partial \theta^q}} \nonumber \\
	&\leq  \lambda_{max}^{p}  \mathrm{dim}(H_h) \norm{ \trh{e^{-H}}},
\end{align}

\begin{align}
\norm{\frac{\partial^{\mu+1-p} Z^{-1}}{\partial \theta^{\mu+1-p}}} &\leq  \frac{(\mu+1-p)! \lvert \lambda_{max} \rvert^{\mu+1-p}}{Z^{\mu+2-p}} \tr{e^{-H}}\nonumber\\
	&\leq \left(\frac{\mu+1-p}{eZ}\right)^{\mu+1-p} \frac{e}{Z} \lvert \lambda_{max} \rvert^{\mu+1-p}\tr{e^{-H}} \nonumber\\
	&= \left(\frac{\mu+1-p}{eZ}\right)^{\mu+1-p} e \lvert \lambda_{max} \rvert^{\mu+1-p}
\end{align}
We can therefore find a bound for \eqref{eq:bound_mu_derivative} as
\begin{align}
      \norm{\frac{\partial^{\mu+1} \sigma_v}{\partial \theta^{\mu+1}}} &\leq  e 2^{\mu+1+n_h} \lambda_{max}^{\mu+1} \norm{\trh{e^{-H}}}  \max_p \left(\frac{\mu+1-p}{eZ}\right)^{\mu+1-p}.
\end{align}
Plugging this result into the bound from above yields
  \begin{align}
    &\left\lVert \partial_{\theta} [\log_{K_1,M_1} \sigma_v - \log_{K_1,M_1} \tilde{\sigma}_v] \right\rVert  \nonumber \\
    &\leq \frac{M_1 \norm{a}_1 \pi}{2} \left(e2^{\mu+1+n_h} \lambda_{max}^{\mu+1} \norm{\trh{e^{-H}}}  \max_p \left(\frac{\mu+1-p}{eZ}\right)^{\mu+1-p} \left( \frac{\Delta}{\mu-1}\right)^{\mu} \frac{1}{\mu+1}  \right) \nonumber \\
 	&+  \frac{M_1 \norm{a}_1 \pi}{2} \left( \mu \norm{\mathcal{L}'_{\mu,j}(\theta_j)}_{\infty}  \epsilon_2 \right) ,
  \end{align}
Note that under the reasonable assumption that $2 \leq \mu \ll Z$, the maximum is achieved for $p=\mu+1$, and we hence obtain the upper bound
 \begin{align}
	\label{eq:error_ddfs}
  &\frac{M_1 \norm{a}_1 \pi}{2} \left( 2^{n_h} e ( 2 \lvert \lambda_{max} \rvert)^{\mu+1}  \norm{\trh{e^{-H}}} \left( \frac{\Delta}{\mu-1}\right)^{\mu} \frac{1}{\mu+1}  + \mu  \norm{\mathcal{L}'_{\mu,j}(\theta_j)}_{\infty}   \epsilon_2 \right) \nonumber \\
 &\leq  \frac{M_1 \norm{a}_1 \pi}{2} \left( 2^{n_h}e ( 2 \lvert \lambda_{max} \rvert)^{\mu+1}  \norm{\trh{e^{-H}}} \left( \frac{\Delta}{\mu-1}\right)^{\mu} + \mu  \norm{\mathcal{L}'_{\mu,j}(\theta_j)}_{\infty}  \epsilon_2 \right) ,
  \end{align}
and we can hence obtain a bound on $\mu$, the grid point number, in order to achieve an error of $\epsilon/6 > 0$ for the former term, which is given by
\begin{equation}
	\mu \geq (\lvert \lambda_{max} \rvert \Delta)  \exp \left( W \left(
			 \frac{\log \left( 2^{n_h}\frac{
			6  M_1 \norm{a}_1 e^2 \lvert \lambda_{max}\rvert \pi  \norm{\trh{e^{-H}}}}{\epsilon}
			\right)}{2 \lambda_{max} \Delta}
	\right)\right),
\end{equation}
where $W$ is the Lambert function, also known as product-log function, which generally grows slower than the logarithm in the asymptotic limit.
Note that $\mu$ can hence be lower bounded by

\begin{equation}
\label{eq:bound_mu}
	\mu \geq n_h + \log \left( \frac{6 M_1 \norm{a}_1 e^2 \lvert \lambda_{max}\rvert \pi  \norm{\trh{e^{-H}}}}{\epsilon} \right):= n_h + \log \left(\frac{M_1\Lambda}{\epsilon} \right).
\end{equation}
For convenience, let us choose $\epsilon$ such that $n_h + \log(M_1 \Lambda/\epsilon)$ is an integer.  We do this simply to avoid having to keep track of ceiling or floor functions in the following discussion where we will choose $\mu = n_h + \log(M_1\Lambda /\epsilon)$.

For the second part, we will bound the derivative of the Lagrangian interpolation polynomial.
First, note that $\mathcal{L}'_{\mu,j}(\theta)  = \sum_{l=0; l \neq j}^{\mu} \left( \prod_{k=0; k \neq j,l} \frac{\theta - \theta_k}{\theta_j - \theta_k} \right)\frac{1}{\theta_j-\theta_l}$ for a chosen discretization of the space such that $\theta_k -\theta_j= (k-j)\Delta/\mu$ can be bound by using a central difference formula, such that we use an uneven number of points (i.e. we take $\mu = 2\kappa +1$ for positive integer $\kappa$) and chose the point $m$ at which we evaluate the gradient as the central point of the mesh.
Note that in this case he have that for $\mu \ge 5$ and $\theta_m$ being the parameters at the midpoint of the stencil
\begin{align}
  \norm{\mathcal{L}_{\mu,j}'}_{\infty}
  &\leq \sum_{\substack{l\neq j}} \prod_{\substack{k \neq j,l}} \frac{|\theta_m - \theta_k|}{|\theta_j - \theta_k|} \frac{1}{\lvert \theta_l - \theta_j \rvert}
  \leq \frac{(\kappa!)^2}{(\kappa!)^2} \frac{\mu}{\Delta} \sum_{l \ne j} \frac{1}{|l -j|}\nonumber\\
  &\leq  \frac{2\mu}{\Delta} \sum_{l=1}^\kappa \frac{1}{l}\nonumber\\
  &\leq \frac{2\mu}{\Delta} \left( 1+ \int_{1}^{\kappa-1} \frac{1}{\ell} \mathrm{d}\ell\right) = \frac{2\mu}{\Delta}\left(1+\log((\mu-3)/2)\right)\le \frac{5\mu}{\Delta} \log(\mu/2),
\end{align}
where the last inequality follows from the fact that $\mu \ge 5$ and $1+\ln(5/2) < (5/2)\ln(5/2)$.
Now, plugging in the $\mu$ from \eqref{eq:bound_mu}, we find that this error is bound by
\begin{equation}
	\norm{\mathcal{L}_{\mu,j}'}_{\infty} \leq \frac{5 n_h + 5\log \left(\frac{M_1\Lambda}{\epsilon} \right)}{\Delta} \log(n_h/2 + \log \left(\frac{M_1\Lambda}{\epsilon} \right)/2) = \tilde{O}\left(\frac{n_h+\log\left(\frac{M_1\Lambda}{\epsilon} \right)}{\Delta}\right),
\end{equation}
If we want an upper bound of $\epsilon/6$ for the second term of the error in \eqref{eq:error_ddfs}, we hence require

\begin{align}
\epsilon_2 &\leq \frac{\epsilon}{15  M_1 \norm{a}_1 \pi \mu   \norm{\mathcal{L}'_{\mu,j}(\theta_j)}_{\infty} }\nonumber\\
	&\leq \frac{\epsilon \Delta}{15 M_1 \|a\|_1 \pi \left(n_h+\log(M_1 \Lambda/\epsilon)\right)^2 \log((n_h/2) + \log(M_1\Lambda/\epsilon)/2)}\nonumber \\
 	&\leq \frac{\epsilon \Delta}{15 M_1 \|a\|_1 \pi \mu^2 \log((\mu-1)/2)}
\end{align}
We hence obtain that the approximation error due to the divided differences and Fourier series approximation of $\sigma_v$ is bounded by $\epsilon/3$ for the above choice of $\epsilon_2$ and $\mu$. This bounds the second term in  \eqref{eq:bounds_ddfs} by $\epsilon/3$.\\

  Finally, we need to take into account the error $ \left\lVert \partial_{\theta} [\log_{K_1,M_1} \tilde{\sigma_v} - \log^s_{K_1,M_1} \tilde{\sigma}_v] \right\rVert$ which we introduce through the sampling process, i.e., through the finite sample estimate of $\mathbb{E}_s[\cdot]$ here indicated with the superscript $s$ over the logarithm.
Note that this error can be bound straight forward by \eqref{eq:full_approximation}.
We only need to bound the error introduced via the finite amount of samples we take, which is a well-known procedure.
The concrete bounds for the sample error when estimating the expectation value are stated in the following lemma.
\begin{lemma}
	Let $\sigma_{m}$ be the sample standard deviation of the random variable
	\begin{equation}
	\label{eq:random_variable_for_trace}
	\tilde{\mathbb{E}}_{s \in[0,1]} \left[\tr{ \rho e^{\frac{i s \pi m}{2}\sigma_v} e^{\frac{i \pi m'}{2} \sigma_v(\theta_j)} e^{\frac{i (1-s) \pi m}{2}\sigma_v} } \right],
	\end{equation}
	 such that the sample standard deviation is given by $\sigma_k = \frac{\sigma_m}{\sqrt{k}}$.
	Then with probability at least $1-\delta_s$, we can obtain an estimate which is within $\epsilon_s \sigma_m$ of the mean by taking $k=\frac{4}{\epsilon_s^2}$ samples for each sample estimate and taking the median of $O(\log(1/\delta_s))$ such samples.
\end{lemma}
\begin{proof}
From Chebyshev's inequality taking $k=\frac{4}{\epsilon_s^2}$ samples implies that with probability of at least $p=3/4$ each of the mean estimates is within $2\sigma_k = \epsilon_s \sigma_m$ from the true mean.
Therefore, using standard techniques, we take the median of $O( \log(1/\delta_s))$ such estimates which gives us with probability $1-\delta_s$ an estimate of the mean with error at most $\epsilon_s \sigma_m$, which implies that we need to repeat the procedure $O \left(\frac{1}{\epsilon_s^2}\log \left(\frac{1}{\delta_s}\right) \right)$ times.
\end{proof}
We can then bound the error of the sampling step in the final estimate, denoting with $\epsilon_s$ the sample error, as
\begin{align}
	\sum_{m=-M_1}^{M_1} \sum_{m'=-M_2}^{M_2}  \left\lvert \frac{i c_m c_{m'} m \pi}{2}  \right\rvert \sum_{j=0}^{\mu} \lvert \mathcal{L}'_{\mu,j}(\theta) \rvert \epsilon_s \sigma_m \nonumber\\
	\leq \frac{5 \norm{a}_1 M_1  \epsilon_s \sigma_m \pi \mu^2\log\left(\frac{\mu}{2}\right) }{\Delta} \leq \frac{\epsilon}{3},
\end{align}

We hence find that for
\begin{align}
	\epsilon_s &\leq \frac{\epsilon \Delta }{15 \norm{a}_1 M_1  \sigma_m \pi \mu^2\log\left(\frac{\mu}{2}\right) } \nonumber \\
	&\leq \frac{\epsilon \Delta}{15 M_1 \|a\|_1  \sigma_m \pi \left( n_h + \log(M_1 \Lambda/\epsilon)\right)^2 \log((n_n/2) + \log(M_1\Lambda/\epsilon)/2)}\nonumber \\
 	&\leq \frac{\epsilon \Delta}{15 M_1 \|a\|_1  \sigma_m \pi \mu^2 \log(\mu/2)}
\end{align}
also the last term in \eqref{eq:bounds_approximation_error} can be bounded by $\epsilon/3$, which together results in an overall error of $\epsilon$ for the various approximation steps, which concludes the proof.
\end{proof}
Notably all quantities which occure in our bounds are only polynomial in the number of the qubits.
The lower bounds for the choice of parameters are summarized in the following.
\begin{align}
M_1 &\geq \sqrt{L \log\left(\frac{9 \norm{\frac{\partial \sigma_v}{\partial \theta}} K_1 L \pi}{2\epsilon}\right)} \label{eq:bound_M1} \\
M_2 &\geq \left\lceil \log \left(\frac{4}{\epsilon_2} \right) \sqrt{(2\log{\delta_u^{-1}})^{-1}} \right\rceil \label{eq:bound_M2} \\
K_1 &\geq \log((1-\norm{\sigma_v})\epsilon/9)/\log(\norm{\sigma_v}) \label{eq:bound_K1} \\
K_2 &\geq \frac{\log(4/\epsilon_2)}{\log(\delta_u^{-1})} \label{eq:bound_K2} \\
L &\geq \frac{\log \left(\frac{\epsilon}{9\pi K_1 \norm{\frac{\partial \sigma_v}{\partial\theta}}}\right)}{\log(\norm{\sigma_v} \pi)} \label{eq:bound_L} \\
\mu &\geq n_h + \log \left( \frac{6 M_1 \norm{a}_1 e^2 \lvert \lambda_{max}\rvert \pi  \norm{\trh{e^{-H}}}}{\epsilon} \right):= n_h + \log(M_1 \Lambda/\epsilon) \label{eq:bound_mu} \\
\epsilon_2 &\le \frac{\epsilon \Delta}{15 M_1 \|a\|_1 \pi \left(n_h+\log(M_1 \Lambda/\epsilon)\right)^2 \log((n_h/2) + \log(M_1\Lambda/\epsilon)/2)}  \nonumber \\
	 &\leq \frac{\epsilon \Delta}{15 M_1 \|a\|_1 \pi \mu^2 \log((\mu-1)/2)} \label{eq:bounds_eps2}\\
\epsilon_s &\leq  \frac{\epsilon \Delta}{15 M_1 \|a\|_1  \sigma_m \pi \left( n_h + \log(M_1 \Lambda/\epsilon)\right)^2 \log((n_n/2) + \log(M_1\Lambda/\epsilon)/2)} \nonumber \\
		&\leq \frac{\epsilon \Delta}{15 M_1 \|a\|_1  \sigma_m \pi \mu^2 \log(\mu/2)} \label{eq:bounds_epss}
\end{align}

\subsection{Operationalising}
\label{app:proof_general_algo}
In the following we will make use of two established subroutines, namely sample based Hamiltonian simulation (aka the LMR protocol)~\cite{lloyd2014quantum}, as well as the Hadamard test, in order to evaluate the gradient approximation as defined in \eqref{eq:full_approximation}.
In order to hence derive the query complexity for this algorithm, we only need to multiply the cost of the number of factors we need to evaluate with the query complexity of these routines.
For this we will rely on the following result.
\begin{theorem}[Sample based Hamiltonian simulation~\cite{kimmel2017hamiltonian}]
  \label{thm:sample_based_ham_sim}
	Let $0 \leq \epsilon_h \leq 1/6$ be an error parameter and let $\rho$ be a density for which we can obtain multiple copies through queries to a oracle $O_{\rho}$.
	We can then simulate the time evolution $e^{-i\rho t}$ up to error $\epsilon_h$ in trace norm as long as $\epsilon_h/t \leq 1/(6 \pi)$ with $\Theta(t^2/\epsilon_h)$ copies of $\rho$ and hence queries to $O_{\rho}$.
\end{theorem}
We in particularly need to evaluate terms of the form
\begin{equation}
  \label{eq:error_formula_sample_based_ham_sim}
	\tr{ \rho e^{\frac{i s \pi m}{2}\sigma_v} e^{\frac{i \pi m'}{2} \sigma_v(\theta_j)} e^{\frac{i (1-s) \pi m}{2}\sigma_v}}
\end{equation}
Note that we can simulate every term in the trace (except $\rho$) via the sample based Hamiltonian simulation approach to error $\epsilon_h$ in trace norm.
This will introduce a additional error which we need to take into account for the analysis.
Let $\tilde{U}_{i}, i \in \{1,2,3\}$ be the unitaries such that $\norm{U_i - \tilde{U}_i}_* \leq \epsilon_h$ where the $U_i$ are corresponding to the factors in \eqref{eq:error_formula_sample_based_ham_sim}, i.e.,
$U_1 := e^{\frac{i s \pi m}{2}\sigma_v}$, $U_2 := e^{\frac{i \pi m'}{2} \sigma_v(\theta_j)}$, and $U_3 := e^{\frac{i (1-s) \pi m}{2}\sigma_v}$.
We can then bound the error as follows. First note that $\norm{\tilde U_i} \leq \norm{\tilde U_i - U_i} + \norm{U_i} \leq 1+\epsilon_h$, using Theorem~\ref{thm:sample_based_ham_sim} and the fact that the spectral norm is upper bounded by the trace norm.
\begin{align}
\label{eq:error_propagation_ham_sim}
  \tr{ \rho U_1 U_2 U_3 } &-	\tr{ \rho \tilde U_1 \tilde U_2 \tilde U_3 }  \leq \nonumber \\
  &= \tr{\rho U_1 U_2 U_3 -	\rho \tilde U_1 \tilde U_2 \tilde U_3} \nonumber \\
  &\leq \norm{U_1 U_2 U_3 -	\tilde U_1 \tilde U_2 \tilde U_3} \nonumber \\
  &\leq \norm{U_1 -	\tilde U_1} \norm{\tilde U_2} \norm{\tilde U_3}+ \norm{U_2 -	\tilde U_2} \tilde{U_3} + \norm{U_3 -	\tilde U_3} \nonumber \\
  &\leq \norm{U_1 -	\tilde U_1}_* (1+\epsilon_h)^2 + \norm{U_2 -	\tilde U_2}_* (1+\epsilon_h)+ \norm{U_3 -	\tilde U_3}_* \nonumber \\
  &\leq  \epsilon_h (1+\epsilon_h)^2 + \epsilon_h (1+\epsilon_h) + \epsilon_h = O(\epsilon_h),
\end{align}
neglecting higher orders of $\epsilon_h$, and where in the first step we applied the Von-Neumann trace inequality and the fact that $\rho$ is Hermitian, and in the last step we used the results of Theorem~\ref{thm:sample_based_ham_sim}.
We hence require $O((\max \{ M_1, M_2\} \pi)^2/\epsilon_h)$ queries to the oracles for $\sigma_v$ for the evaluation of each term in the multi sum in \eqref{eq:full_approximation}.
Note that the Hadamard test has a query cost of $O(1)$.
In order to hence achieve an overall error of $\epsilon$ in the gradient estimation we require the error introduced by the sample based Hamiltonian simulation also to be of $O(\epsilon)$.
In order to do so we require $\epsilon_h \leq O(\frac{\epsilon \Delta }{5 \norm{a}_1 M_1 \pi \mu^2 \log(\mu/2)})$, similar to the sample based error which yield the query complexity of
\begin{equation}
  \label{eq:final_comp_complexity_grad_entry_est}
  O\left( \frac{\max \{ M_1, M_2\}^2 \norm{a}_1 M_1  \pi^3 \mu^2 \log(\mu/2)}{\epsilon \Delta}\right)
\end{equation}
Adjusting the constants gives then the required bound of $\epsilon$ of the total error and the query complexity for the algorithm to the Gibbs state preparation procedure is consequentially given by the number of terms in \eqref{eq:full_approximation} times the query complexity for the individual term, yielding
\begin{align}
    O\left(\frac{M_1^2 M_2 \max \{ M_1, M_2\}^2 \norm{a}_1 \sigma_m \pi^3 \mu^3 \log \left(\frac{\mu}{2}\right)}{\epsilon\ \epsilon_s^2 \Delta} \right),
\end{align}
and classical precomputation polynomial in $M_1,M_2,K_1,L,s,\Delta,\mu$, where the different quantities are defines in eq.~(\ref{eq:bound_M1}-\ref{eq:bounds_epss}).

  Taking into account the query complexity of the individual steps then results in Theorem~\ref{thm:complexity_general_algo}.
  We proceed by proving this theorem next.
  \begin{proof}[Proof of Theorem~\ref{thm:complexity_general_algo}]
    The runtime follows straight forward by using the bounds derived in \eqref{eq:final_comp_complexity_grad_entry_est} and Lemma~\ref{lem:proof_error_bound_grads}, and by using the bounds for the parameters $M_1,M_2, K_1,L,\mu,\Delta,s$ given in eq.~(\ref{eq:bound_M1}-\ref{eq:bounds_epss}).
For the success probability for estimating the whole gradient with dimensionality $d$, we can now again make use of the boosting scheme used in \eqref{eq:probability_boost} to be
\begin{align}
\label{eq:final_query_complexity_Gibbs_exact}
    \tilde{O} \left(
      \frac{
        d \norm{a}_1^3 \sigma_m^3 \mu^5 \log^3(\mu/2)
        \mathrm{polylog} \left( \frac{\norm{\frac{\partial \sigma_v}{\partial \theta}}}{\epsilon}, \,
         \frac{n_h^2 \norm{a}_1 \sigma_m}{\epsilon \Delta}\right)
      }
      {\epsilon^3 \Delta^3}
      \log \left( d \right)
    \right),
\end{align}
where $\mu=n_h + \log(M_1 \Lambda/\epsilon)$.\\
Next we need to take into account the errors from the Gibbs state preparation given in Lemma~\ref{thm:sample_based_ham_sim}.
For this note that the error between the perfect Hamiltonian simulation of $\sigma_v$ and the sample based Hamiltonian simulation with an erroneous density matrix denoted by $\tilde{U}$, i.e., including the error from the Gibbs state preparation procedure, is given by
\begin{align}
	\norm{\tilde U - e^{-i\sigma_v t}} &\leq \norm{\tilde U - e^{-i \tilde \sigma_v t}} +  \norm{e^{-i \tilde \sigma_v t} - e^{-i\sigma_v t}} \nonumber \\
	&\leq \epsilon_h  + \epsilon_G t
\end{align}
where $\epsilon_h$ is the error of the sample based Hamiltonian simulation, which holds since the trace norm is an upper bound for the spectral norm, and $\norm{\sigma_v - \tilde{\sigma}_v} \leq \epsilon_G$ is the error for the Gibbs state preparation from Theorem~\ref{thm:Gibbs_state_prep} for a $d$-sparse Hamiltonian, for a cost  $$\tilde{\mathcal{O}} \left(\sqrt{\frac{N}{z}} \norm{H}d \log \left( \frac{\norm{H}}{\epsilon_G}\right) \log \left( \frac{1}{\epsilon_G} \right) \right).$$
From \eqref{eq:error_propagation_ham_sim} we know that the error $\epsilon_h$ propagates nearly linear, and hence it suffices for us to take $\epsilon_G \leq \epsilon_h/t$ where $t = O(\max\{M_1,M_2\})$ and adjust the constants $\epsilon_h \leftarrow \epsilon_h/2$ in order to achieve the same precision $\epsilon$ in the final result.
We hence require
\begin{equation}
	\tilde{\mathcal{O}} \left( \sqrt{\frac{N}{z}} \norm{H(\theta)} \log \left( \frac{\norm{H(\theta)} \max \lbrace M_1,M_2\rbrace }{\epsilon_h}\right) \log \left( \frac{\max \lbrace M_1,M_2\rbrace }{\epsilon_h} \right) \right)
\end{equation}
and using the $\epsilon_h$ from before we hence find that this s bound by
\begin{equation}
	\tilde{\mathcal{O}} \left( \sqrt{\frac{N}{z}} \norm{H(\theta)} \log \left( \frac{\norm{H(\theta)} n_h^2}{\epsilon \Delta}\right) \log \left( \frac{n_h^2}{\epsilon \Delta} \right) \right)
\end{equation}
query complexity to the oracle of $H$ for the Gibbs state preparation. \\
The procedure succeeds with probability at least $1-\delta_s$ for a single repetition for each entry of the gradient.
In order to have a failure probability of the final algorithm of less than $1/3$, we need to repeat the procedure for all $D$ dimensions of the gradient and take for each the median over a number of samples.
Let $n_f$ be as previously the number of instances of the one component of the gradient such that the error is larger than $\epsilon_s \sigma_m$ and $n_s$ be the number of instances with an error $\leq \epsilon_s \sigma_m$ , and the result that we take is the median of the estimates, where we take $n=n_s+n_f$ samples.
The algorithm gives a wrong answer for each dimension if $n_s \leq \left\lfloor \frac{n}{2} \right\rfloor$, since then the median is a sample such that the error is larger than $\epsilon_s \sigma_m$.
Let $p=1-\delta_s$ be the success probability to draw a positive sample, as is the case of our algorithm.
Since each instance of (recall that each sample here consists of a number of samples itself) from the algorithm will independently return an estimate for the entry of the gradient, the total failure probability is bounded by the union bound, i.e.,
\begin{equation}
\pr_{fail} \leq D \cdot \pr \left[n_s \leq  \left\lfloor \frac{n}{2} \right\rfloor \right] \leq D \cdot e^{- \frac{n}{2(1-\delta_s)} \left((1-\delta_s) - \frac{1}{2} \right)^2} \leq \frac{1}{3},
\end{equation}
which follows from the Chernoff inequality for a binomial variable with $1-\delta_s>1/2$, which is given in our case for a proper choice of $\delta_s <1/2$.
Therefore, by taking $n \geq \frac{2-2\delta_s}{(1/2-\delta_s)^2} \log(3D) = O(\log(3D))$, we achieve a total failure probability of at least $1/3$ for a constant, fixed $\delta_s$.
Note that this hence results in an multiplicative factor of $O(\log(D))$ in the query complexity of \eqref{eq:final_query_complexity}.\\
The total query complexity to the oracle $O_{\rho}$ for a purified density matrix of the data $\rho$ and the Hamiltonian oracle $O_H$ is then given by
 \begin{equation}
     \tilde{O} \left(
       \sqrt{\frac{N}{z}}
	\frac{
        d  \log \left( d \right) \norm{H(\theta)}  \norm{a}_1^3 \sigma_m^3 \mu^5 \log^3(\mu/2)
        \mathrm{polylog} \left( \frac{\norm{\frac{\partial \sigma_v}{\partial \theta}}}{\epsilon}, \,
         \frac{n_h^2 \norm{a}_1 \sigma_m}{\epsilon \Delta}, \, \norm{H(\theta)} \right)
      }
      {\epsilon^3 \Delta^3}
    \right),
\end{equation}
which reduces to
  \begin{align}
    \tilde{O} \left(   \sqrt{\frac{N}{z}}
      \frac{D \norm{H(\theta)}
        d   \mu^5
      \alpha
      }
      {\epsilon^3}
    \right),
  \end{align}
hiding the logarithmic factors in the $\tilde{O}$ notation.
\end{proof}


\addcontentsline{toc}{chapter}{Bibliography}

\bibliography{example}

\begin{thebibliography}{100}

\bibitem{krizhevsky2012imagenet}
Alex Krizhevsky, Ilya Sutskever, and Geoffrey~E Hinton.
\newblock Imagenet classification with deep convolutional neural networks.
\newblock In {\em Advances in neural information processing systems}, pages
  1097--1105, 2012.

\bibitem{mnih2013playing}
Volodymyr Mnih, Koray Kavukcuoglu, David Silver, Alex Graves, Ioannis
  Antonoglou, Daan Wierstra, and Martin Riedmiller.
\newblock Playing atari with deep reinforcement learning.
\newblock {\em arXiv preprint arXiv:1312.5602}, 2013.

\bibitem{silver2016mastering}
David Silver, Aja Huang, Chris~J Maddison, Arthur Guez, Laurent Sifre, George
  Van Den~Driessche, Julian Schrittwieser, Ioannis Antonoglou, Veda
  Panneershelvam, Marc Lanctot, et~al.
\newblock Mastering the game of go with deep neural networks and tree search.
\newblock {\em Nature}, 529(7587):484--489, 2016.

\bibitem{markov2014limits}
Igor~L Markov.
\newblock Limits on fundamental limits to computation.
\newblock {\em Nature}, 512(7513):147--154, 2014.

\bibitem{ciliberto2018quantum}
Carlo Ciliberto, Mark Herbster, Alessandro~Davide Ialongo, Massimiliano Pontil,
  Andrea Rocchetto, Simone Severini, and Leonard Wossnig.
\newblock Quantum machine learning: a classical perspective.
\newblock {\em Proceedings of the Royal Society A: Mathematical, Physical and
  Engineering Sciences}, 474(2209):20170551, 2018.

\bibitem{giovannetti2008qram1}
Vittorio Giovannetti, Seth Lloyd, and Lorenzo Maccone.
\newblock Quantum random access memory.
\newblock {\em Physical review letters}, 100(16):160501, 2008.

\bibitem{giovannetti2008qram2}
Vittorio Giovannetti, Seth Lloyd, and Lorenzo Maccone.
\newblock Architectures for a quantum random access memory.
\newblock {\em Physical Review A}, 78(5):052310, 2008.

\bibitem{lloyd1996universal}
Seth Lloyd.
\newblock Universal quantum simulators.
\newblock {\em Science}, pages 1073--1078, 1996.

\bibitem{wiebe2014quantum}
Nathan Wiebe, Ashish Kapoor, and Krysta~M Svore.
\newblock Quantum deep learning.
\newblock {\em arXiv preprint arXiv:1412.3489}, 2014.

\bibitem{rudi2017falkon}
Alessandro Rudi, Luigi Carratino, and Lorenzo Rosasco.
\newblock Falkon: An optimal large scale kernel method.
\newblock In {\em Advances in Neural Information Processing Systems}, pages
  3888--3898, 2017.

\bibitem{clader2013preconditioned}
B~David Clader, Bryan~C Jacobs, and Chad~R Sprouse.
\newblock Preconditioned quantum linear system algorithm.
\newblock {\em Physical review letters}, 110(25):250504, 2013.

\bibitem{harrow2017limitations}
Aram Harrow and Rolando~P. La~Placa.
\newblock Limitations of quantum preconditioning using a sparse approximate
  inverse.
\newblock {\em Unpublished work}, 2017.

\bibitem{lloyd2014quantum}
Seth Lloyd, Masoud Mohseni, and Patrick Rebentrost.
\newblock Quantum principal component analysis.
\newblock {\em Nature Physics}, 10(9):631, 2014.

\bibitem{rebentrost2014quantum}
Patrick Rebentrost, Masoud Mohseni, and Seth Lloyd.
\newblock Quantum support vector machine for big data classification.
\newblock {\em Physical review letters}, 113(13):130503, 2014.

\bibitem{tang2019quantum}
Ewin Tang.
\newblock A quantum-inspired classical algorithm for recommendation systems.
\newblock In {\em Proceedings of the 51st Annual ACM SIGACT Symposium on Theory
  of Computing}, pages 217--228, 2019.

\bibitem{kerenidis2016quantum}
Iordanis Kerenidis and Anupam Prakash.
\newblock Quantum recommendation systems.
\newblock In {\em Proceedings of the 8th Innovations in Theoretical Computer
  Science Conference (ITCS 2017)}. Schloss Dagstuhl-Leibniz-Zentrum fuer
  Informatik, 2017.

\bibitem{tang2018quantum}
Ewin Tang.
\newblock Quantum-inspired classical algorithms for principal component
  analysis and supervised clustering.
\newblock {\em arXiv preprint arXiv:1811.00414}, 2018.

\bibitem{chia2018quantum}
Nai-Hui Chia, Han-Hsuan Lin, and Chunhao Wang.
\newblock Quantum-inspired sublinear classical algorithms for solving low-rank
  linear systems.
\newblock {\em arXiv preprint arXiv:1811.04852}, 2018.

\bibitem{gilyen2018quantum}
Andr{\'a}s Gily{\'e}n, Seth Lloyd, and Ewin Tang.
\newblock Quantum-inspired low-rank stochastic regression with logarithmic
  dependence on the dimension.
\newblock {\em arXiv preprint arXiv:1811.04909}, 2018.

\bibitem{chia2019quantum}
Nai-Hui Chia, Tongyang Li, Han-Hsuan Lin, and Chunhao Wang.
\newblock Quantum-inspired classical sublinear-time algorithm for solving
  low-rank semidefinite programming via sampling approaches.
\newblock {\em arXiv preprint arXiv:1901.03254}, 2019.

\bibitem{chia2020sampling}
Nai-Hui Chia, Andr{\'a}s Gily{\'e}n, Tongyang Li, Han-Hsuan Lin, Ewin Tang, and
  Chunhao Wang.
\newblock Sampling-based sublinear low-rank matrix arithmetic framework for
  dequantizing quantum machine learning.
\newblock In {\em Proceedings of the 52nd Annual ACM SIGACT Symposium on Theory
  of Computing}, pages 387--400, 2020.

\bibitem{dervovic2018quantum}
Danial Dervovic, Mark Herbster, Peter Mountney, Simone Severini, Na{\"\i}ri
  Usher, and Leonard Wossnig.
\newblock Quantum linear systems algorithms: a primer.
\newblock {\em arXiv preprint arXiv:1802.08227}, 2018.

\bibitem{wang2018quantum}
Chunhao Wang and Leonard Wossnig.
\newblock A quantum algorithm for simulating non-sparse hamiltonians.
\newblock {\em Quantum Information \& Computation}, 20(7\&8):597--615, 2020.

\bibitem{rudi2020approximating}
Alessandro Rudi, Leonard Wossnig, Carlo Ciliberto, Andrea Rocchetto,
  Massimiliano Pontil, and Simone Severini.
\newblock Approximating hamiltonian dynamics with the nystr{\"o}m method.
\newblock {\em Quantum}, 4:234, 2020.

\bibitem{ciliberto2020fast}
Carlo Ciliberto, Andrea Rocchetto, Alessandro Rudi, and Leonard Wossnig.
\newblock Fast quantum learning with statistical guarantees.
\newblock {\em arXiv preprint arXiv:2001.10477}, 2020.

\bibitem{wiebe2019generative}
Nathan Wiebe and Leonard Wossnig.
\newblock Generative training of quantum boltzmann machines with hidden units.
\newblock {\em arXiv preprint arXiv:1905.09902}, 2019.

\bibitem{chen2018universal}
Hongxiang Chen, Leonard Wossnig, Simone Severini, Hartmut Neven, and Masoud
  Mohseni.
\newblock Universal discriminative quantum neural networks.
\newblock {\em arXiv preprint arXiv:1805.08654}, 2018.

\bibitem{benedetti2019adversarial}
Marcello Benedetti, Edward Grant, Leonard Wossnig, and Simone Severini.
\newblock Adversarial quantum circuit learning for pure state approximation.
\newblock {\em New Journal of Physics}, 21(4):043023, 2019.

\bibitem{grant2019initialization}
Edward Grant, Leonard Wossnig, Mateusz Ostaszewski, and Marcello Benedetti.
\newblock An initialization strategy for addressing barren plateaus in
  parametrized quantum circuits.
\newblock {\em Quantum}, 3:214, 2019.

\bibitem{cao2020cost}
Shuxiang Cao, Leonard Wossnig, Brian Vlastakis, Peter Leek, and Edward Grant.
\newblock Cost-function embedding and dataset encoding for machine learning
  with parametrized quantum circuits.
\newblock {\em Physical Review A}, 101(5):052309, 2020.

\bibitem{rungger2019dynamical}
I~Rungger, N~Fitzpatrick, H~Chen, CH~Alderete, H~Apel, A~Cowtan, A~Patterson,
  D~Munoz Ramo, Y~Zhu, NH~Nguyen, et~al.
\newblock Dynamical mean field theory algorithm and experiment on quantum
  computers.
\newblock {\em arXiv preprint arXiv:1910.04735}, 2019.

\bibitem{patterson2019quantum}
Andrew Patterson, Hongxiang Chen, Leonard Wossnig, Simone Severini, Dan Browne,
  and Ivan Rungger.
\newblock Quantum state discrimination using noisy quantum neural networks.
\newblock {\em arXiv preprint arXiv:1911.00352}, 2019.

\bibitem{tilly2020computation}
Jules Tilly, Glenn Jones, Hongxiang Chen, Leonard Wossnig, and Edward Grant.
\newblock Computation of molecular excited states on ibmq using a
  discriminative variational quantum eigensolver.
\newblock {\em arXiv preprint arXiv:2001.04941}, 2020.

\bibitem{haber2018notes}
Howard~E Haber.
\newblock Notes on the matrix exponential and logarithm.
\newblock 2018.

\bibitem{higham2008functions}
Nicholas~J Higham.
\newblock {\em Functions of matrices: theory and computation}, volume 104.
\newblock Siam, 2008.

\bibitem{harrow2009quantum}
Aram~W Harrow, Avinatan Hassidim, and Seth Lloyd.
\newblock Quantum algorithm for linear systems of equations.
\newblock {\em Physical review letters}, 103(15):150502, 2009.

\bibitem{CGJ19}
Shantanav Chakraborty, Andr{\'a}s Gily{\'e}n, and Stacey Jeffery.
\newblock The power of block-encoded matrix powers: Improved regression
  techniques via faster {H}amiltonian simulation.
\newblock In {\em Proceedings of the 46th International Colloquium on Automata,
  Languages, and Programming (ICALP 2019)}, volume 132, pages 33:1--33:14,
  2019.

\bibitem{wiebe2012quantum}
Nathan Wiebe, Daniel Braun, and Seth Lloyd.
\newblock Quantum algorithm for data fitting.
\newblock {\em Physical review letters}, 109(5):050505, 2012.

\bibitem{schuld2016prediction}
Maria Schuld, Ilya Sinayskiy, and Francesco Petruccione.
\newblock Prediction by linear regression on a quantum computer.
\newblock {\em Physical Review A}, 94(2):022342, 2016.

\bibitem{wang2017quantum}
Guoming Wang.
\newblock Quantum algorithm for linear regression.
\newblock {\em Physical review A}, 96(1):012335, 2017.

\bibitem{kerenidis2017quantum}
Iordanis Kerenidis and Anupam Prakash.
\newblock Quantum recommendation systems.
\newblock In {\em Innovations in Theoretical Computer Science (ITCS'17)}, 2017.

\bibitem{zhang2018nonlinear}
Dan-Bo Zhang, Shi-Liang Zhu, and ZD~Wang.
\newblock Nonlinear regression based on a hybrid quantum computer.
\newblock {\em arXiv preprint arXiv:1808.09607}, 2018.

\bibitem{aaronson2015read}
Scott Aaronson.
\newblock Read the fine print.
\newblock {\em Nature Physics}, 11(4):291--293, 2015.

\bibitem{shalev2014understanding}
Shai Shalev-Shwartz and Shai Ben-David.
\newblock {\em Understanding machine learning: From theory to algorithms}.
\newblock Cambridge University Press, 2014.

\bibitem{vapnik1998}
Vladimir~Naumovich Vapnik and Vlamimir Vapnik.
\newblock {\em Statistical learning theory}, volume~1.
\newblock Wiley New York, 1998.

\bibitem{blumer1989learnability}
Anselm Blumer, Andrzej Ehrenfeucht, David Haussler, and Manfred~K Warmuth.
\newblock Learnability and the vapnik-chervonenkis dimension.
\newblock {\em Journal of the ACM (JACM)}, 36(4):929--965, 1989.

\bibitem{cucker2002mathematical}
Felipe Cucker and Steve Smale.
\newblock On the mathematical foundations of learning.
\newblock {\em Bulletin of the American mathematical society}, 39(1):1--49,
  2002.

\bibitem{bishop2006pattern}
Christopher~M Bishop et~al.
\newblock {\em Pattern recognition and machine learning}, volume~1.
\newblock springer New York, 2006.

\bibitem{rasmussen2006gaussian}
Carl Rasmussen and Chris Williams.
\newblock Gaussian processes for machine learning.
\newblock {\em Gaussian Processes for Machine Learning}, 2006.

\bibitem{shawe2004kernel}
John Shawe-Taylor, Nello Cristianini, et~al.
\newblock {\em Kernel methods for pattern analysis}.
\newblock Cambridge university press, 2004.

\bibitem{zhang2013divide}
Yuchen Zhang, John Duchi, and Martin Wainwright.
\newblock Divide and conquer kernel ridge regression.
\newblock In {\em Conference on Learning Theory}, pages 592--617, 2013.

\bibitem{rahimi2007random}
Ali Rahimi, Benjamin Recht, et~al.
\newblock Random features for large-scale kernel machines.
\newblock In {\em NIPS}, volume~3, page~5, 2007.

\bibitem{smola2000sparse}
Alex~J Smola and Bernhard Sch{\"o}lkopf.
\newblock Sparse greedy matrix approximation for machine learning.
\newblock pages 911--918. Morgan Kaufmann, 2000.

\bibitem{williams2001using}
Christopher~KI Williams and Matthias Seeger.
\newblock Using the nystr{\"o}m method to speed up kernel machines.
\newblock In {\em Advances in neural information processing systems}, pages
  682--688, 2001.

\bibitem{rudi2015less}
Alessandro Rudi, Raffaello Camoriano, and Lorenzo Rosasco.
\newblock Less is more: Nystr\"{o}m computational regularization.
\newblock In {\em Proceedings of the 28th International Conference on Neural
  Information Processing Systems - Volume 1}, NIPS’15, page 1657–1665,
  Cambridge, MA, USA, 2015. MIT Press.

\bibitem{RWC+18}
Alessandro Rudi, Leonard Wossnig, Carlo Ciliberto, Andrea Rocchetto,
  Massimiliano Pontil, and Simone Severini.
\newblock Approximating {H}amiltonian dynamics with the {N}ystr{\"o}m method.
\newblock {\em Quantum}, 4:234, 2020.

\bibitem{kimmel2017hamiltonian}
Shelby Kimmel, Cedric Yen-Yu Lin, Guang~Hao Low, Maris Ozols, and Theodore~J
  Yoder.
\newblock Hamiltonian simulation with optimal sample complexity.
\newblock {\em npj Quantum Information}, 3(1):13, 2017.

\bibitem{li2019sublinear}
Tongyang Li, Shouvanik Chakrabarti, and Xiaodi Wu.
\newblock Sublinear quantum algorithms for training linear and kernel-based
  classifiers.
\newblock In Kamalika Chaudhuri and Ruslan Salakhutdinov, editors, {\em
  Proceedings of the 36th International Conference on Machine Learning},
  volume~97 of {\em Proceedings of Machine Learning Research}, pages
  3815--3824, Long Beach, California, USA, 09--15 Jun 2019. PMLR.

\bibitem{gerfo2008spectral}
L~Lo Gerfo, Lorenzo Rosasco, Francesca Odone, E~De Vito, and Alessandro Verri.
\newblock Spectral algorithms for supervised learning.
\newblock {\em Neural Computation}, 20(7):1873--1897, 2008.

\bibitem{hochstadt2011integral}
Harry Hochstadt.
\newblock {\em Integral equations}, volume~91.
\newblock John Wiley \& Sons, 2011.

\bibitem{tropp2015introduction}
Joel~A Tropp et~al.
\newblock An introduction to matrix concentration inequalities.
\newblock {\em Foundations and Trends{\textregistered} in Machine Learning},
  8(1-2):1--230, 2015.

\bibitem{grilo2017learning}
Alex~B Grilo and Iordanis Kerenidis.
\newblock Learning with errors is easy with quantum samples.
\newblock {\em arXiv preprint arXiv:1702.08255}, 2017.

\bibitem{kanade2018learning}
Varun Kanade, Andrea Rocchetto, and Simone Severini.
\newblock Learning dnfs under product distributions via
  $\{$$\backslash$mu$\}$-biased quantum fourier sampling.
\newblock {\em arXiv preprint arXiv:1802.05690}, 2018.

\bibitem{tong2020fast}
Yu~Tong, Dong An, Nathan Wiebe, and Lin Lin.
\newblock Fast inversion, preconditioned quantum linear system solvers, and
  fast evaluation of matrix functions.
\newblock {\em arXiv preprint arXiv:2008.13295}, 2020.

\bibitem{ma2017diving}
Siyuan Ma and Mikhail Belkin.
\newblock Diving into the shallows: a computational perspective on large-scale
  shallow learning.
\newblock In {\em Advances in Neural Information Processing Systems}, pages
  3778--3787, 2017.

\bibitem{gonen2016solving}
Alon Gonen, Francesco Orabona, and Shai Shalev-Shwartz.
\newblock Solving ridge regression using sketched preconditioned svrg.
\newblock In {\em International Conference on Machine Learning}, pages
  1397--1405, 2016.

\bibitem{avron2017faster}
Haim Avron, Kenneth~L Clarkson, and David~P Woodruff.
\newblock Faster kernel ridge regression using sketching and preconditioning.
\newblock {\em SIAM Journal on Matrix Analysis and Applications},
  38(4):1116--1138, 2017.

\bibitem{fasshauer2012stable}
Gregory~E Fasshauer and Michael~J McCourt.
\newblock Stable evaluation of gaussian radial basis function interpolants.
\newblock {\em SIAM Journal on Scientific Computing}, 34(2):A737--A762, 2012.

\bibitem{woodruff2014sketching}
David~P Woodruff.
\newblock Sketching as a tool for numerical linear algebra.
\newblock {\em arXiv preprint arXiv:1411.4357}, 2014.

\bibitem{drineas2006fast}
Petros Drineas, Ravi Kannan, and Michael~W Mahoney.
\newblock Fast monte carlo algorithms for matrices i: Approximating matrix
  multiplication.
\newblock {\em SIAM Journal on Computing}, 36(1):132--157, 2006.

\bibitem{mahoney2016lecture}
Michael~W Mahoney.
\newblock Lecture notes on randomized linear algebra.
\newblock {\em arXiv preprint arXiv:1608.04481}, 2016.

\bibitem{adcock2015advances}
Jeremy Adcock, Euan Allen, Matthew Day, Stefan Frick, Janna Hinchliff, Mack
  Johnson, Sam Morley-Short, Sam Pallister, Alasdair Price, and Stasja
  Stanisic.
\newblock Advances in quantum machine learning.
\newblock {\em arXiv preprint arXiv:1512.02900}, 2015.

\bibitem{steiger2016perimeter}
Damian~S. Steiger and Matthias Troyer.
\newblock Racing in parallel: Quantum versus classical.
\newblock In {\em Quantum Machine Learning Workshop}, Waterloo, CA, August
  2016. Perimeter Institute for theoretical Physics.

\bibitem{regev2008impossibility}
Oded Regev and Liron Schiff.
\newblock Impossibility of a quantum speed-up with a faulty oracle.
\newblock In {\em International Colloquium on Automata, Languages, and
  Programming}, pages 773--781. Springer, 2008.

\bibitem{arunachalam2015robustness}
Srinivasan Arunachalam, Vlad Gheorghiu, Tomas Jochym-O’Connor, Michele Mosca,
  and Priyaa~Varshinee Srinivasan.
\newblock On the robustness of bucket brigade quantum ram.
\newblock {\em New Journal of Physics}, 17(12):123010, 2015.

\bibitem{feynman1982simulating}
Richard~P Feynman.
\newblock Simulating physics with computers.
\newblock {\em International journal of theoretical physics}, 21(6):467--488,
  1982.

\bibitem{feynman1986quantum}
Richard~P Feynman.
\newblock Quantum mechanical computers.
\newblock {\em Foundations of physics}, 16(6):507--531, 1986.

\bibitem{shor1999polynomial}
Peter~W Shor.
\newblock Polynomial-time algorithms for prime factorization and discrete
  logarithms on a quantum computer.
\newblock {\em SIAM Review}, 41(2):303--332, 1999.

\bibitem{aharonov2003adiabatic}
Dorit Aharonov and Amnon Ta-Shma.
\newblock Adiabatic quantum state generation and statistical zero knowledge.
\newblock In {\em Proceedings of the 35th Annual ACM Symposium on Theory of
  Computing (STOC 2003)}, pages 20--29. ACM, 2003.

\bibitem{berry2007efficient}
Dominic~W Berry, Graeme Ahokas, Richard Cleve, and Barry~C Sanders.
\newblock Efficient quantum algorithms for simulating sparse hamiltonians.
\newblock {\em Communications in Mathematical Physics}, 270(2):359--371, 2007.

\bibitem{berry2009black}
Dominic~W Berry and Andrew~M Childs.
\newblock Black-box {H}amiltonian simulation and unitary implementation.
\newblock {\em Quantum Information \& Computation}, 12(1-2):29--62, 2012.

\bibitem{berry2017exponential}
Dominic~W Berry, Andrew~M Childs, Richard Cleve, Robin Kothari, and Rolando~D
  Somma.
\newblock Exponential improvement in precision for simulating sparse
  {H}amiltonians.
\newblock {\em Forum of Mathematics, Sigma}, 5, 2017.

\bibitem{childs2010relationship}
Andrew~M Childs.
\newblock On the relationship between continuous-and discrete-time quantum
  walk.
\newblock {\em Communications in Mathematical Physics}, 294(2):581--603, 2010.

\bibitem{childs2012hamiltonian}
Andrew~M Childs and Nathan Wiebe.
\newblock Hamiltonian simulation using linear combinations of unitary
  operations.
\newblock {\em Quantum Information \& Computation}, 12(11-12):901--924, 2012.

\bibitem{poulin2011quantum}
David Poulin, Angie Qarry, Rolando Somma, and Frank Verstraete.
\newblock Quantum simulation of time-dependent {H}amiltonians and the
  convenient illusion of hilbert space.
\newblock {\em Physical Review Letters}, 106(17):170501, 2011.

\bibitem{wiebe2011simulating}
Nathan Wiebe, Dominic~W Berry, Peter H{\o}yer, and Barry~C Sanders.
\newblock Simulating quantum dynamics on a quantum computer.
\newblock {\em Journal of Physics A: Mathematical and Theoretical},
  44(44):445308, 2011.

\bibitem{low2016hamiltonian}
Guang~Hao Low and Isaac~L Chuang.
\newblock Hamiltonian simulation by qubitization.
\newblock {\em Quantum}, 3:163, 2019.

\bibitem{berry2016corrected}
Dominic~W Berry and Leonardo Novo.
\newblock Corrected quantum walk for optimal {H}amiltonian simulation.
\newblock {\em Quantum Information \& Computation}, 16(15-16):1295--1317, 2016.

\bibitem{low2017hamiltonian}
Guang~Hao Low and Isaac~L Chuang.
\newblock Hamiltonian simulation by uniform spectral amplification.
\newblock {\em arXiv preprint arXiv:1707.05391}, 2017.

\bibitem{berry2015hamiltonian}
Dominic~W Berry, Andrew~M Childs, and Robin Kothari.
\newblock Hamiltonian simulation with nearly optimal dependence on all
  parameters.
\newblock In {\em Proceedings of the 56th Annual Symposium on Foundations of
  Computer Science (FOCS 2015)}, pages 792--809. IEEE, 2015.

\bibitem{low2017optimal}
Guang~Hao Low and Isaac~L Chuang.
\newblock Optimal {H}amiltonian simulation by quantum signal processing.
\newblock {\em Physical Review Letters}, 118(1):010501, 2017.

\bibitem{low2018hamiltonian}
Guang~Hao Low.
\newblock Hamiltonian simulation with nearly optimal dependence on spectral
  norm.
\newblock In {\em Proceedings of the 51st Annual ACM Symposium on Theory of
  Computing (STOC 2019)}, pages 491--502, 2019.

\bibitem{childs2017quantum}
Andrew~M Childs, Robin Kothari, and Rolando~D Somma.
\newblock Quantum algorithm for systems of linear equations with exponentially
  improved dependence on precision.
\newblock {\em SIAM Journal on Computing}, 46(6):1920--1950, 2017.

\bibitem{wossnig2018quantum}
Leonard Wossnig, Zhikuan Zhao, and Anupam Prakash.
\newblock Quantum linear system algorithm for dense matrices.
\newblock {\em Physical Review Letters}, 120(5):050502, 2018.

\bibitem{CK10}
Andrew~M. Childs and Robin Kothari.
\newblock Limitations on the simulation of non-sparse {H}amiltonians.
\newblock {\em Quantum Information {\&} Computation}, 10(7{\&}8):669--684,
  2010.

\bibitem{horn1990matrix}
Roger~A Horn and Charles~R Johnson.
\newblock {\em Matrix analysis}.
\newblock Cambridge university press, 1990.

\bibitem{GSLW19}
Andr{\'a}s Gily{\'e}n, Yuan Su, Guang~Hao Low, and Nathan Wiebe.
\newblock Quantum singular value transformation and beyond: exponential
  improvements for quantum matrix arithmetics.
\newblock In {\em Proceedings of the 51st Annual ACM Symposium on Theory of
  Computing (STOC 2019)}, pages 193--204. ACM, 2019.

\bibitem{higham2005scaling}
Nicholas~J. Higham.
\newblock The scaling and squaring method for the matrix exponential revisited.
\newblock {\em SIAM J. Matrix Anal. Appl.}, 26(4):1179–1193, April 2005.

\bibitem{higham2009scaling}
Nicholas~J. Higham.
\newblock The scaling and squaring method for the matrix exponential revisited.
\newblock {\em SIAM Rev.}, 51(4):747–764, November 2009.

\bibitem{al2009new}
Awad~H Al-Mohy and Nicholas~J Higham.
\newblock A new scaling and squaring algorithm for the matrix exponential.
\newblock {\em SIAM Journal on Matrix Analysis and Applications},
  31(3):970--989, 2009.

\bibitem{al2011computing}
Awad~H Al-Mohy and Nicholas~J Higham.
\newblock Computing the action of the matrix exponential, with an application
  to exponential integrators.
\newblock {\em SIAM journal on scientific computing}, 33(2):488--511, 2011.

\bibitem{drineas2011faster}
Petros Drineas, Michael~W. Mahoney, S.~Muthukrishnan, and Tam\'{a}s Sarl\'{o}s.
\newblock Faster least squares approximation.
\newblock {\em Numer. Math.}, 117(2):219–249, February 2011.

\bibitem{orecchia2012approximating}
Lorenzo Orecchia, Sushant Sachdeva, and Nisheeth~K. Vishnoi.
\newblock Approximating the exponential, the lanczos method and an
  \~{O}(m)-time spectral algorithm for balanced separator.
\newblock In {\em Proceedings of the Forty-Fourth Annual ACM Symposium on
  Theory of Computing}, STOC ’12, page 1141–1160, New York, NY, USA, 2012.
  Association for Computing Machinery.

\bibitem{spielman2004nearly}
Daniel~A. Spielman and Shang-Hua Teng.
\newblock Nearly-linear time algorithms for graph partitioning, graph
  sparsification, and solving linear systems.
\newblock In {\em Proceedings of the Thirty-Sixth Annual ACM Symposium on
  Theory of Computing}, STOC ’04, page 81–90, New York, NY, USA, 2004.
  Association for Computing Machinery.

\bibitem{spielman2011spectral}
Daniel~A Spielman and Shang-Hua Teng.
\newblock Spectral sparsification of graphs.
\newblock {\em SIAM Journal on Computing}, 40(4):981--1025, 2011.

\bibitem{drineas2006fast2}
Petros Drineas, Ravi Kannan, and Michael~W Mahoney.
\newblock Fast monte carlo algorithms for matrices ii: Computing a low-rank
  approximation to a matrix.
\newblock {\em SIAM Journal on computing}, 36(1):158--183, 2006.

\bibitem{kumar2012sampling}
Sanjiv Kumar, Mehryar Mohri, and Ameet Talwalkar.
\newblock Sampling methods for the nystr{\"o}m method.
\newblock {\em Journal of Machine Learning Research}, 13(Apr):981--1006, 2012.

\bibitem{drineas2005nystrom}
Petros Drineas and Michael~W. Mahoney.
\newblock On the nystr\"{o}m method for approximating a gram matrix for
  improved kernel-based learning.
\newblock {\em J. Mach. Learn. Res.}, 6:2153–2175, December 2005.

\bibitem{williams2002observations}
CKI. Williams, CE. Rasmussen, A.~Schwaighofer, and V.~Tresp.
\newblock Observations on the nystr{\"o}m method for gaussian process
  prediction.
\newblock 2002.

\bibitem{nystrom1930praktische}
Evert~J Nystr{\"o}m.
\newblock {\"U}ber die praktische aufl{\"o}sung von integralgleichungen mit
  anwendungen auf randwertaufgaben.
\newblock {\em Acta Mathematica}, 54(1):185--204, 1930.

\bibitem{zhang2010clustered}
Kai Zhang and James~T Kwok.
\newblock Clustered nystr{\"o}m method for large scale manifold learning and
  dimension reduction.
\newblock {\em IEEE Transactions on Neural Networks}, 21(10):1576--1587, 2010.

\bibitem{talwalkar2008large}
Ameet Talwalkar, Sanjiv Kumar, and Henry Rowley.
\newblock Large-scale manifold learning.
\newblock In {\em Computer Vision and Pattern Recognition, 2008. CVPR 2008.
  IEEE Conference on}, pages 1--8. IEEE, 2008.

\bibitem{fowlkes2004spectral}
Charless Fowlkes, Serge Belongie, Fan Chung, and Jitendra Malik.
\newblock Spectral grouping using the nystr\"{o}m method.
\newblock {\em IEEE Trans. Pattern Anal. Mach. Intell.}, 26(2):214–225,
  January 2004.

\bibitem{belabbas2007fast}
Mohamed-Ali Belabbas and Patrick~J Wolfe.
\newblock Fast low-rank approximation for covariance matrices.
\newblock In {\em Computational Advances in Multi-Sensor Adaptive Processing,
  2007. CAMPSAP 2007. 2nd IEEE International Workshop on}, pages 293--296.
  IEEE, 2007.

\bibitem{belabbas2008sparse}
Mohamed-Ali Belabbas and Patrick~J Wolfe.
\newblock On sparse representations of linear operators and the approximation
  of matrix products.
\newblock In {\em Information Sciences and Systems, 2008. CISS 2008. 42nd
  Annual Conference on}, pages 258--263. IEEE, 2008.

\bibitem{kumar2009sampling}
Sanjiv Kumar, Mehryar Mohri, and Ameet Talwalkar.
\newblock On sampling-based approximate spectral decomposition.
\newblock In {\em Proceedings of the 26th Annual International Conference on
  Machine Learning}, ICML ’09, page 553–560, New York, NY, USA, 2009.
  Association for Computing Machinery.

\bibitem{li2010making}
Mu~Li, James~T. Kwok, and Bao-Liang Lu.
\newblock Making large-scale nystr\"{o}m approximation possible.
\newblock In {\em Proceedings of the 27th International Conference on
  International Conference on Machine Learning}, ICML’10, page 631–638,
  Madison, WI, USA, 2010. Omnipress.

\bibitem{mackey2011divide}
Lester Mackey, Ameet Talwalkar, and Michael~I. Jordan.
\newblock Divide-and-conquer matrix factorization.
\newblock In {\em Proceedings of the 24th International Conference on Neural
  Information Processing Systems}, NIPS’11, page 1134–1142, Red Hook, NY,
  USA, 2011. Curran Associates Inc.

\bibitem{zhang2008improved}
Kai Zhang, Ivor~W. Tsang, and James~T. Kwok.
\newblock Improved nystr\"{o}m low-rank approximation and error analysis.
\newblock In {\em Proceedings of the 25th International Conference on Machine
  Learning}, ICML ’08, page 1232–1239, New York, NY, USA, 2008. Association
  for Computing Machinery.

\bibitem{kassal2011simulating}
Ivan Kassal, James~D Whitfield, Alejandro Perdomo-Ortiz, Man-Hong Yung, and
  Al{\'a}n Aspuru-Guzik.
\newblock Simulating chemistry using quantum computers.
\newblock {\em Annual review of physical chemistry}, 62:185--207, 2011.

\bibitem{jordan2009efficient}
Stephen~P Jordan and Pawel Wocjan.
\newblock Efficient quantum circuits for arbitrary sparse unitaries.
\newblock {\em Physical Review A}, 80(6):062301, 2009.

\bibitem{kitaev1997quantum}
A~Yu Kitaev.
\newblock Quantum computations: algorithms and error correction.
\newblock {\em Russian Mathematical Surveys}, 52(6):1191--1249, 1997.

\bibitem{szegedy2004quantum}
Mario Szegedy.
\newblock Quantum speed-up of markov chain based algorithms.
\newblock In {\em Foundations of Computer Science, 2004. Proceedings. 45th
  Annual IEEE Symposium on}, pages 32--41. IEEE, 2004.

\bibitem{abramowitz1964handbook}
Milton Abramowitz and Irene~A Stegun.
\newblock {\em Handbook of mathematical functions: with formulas, graphs, and
  mathematical tables}, volume~55.
\newblock Courier Corporation, 1964.

\bibitem{kothari2014efficient}
Robin Kothari.
\newblock {\em Efficient algorithms in quantum query complexity}.
\newblock PhD thesis, University of Waterloo, 2014.

\bibitem{british1960bessel}
W.~G. Bickley, L.~J. Comrie, J.~C.~P. Miller, D.~H. Sadler, and A.~J. Thompson.
\newblock {\em Bessel Functions: Part II. Functions of Positive Integer Order}.
\newblock Cambridge University Press, 1960.

\bibitem{olver1964error}
F.~W.~J. Olver.
\newblock Error analysis of miller's recurrence algorithm.
\newblock {\em Mathematics of Computation}, 18(85):65--74, 1964.

\bibitem{huang2018explicit}
Cupjin Huang, Michael Newman, and Mario Szegedy.
\newblock Explicit lower bounds on strong quantum simulation.
\newblock {\em arXiv preprint arXiv:1804.10368}, 2018.

\bibitem{frieze2004fast}
Alan Frieze, Ravi Kannan, and Santosh Vempala.
\newblock Fast monte-carlo algorithms for finding low-rank approximations.
\newblock {\em Journal of the ACM (JACM)}, 51(6):1025--1041, 2004.

\bibitem{nakamoto2003norm}
Ritsuo Nakamoto.
\newblock A norm inequality for hermitian operators.
\newblock {\em The American mathematical monthly}, 110(3):238, 2003.

\bibitem{mathias1993approximation}
Roy Mathias.
\newblock Approximation of matrix-valued functions.
\newblock {\em SIAM journal on matrix analysis and applications},
  14(4):1061--1063, 1993.

\bibitem{tropp2012user}
Joel~A. Tropp.
\newblock User-friendly tail bounds for sums of random matrices.
\newblock {\em Found. Comput. Math.}, 12(4):389–434, August 2012.

\bibitem{aleksandrov2016operator}
Alexei~Borisovich Aleksandrov and Vladimir~Vsevolodovich Peller.
\newblock Operator lipschitz functions.
\newblock {\em Russian Mathematical Surveys}, 71(4):605, 2016.

\bibitem{hsu2014weighted}
Daniel Hsu.
\newblock Weighted sampling of outer products.
\newblock {\em arXiv preprint arXiv:1410.4429}, 2014.

\bibitem{yang2012nystrom}
Tianbao Yang, Yu-Feng Li, Mehrdad Mahdavi, Rong Jin, and Zhi-Hua Zhou.
\newblock Nystr{\"o}m method vs random fourier features: A theoretical and
  empirical comparison.
\newblock In {\em Advances in neural information processing systems}, pages
  476--484, 2012.

\bibitem{rahimi2009weighted}
Ali Rahimi and Benjamin Recht.
\newblock Weighted sums of random kitchen sinks: Replacing minimization with
  randomization in learning.
\newblock In {\em Advances in neural information processing systems}, pages
  1313--1320, 2009.

\bibitem{halko2011finding}
Nathan Halko, Per-Gunnar Martinsson, and Joel~A Tropp.
\newblock Finding structure with randomness: Probabilistic algorithms for
  constructing approximate matrix decompositions.
\newblock {\em SIAM review}, 53(2):217--288, 2011.

\bibitem{wang2013improving}
Shusen Wang and Zhihua Zhang.
\newblock Improving cur matrix decomposition and the nystr{\"o}m approximation
  via adaptive sampling.
\newblock {\em The Journal of Machine Learning Research}, 14(1):2729--2769,
  2013.

\bibitem{wang2016towards}
Shusen Wang, Zhihua Zhang, and Tong Zhang.
\newblock Towards more efficient spsd matrix approximation and cur matrix
  decomposition.
\newblock {\em The Journal of Machine Learning Research}, 17(1):7329--7377,
  2016.

\bibitem{zhao2015quantum}
Zhikuan Zhao, Jack~K Fitzsimons, and Joseph~F Fitzsimons.
\newblock Quantum assisted gaussian process regression.
\newblock {\em arXiv preprint arXiv:1512.03929}, 2015.

\bibitem{lloyd2014topological}
Seth Lloyd, Silvano Garnerone, and Paolo Zanardi.
\newblock Quantum algorithms for topological and geometric analysis of big
  data.
\newblock {\em arXiv preprint arXiv:1408.3106}, 2014.

\bibitem{arrazola2019quantum}
Juan~Miguel Arrazola, Alain Delgado, Bhaskar~Roy Bardhan, and Seth Lloyd.
\newblock Quantum-inspired algorithms in practice.
\newblock {\em arXiv preprint arXiv:1905.10415}, 2019.

\bibitem{dahiya2018empirical}
Yogesh Dahiya, Dimitris Konomis, and David~P Woodruff.
\newblock An empirical evaluation of sketching for numerical linear algebra.
\newblock In {\em Proceedings of the 24th ACM SIGKDD International Conference
  on Knowledge Discovery \& Data Mining}, pages 1292--1300, 2018.

\bibitem{biamonte2017quantum}
Jacob Biamonte, Peter Wittek, Nicola Pancotti, Patrick Rebentrost, Nathan
  Wiebe, and Seth Lloyd.
\newblock Quantum machine learning.
\newblock {\em Nature}, 549(7671):195, 2017.

\bibitem{servedio2004equivalences}
Rocco~A Servedio and Steven~J Gortler.
\newblock Equivalences and separations between quantum and classical
  learnability.
\newblock {\em SIAM Journal on Computing}, 33(5):1067--1092, 2004.

\bibitem{arunachalam2018optimal}
Srinivasan Arunachalam and Ronald De~Wolf.
\newblock Optimal quantum sample complexity of learning algorithms.
\newblock {\em The Journal of Machine Learning Research}, 19(1):2879--2878,
  2018.

\bibitem{lloyd2013quantum}
Seth Lloyd, Masoud Mohseni, and Patrick Rebentrost.
\newblock Quantum algorithms for supervised and unsupervised machine learning.
\newblock {\em arXiv preprint arXiv:1307.0411}, 2013.

\bibitem{wiebe2019language}
Nathan Wiebe, Alex Bocharov, Paul Smolensky, Krysta Svore, and Matthias Troyer.
\newblock Quantum language processing.
\newblock {\em arXiv preprint arXiv:1902.05162}, 2019.

\bibitem{mcclean2018barren}
Jarrod~R McClean, Sergio Boixo, Vadim~N Smelyanskiy, Ryan Babbush, and Hartmut
  Neven.
\newblock Barren plateaus in quantum neural network training landscapes.
\newblock {\em Nature communications}, 9(1):4812, 2018.

\bibitem{schuld2018circuit}
Maria Schuld, Alex Bocharov, Krysta Svore, and Nathan Wiebe.
\newblock Circuit-centric quantum classifiers.
\newblock {\em arXiv preprint arXiv:1804.00633}, 2018.

\bibitem{kieferova2016tomography}
M\'aria Kieferov\'a and Nathan Wiebe.
\newblock Tomography and generative training with quantum boltzmann machines.
\newblock {\em Phys. Rev. A}, 96:062327, Dec 2017.

\bibitem{schuld2019quantum}
Maria Schuld and Nathan Killoran.
\newblock Quantum machine learning in feature hilbert spaces.
\newblock {\em Physical review letters}, 122(4):040504, 2019.

\bibitem{romero2017quantum}
Jonathan Romero, Jonathan~P Olson, and Alan Aspuru-Guzik.
\newblock Quantum autoencoders for efficient compression of quantum data.
\newblock {\em Quantum Science and Technology}, 2(4):045001, 2017.

\bibitem{chefles1999strategies}
Anthony Chefles and Stephen~M Barnett.
\newblock Strategies and networks for state-dependent quantum cloning.
\newblock {\em Physical Review A}, 60(1):136, 1999.

\bibitem{kappen2018learning}
Hilbert~J Kappen.
\newblock Learning quantum models from quantum or classical data.
\newblock {\em arXiv preprint arXiv:1803.11278}, 2018.

\bibitem{amin2018quantum}
Mohammad~H Amin, Evgeny Andriyash, Jason Rolfe, Bohdan Kulchytskyy, and Roger
  Melko.
\newblock Quantum boltzmann machine.
\newblock {\em Physical Review X}, 8(2):021050, 2018.

\bibitem{crawford2016reinforcement}
Daniel Crawford, Anna Levit, Navid Ghadermarzy, Jaspreet~S Oberoi, and Pooya
  Ronagh.
\newblock Reinforcement learning using quantum boltzmann machines.
\newblock {\em arXiv preprint arXiv:1612.05695}, 2016.

\bibitem{benedetti2017quantum}
Marcello Benedetti, John Realpe-G{\'o}mez, Rupak Biswas, and Alejandro
  Perdomo-Ortiz.
\newblock Quantum-assisted learning of hardware-embedded probabilistic
  graphical models.
\newblock {\em Physical Review X}, 7(4):041052, 2017.

\bibitem{wiebe2015quantum}
Nathan Wiebe, Ashish Kapoor, Christopher Granade, and Krysta~M Svore.
\newblock Quantum inspired training for boltzmann machines.
\newblock {\em arXiv preprint arXiv:1507.02642}, 2015.

\bibitem{aarts1988simulated}
Emile Aarts and Jan Korst.
\newblock Simulated annealing and boltzmann machines.
\newblock 1988.

\bibitem{salakhutdinov2007restricted}
Ruslan Salakhutdinov, Andriy Mnih, and Geoffrey Hinton.
\newblock Restricted boltzmann machines for collaborative filtering.
\newblock In {\em Proceedings of the 24th international conference on Machine
  learning}, pages 791--798. ACM, 2007.

\bibitem{tieleman2008training}
Tijmen Tieleman.
\newblock Training restricted boltzmann machines using approximations to the
  likelihood gradient.
\newblock In {\em Proceedings of the 25th international conference on Machine
  learning}, pages 1064--1071. ACM, 2008.

\bibitem{le2008representational}
Nicolas Le~Roux and Yoshua Bengio.
\newblock Representational power of restricted boltzmann machines and deep
  belief networks.
\newblock {\em Neural computation}, 20(6):1631--1649, 2008.

\bibitem{salakhutdinov2009deep}
Ruslan Salakhutdinov and Geoffrey Hinton.
\newblock Deep boltzmann machines.
\newblock In {\em Artificial intelligence and statistics}, pages 448--455,
  2009.

\bibitem{salakhutdinov2010efficient}
Ruslan Salakhutdinov and Hugo Larochelle.
\newblock Efficient learning of deep boltzmann machines.
\newblock In {\em Proceedings of the thirteenth international conference on
  artificial intelligence and statistics}, pages 693--700, 2010.

\bibitem{lee2009convolutional}
Honglak Lee, Roger Grosse, Rajesh Ranganath, and Andrew~Y Ng.
\newblock Convolutional deep belief networks for scalable unsupervised learning
  of hierarchical representations.
\newblock In {\em Proceedings of the 26th annual international conference on
  machine learning}, pages 609--616. ACM, 2009.

\bibitem{hinton2012practical}
Geoffrey~E Hinton.
\newblock A practical guide to training restricted boltzmann machines.
\newblock In {\em Neural networks: Tricks of the trade}, pages 599--619.
  Springer, 2012.

\bibitem{lee2009unsupervised}
Honglak Lee, Peter Pham, Yan Largman, and Andrew~Y Ng.
\newblock Unsupervised feature learning for audio classification using
  convolutional deep belief networks.
\newblock In {\em Advances in neural information processing systems}, pages
  1096--1104, 2009.

\bibitem{mohamed2011acoustic}
Abdel-rahman Mohamed, George~E Dahl, and Geoffrey Hinton.
\newblock Acoustic modeling using deep belief networks.
\newblock {\em IEEE transactions on audio, speech, and language processing},
  20(1):14--22, 2011.

\bibitem{srivastava2012multimodal}
Nitish Srivastava and Ruslan~R Salakhutdinov.
\newblock Multimodal learning with deep boltzmann machines.
\newblock In {\em Advances in neural information processing systems}, pages
  2222--2230, 2012.

\bibitem{carleo2017solving}
Giuseppe Carleo and Matthias Troyer.
\newblock Solving the quantum many-body problem with artificial neural
  networks.
\newblock {\em Science}, 355(6325):602--606, 2017.

\bibitem{torlai2016learning}
Giacomo Torlai and Roger~G Melko.
\newblock Learning thermodynamics with boltzmann machines.
\newblock {\em Physical Review B}, 94(16):165134, 2016.

\bibitem{nomura2017restricted}
Yusuke Nomura, Andrew~S Darmawan, Youhei Yamaji, and Masatoshi Imada.
\newblock Restricted boltzmann machine learning for solving strongly correlated
  quantum systems.
\newblock {\em Physical Review B}, 96(20):205152, 2017.

\bibitem{aharonov2009polynomial}
Dorit Aharonov, Vaughan Jones, and Zeph Landau.
\newblock A polynomial quantum algorithm for approximating the jones
  polynomial.
\newblock {\em Algorithmica}, 55(3):395--421, 2009.

\bibitem{van2017quantum}
Joran van Apeldoorn, Andr{\'a}s Gily{\'e}n, Sander Gribling, and Ronald
  de~Wolf.
\newblock Quantum sdp-solvers: Better upper and lower bounds.
\newblock {\em arXiv preprint arXiv:1705.01843}, 2017.

\bibitem{aaronson2007learnability}
Scott Aaronson.
\newblock The learnability of quantum states.
\newblock In {\em Proceedings of the Royal Society of London A: Mathematical,
  Physical and Engineering Sciences}, volume 463, pages 3089--3114. The Royal
  Society, 2007.

\bibitem{liu2020rigorous}
Yunchao Liu, Srinivasan Arunachalam, and Kristan Temme.
\newblock A rigorous and robust quantum speed-up in supervised machine
  learning.
\newblock {\em arXiv preprint arXiv:2010.02174}, 2020.

\bibitem{huang2020power}
Hsin-Yuan Huang, Michael Broughton, Masoud Mohseni, Ryan Babbush, Sergio Boixo,
  Hartmut Neven, and Jarrod~R McClean.
\newblock Power of data in quantum machine learning.
\newblock {\em arXiv preprint arXiv:2011.01938}, 2020.

\bibitem{brassard2002quantum}
Gilles Brassard, Peter Hoyer, Michele Mosca, and Alain Tapp.
\newblock Quantum amplitude amplification and estimation.
\newblock {\em Contemporary Mathematics}, 305:53--74, 2002.

\bibitem{landau1980statistical}
LD~Landau and EM~Lifshitz.
\newblock Statistical physics, vol. 5.
\newblock {\em Course of theoretical physics}, 30, 1980.

\bibitem{poulin2009sampling}
David Poulin and Pawel Wocjan.
\newblock Sampling from the thermal quantum gibbs state and evaluating
  partition functions with a quantum computer.
\newblock {\em Physical review letters}, 103(22):220502, 2009.

\bibitem{chowdhury2016quantum}
Anirban~Narayan Chowdhury and Rolando~D Somma.
\newblock Quantum algorithms for gibbs sampling and hitting-time estimation.
\newblock {\em arXiv preprint arXiv:1603.02940}, 2016.

\bibitem{yung2012quantum}
Man-Hong Yung and Al{\'a}n Aspuru-Guzik.
\newblock A quantum--quantum metropolis algorithm.
\newblock {\em Proceedings of the National Academy of Sciences},
  109(3):754--759, 2012.

\end{thebibliography}

\end{document}